\documentclass{article}
\usepackage{arxiv}
\usepackage{algorithm}
\usepackage{doi}
\usepackage{amsthm}

\newcommand{\tbl}{\caption}

\newtheorem{proposition}{Proposition}

\usepackage[T1]{fontenc}    
\usepackage{hyperref}       
\usepackage{url}            
\usepackage{booktabs}       
\usepackage{amsfonts}       
\usepackage{nicefrac}       
\usepackage{microtype}      
\usepackage{graphicx}
\usepackage{natbib}
\usepackage{bm}
\usepackage{bbm}

\usepackage{algpseudocode}
\usepackage{subfigure}
\usepackage{tikz}
\usepackage{mathdots}
\usepackage{yhmath}
\usepackage{cancel}
\usepackage{color}
\usepackage{siunitx}
\usepackage{array}
\usepackage{arydshln}
\usepackage{multirow}
\usepackage{amssymb}
\usepackage{tabularx}
\usepackage{extarrows}
\usepackage{booktabs}
\usepackage{dsfont}

\usepackage{calc}

\usetikzlibrary{fadings}
\usetikzlibrary{patterns}
\usetikzlibrary{shadows.blur}
\usetikzlibrary{shapes}
\usetikzlibrary{arrows.meta}

\DeclareMathOperator*{\argmin}{arg\,min}

\DeclareMathOperator*{\nd}{d}

\newcommand{\cT}{\mathcal{T}}
\newcommand{\cC}{\mathcal{C}}
\newcommand{\cS}{\mathcal{S}}
\newcommand{\bR}{\mathbb{R}}
\newcommand{\bP}{\mathbb{P}}

\newcommand{\cM}{\mathcal{M}}

\newcommand{\cN}{\mathcal{N}}
\newcommand{\cX}{\mathcal{X}}
\newcommand{\cB}{\mathcal{B}}

\newcommand{\cY}{\mathcal{Y}}
\newcommand{\cZ}{\mathcal{Z}}
\newcommand{\cO}{\mathcal{O}}
\newcommand{\dsp}{\mathds{P}}

\newcommand{\nD}{\textnormal{D}}
\newcommand{\nJ}{\textnormal{J}}

\newcommand{\ID}{\mathbbm{1}}
\newcommand{\reach}{{\operatorname{reach}}}
\newcommand{\rank}{{\operatorname{rank}}}
\newcommand{\supp}{{\operatorname{supp}}}

\newcommand{\vol}{{\operatorname{vol}}}

\newcommand{\sm}{{\operatorname{sm}}}

\algrenewcommand\algorithmicrequire{\textbf{Input:}}
\algrenewcommand\algorithmicensure{\textbf{Output:}}

\title{Principal Decomposition with Nested Submanifolds}

\newcommand{\copyrightnotice}[1]{%
\begin{tikzpicture}[remember picture,overlay]
  \node[anchor=south, xshift=0pt, yshift=5pt] at (current page.south) {\fbox{\parbox{\dimexpr\textwidth-\fboxsep-\fboxrule\relax}{\footnotesize This article is \copyright{} \the\year{} by author(s) as listed above. The article is licensed under a Creative Commons Attribution (CC BY 4.0) International license (https://creativecommons.org/licenses/by/4.0/legalcode), except where otherwise indicated with respect to particular material included in the article. The article should be attributed to the author(s) identified above.}}};
\end{tikzpicture}%
}

\renewcommand{\ID}{\mathds{1}}

\renewcommand{\copyrightnotice}[1]{%
  \AddToHookNext{shipout/foreground}{%
    \begin{tikzpicture}[remember picture,overlay]
      \node[anchor=north, xshift=0pt, yshift=-10pt] at (current page.north) {%
        \fcolorbox{red}{white}{%
          \parbox{\dimexpr\textwidth-2\fboxsep-2\fboxrule\relax}{%
            \footnotesize\color{red}\raggedright #1%
          }%
        }%
      };
    \end{tikzpicture}%
  }%
}

\definecolor{backgroundcolor}{RGB}{227,237,205}
\author{ 
	\href{https://orcid.org/0000-0003-3675-1407}{\includegraphics[scale=0.06]{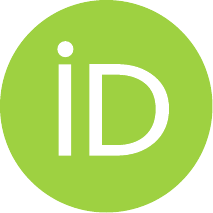}\hspace{1mm}Jiaji Su} \\
	College of Mathematics and Statistics \\
	Chongqing University\\
	Chongqing, 401331 China\\
	\texttt{su\_jiaji@cqu.edu.cn} \\
	\And
        Zhigang Yao\footnotemark \\
	Department of Statistics and Data Science \\
	National University of Singapore\\
	Singapore, 117546\\
	\texttt{zhigang.yao@nus.edu.sg} \\
}

\renewcommand{\shorttitle}{\textit{arXiv} Manuscript}
\newtheorem{theorem}{Theorem}
\newtheorem*{remark}{Remark}
\newtheorem{assumption}{Assumption}
\newtheorem{lemma}{Lemma}
\newtheorem{definition}{Definition}
\newtheorem{corollary}{Corollary}[lemma]

\newtheorem{condition}{Condition}

\begin{document}
\maketitle
\begin{abstract}
Over the past decades, the increasing dimensionality of data has increased the need for effective data decomposition methods. Existing approaches, however, often rely on linear models or lack sufficient interpretability or flexibility.
To address this issue, we introduce a novel nonlinear decomposition technique called the principal nested submanifolds, which builds on the foundational concepts of principal component analysis. This method exploits the local geometric information of data sets by projecting samples onto a series of nested principal submanifolds with progressively decreasing dimensions.
It effectively isolates complex information within the data in a backward stepwise manner by targeting variations associated with smaller eigenvalues in local covariance matrices. The fitted structures are smooth nested manifolds rather than linear or special-shape subspaces.
Simulation studies and an analysis of single-cell RNA-sequencing data illustrate how the method summarizes locally varying nonlinear structure while retaining a compatible hierarchy across dimensions. It provides flexible subspace constraints that facilitate the extraction of data components and noise reduction.
This approach advances the non-Euclidean statistical analysis of data with low-dimensional intrinsic structure within Euclidean spaces and offers a complementary perspective for high-dimensional noisy data sets in fields such as bioinformatics and machine learning.
\end{abstract}

\copyrightnotice{\textbf{Author-version and copyright notice.} This article has been accepted for publication in \textit{Biometrika}, published by Oxford University Press. This author-prepared version uses the authors' layout and incorporates material supplied separately as Supplementary Material for the journal publication; its pagination and presentation may therefore differ from the Version of Record. The Version of Record, once published, can be found through the journal's website and should be used for citation.\\
Copyright and reuse of the published version are governed by the applicable Oxford University Press licence.}

\section{Introduction}
\label{Sec:Intro}
In recent decades, the dimensionality of data has increased significantly, driven by the advent of fields such as machine learning and genomics and supported by enhanced computational capabilities. For instance, the scope of gene analysis spans from tens in mass cytometry to tens of thousands in single-cell RNA sequencing, and even millions in spatial transcriptomics. This escalation in data dimensionality has not only transformed everyday life but also posed fresh challenges for statistics and data science. A primary challenge is how to conduct effective dimensionality reduction whilst minimizing information loss. Historically, principal component analysis has been one of the most insightful and prevalently employed methods. The concept of principal component analysis dates back over a century; \citet{pearson1901liii} explores fitting lines or planes to data points and notes the difficulty of handling data exceeding four dimensions at the time, yet predicts the utility of this method. Following the development in matrix theory and decomposition methods, together with a series of significant statistical works, such as \citep{hotelling1933analysis, hotelling1936simplified, girshick1936principal, girshick1939sampling, anderson1963asymptotic}, principal component analysis has increasingly penetrated various fields. While principal component analysis provides low-dimensional linear representations with adjustable dimensionality that effectively reduces data dimensionality, its limitations in capturing nonlinear structures have become an increasing concern.

To address nonlinear structure in high-dimensional data, nonlinear dimensionality reduction techniques have been extensively explored. Key methodologies in this field include Isomap \citep{isomap}, locally linear embedding \citep{lle}, local tangent space alignment \citep{zhang2004LTSA}, t-distributed stochastic neighbour embedding \citep{tsne}, and uniform manifold approximation and projection \citep{mcinnes2018umap}. These methods are effective for constructing low-dimensional representations and facilitating visualization. However, they are typically designed for one target dimension at a time, and do not naturally provide a sequence of mutually compatible representations across dimensions.

Beyond these nonlinear embedding techniques, several researchers have endeavoured to extend the principles of principal component analysis into nonlinear realms. \citet{donnell1994analysis} introduce and elaborate on the concept of the smallest additive principal component, which facilitates learning a nonlinear subspace with codimension one from data points, akin to fitting a surface to the data. In the works of \citet{panaretos2014principal} and \citet{yao2024principalsubmanifolds}, the authors developed methodologies for principal flows and principal submanifolds. These are curves or surfaces where tangent velocity vectors align with the local leading eigenvector derived from tangent space principal component analysis, effectively traversing paths of maximum variability in the data. While these methods capture important nonlinear features, they mainly produce structures at selected dimensions and are less suited to a dimension-by-dimension decomposition in which adjacent dimensions are intrinsically connected.

Additionally, some researchers have explored extending principal component analysis to special non-Euclidean spaces. One notable method is the principal nested spheres \citep{PNS}, which is tailored for angle-related data positioned on high-dimensional spheres. The method employs an angle-based metric on the sphere to cut out the maximum sub-spheres with hyper-planes, leading to a series of spheres with decreasing intrinsic dimensions and a nested structure, with the residuals being `orthogonal' in the spherical sense. This approach is further extended to tori by \citet{tpca-aoas}, accommodating data in the form of a Cartesian product of angles. These methods are valuable for exploring the decomposition of high-dimensional non-Euclidean spaces and have demonstrated their effectiveness on biological data such as protein structures. Nevertheless, global linear or prescribed nested-subsphere families may be restrictive when the relevant covariance directions vary with position.

To address these challenges, this paper introduces a novel framework, referred to as principal nested submanifolds, designed to fit a series of smooth submanifolds with sequentially decreasing intrinsic dimensions and nested structures from high-dimensional observations. The main objective is a hierarchical geometric decomposition of the observed distribution, guided by local concentration and local covariance information, rather than the recovery of a unique latent manifold.
Principal nested submanifolds draw on manifold-fitting methods that use local covariance to reveal data geometry \citep{fefferman2018fitting, fefferman2021fitting, yao2019manifold, yao2023manifold, yao2023submanifold}. We seek a compatible dimension-by-dimension sequence rather than unrelated fits. Following backward principal component analysis \citep{huckemann2006principal,huckemann2010intrinsic}, its levels are ordered from higher to lower dimension, each simplifying the preceding one. Related barycentric subspace and backward nested-descriptor methods study nested geometric descriptors on nonlinear spaces \citep{10.1214/17-AOS1636,10.1214/17-AOS1609}, whereas our hierarchy is induced by the distribution through local covariance. Thus nestedness supports layered interpretation and sample-compatible adjacent-level residuals. Allowing the relevant directions to vary with position makes the construction more flexible than global linear or prescribed nested-subsphere summaries and may reduce approximation error for locally anisotropic nonlinear structure; this is not a claim of uniform superiority.
This paper introduces the principal nested submanifolds as a nested sequence of mutually compatible geometric representations across dimensions, together with an estimation procedure for constructing this hierarchy from data. We then study this construction under a fixed-radius framework, which provides a common scale for the successive levels and supports the nested hierarchy as a coherent object.
Through simulation studies and the analysis of real data, we further illustrate how the resulting hierarchy can be used for interpretation and exploratory data analysis. Practically, this allows broad organization and finer refinement to be examined within one compatible sequence of representations across dimensions. In this way, the proposed framework enriches our understanding of nonlinear space decomposition and provides a basis for layered geometric insight in complex data.

The remaining sections are organized as follows. Section 2 introduces the notation and mathematical preliminaries and provides a brief review of manifold fitting methods. Section 3 presents the framework of principal nested submanifolds, including its population definition, main theorem, and estimation procedure. Section 4 contains simulation studies that demonstrate the effectiveness of principal nested submanifolds in various settings. Section 5 illustrates the application of our method through a real data analysis case. Finally, Section 6 summarizes the key findings and conclusions of our study and discusses several directions for future research.

\section{Preliminary}
\subsection{Notations and mathematical concepts}
In this paper, the symbols $x$ or $y$ represent points associated with the observation, and $z$ indicates an arbitrary point of interest in the ambient space. The symbol $r$ is used to specify the radius in relevant contexts. 
Mathematical entities related to sets are denoted using capitalized calligraphy letters, such as $\mathcal{M}$ for the manifold and $\mathcal{B}_D(z, r)$ for a $D$-dimensional Euclidean ball centred at $z$ with radius $r$.
The distance between a point $a$ and a set $\mathcal{A}$ is given by 
\[
\nd(a,\mathcal{A}) = \min_{a^\prime \in \mathcal{A}} \|a-a^\prime\|,
\]
where $\|\cdot\|$ is the Euclidean norm.
For a matrix $A$, $\bP_{k}(A)$ denotes the projection of $A$ onto the linear subspace spanned by its eigenvectors corresponding to the largest $k$ eigenvalues. 
For a function $f:\bR^{d_1} \to \bR^{d_2}$, the derivative with respect to a subset of variables $(x_{s_1},\dots, x_{s_k})$ is denoted by
\[
    \nD_{x_{s_1},\dots, x_{s_k}} f = \left( \frac{\partial f_i}{\partial x_{s_j}} \right)_{d_2\times k}.
\]
The Jacobian matrix of $f$ is given by $\nJ_f = \bigl( \partial f_i/ \partial x_j \bigl)_{d_2\times d_1}$. For a unit vector $u\in\bR^{d_1}$, the directional derivative of $f$ along $u$ is defined as $\partial_u f = \nJ_f u$.

\begin{figure}[htbp]
    \centering
    \tikzset{every picture/.style={line width=0.75pt}}       
    
    \begin{tikzpicture}[x=0.75pt,y=0.75pt,yscale=-1,xscale=1]
    \fill[gray!15] (0,110) to[bend right=20] (195,30) 
                   to[bend right=20] (390,110)
                   to[bend left=15] (195,115)
                   to[bend left=15] (0,110);
                   
    \draw (0,110) to[bend right=20] (195,30) 
                   to[bend right=20] (390,110)
                   to[bend left=15] (195,115)
                   to[bend left=15] (0,110);
                   
    \draw [-Stealth] (195,30) -- (135,55) ;
    \draw [teal] (195,30) edge[bend left=15] (135,85) ;
    \draw [dotted] (100,0) -- (400,0) -- (300,60) -- (0,60) -- cycle ;
    
    \draw (342, 88) edge[bend right=40, dashed, -Stealth] (342, 17); 
    \draw (362, 17) edge[bend right=40, dashed, -Stealth] (362, 88) ;
    
    \fill [black] (195,30) circle (2pt) node[above] {$p$};
    \fill [black] (135,85) circle (2pt) node[left] {$q$};
    \draw (360,  2) node [anchor=north][inner sep=0.75pt,xslant =1, yscale=0.8] {$T_{p}\cM$};
    \draw (352,100) node [anchor=south] [inner sep=0.75pt] {$\cM$};
    \draw (332,25) node [anchor=east] [inner sep=0.75pt] {$\log_p$};
    \draw (372,85) node [anchor=west] [inner sep=0.75pt] {$\exp_p$};
    \draw (160,43) node [anchor=south] [inner sep=0.75pt] {$v$};
    \draw (165,55) node [anchor=west, teal] [inner sep=0.75pt] {$\gamma$};
    
    \draw [red,-Stealth] (195,30) -- (260,-20) ;
    \draw (230,15) node [anchor=south] [red, inner sep=0.75pt] {$u$};
    \draw [blue,dashed,-Stealth] (195,30) -- (260,30) ;
    \draw (230,52) node [anchor=south] [blue,inner sep=0.75pt] {$\Pi_p^- u$};
    \draw [purple,dashed,-Stealth] (260,30) -- (260,-20) ;
    \draw (275,20) node [anchor=south] [purple,inner sep=0.75pt] {$\Pi_p^\perp u$};
    \end{tikzpicture}
    \caption{Illustration for the involved geometrical concepts.}
\end{figure}

We also require some essential geometrical concepts relevant to the study of principal submanifolds. A $d$-dimensional \textit{topological manifold} is defined as a second-countable, Hausdorff topological space that is locally Euclidean of dimension $d$. Specifically, this implies that every point has a neighbourhood homeomorphic to an open subset of $\mathbb{R}^d$. The \textit{tangent space} at a point $p \in \mathcal{M}$, denoted $T_p\mathcal{M}$, is a $d$-dimensional affine space consisting of all vectors tangent to $\mathcal{M}$ at $p$. A \textit{Riemannian metric} $g$ on $\mathcal{M}$ is a smoothly varying family of inner products on the tangent spaces, where each inner product $g_x: T_x\mathcal{M} \times T_x\mathcal{M} \to \mathbb{R}$ is defined at $x \in \mathcal{M}$. Formally, a \textit{Riemannian manifold} is a pair $(\mathcal{M}, g)$, where $\mathcal{M}$ is a smooth manifold equipped with the metric $g$. In our setting, $\mathcal{M}$ is assumed to be a subset of $\mathbb{R}^D$ with $g$ induced by the Euclidean metric of $\mathbb{R}^D$, thereby simplifying $\mathcal{M}$ to a $d$-dimensional Riemannian manifold.

For each pair of points $p, q \in \cM$, the \textit{Riemannian distance} from $p$ to $q$, denoted $g(p, q)$ is defined as the infimum of the lengths of all admissible curves connecting $p$ to $q$.
When $\cM$ is connected, an admissible curve $\gamma$ is called a \emph{minimizing curve} if and only if its length equals the Riemannian distance between its endpoints. A unit-speed minimizing curve is referred to as a \emph{geodesic}. Thus, we use \emph{geodesic distance} and Riemannian distance interchangeably. For each $p \in \cM$, the \emph{exponential map} at ${p}$, denoted by $\exp_p$, is defined by
$
\exp_p (v)=\gamma_v(1),
$
where $v\in T_p\cM$ and $\gamma_v$ are the unique geodesic with initial location $\gamma_v(0) = p$ and $\dot{\gamma}_v(0) = v$. The exponential map is a diffeomorphism in a neighbourhood of the tangent space. The \emph{logarithm map} $\log_p$ is defined as the inverse of $\exp_p$. For a curve $\gamma$, we use $\dot{\gamma}$ and $\ddot{\gamma}$ to represent $\nd\!\gamma(t)/\nd\!t$ and $\nd^2\!\gamma(t)/\nd\!t^2$ respectively.

The projection matrices $\Pi_x^-$ and $\Pi_x^\perp$ project any vector $u \in \mathbb{R}^D$ onto the tangent space $T_x\mathcal{M}$ and its orthogonal complement, respectively. These matrices satisfy the relation $\Pi_x^\perp = I_D - \Pi_x^-$, where $I_D$ is the identity matrix in $\mathbb{R}^D$, and we denote estimators for these projections as $\widehat{\Pi}_z^\perp$ and $\widehat{\Pi}_z^-$. For an arbitrary point $z \not\in \mathcal{M}$, its projection onto the manifold is defined by $z^* = \arg\min_{x \in \mathcal{M}} \|x - z\|$.

The curvature of $\cM$ is described by the \emph{reach} of $\cM$. The concept of reach, as introduced by \citet{federer1959curvature}, is pivotal in assessing the regularity of manifolds embedded in Euclidean space and finds extensive applications in signal processing and machine learning. It can be defined as follows:
\begin{definition}
    [Reach] Let $\mathcal A$ be a closed subset of $\bR^D$. The reach of $\mathcal A$, denoted by $\reach(\mathcal A)$, is the largest number $\tau$ such that for any $z$ with $\nd(z,\mathcal A)<\tau$, its projection on $\mathcal{A}$ is uniquely defined.
\end{definition}
\begin{remark}
    The value of $\reach(\cM)$ can be interpreted as a second-order differential quantity if $\cM$ is treated as a function. Namely, let $\gamma$ be an arc-length parametrized geodesic of $\cM$; then, according to \citet{niyogi2008finding}, $\|\ddot{\gamma}(t)\|\leq \reach(\cM)^{-1}$ for all $t$.
\end{remark}
For example, the reach of a circle is its radius, and the reach of a linear subspace is infinite. Intuitively, a large reach implies that the manifold is locally close to the tangent space, which leads to another definition of reach given by \citet{federer1959curvature}:
\begin{lemma}[Federer's reach condition]
    \label{Lemma:ReachCond}
     Let $\cM$ be an embedded submanifold of $\bR^{D}$. Then,
    $$\reach(\cM)^{-1} = \sup \left\{\frac{2\textnormal{d} \mathit{(b,T_a\cM)}}{\|\mathit{a-b}\|^2} \mid \mathit{a,b}\in\cM,~\mathit{a\neq b}\right\}.$$
\end{lemma}

Calculating the average of a group of points on a manifold is not straightforward due to the nonlinearity of the space. To address this, the concept of the \emph{Fr\'echet mean} is widely used. The Fr\'echet mean generalizes the notion of centroids to metric spaces and provides a representative point that captures the central tendency of a set of points.
\begin{definition}
    [Fr\'echet Mean] Let $\{x_1, \dots, x_n\}$ be a collection of points on a manifold $\cM$. For any point $z$ on $\cM$, define the Fr\'echet function to be the sum of squared distances from $z$ to each $x_i$:
    \[
        F(z) = \sum_{i=1}^n g^2(z, x_i).
    \]
    The set of points $\{\mu_F \in \cM : F(\mu_F) = \arg\min_{z \in \cM} F(z)\}$, where $F(z)$ is minimized, is called the Fr\'echet mean set. If the minimizer is unique, $\mu_F$ is simply called the Fr\'echet mean of the set $\{x_1, \dots, x_n\}$ on $\cM$.
\end{definition}

\subsection{Fitting manifolds from data}
In the manifold fitting literature, a common goal is to model low-dimensional latent data-generating manifolds in a high-dimensional ambient space from noisy observations. A standard formulation centres on a random vector $Y = X + \xi$, where $X$ is drawn from a distribution supported on a smooth $m$-manifold $\cM\subset\mathcal{A}$, and $\xi$ represents the observation noise in the ambient space $\mathcal{A}$. Manifold fitting methods then utilize the local covariance structure to produce smooth manifold estimators from $\cY_n = \{y_i\}_{i=1}^n$, a set of independent and identically distributed realizations of $Y$. In contrast, the present paper uses these ideas as background and defines its target directly through the observed distribution and its local covariance structure.

\citet{yao2019manifold} explore the scenario where $\mathcal{A} = \bR^D$ and $\xi \sim \cN(0,\sigma I_D)$. For each $y_i$, let $y_i^*$ be its projection on $\cM$. The projection matrix onto $N_{y_i^*}\cM$, denoted as $\widehat{\Pi}_{i}^\perp$, is estimated with the smallest $D-m$ eigenvectors from local principal component analysis centred at $y_i$ with a radius $r$.
Then, for an arbitrary point $z$ around $\cM$, its bias vector is defined as  
\begin{align*}
\widehat{b}(z) = \widehat{\Pi}_z^\perp(z - \widehat{\mu}_{z}),    
\end{align*}
with $\widehat{\mu}_{z} = \sum_{i=1}^n\alpha_i(z)y_i$ and $\widehat{\Pi}_z^\perp = \bP_{D-m}\left(\sum_{i=1}^n\alpha_i(z)\widehat\Pi_{i}^\perp\right)$, where the weights are defined as 
\begin{equation*}
    \widetilde{\alpha}_i(z)=\left(1 - \frac{\|z - y_i\|^2}{r^2}\right)^{\beta} \ID(\|z-y_i\|\leq r), \quad \alpha_i(z) = \frac{\widetilde{\alpha}_i(z)}{\sum_{i=1}^n\widetilde{\alpha}_i(z)},
\end{equation*}
where $\beta\geq 2$ is a fixed smoothness parameter and $\ID$ denotes the indicator function. The manifold estimator is then defined as
\begin{equation*}
\widehat{\cM} = \{z\in\bR^D:~\nd(z,\cM)\leq cr,~ \widehat{b}(z) = 0\},
\end{equation*}
which is proved to be $\cO(\sigma)$-close to $\cM$ in Hausdorff distance and maintains a lower bounded reach with high probability.

\cite{yao2024principalsubmanifolds} and \citet{yao2023submanifold} investigate the problem of learning low-dimensional structures for random objects on a known manifold. The former approach is based on collections of principal flows, while the latter introduces a tubular noise model on manifolds and develops a corresponding manifold fitting method. In their framework, the sample set $\cY_n$ is posited to lie precisely on a known $d$-dimensional manifold $\mathcal{A}$, which is embedded in $\bR^D$. \cite{yao2023submanifold} concentrate on cases where the samples cluster around a certain submanifold $\cM \subset \mathcal{A}$ with codimension $1$. Their proposed iterative algorithm projects $\cY_n$ onto the latent submanifold. For each point $y_i$, they initialize $z_0$ as $y_i$, and iteratively update it as
\begin{equation*}
z_{k+1} = \argmin_{z\in\mathcal{A}}  \|V_{z}^\top (z - z_k)\|^2 +\|z - \widehat{\mu}_z\|^2,
\end{equation*}
where  $V_{z}$ is made up of the eigenvectors with respect to the top $d-1$ eigenvalues of the local sample covariance matrix at $z$. Under certain assumptions, the sequence $\{z_k\}$ will converge to the latent submanifold $\cM$.

In conclusion, manifold fitting methods provide a useful background for describing low-dimensional geometric structure in data by employing local sample means as reference points and managing the dimensions of fitting through a projection matrix that encapsulates the local covariance structure, which ultimately yields manifold estimators described as level sets. In the present paper, this perspective is adapted to construct a nested hierarchy from the observed distribution rather than to recover a uniquely specified latent manifold.

\section{Methodology}
\label{Sec:Method}
Manifold fitting techniques emphasize the importance of local covariance structure in the analysis of data adhering to low-dimensional manifolds. Building on this concept and the principles of local principal component analysis, we introduce a framework of fitting submanifolds with a nested structure, referred to as the principal nested submanifolds. In this section, we establish the essential notations, define principal nested submanifolds, and outline their estimators along with a corresponding algorithm for practical implementation.

\subsection{The population principal submanifolds}
Given the model complexity and its relevance to practical applications, we proceed by considering a random vector $X\in\bR^D$, with its distribution denoted by $F_X$. To investigate the local nonlinear structure of $F_X$, a radius parameter $r>0$ is necessary. Then, for any point $z\in\bR^D$ satisfying $\nd\bigl(z,\supp(F_X)\bigl)<r$, we define the local average and the local covariance matrix at $z$ as follows:
\begin{gather*}
    \mu_{r,z} = E\bigl\{X \mid X \in \cB(z,r)\bigl\},\\
    \Sigma_{r,z} = E\bigl\{ (X-\mu_{r,z})(X-\mu_{r,z})^\top \mid X \in \cB(z,r) \bigl\}.
\end{gather*}
These measures aim to capture the essential behaviour of the distribution in the neighbourhood of $z$.

The radius $r$ is a target-defining resolution parameter: different values generally define different scale-specific hierarchies, whilst one value is kept across dimensions to give a shared geometric scale. Candidate radii are assessed for adequate neighbourhood support, eigengaps separating the relevant eigenspaces from their complements, nearby-radius eigenspace stability and retention of position-dependent structure. At fixed $r$, a nonzero population eigengap is necessary for directional identifiability; finite-sample reliability also requires local covariance error to be small relative to the gap. Thus, small radii may be under-supported and empirically unstable, whereas large radii average changing local covariance structure into an increasingly global covariance-based summary. Under an auxiliary manifold-plus-ambient-variation interpretation, ambient variation and sampling sparsity give heuristic lower-scale guidance, while the geometric scale over which the relevant eigenspaces vary gives upper-scale guidance; curvature and reach provide model-dependent interpretations of that scale. These are neither assumptions nor optimality claims; controlled and data-specific diagnostics are reported in Sections~\ref{sec:spiral-radius} and \ref{sec:radius-diagnostics}.

Through eigendecomposition of $\Sigma_{r,z}$, we identify the principal directions of $F_X$ at $z$ as the unit eigenvectors $v_{r,z,1},\dots,v_{r,z,D}$, arranged such that their corresponding eigenvalues are in non-decreasing order. For $k=1,\dots,D$, the $k$th principal direction gives rise to a projection matrix $\Pi_{r,z,k}$ and yields the $k$th bias vector as
\[
    b_{r,k}(z) = \Pi_{r,z,k} (z - \mu_{r,z}) = v_{r,z,k} v_{r,z,k}^\top (z - \mu_{r,z}).
\]
These bias vectors lead to a root set given by
\begin{equation}\label{eq:def-pop-d-fold}
    \cM_{r,d} = \left\{z\in\bR^D : \nd\bigl(z,\supp(F_X)\bigl)<r,~\sum_{k=1}^{D-d}b_{r,k}(z) = 0\right\}.
\end{equation}
\begin{proposition}
    If $r$ is large enough that the local mean and covariance equal their global counterparts over the relevant domain, then $\cM_{r,d}$ coincides there with the affine principal subspace determined by the leading $d$ principal components.
\end{proposition}

\begin{figure}[htbp]
    \centering
    \tikzset{every picture/.style={line width=0.75pt}}       
    
    \begin{tikzpicture}[x=0.75pt,y=0.75pt,yscale=-1,xscale=1]
    \fill[gray!15] (40,110) to[bend right=28] (195,30) 
                   to[bend right=28] (350,110)
                   to[bend left=15] (195,115)
                   to[bend left=15] (40,110);
                   
    \draw [teal!50] (275,45) .. controls (170,50) and (200,80) .. (120,100) ;
    \draw (240,60) node [anchor=west, teal] [inner sep=0.75pt] {$\cM_{r,D-2}$};
                
    \draw [gray] (40,110) to[bend right=28] (195,30) 
                   to[bend right=28] (350,110)
                   to[bend left=15] (195,115)
                   to[bend left=15] (40,110);

    \draw (85,100) node [anchor=south] [inner sep=0.75pt,gray] {$\cM_{r,D-1}$};
    
    \draw [blue,-Stealth] (165, 20) -- (165, 60);
    \draw [blue,-Stealth] (165, 20) -- (195, 25);
    \draw [blue,densely dotted] (165, 60) -- (195, 65) -- (195, 25);

    \fill[blue, opacity=0.1] (165, 20) -- (165, 60) -- (195, 65) -- (195, 25) -- cycle;

    \draw [blue] (165, 48) node[left] {$b_{r,1}(z)$};
    \draw [blue] (220, 12) node[left] {$b_{r,2}(z)$};

    \fill [red] (185, 67) circle (2pt) node[below=10pt,right=-5pt] {$\mu_{r,z}$};
    \draw [red,-Stealth] (165, 20) -- (185, 67);
    \fill [black] (165, 20) circle (2pt);
    \draw [black] (165, 13) node[left] {$z$};

    \draw[-Stealth] (300,45)  to[bend right=40] (360,45);
    

    \begin{scope}[shift={(370,10)}]
        \draw [teal] (22.53,15.91) .. controls (42.53,1.57) and (135.2,-9.09) .. (118.87,20.57) .. controls (102.53,50.24) and (117.2,57.91) .. (139.87,96.24) .. controls (162.53,134.57) and (32.53,119.24) .. (12.53,89.24) .. controls (-7.47,59.24) and (2.53,30.24) .. (22.53,15.91) -- cycle;
        \fill [teal!15] (22.53,15.91) .. controls (42.53,1.57) and (135.2,-9.09) .. (118.87,20.57) .. controls (102.53,50.24) and (117.2,57.91) .. (139.87,96.24) .. controls (162.53,134.57) and (32.53,119.24) .. (12.53,89.24) .. controls (-7.47,59.24) and (2.53,30.24) .. (22.53,15.91) -- cycle;
        \draw (60,15) node [anchor=west, teal] [inner sep=0.75pt] {$\cM_{r,D-2}$};

        \draw [purple] (4.2,41.24) .. controls (56.87,31.91) and (94.2,117.24) .. (140.87,100.24);
        \draw (80,110) node [anchor=west, purple] [inner sep=0.75pt] {$\cM_{r,D-3}$};
        
        \draw[-Stealth] (130,35)  to[bend right=40] (190,35) node[below]{$\dots$};

        \draw [blue,-Stealth] (140, 60) -- (110, 60);
        \draw [blue,-Stealth] (140, 60) -- (120, 85);
        \draw [blue,densely dotted] (110, 60) -- (90, 85) -- (120, 85);
    
        \fill[blue, opacity=0.1] (140, 60) -- (110, 60) -- (90, 85) -- (120, 85) -- cycle;

        \draw [blue] (215, 80) node[left] {$\sum_{k=1}^2 b_{r,k}(z)$};
        \draw [blue] (137, 50) node[left] {$b_{r,3}(z)$};

        \fill [red] (90, 80) circle (2pt) node[left=3pt] {$\mu_{r,z}$};
        \draw [red,-Stealth] (140, 60) -- (90, 80);
        \fill [black] (140, 60) circle (2pt);
        \draw [black] (140, 60) node[right] {$z$};
    \end{scope}

    \end{tikzpicture}
    \caption{Conceptual illustration of the principal nested submanifolds. Each $\cM_{r,d}$ is determined as the root set of $\sum_{k=1}^{D-d} b_{r,k}(z)$ and is nested within $\cM_{r,d+1}$.}
\end{figure}

Similar to other decomposition methods, the definition of a well-defined principal submanifold requires a significant eigengap. Meanwhile, to ensure that such root sets are smooth manifolds, $F_X$ should satisfy additional smoothness conditions.

The resulting hierarchy is a distribution-dependent nonlinear extension of principal-component-type summaries rather than a universally canonical nested flag. It is most informative when low-rank or locally anisotropic structure is identifiable; under complete symmetry, ordered directions need not be identifiable, as in principal component analysis under spherical covariance.
\begin{condition}\label{asm:1}
    The local mean function is smooth enough with respect to $z$, that is,
        \[
            \|\partial_u\mu_{r,z}\| < \ell_1,
            \text{ and }
            \|\partial_u\partial_{u^\prime}\mu_{r,z}\| < \ell_2,
        \]
    for any unit vectors $u,u^\prime \in\bR^D$ and some constant $\ell_1$ and $\ell_2$.
\end{condition}
\begin{condition}\label{asm:2}
    The $k$th projection matrix is smooth enough with respect to $z$, that is, 
        \[
            \|\partial_u\Pi_{r,z,k}\| < \ell_3, 
            \text{ and }
            \|\partial_u\partial_{u^\prime }\Pi_{r,z,k}\| < \ell_4,
        \]
        for any unit vectors $u,u^\prime \in\bR^D$ and some constant $\ell_3$ and $\ell_4$.
\end{condition}
\begin{condition}\label{asm:3}
    The distribution of $X$ is concentrated towards $\mu_{r,z}$ along the $k$th principal direction $v_{r,z,k}$, that is,
        \[\|v_{r,z,k}^\top \nJ_\mu(z)\| < \ell_5,\]
        for some constant $\ell_5$.
\end{condition}

With these conditions and additional mild assumptions, the level set can be shown to be a submanifold, which leads to the following theorem:
\begin{theorem}\label{thm:population}
    Let $\cM_{r,d}$ be the level set defined by \eqref{eq:def-pop-d-fold}. Assume there exists a radius $r>0$ such that $\cM_{r,d} \neq \varnothing$, and a constant $\eta$ with
    \begin{equation}\label{eq:def-eta}
        \frac{1 - (D-d) (r \ell_3 + \ell_5)}{(D-d) \bigl\{r\ell_4 + \ell_2 +2 \ell_3(1+\ell_1)\bigl\}} > \eta > 0,
    \end{equation}
    such that Conditions 1--3 hold for any point $z$ where $\nd(z,\cM_{r,d}) < \eta$ and $k \in \{ 1,\dots,D-d\}$. Then, $\cM_{r,d}$ is locally homeomorphic to $\bR^d$. In other words, $\cM_{r,d}$ is a $d$-dimensional topological submanifold of $\bR^D$, and we call it the $d$-dimensional population principal submanifold of $X$ with respect to radius $r$.

    Furthermore, if $\cM_{r,d}$ is connected, then it is a $\mathcal{C}^1$ manifold and the reach of $\cM_{r,d}$ satisfies $\reach(\cM_{r,d}) \geq \eta/2$.
\end{theorem}
\begin{corollary}
    If there exists a radius $r$ and a sequence of dimensions $\{d_j\}_{j=1}^{k} \subseteq \{1,\dots,D\}$ such that for each $d\in\{d_j\}_{j=1}^{k}$, $\cM_{r,d}$ is a $d$-dimensional population principal submanifold given by Theorem \ref{thm:population}, the principal submanifolds are nested. That is,
    \[
        \cM_{r,d} \subset \cM_{r,d^\prime}, \quad \text{for } d < d^\prime \in \{d_j\}_{j=1}^{k}.
    \]
    The collection of principal submanifolds, $\{\cM_{r,d_j}\}_{j=1}^{k}$, is called the principal nested submanifolds.
\end{corollary}

\subsection{Estimating the sequence of principal nested submanifolds}
Consider a data set $\cX_n = \{x_i\}_{i=1}^n$ comprising points sampled from the distribution $F_X$. For each sample point $x_i$, we define its local covariance matrix within radius $r$ as
\begin{equation*}
\widehat\Sigma_{r,i} = \frac{\sum_{j=1}^n (x_j-x_i)(x_j-x_i)^\top \ID(\|x_j-x_i\|\leq r)}{\sum_{j=1}^n \ID(\|x_j-x_i\|\leq r)},
\end{equation*}
where $\ID$ is the indicator function. The unit eigenvectors of $\widehat\Sigma_{r,i}$, denoted as $\widehat{v}_{r,i,1},\dots,\widehat{v}_{r,i,D}$, are arranged such that their corresponding eigenvalues are in ascending order. The projection matrix for each eigenvector is computed as
\begin{equation*}
\widehat\Pi_{r,i,k} = \widehat{v}_{r,i,k}\widehat{v}_{r,i,k}^\top, \quad \text{for }k = 1,\dots,D.
\end{equation*}

For any point of interest $z \in \bigcup_{i=1}^n \cB(x_i,cr)$ with a constant $c<1$, we estimate its projection matrix and bias vectors relative to the principal nested submanifolds using the manifold fitting approach. Specifically, we define the weight functions by
\begin{equation}\label{eq:alpha}
    \widetilde{\alpha}_i(z)=\left(1 - \frac{\|z - x_i\|^2}{r^2}\right)^{3} \ID({\|z - x_i\|\leq r}), \quad \alpha_i(z) = \frac{\widetilde{\alpha}_i(z)}{\sum_{i=1}^n\widetilde{\alpha}_i(z)}.
\end{equation}
The cubic power is a convenient practical smooth-weighting choice. The resulting compactly supported taper decreases smoothly within the neighbourhood and vanishes at its boundary, reducing sensitivity to observations near the cutoff.
The projection matrix onto the $k$th principal direction is then given by
\begin{equation*}
    \widehat\Pi_{r,z,k} = \bP_1 \left(\sum_{i=1}^n\alpha_i(z)\widehat\Pi_{r,i,k}\right).
\end{equation*}
Setting the reference point $\widehat{\mu}_{r,z} = \sum_{i=1}^n\alpha_i(z)x_i$, the $k$th bias vector is
\begin{equation*}
\widehat{b}_{r,k}(z) = \widehat\Pi_{r,z,k} \left(z - \widehat{\mu}_{r,z}\right).
\end{equation*}
Finally, the sequence of principal nested submanifolds can be established as
\begin{equation}\label{def:M_hat}
\widehat{\cM}_{r,d} = \left\{z \in \bigcup_{i=1}^n \cB(x_i,cr):~ \sum_{k=1}^{D-d} \widehat{b}_{r,k}(z) = 0\right\}, \quad \text{for } d = 1,\dots,D-1.
\end{equation}

The population principal submanifold $\cM_{r,d}$ defined in \eqref{eq:def-pop-d-fold} is constructed through the hard neighbourhood $\cB(z,r)$. In contrast, for numerical stability, the empirical construction uses a smooth localized weighting scheme. To relate these two objects, Appendix~\ref{sec:appendix-consistency} introduces auxiliary population fields based on the same kernel and projector-smoothing operations as the empirical procedure. This yields a bias--variance decomposition: the discrepancy between the empirical and population bias fields consists of a stochastic term and deterministic approximation terms accounting for local-covariance centring, hard-versus-smooth localization and projector smoothing.

\begin{theorem}[Asymptotic properties of empirical principal submanifolds]
\label{thm:consistency-main}
Let
\[
g_r(z)=\sum_{k=1}^{D-d} b_{r,k}(z),
\qquad
\widehat g_r(z)=\sum_{k=1}^{D-d}\widehat b_{r,k}(z),
\]
and let $\widehat{\cM}_{r,d}$ denote the restriction to $\cZ$ of the set defined by \eqref{def:M_hat}. Suppose Assumptions~\ref{asm:consistency-mass}--\ref{asm:consistency-transverse} in Appendix~\ref{sec:appendix-consistency} hold. Then, for a fixed target dimension $d$ and radius $r>0$,
\[
\sup_{z\in\cZ}\|\widehat g_r(z)-g_r(z)\|
\le
\cO_p\!\left(\sqrt{\frac{\log n}{nr^D}}\right)+B_r,
\]
where $B_r$ is the total deterministic approximation term defined in Appendix~\ref{sec:appendix-consistency}. Moreover, there exists a constant $C>0$ such that
\[
\textnormal{d}_\textnormal{H}(\widehat{\cM}_{r,d},\cM_{r,d})
\le
C\left\{
\cO_p\!\left(\sqrt{\frac{\log n}{nr^D}}\right)+B_r
\right\}.
\]
A more precise formulation, together with the auxiliary smooth population object and the proofs, is given in Appendix~\ref{sec:appendix-consistency}.
\end{theorem}

In practice, working with real data sets often involves data that resides on a known manifold. To make better use of the geometric information contained in the data, it is desirable to embed the data in a higher dimensional space and then fit the principal nested submanifolds. Here the embedding set serves as the ambient Euclidean representation in which the local covariance construction and recursive fitting are carried out. For some topological structures, such an embedding is needed before these Euclidean-space operations can be applied, but it is used only as a modelling and computational scaffold for the geometric hierarchy under study. In certain cases, the maximum dimensionality of the manifolds to be fitted is predetermined, thereby allowing certain dimensions to be omitted during the fitting process. For example, consider data representing two-dimensional angles positioned on a torus; an isometric embedding in four-dimensional Euclidean space is required. In this scenario, only fitting the one-dimensional principal submanifold makes sense. The algorithms that detail this process can be found in the Appendix.

\section{Simulation studies}
\subsection{Simulation in Euclidean space}
In this subsection, we explore the decomposition of sample points around a low-dimensional structure in Euclidean space, comparing our proposed method with principal component analysis in three different cases. For each scenario, we consider a generating curve $\{\gamma(t)\mid t\in\cT\}\subset\bR^3$ and uniformly generate $\{t_i\}_{i=1}^n\subset\cT$. Next, at each $\gamma(t_i)$, we introduce anisotropic noise along two perpendicular, position-dependent unit directions, $v_1(t_i)$ and $v_2(t_i)$, with amplitudes $\xi_1\sim \cN(0,\sigma_1^2)$ and $\xi_2 \sim \cN(0,\sigma_2^2)$ respectively. Specifically, an observed point is given by
\[x_i = \gamma(t_i) + \xi_{1,i} v_1(t_i) + \xi_{2,i} v_2(t_i),\quad i=1,\dots,n.\]
This process results in a point cloud $\cX_n \subset\bR^3$, with a local dimensionality of three. We apply our proposed method, using a radius $r = 0.5$, to $\cX_n$ to project it onto the principal nested submanifolds $\widehat{\cM}_2$ and $\widehat{\cM}_1$. We also perform principal component analysis for comparison.

The simplest generating curve considered is a line segment in $\bR^3$. Specifically, we define
\[
\gamma_1(t) = \bigl(t,0,0\bigl)^\top,~t\in(0,1),
\]
where $n=10^4$ random points are generated, and noise is introduced in the directions $v_1 = (0,1,0)^\top$ and $v_2 = (0,0,1)^\top$, with standard deviations $\sigma_1=0.1$ and $\sigma_2=0.05$, respectively. The resulting data set, shown in Figure \ref{fig:Simu-Eucli-line-Input}, resembles an elliptic cylinder with three identifiable fixed principal directions and can be efficiently decomposed in to linear subspaces.
\begin{figure}[htbp]
    \centering
    \subfigure[]{\includegraphics[width=0.3\textwidth]{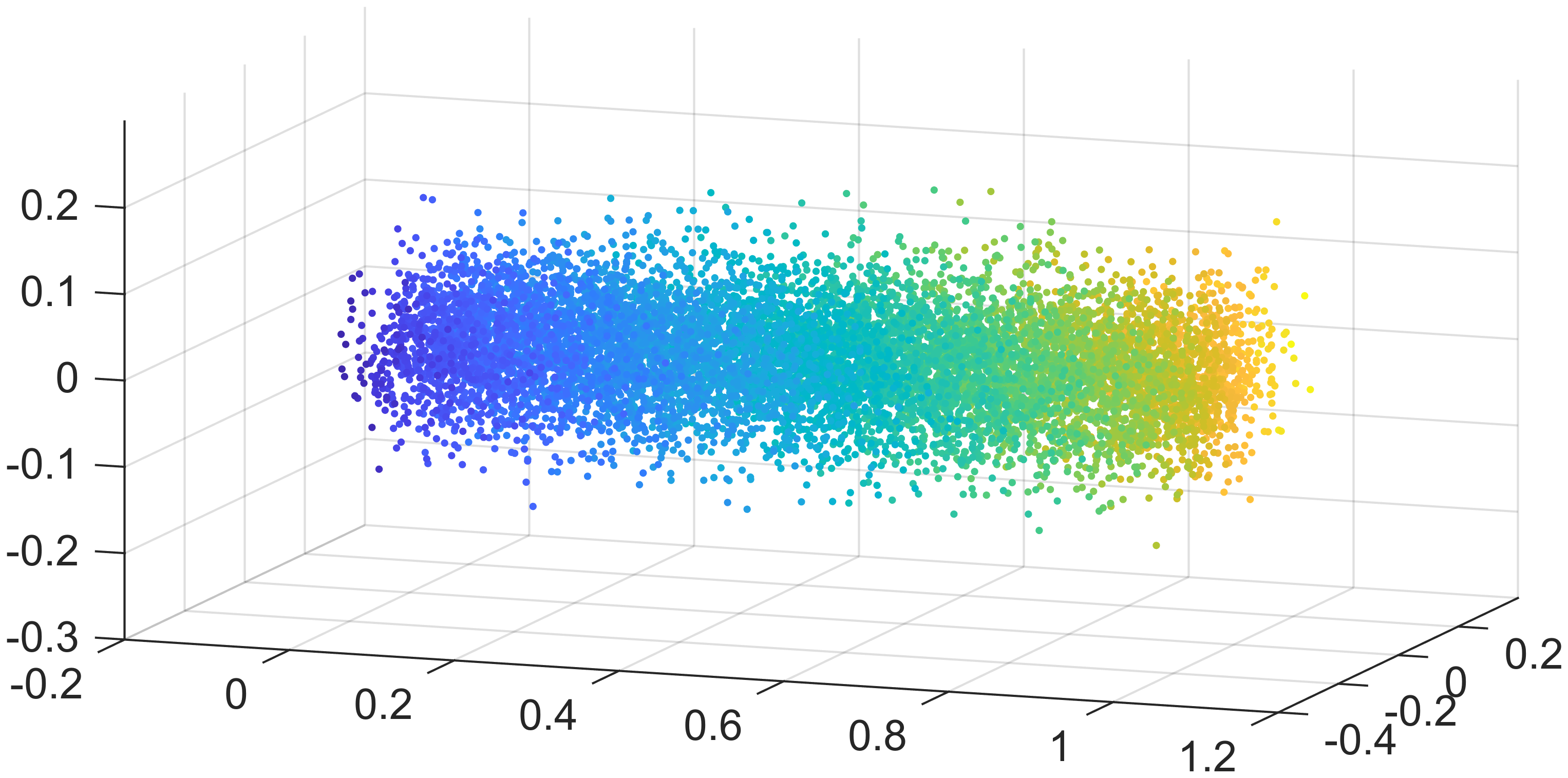}\label{fig:Simu-Eucli-line-Input}}
    \subfigure[]{\includegraphics[width=0.3\textwidth]{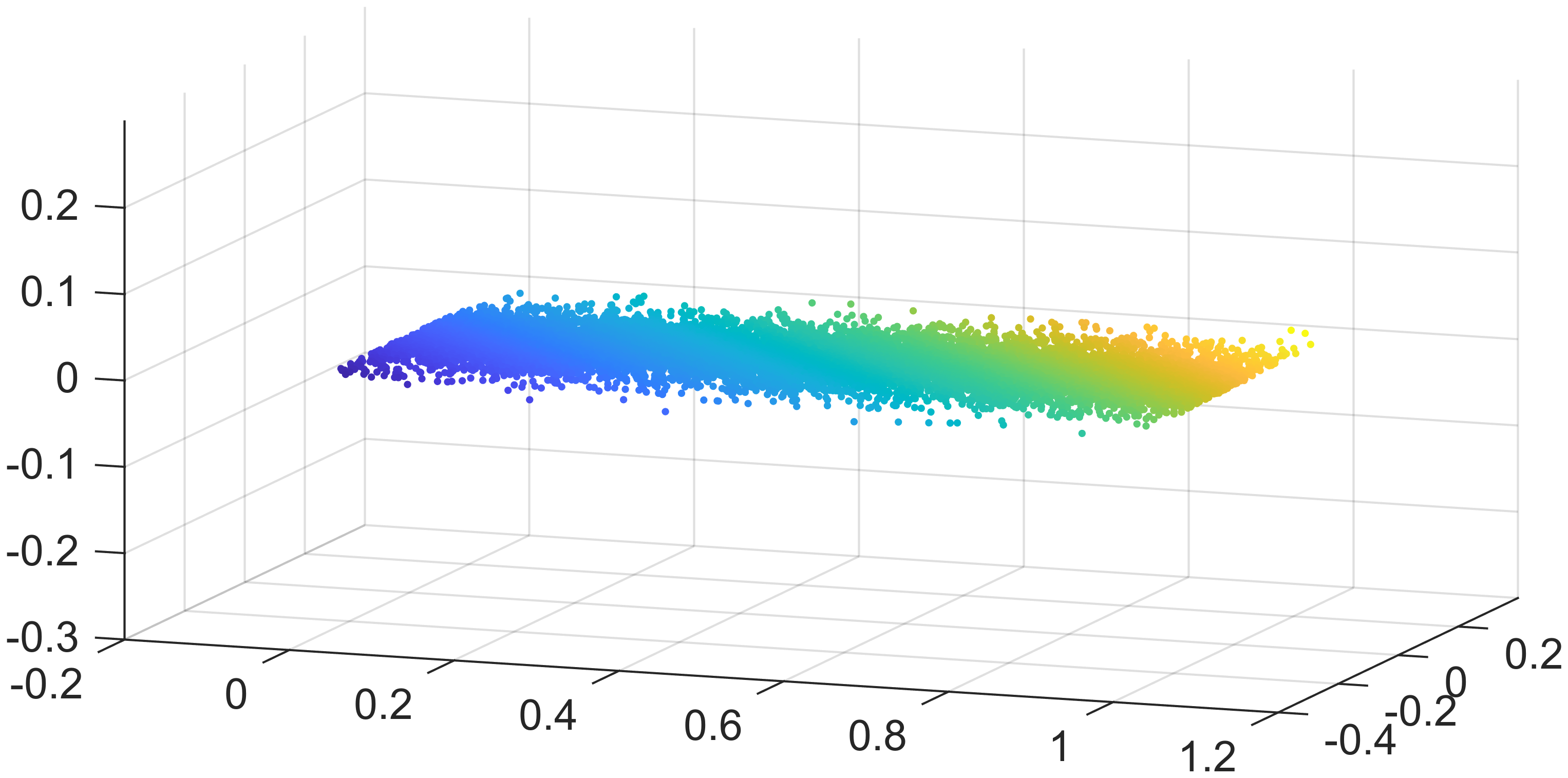}\label{fig:Simu-Eucli-line-PNSM2}}
    \subfigure[]{\includegraphics[width=0.3\textwidth]{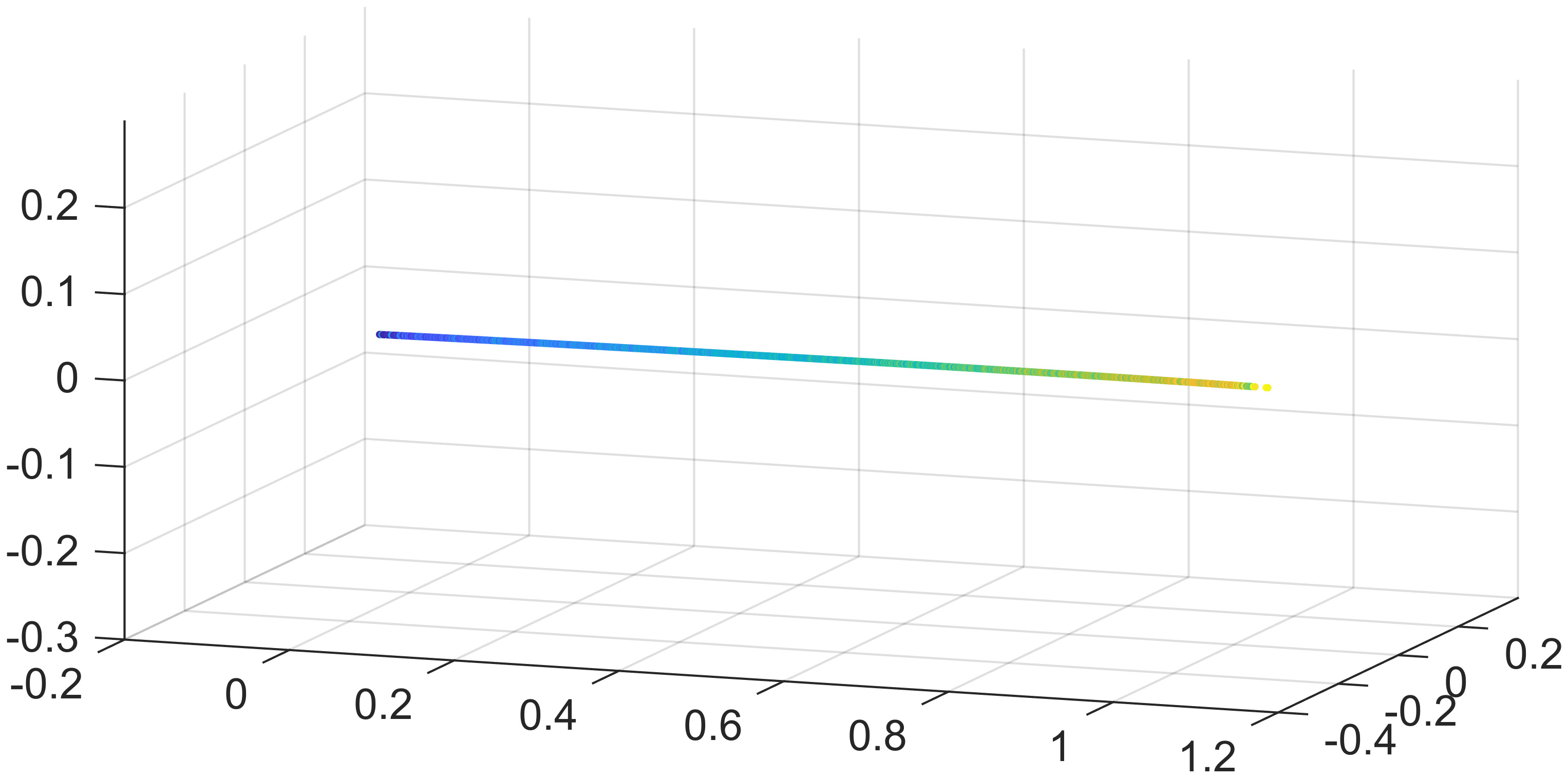}\label{fig:Simu-Eucli-line-PNSM1}}\\
    \subfigure[]{\includegraphics[width=0.3\textwidth]{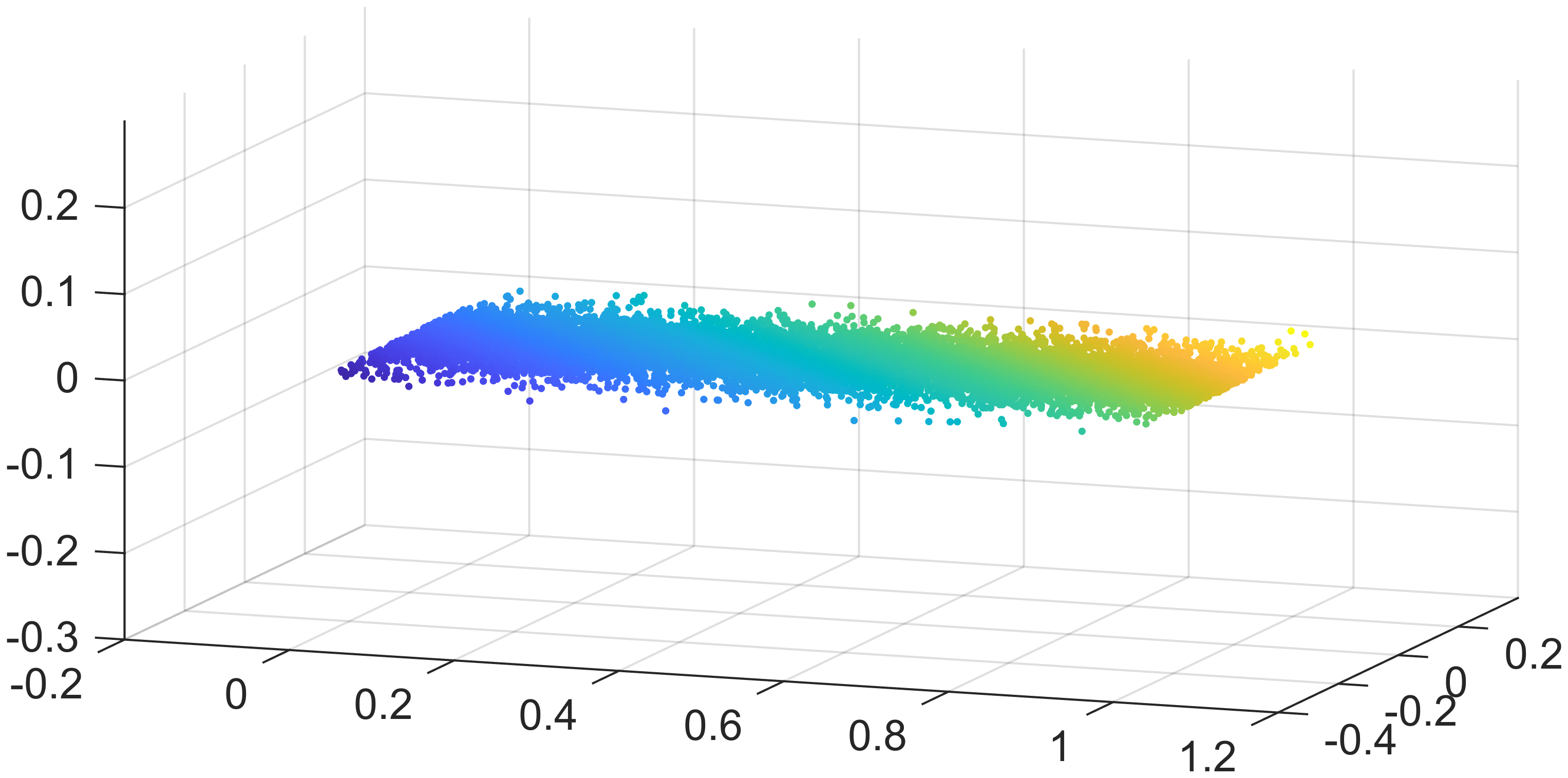}\label{fig:Simu-Eucli-line-PCA2}}
    \subfigure[]{\includegraphics[width=0.3\textwidth]{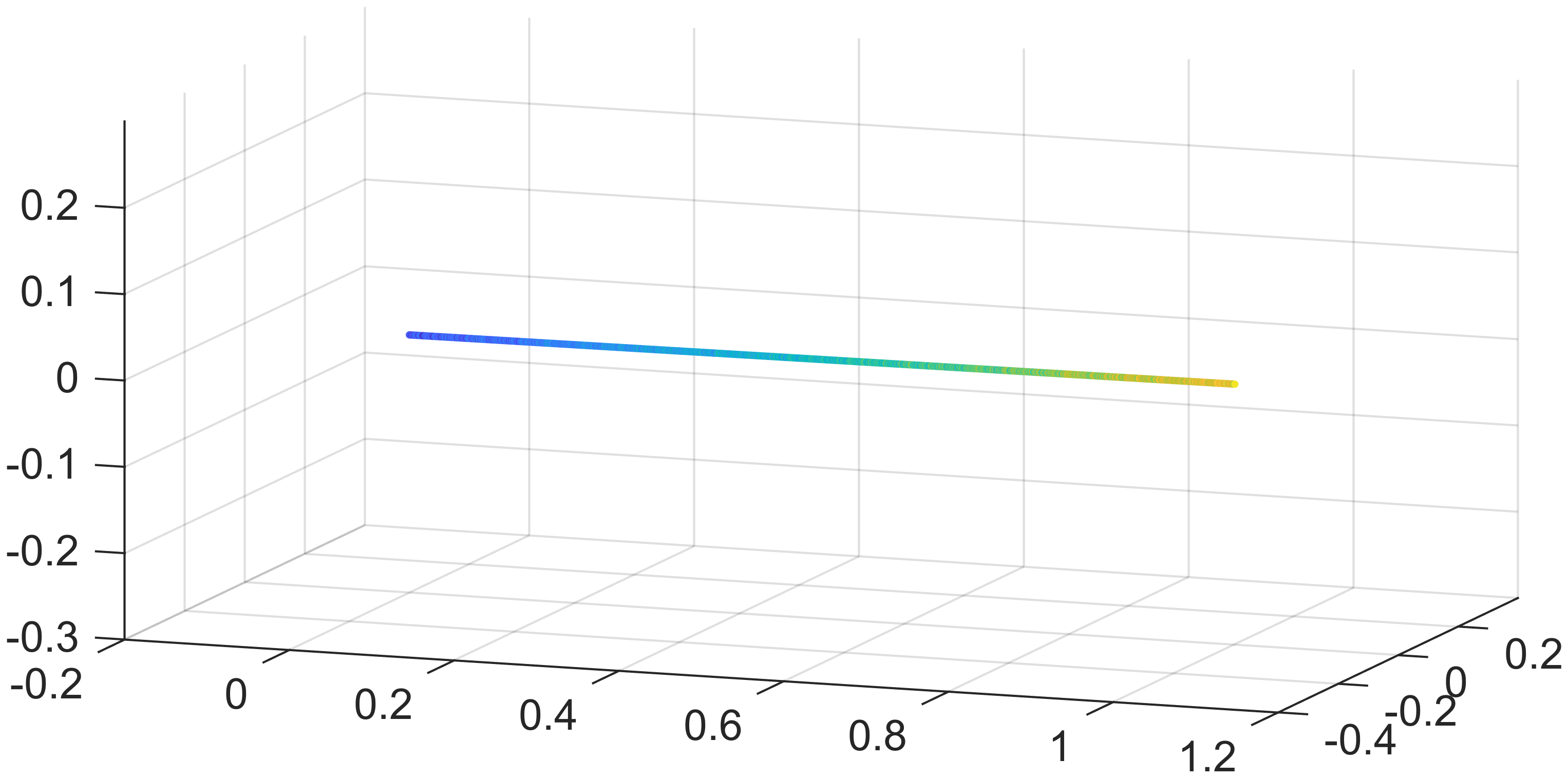}\label{fig:Simu-Eucli-line-PCA1}}
    \caption{Scatter plots for the line segment case in $\bR^3$ with colors added to enhance the visualization of sample adjacency. (a) Input data set; (b,c) Projection results onto the principal nested submanifolds of dimensions 2 and 1, respectively; (d,e) Projections onto the first two and single principal components, respectively.}
    \label{fig:Simu-Eucli-line}
\end{figure}

Figure \ref{fig:Simu-Eucli-line} presents the input data and projection results from the two methods. The points projected with our method are shown in Figure \ref{fig:Simu-Eucli-line-PNSM2} and Figure \ref{fig:Simu-Eucli-line-PNSM1}, where the points lie very close to the xy-plane and x-axis, respectively. This illustrates that the proposed method handles the scenario with identifiable fixed principal directions and approximates the linear subspaces effectively. The projection results are remarkably similar to those obtained from principal component analysis, as illustrated in Figure \ref{fig:Simu-Eucli-line-PCA2} and Figure \ref{fig:Simu-Eucli-line-PCA1}. This similarity indicates that our method remains compatible with linearly generated scenarios.

Subsequently, we introduce nonlinearity into the generating curve. Consider the unit circle
\[\gamma_2(t) = \bigl(\cos(t),\sin(t),0\bigl)^\top,~t\in(0,2\pi),\]
where $n=10^4$ points are sampled, and noise is introduced in the directions $v_1 = (0,0,1)^\top$ and $v_2 = (\cos(t),\sin(t),0)^\top$, with standard deviations $\sigma_1=0.1$ and $\sigma_2=0.05$, respectively. The resulting data set, shown in Figure \ref{fig:Simu-Eucli-circle-Input}, resembles a torus that is thicker in the z-axis direction.
\begin{figure}[htbp]
    \centering
    \subfigure[]{\includegraphics[width=0.3\textwidth]{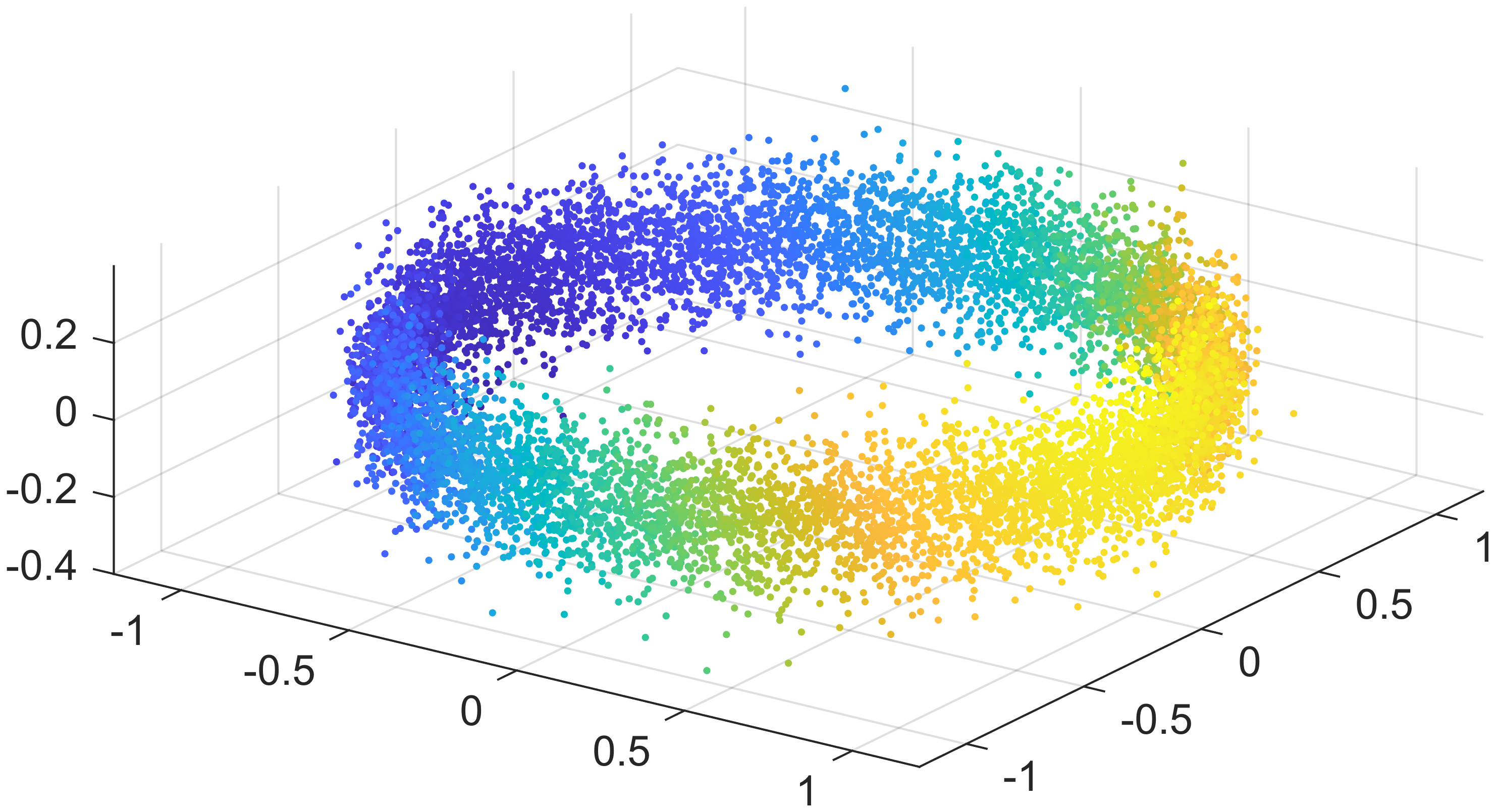}\label{fig:Simu-Eucli-circle-Input}}
    \subfigure[]{\includegraphics[width=0.3\textwidth]{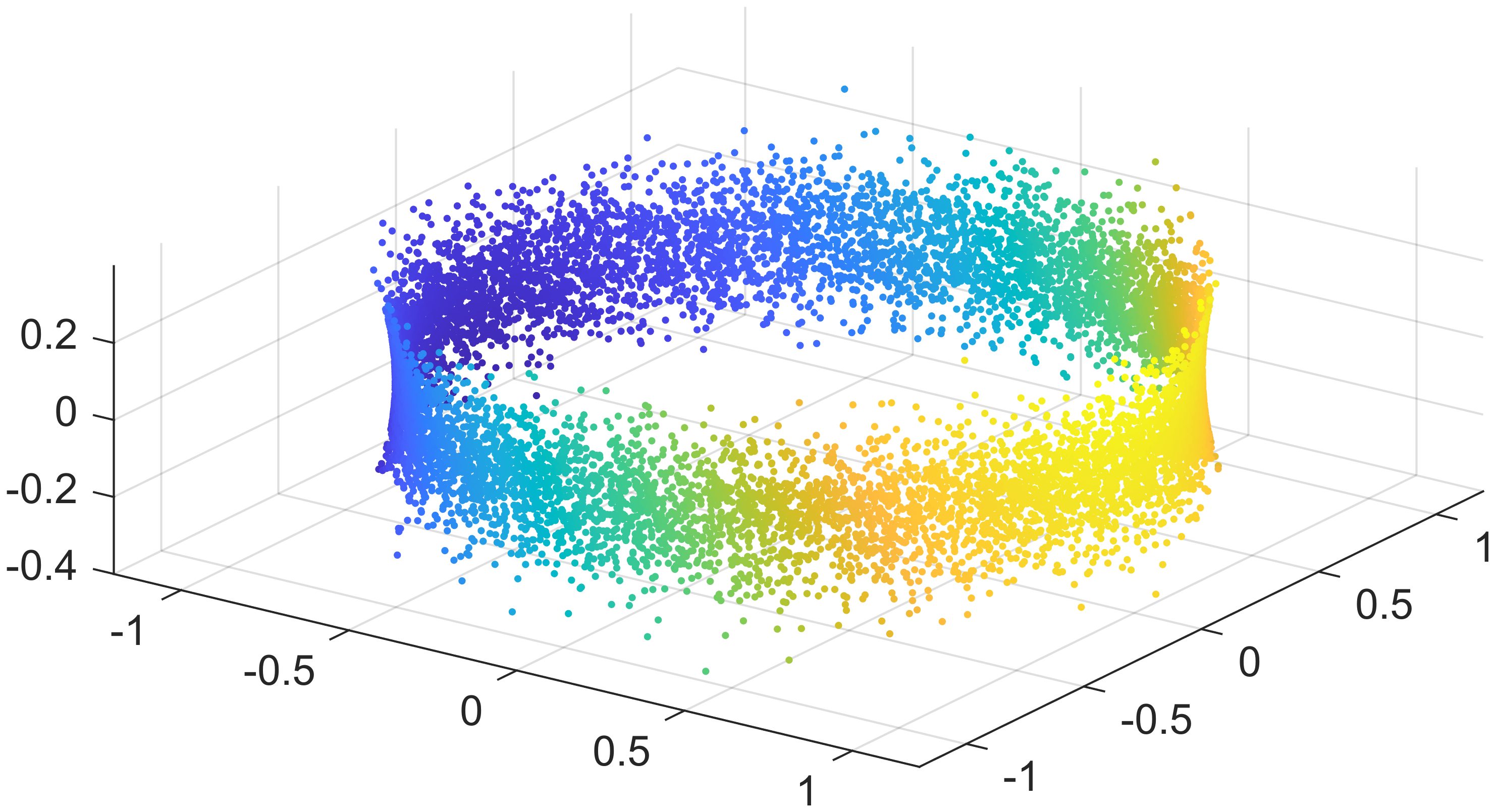}\label{fig:Simu-Eucli-circle-PNSM2}}
    \subfigure[]{\includegraphics[width=0.3\textwidth]{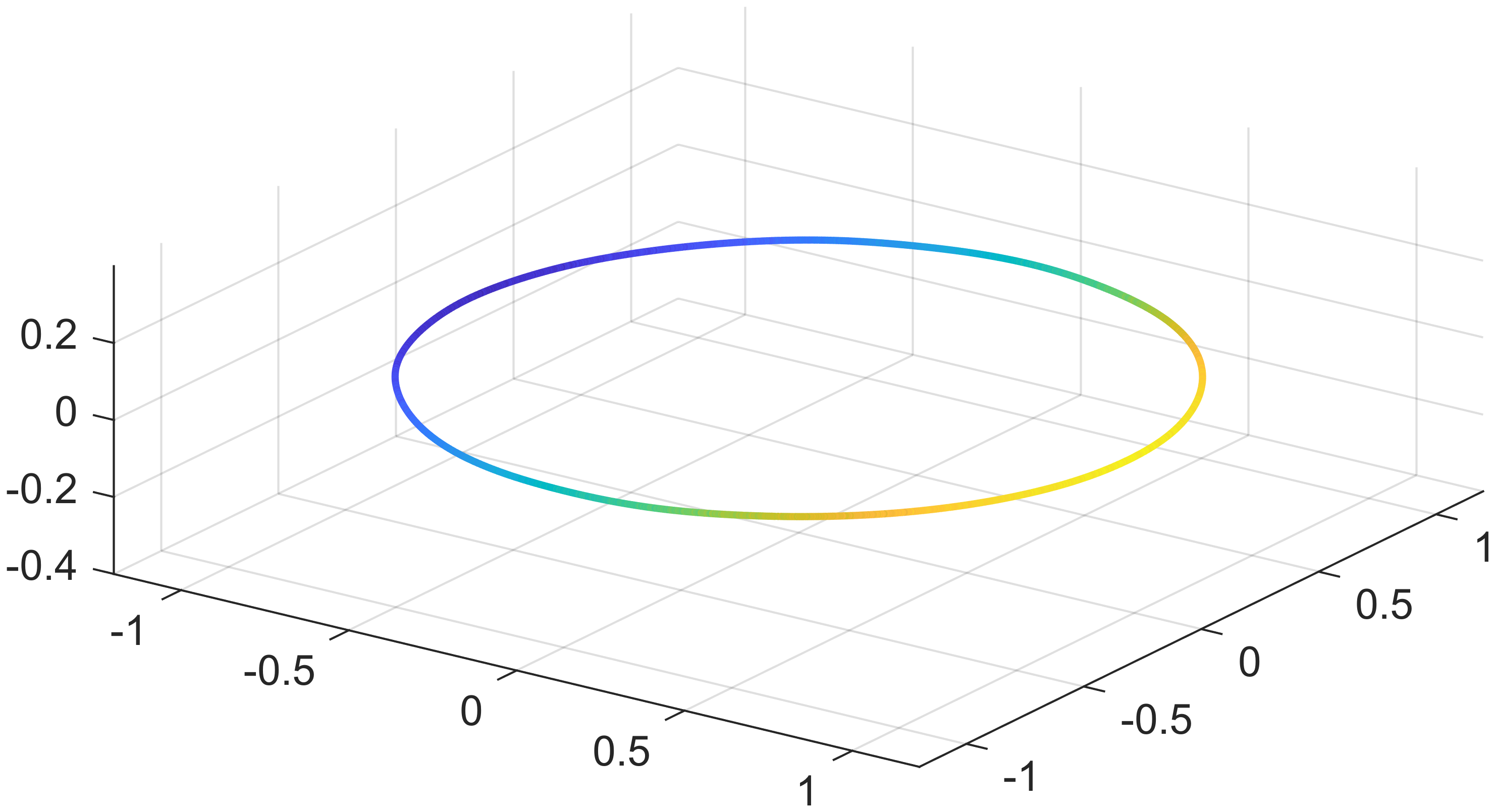}\label{fig:Simu-Eucli-circle-PNSM1}}\\
    \subfigure[]{\includegraphics[width=0.3\textwidth]{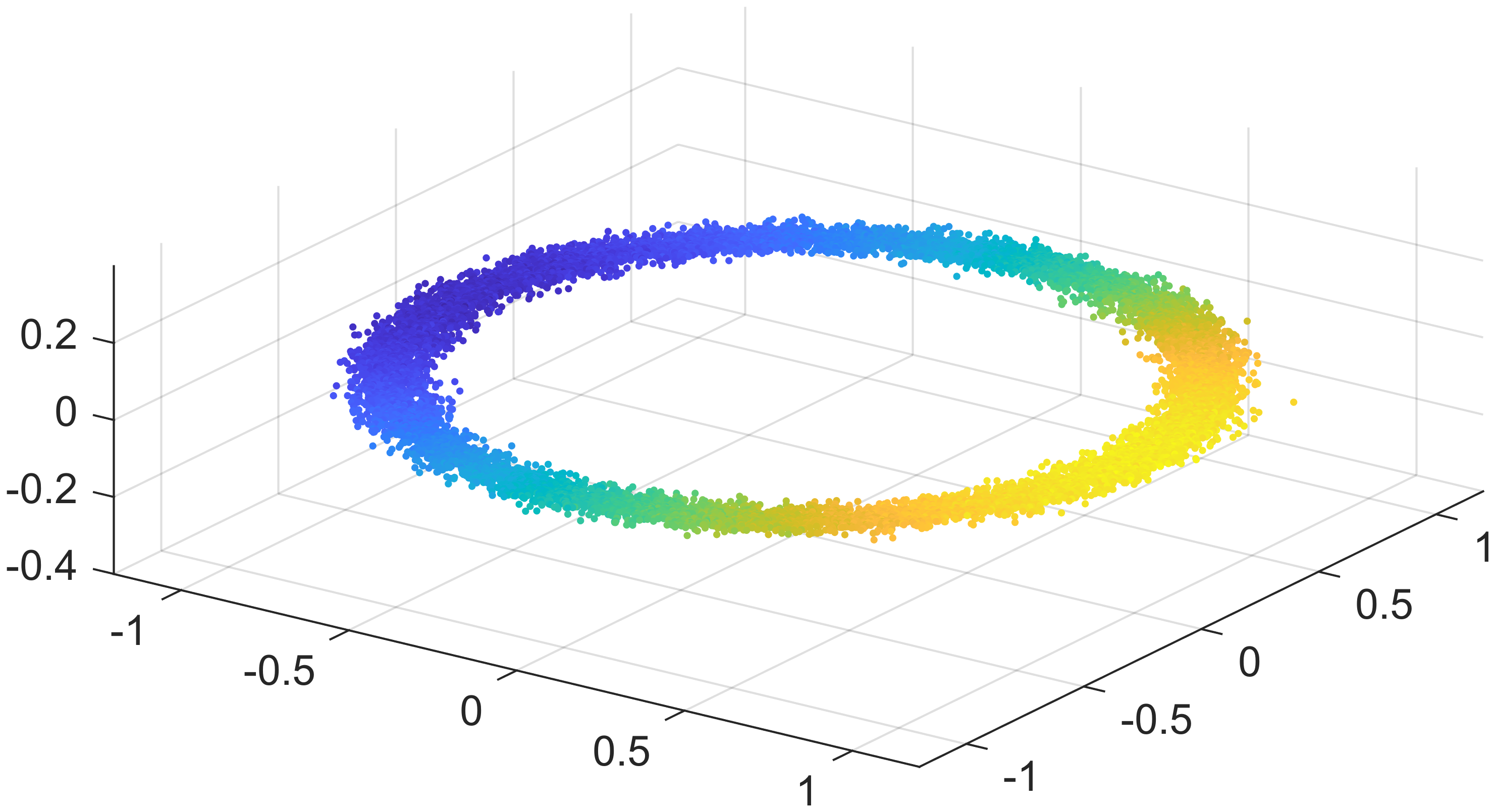}\label{fig:Simu-Eucli-circle-PCA2}}
    \subfigure[]{\includegraphics[width=0.3\textwidth]{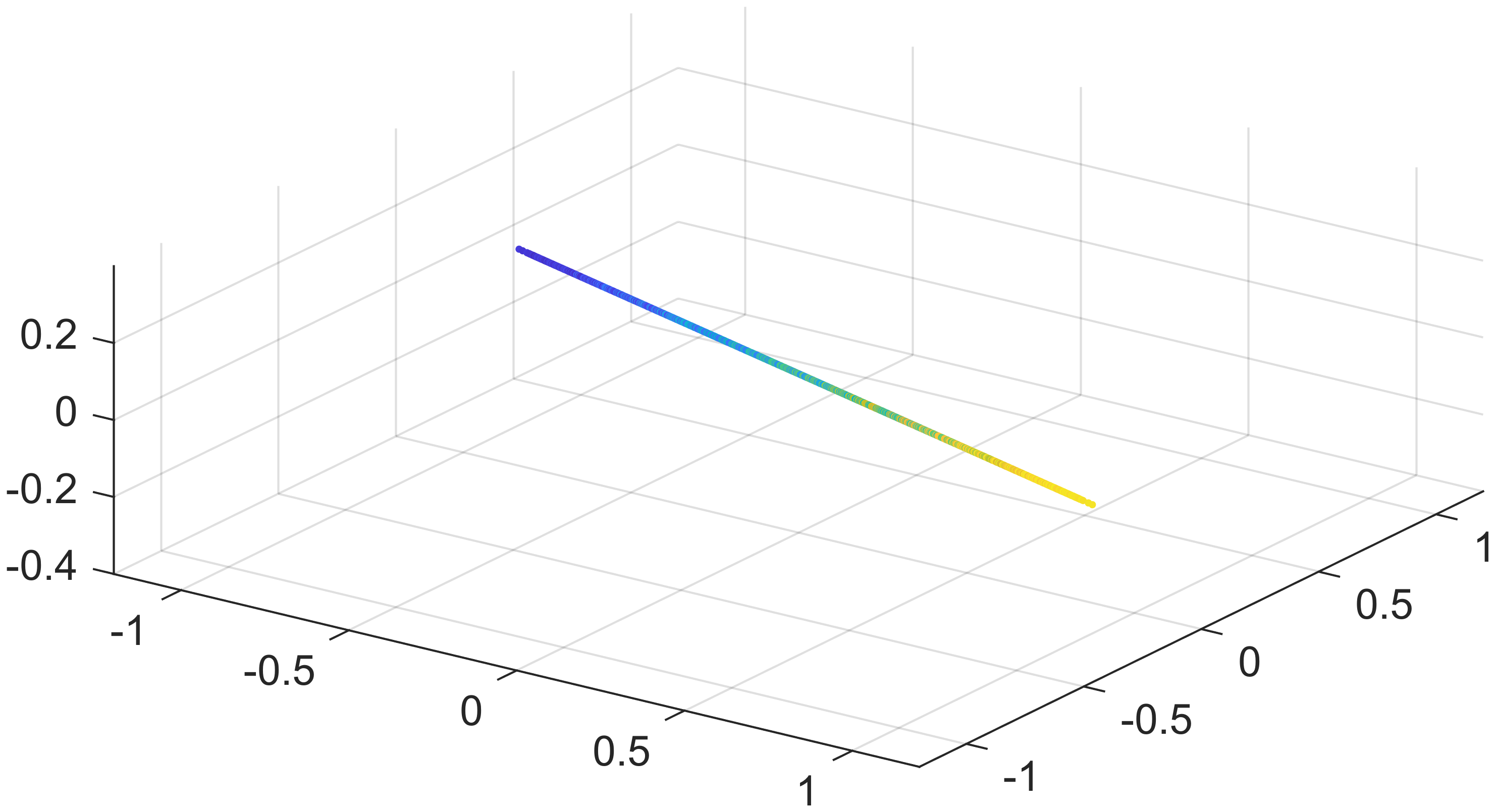}\label{fig:Simu-Eucli-circle-PCA1}}
    \caption{Scatter plots for the circle case in $\bR^3$ with colors added to enhance the visualization of sample adjacency. (a) Input data set; (b,c) Projection results onto the principal nested submanifolds of dimensions 2 and 1, respectively; (d,e) Projections onto the first two and single principal components, respectively.}
    \label{fig:Simu-Eucli-circle}
\end{figure}

The input data and projection results from the two methods are visualized in Figure \ref{fig:Simu-Eucli-circle}. Figure \ref{fig:Simu-Eucli-circle-PNSM2} presents the two dimensional projection with our method, which is a circular strip with width in the z-axis direction. This suggests that our method fits a two-dimensional submanifold that captures the majority of variation in $\cX_n$ by removing the variation along the smallest principal direction. In contrast, the 2-dimensional projection with principal component analysis, as shown in Figure \ref{fig:Simu-Eucli-circle-PCA2}, is an annulus approximately on the xy-plane. This indicates that the principal component analysis primarily focuses on the global covariance structure and overlooks the local principal direction corresponding to the smallest eigenvalue. These points are further projected on one-dimensional structures, and the results are present in Figure \ref{fig:Simu-Eucli-circle-PNSM1} and Figure \ref{fig:Simu-Eucli-circle-PCA1}. The resulting comparison is consistent with the view that our method provides a more faithful geometric summary of the circle structure of $\gamma_2$ in this example.

Finally, we consider a more complex correlation structure. Let the generating curve be
\[
\gamma_3(t) = \left(\frac{t}{6}\cos(t),~\frac{t}{6}\sin(t),~\frac{t}{6}\right)^\top, \quad t\in(0,6\pi),
\]
where $n=10^4$ random points are generated. At point $\gamma_3(t)$, define the unit tangent, Frenet normal and binormal directions by
\[
v_T(t) = \frac{\dot{\gamma}_3(t)}{\|\dot{\gamma}_3(t)\|}, \quad v_N(t) = \frac{\dot v_T(t)}{\|\dot v_T(t)\|}, \quad v_B(t) = v_T(t)\times v_N(t).
\]
Noise with standard deviations $\sigma_1=0.09$ and $\sigma_2=0.03$ is introduced in the directions 
\[
v_1 = \sin(2t)v_B + \cos(2t)v_N~\text{ and }~v_2= \cos(2t)v_B - \sin(2t)v_N,
\]
 respectively. The resulting data set, shown in Figure \ref{fig:Simu-Eucli-involute-Input}, can be regarded as the region swept by a rotating ellipse whose centre translates along $\gamma_3$.
\begin{figure}[htbp]
    \centering
    \subfigure[]{\includegraphics[width=0.3\textwidth]{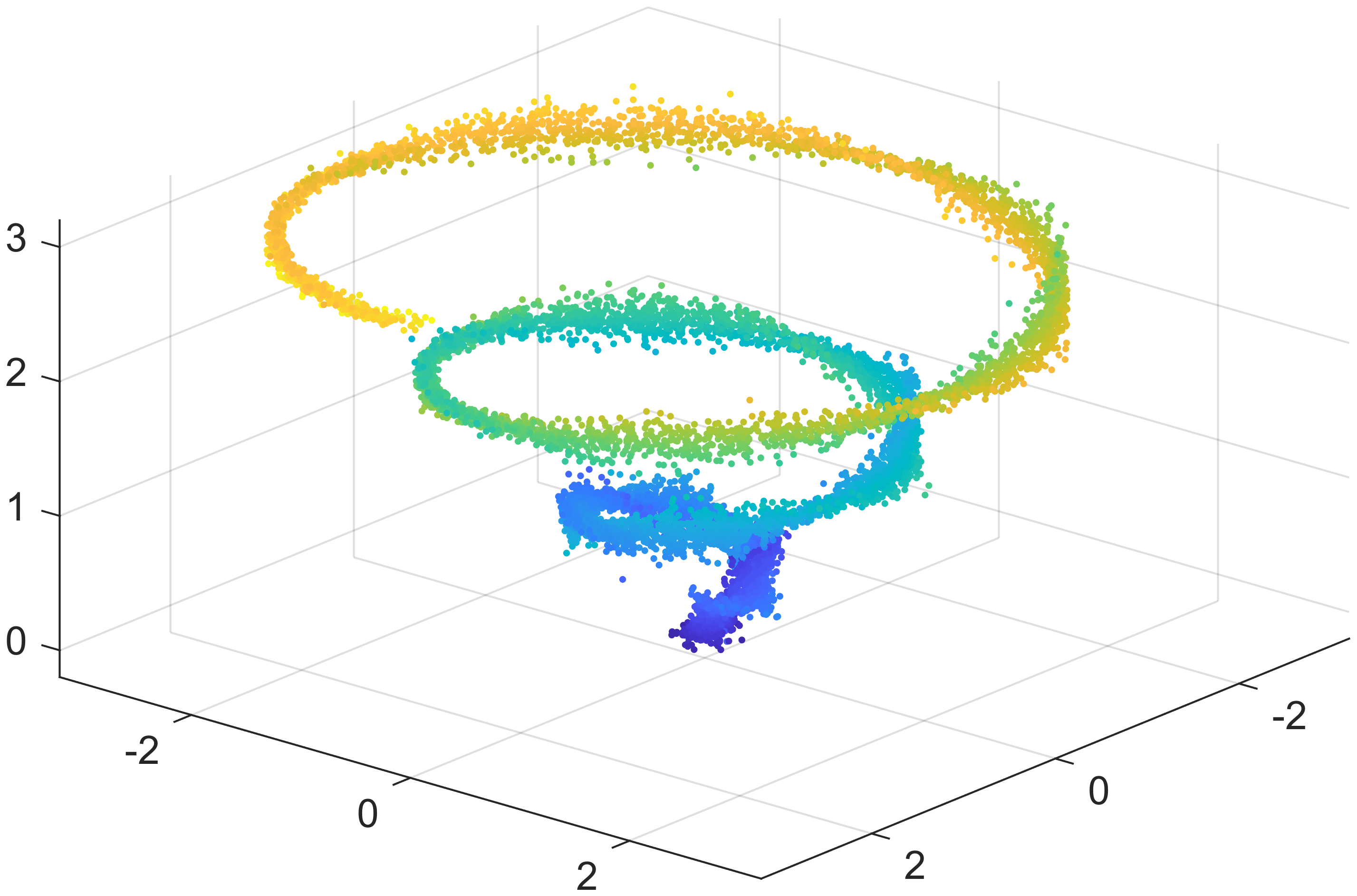}\label{fig:Simu-Eucli-involute-Input}}
    \subfigure[]{\includegraphics[width=0.3\textwidth]{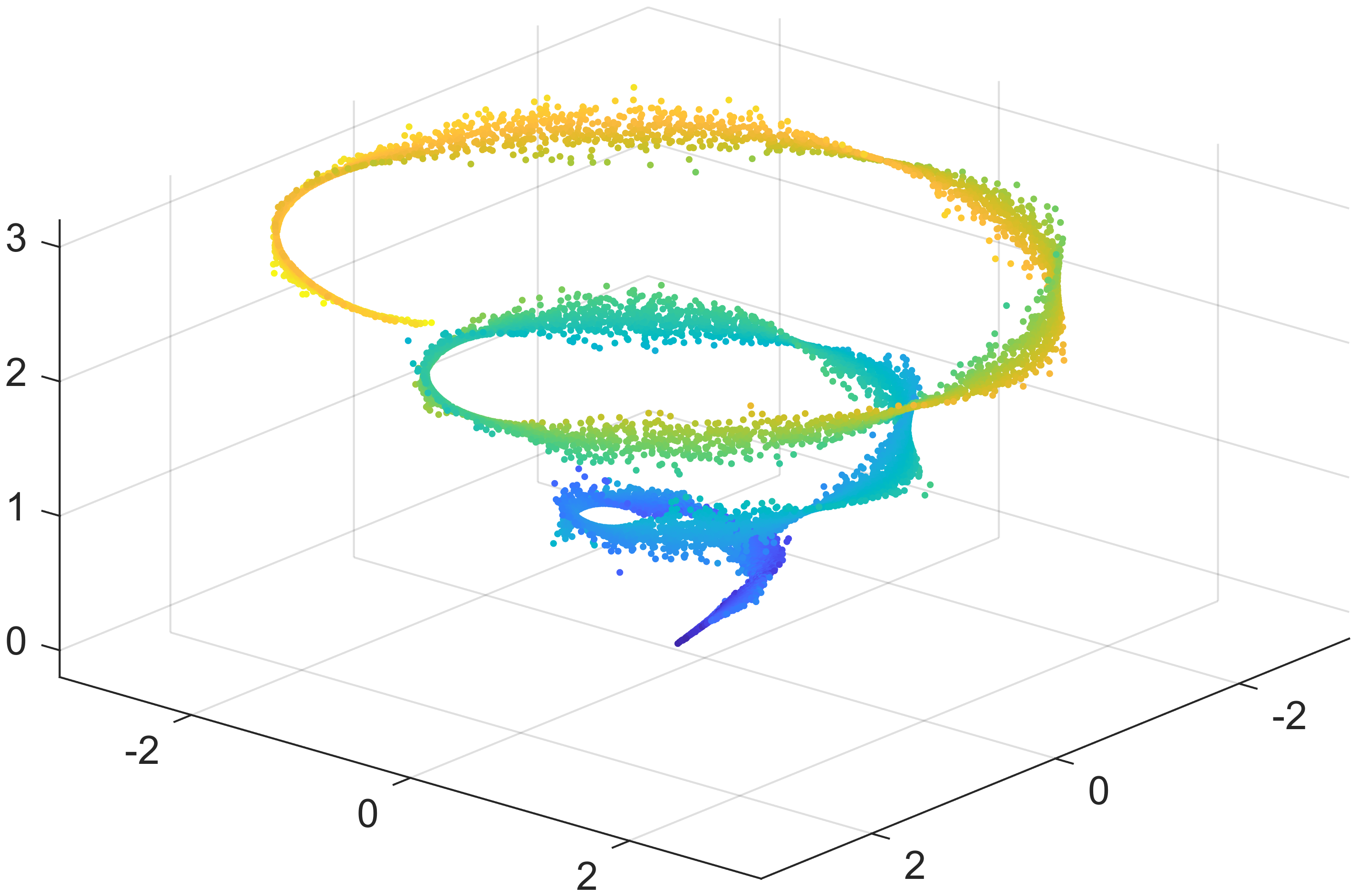}\label{fig:Simu-Eucli-involute-PNSM2}}
    \subfigure[]{\includegraphics[width=0.3\textwidth]{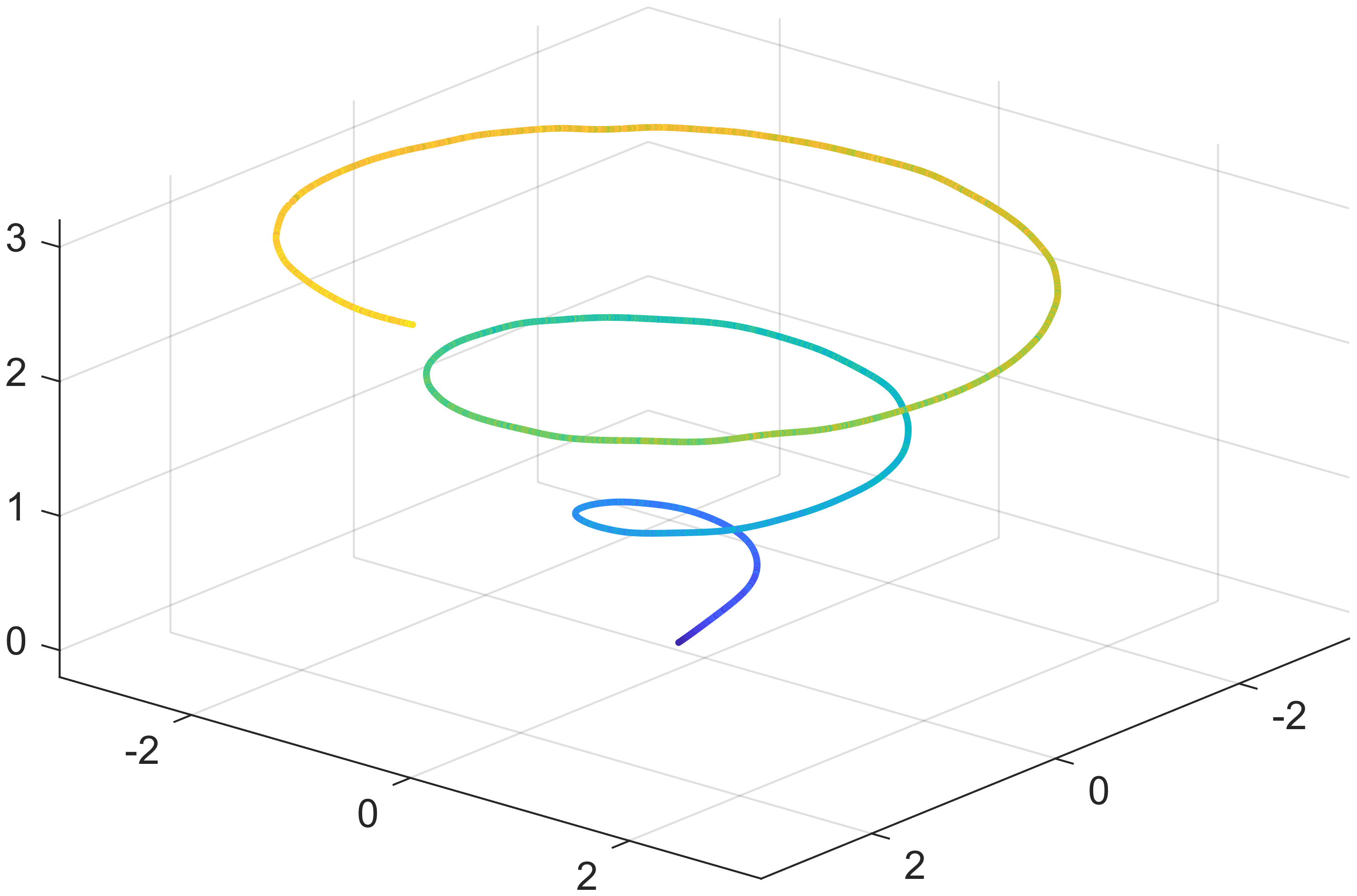}\label{fig:Simu-Eucli-involute-PNSM1}}\\
    \subfigure[]{\includegraphics[width=0.3\textwidth]{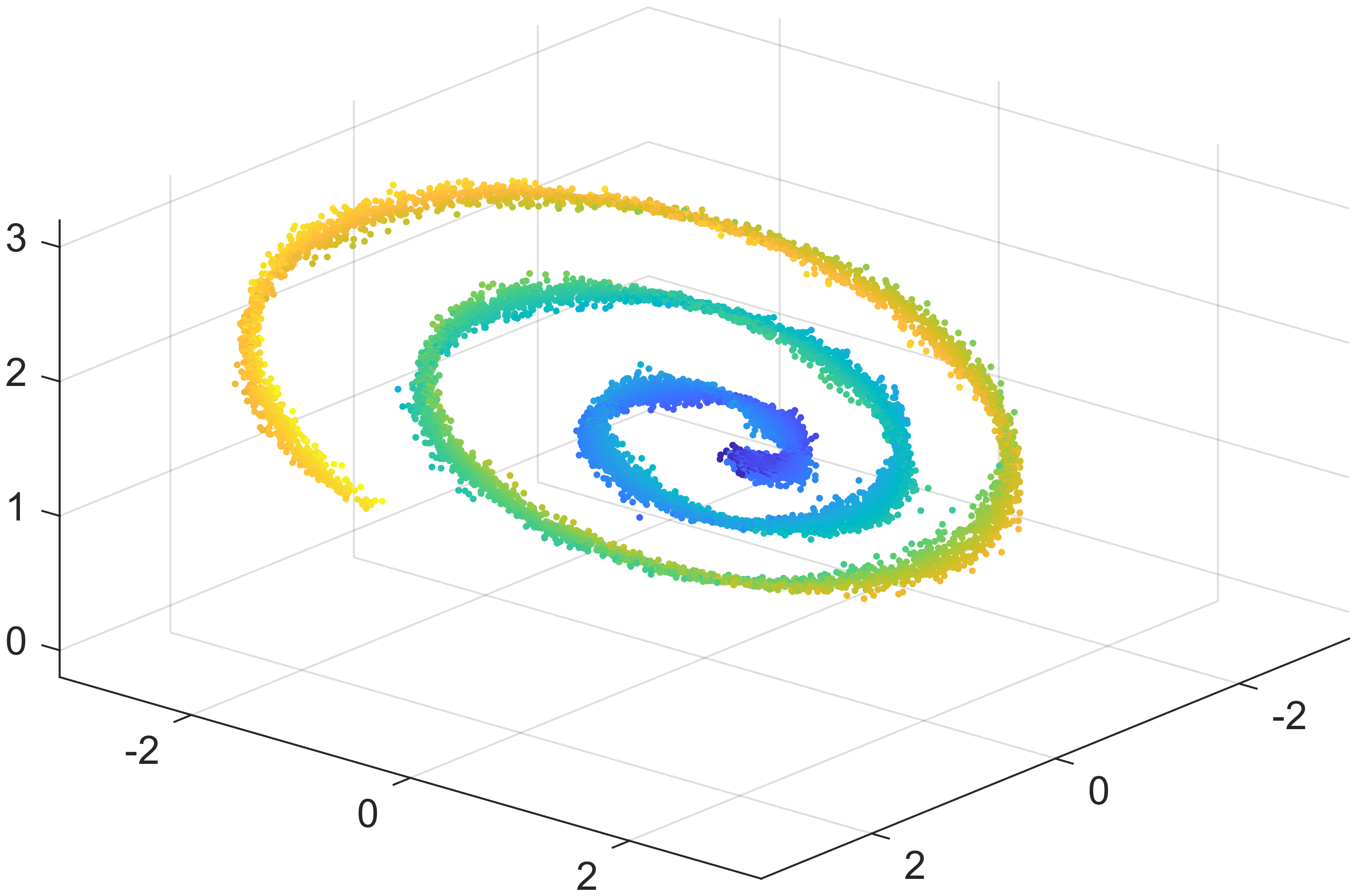}\label{fig:Simu-Eucli-involute-PCA2}}
    \subfigure[]{\includegraphics[width=0.3\textwidth]{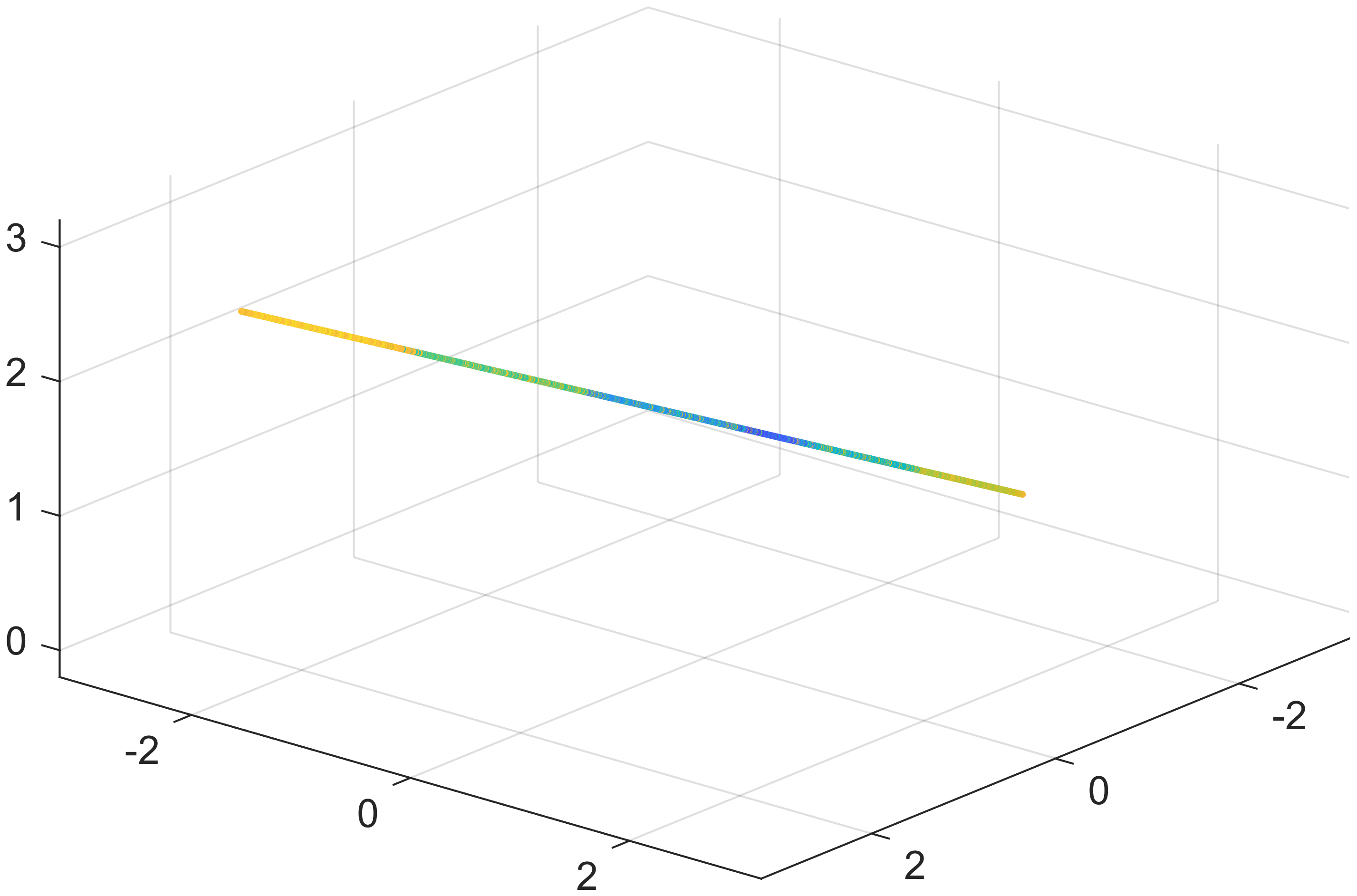}\label{fig:Simu-Eucli-involute-PCA1}}
    \caption{Scatter plots for the involute case in $\bR^3$ with colors added to enhance the visualization of sample adjacency. (a) Input data set; (b,c) Projection results onto the principal nested submanifolds of dimensions 2 and 1, respectively; (d,e) Projections onto the first two and single principal components, respectively.}
    \label{fig:Simu-Eucli-involute}
\end{figure}

Figure \ref{fig:Simu-Eucli-involute} presents the input data and projection results from the two methods. The two-dimensional projection with our method is shown in Figure \ref{fig:Simu-Eucli-involute-PNSM2}, which is a rotating strip along the involute whose surface follows the directions of $v_T$ and $v_1$. This result suggests that our method captures covariance structure that varies smoothly with position and fits a low-dimensional structure by removing variation along the principal direction corresponding to the smallest eigenvalue. In contrast, Figure \ref{fig:Simu-Eucli-involute-PCA2} shows the global planar projection from principal component analysis. The subsequent one-dimensional projections are shown in Figures \ref{fig:Simu-Eucli-involute-PNSM1} and \ref{fig:Simu-Eucli-involute-PCA1}. In this example, the proposed hierarchy follows the position-dependent involute structure, whereas principal component analysis provides a global linear summary.

Overall, these simulations illustrate how the proposed method can be used for decomposition and dimension reduction for data around low-dimensional structures in Euclidean space. This is achieved by adaptively exploiting the local sample covariance structure and combining it with the manifold fitting approach. Principal component analysis remains a strong global linear nested benchmark, while the proposed method provides a complementary locally adapted nonlinear hierarchy.

\subsubsection{Controlled radius sensitivity in the Euclidean spiral example}
\label{sec:spiral-radius}

We apply the radius diagnostics to the Euclidean spiral data set above ($n=10^4$, seed 3, ambient standard deviations 0.09 and 0.03, baseline $r=0.5$), rather than conduct a separate Monte Carlo study. Before examining the fits, we fixed the grid
$\{0.025,0.04,0.06,0.10,0.15,0.25,0.35,0.50,0.70,1.00,1.50,2.25,3.25,4.50,6.25\}$
from the 25th-neighbour distance scale, the baseline and the data diameter; summaries are descriptive across-anchor quantiles for this one data set. Hard support counts nonself observations used for the anchor-centred covariance, whereas cubic effective support, $(\sum_iw_i)^2/\sum_iw_i^2$ for $w_i=(1-d_i^2/r^2)^3$, describes the projection weights. The 25-neighbour line is an illustrative practical reference. For $\lambda_1\leq\lambda_2\leq\lambda_3$, the gaps $(\lambda_2-\lambda_1)/\lambda_3$ and $(\lambda_3-\lambda_2)/\lambda_3$ correspond to $3\to2$ and $2\to1$; anchors below the reference or with rank-deficient covariance are unavailable for these summaries. Fixed-relative stability is
\[
S(r)=\max_{\ell\in\{3\to2,2\to1\}}
\|\widehat P_{\ell,0.9r}-\widehat P_{\ell,1.1r}\|_F/\sqrt{2},
\]
the larger level-specific projector change. The centreline normal-displacement root mean squared error uses the known generating parameter. These observed-data diagnostics neither estimate noise or reach nor define a universal selector.

The displayed radii represent the observed regimes. At $r=0.06$, median support was 7, 87.4\% of anchors fell below the reference, and median stability was 0.579. At $r=0.5$, median hard and cubic effective supports were 324 and 193.5, the median weaker gap was 0.124, median stability was 0.072, and normal-displacement error was 0.026. This error rose to 0.042 at $r=0.7$, showing the onset of over-smoothing, and to 1.289 at $r=6.25$, where all 9999 nonself observations entered each neighbourhood. The latter fit was strongly global, but its median local-to-global normal-projector distance remained 0.702; hence we do not infer exact convergence to principal component analysis.

\begin{figure}[htbp]
    \centering
    \subfigure[$r=0.06$]{\includegraphics[width=0.235\textwidth]{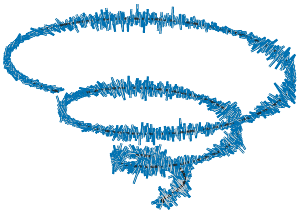}}\hfill
    \subfigure[$r=0.5$]{\includegraphics[width=0.235\textwidth]{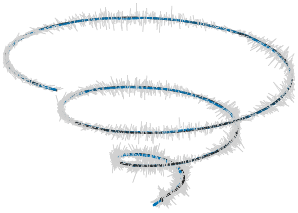}}\hfill
    \subfigure[$r=0.7$]{\includegraphics[width=0.235\textwidth]{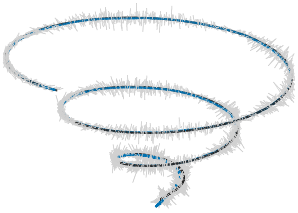}}\hfill
    \subfigure[$r=6.25$]{\includegraphics[width=0.235\textwidth]{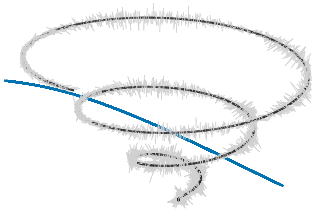}}\\
    \subfigure[Support]{\includegraphics[width=0.4\textwidth]{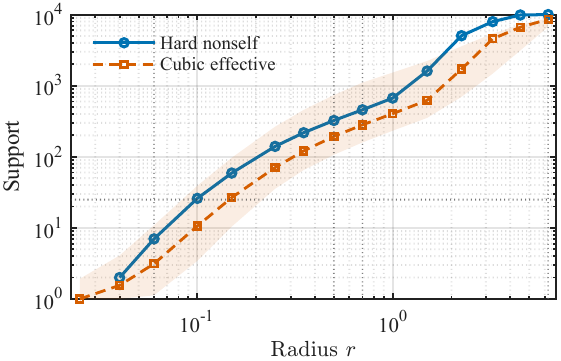}}
    \subfigure[Eigengaps]{\includegraphics[width=0.4\textwidth]{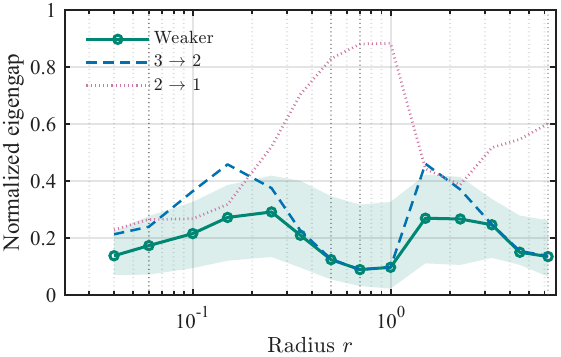}}\\
    \subfigure[Stability]{\includegraphics[width=0.4\textwidth]{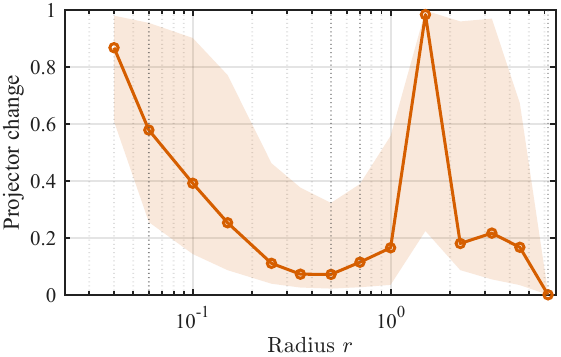}}
    \subfigure[Fit/globalisation]{\includegraphics[width=0.4\textwidth]{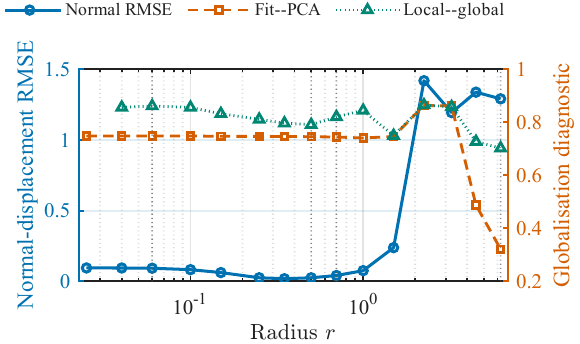}}
    \caption{Radius sensitivity for the Euclidean spiral data set ($n=10^4$, seed 3). Top row: noisy trajectories (grey), generating centreline (black), and fits (blue) representing an under-supported regime, the baseline locally adaptive range, the onset of over-smoothing, and strongly global behaviour. The second and third rows show support, the weaker of the two level-specific normalized eigengaps, the larger level-specific fixed-relative-radius projector change, and fitting/globalisation diagnostics. Lines and shaded ranges in the first three diagnostic panels are descriptive across-anchor medians and 10th--90th percentiles. Vertical lines mark the displayed radii. The final panel gives centreline normal-displacement root mean squared error, normalized fit-to-principal-component distance, and median local-to-global normal-projector distance.}
    \label{fig:spiral-radius}
\end{figure}

\subsubsection{Adjacent-level residual-direction alignment}
\label{sec:residual-alignment}

Let $\widehat x_i^{(d)}$ denote the ambient reconstruction of observation $i$ at target dimension $d$, and define its adjacent-level residual by
\[
\Delta_{i,d}=\widehat x_i^{(d+1)}-\widehat x_i^{(d)}.
\]
When the unit local direction $v_d(t_i)$ targeted by the transition $d+1\to d$ is known, we measure alignment by
\begin{equation}
A_d=\frac{1}{n}\sum_{i=1}^n
\frac{\left|\Delta_{i,d}^{\top}v_d(t_i)\right|}{\|\Delta_{i,d}\|}.
\label{eq:residual-alignment}
\end{equation}
The absolute inner product removes direction-sign ambiguity, and values near one indicate close alignment.

Table~\ref{tab:residual-alignment} reports ten independently generated samples of size $n=1000$, with $r=0.5$. The line is a linear sanity check for fixed directions, the circle is deliberately idealized, and the rotating ellipse in $\bR^3$ is a less symmetric nonlinear example in which the relevant directions vary continuously with position.

\begin{table}[htbp]
\centering
\tbl{Residual-direction alignment over ten simulation seeds. Entries are means with standard deviations in parentheses.}
{\begin{tabular}{llcc}
\toprule
Simulation & Transition & PNSM & PCA \\
\midrule
Line & $3\to2$ & 0.9997 (0.0003) & 0.9998 (0.0003) \\
Line & $2\to1$ & 0.9986 (0.0004) & 0.9998 (0.0003) \\
Circle with rotating radial direction & $3\to2$ & 0.9948 (0.0014) & 0.0038 (0.0022) \\
Circle with rotating radial direction & $2\to1$ & 0.9954 (0.0014) & 0.0029 (0.0023) \\
Rotating ellipse in $\bR^3$ & $3\to2$ & 0.9976 (0.0007) & 0.6394 (0.0104) \\
Rotating ellipse in $\bR^3$ & $2\to1$ & 0.9968 (0.0007) & 0.4014 (0.0696) \\
\bottomrule
\end{tabular}}
\label{tab:residual-alignment}
\end{table}

\begin{figure}[htbp]
    \centering
    \includegraphics[width=0.92\textwidth]{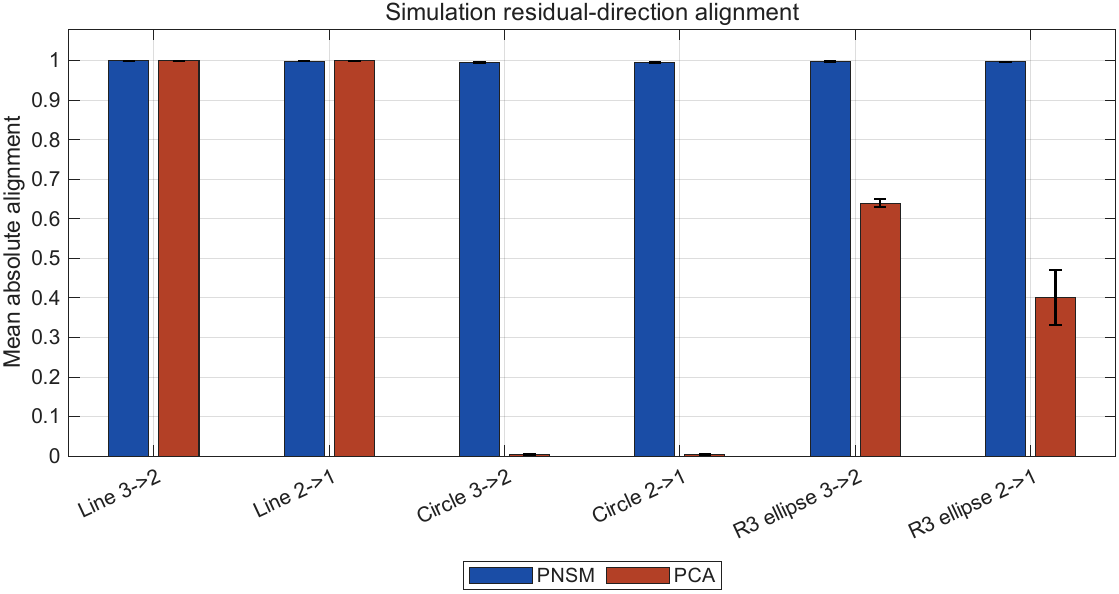}
    \caption{Residual-direction alignment over ten simulation seeds for the $3\to2$ and $2\to1$ transitions. The line is a fixed-direction sanity check, the circle is an idealized rotating-direction example, and the rotating ellipse in $\bR^3$ is a less symmetric nonlinear example. Values near one indicate that an adjacent-level residual follows the corresponding generating direction.}
    \label{fig:residual-alignment}
\end{figure}

The line results show that both methods recover fixed directions. When the relevant directions rotate with position, the locally adapted principal nested submanifold residuals follow those directions more closely, whereas principal component analysis supplies a global linear residual hierarchy. These controlled cases illustrate the distinction between the summaries rather than a universal ranking of the methods.

\subsection{Simulation in shape spaces}
In this subsection, we explore the decomposition of sample points around low-dimensional structures in shape spaces like spheres and tori, comparing our proposed method with the principal nested spheres and the torus principal component analysis in different cases. We still consider a generating curve $\{\gamma(t)\mid t\in\cT\}\subset[0,2\pi]^2$ and uniformly generate $n$ sample points $\{t_i\}_{i=1}^n\subset\cT$ in each case. Next, we introduce noise in the normal direction of $\gamma$ in $[0,2\pi]^2$. Specifically, a noisy angle pair in the angle plane is given by 
\begin{equation}\label{eq:simu-shapes-generating}
    (\phi_i, \psi_i)^\top = \gamma(t_i) + \xi_{i} \frac{\ddot \gamma(t_i)}{\|\ddot\gamma(t_i)\|},\quad i=1,\dots,n,
\end{equation}
where $\{\xi_{i}\}_{i=1}^n$ denotes a set of independent and identically distributed random noise amplitudes, each drawn from a normal distribution, $\cN(0,0.1^2)$. Given that scaling the entire sample set does not influence the analysis, we embed these angles onto $\cS^2$ and $T^2=\cS^1\times\cS^1$ respectively, which leads to the manifold-valued sample set $\cX_n$. This embedding results in a locally two-dimensional subset within the shape space. We then apply our proposed method, using a radius $r = 0.5$, to $\cX_n$ to project it onto the one-dimensional principal nested submanifolds. For comparative analysis, we perform both the principal nested spheres and torus principal component analyses.

\begin{figure}[htbp]
    \centering
    \subfigure{\includegraphics[width=0.25\textwidth]{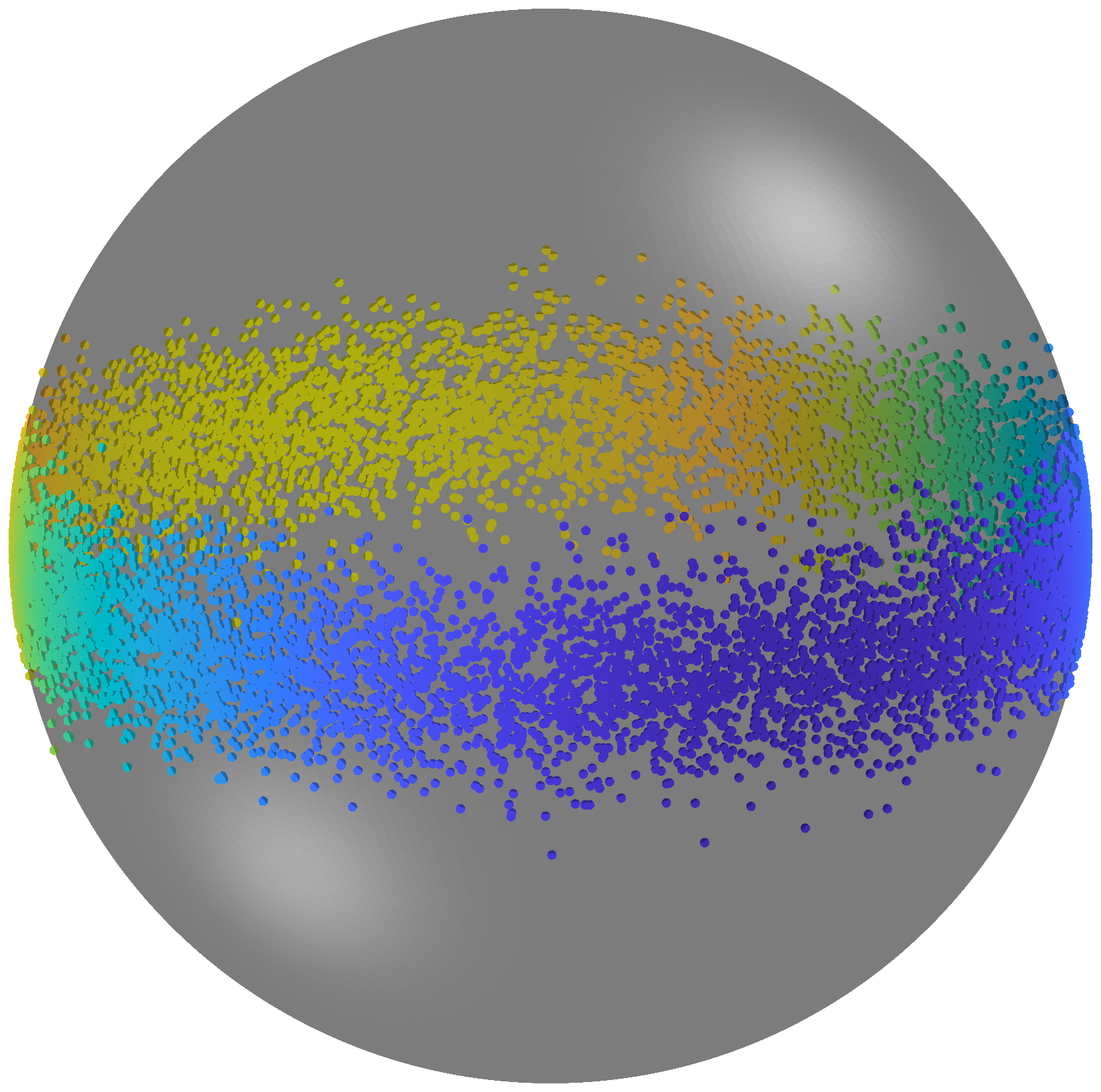}}
    \subfigure{\includegraphics[width=0.25\textwidth]{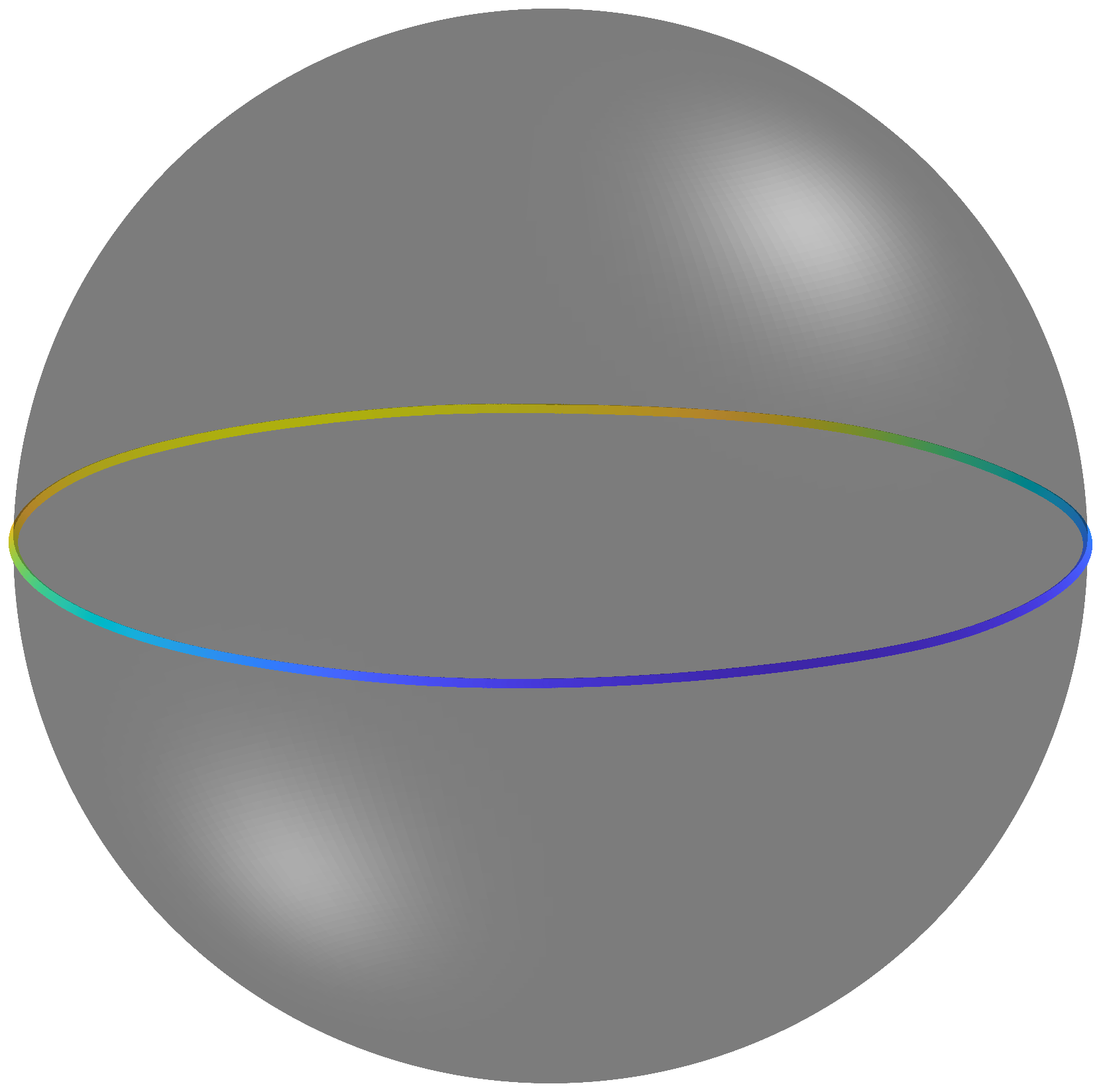}}
    \subfigure{\includegraphics[width=0.25\textwidth]{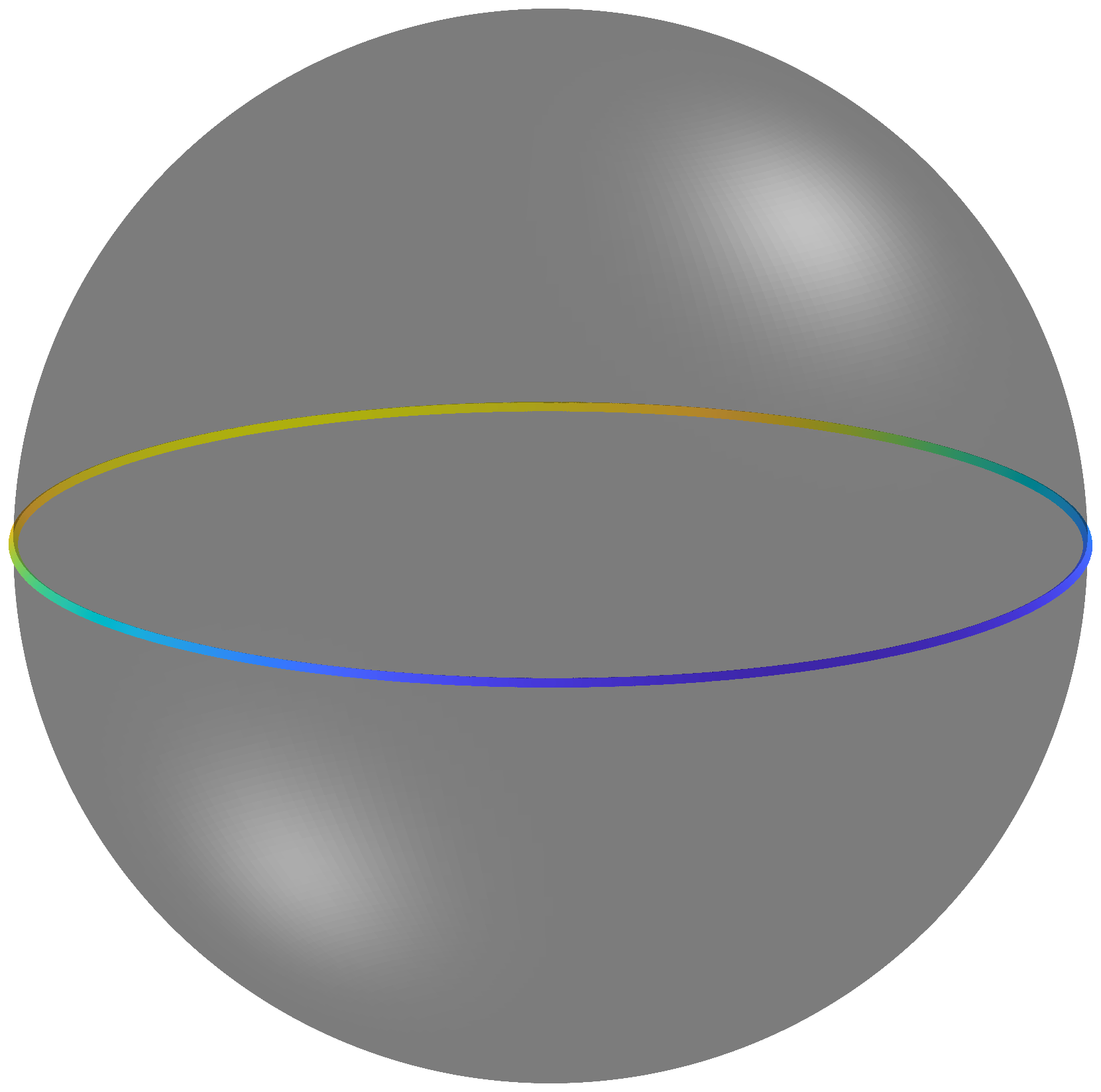}}\\
    \subfigure{\includegraphics[width=0.25\textwidth]{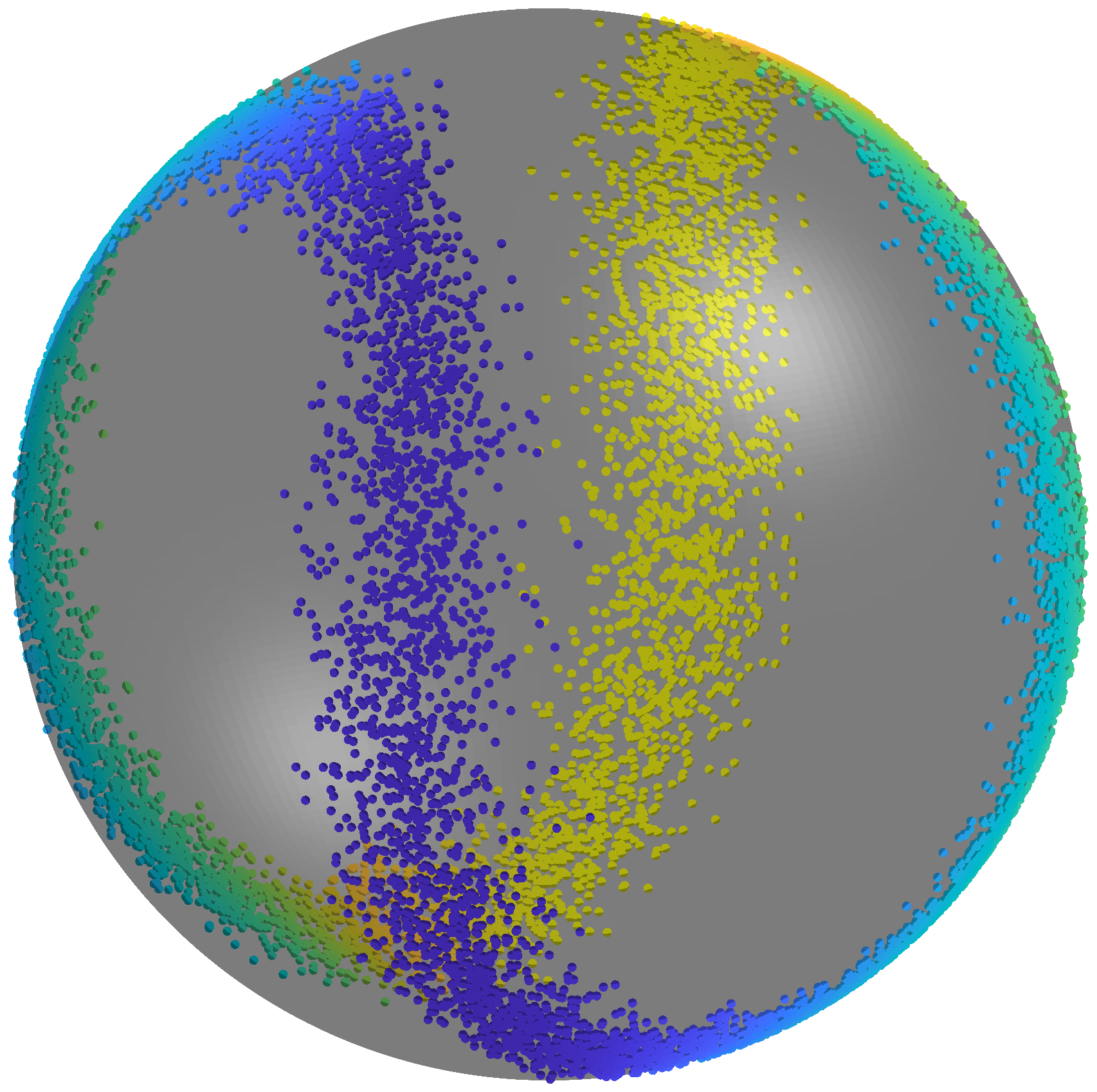}}
    \subfigure{\includegraphics[width=0.25\textwidth]{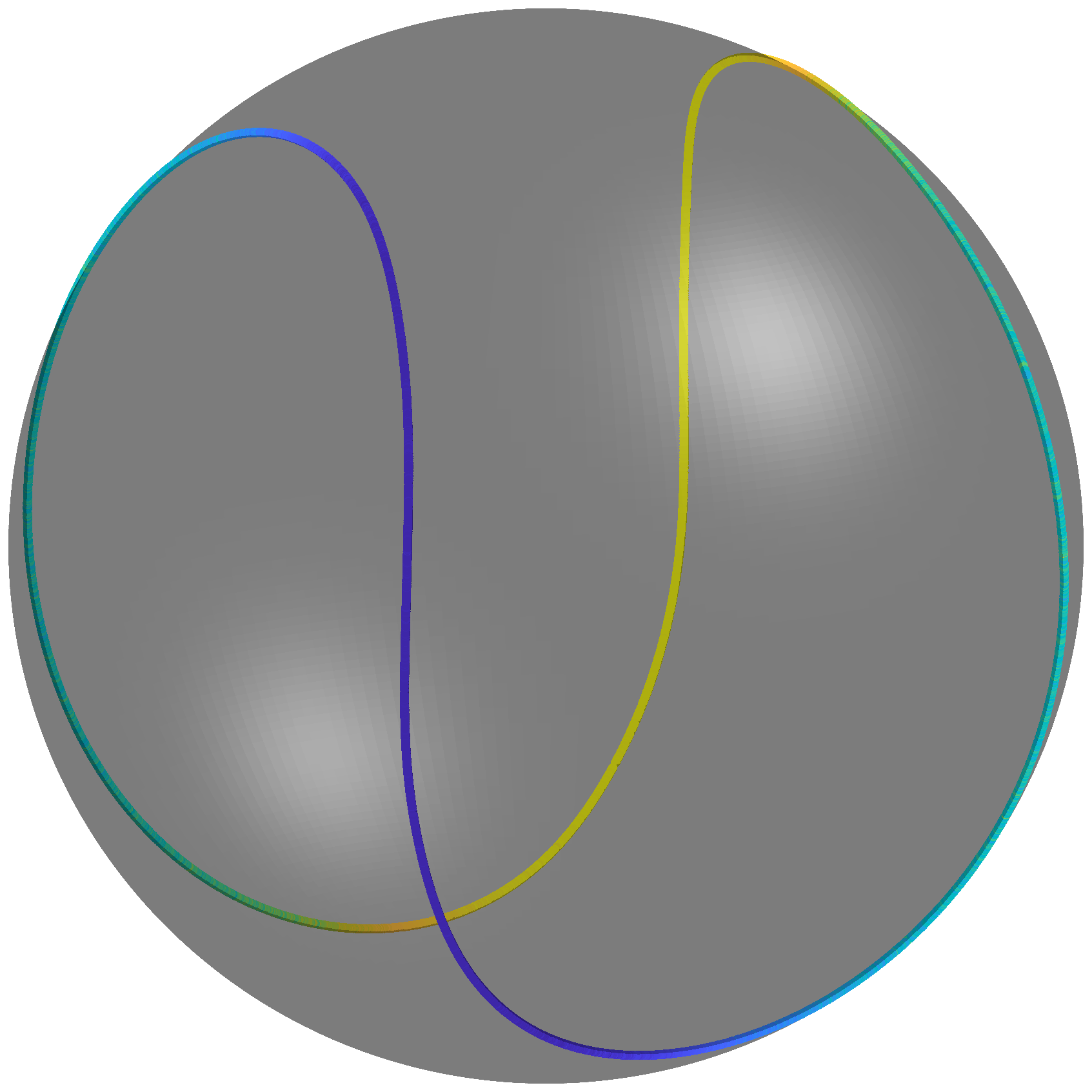}}
    \subfigure{\includegraphics[width=0.25\textwidth]{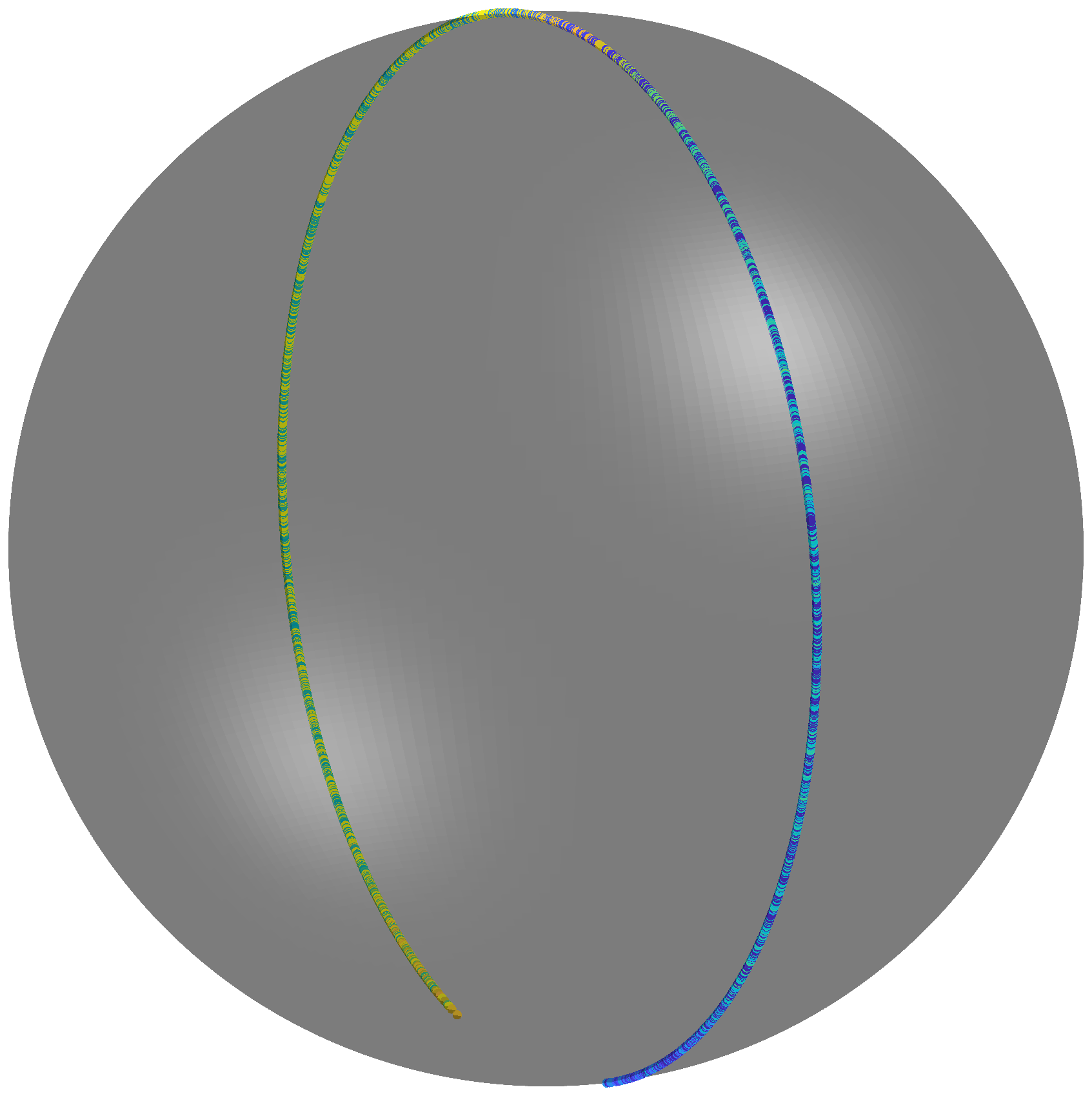}}\\
    \subfigure{\includegraphics[width=0.25\textwidth]{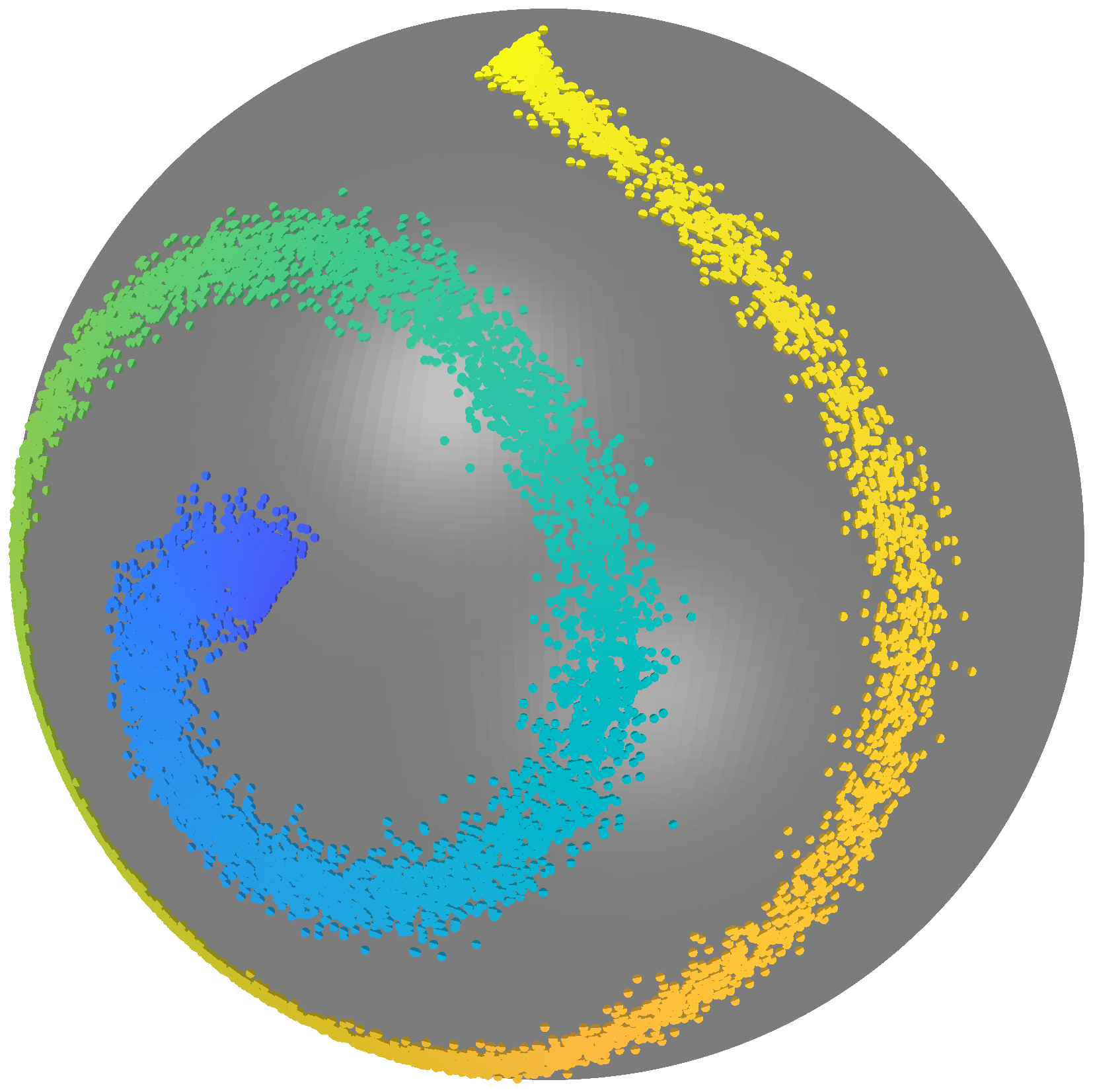}}
    \subfigure{\includegraphics[width=0.25\textwidth]{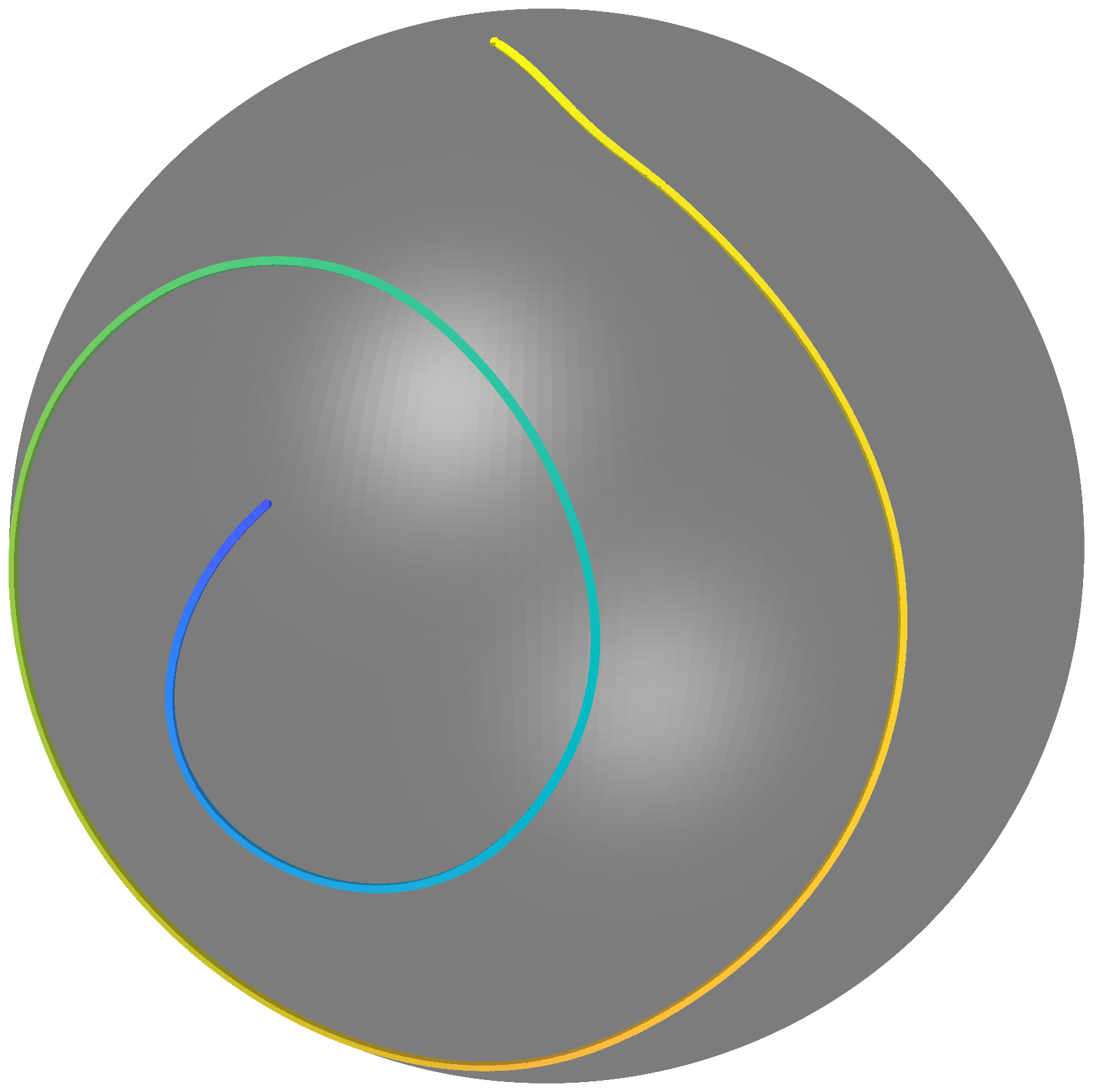}}
    \subfigure{\includegraphics[width=0.25\textwidth]{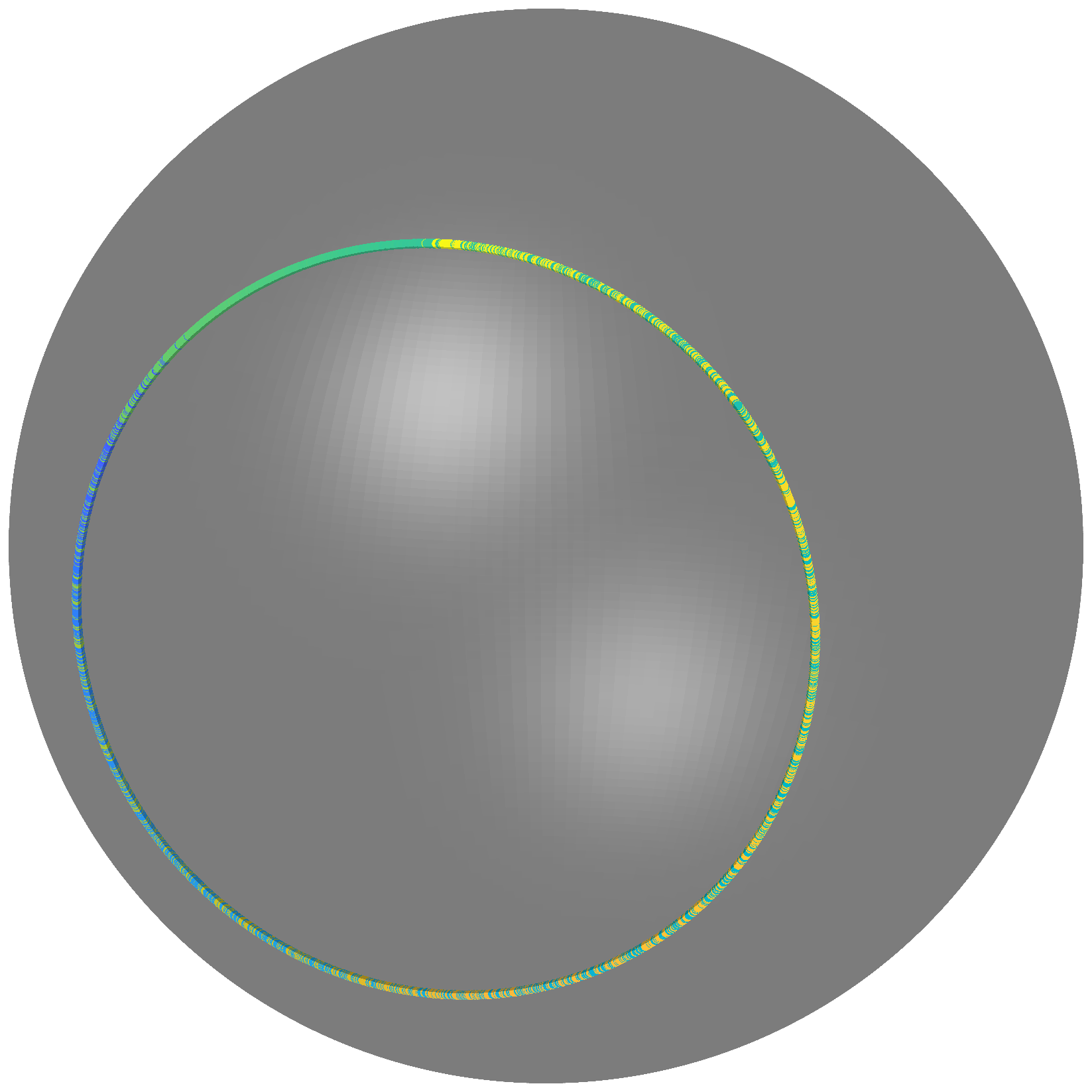}}\\
    \caption{Scatter plots for the sphere case with colors added to enhance the visualization of sample adjacency. Each row corresponds to one generative curve. The left column panels display the input data; the middle panels show the results with principal nested submanifolds, and the right column presents the results from principal nested spheres.}
    \label{fig:Simu-Sphere}
\end{figure}

First, we investigate the decomposition of the data on spheres. In this case, samples are generated with three curves
\[
\begin{aligned}
    \gamma_1(t) &= \left(t,0\right)^\top, \quad t\in(0,2\pi);\\
    \gamma_2(t) &= \begin{cases}
        \left(\arctan\left\{\tan^3(t) - \frac{\pi}{2}\right\},\arccos\left\{\surd 3 \sin(t)\cos(t)\right\}\right)^\top, &t\in(0,\pi),\\
        \left(\arctan\left\{\tan^3(t) + \frac{\pi}{2}\right\},\arccos\left\{\surd 3 \sin(t)\cos(t)\right\}\right)^\top, &t\in(\pi,2\pi);\\
    \end{cases}\\
    \gamma_3(t) &= \left(\frac{t}{10}\cos(t),\frac{t}{10}\sin(t)\right)^\top, \quad t\in\left(\frac{\pi}{2},\frac{9\pi}{2}\right).
\end{aligned}
\]
with the data generating function in \eqref{eq:simu-shapes-generating}, a set of $n=10^4$ noisy angle pairs in the angle plane for each case, and these angle pairs are embedded into a unit sphere $\cS^2\subset\bR^3$ by
\[x_i = \bigl( \cos(\psi_i)\cos(\phi_i),~\cos(\psi_i)\sin(\phi_i),~\sin(\psi_i)\bigl)^\top, \quad i=1,\dots,n.\]
The resulting data sets are shown in the left-most column of Figure \ref{fig:Simu-Sphere}, which resemble a circle, a tennis curve, and an involute on $\cS^2$ respectively. Then, the proposed method and the principal nested spheres are applied to these data to project them on one-dimensional submanifolds of $\cS^2$.

The middle column of Figure \ref{fig:Simu-Sphere} presents the projection results of our method. In each case, the points projected by our method are approximately on a curve that passes through the center of the sample, and the neighbourhood relationship between the samples is well maintained. This suggests that our method can fit submanifolds on spheres to explain the majority of sample variation. In contrast to our approach, the principal nested spheres can only projects the samples on circles cut from the sphere, as shown in the last column of Figure \ref{fig:Simu-Sphere}, which leads to some loss of neighbourhood information. This difference is similar to principal component analysis and reflects the restriction to special subspaces.

\begin{figure}[htbp]
    \centering
    \subfigure{\includegraphics[width=0.25\textwidth]{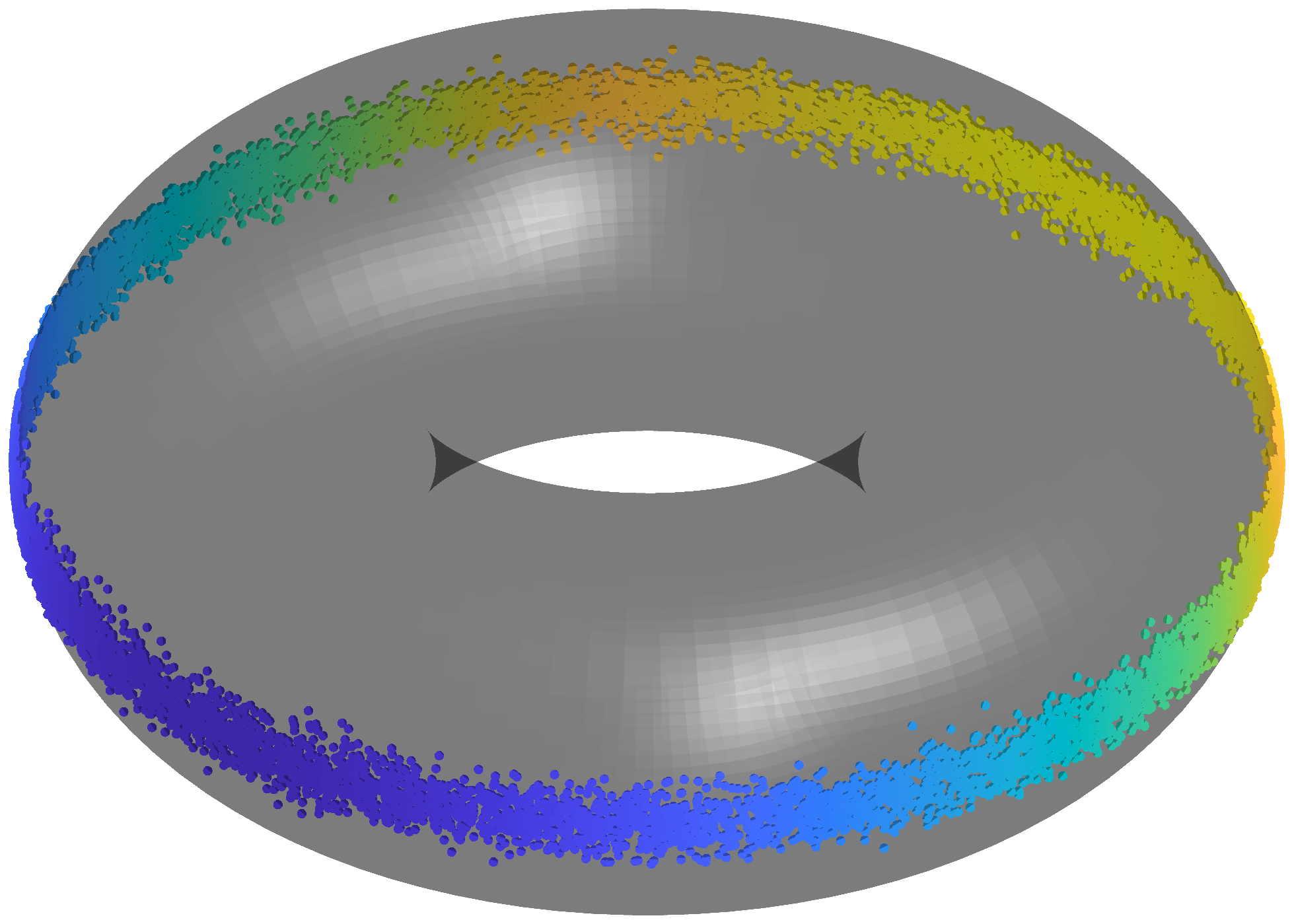}}
    \subfigure{\includegraphics[width=0.25\textwidth]{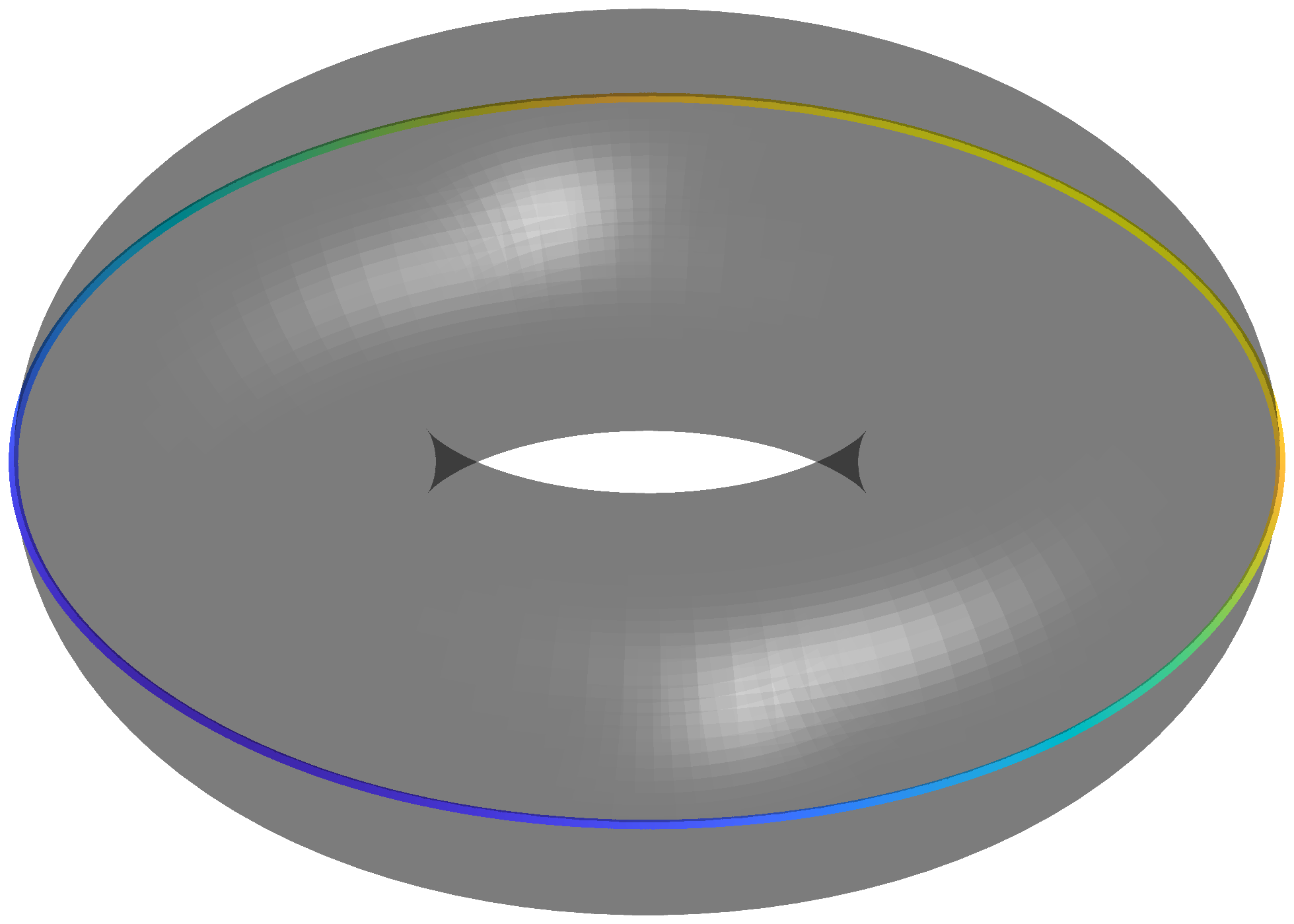}}
    \subfigure{\includegraphics[width=0.25\textwidth]{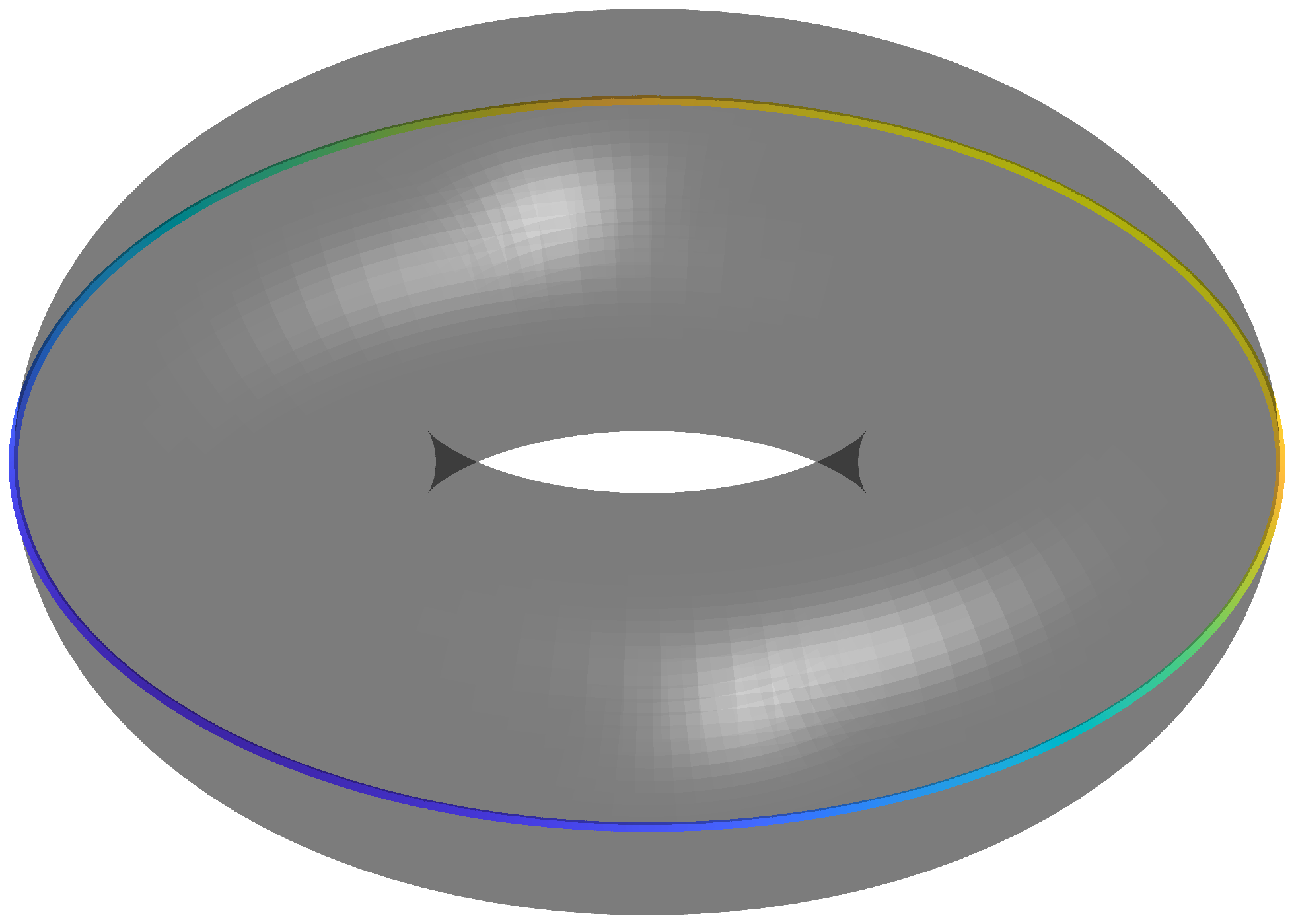}}\\
    \subfigure{\includegraphics[width=0.25\textwidth]{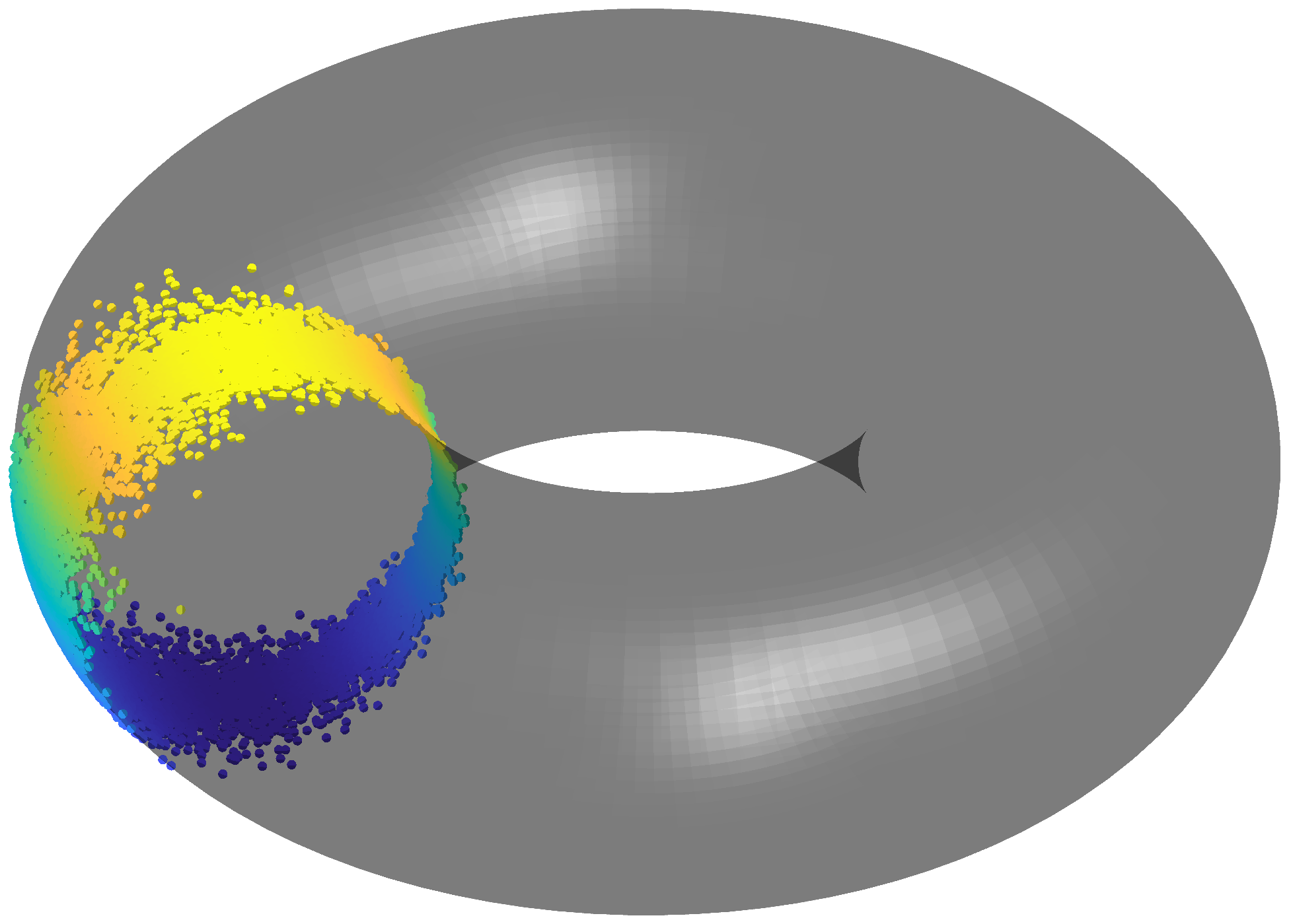}}
    \subfigure{\includegraphics[width=0.25\textwidth]{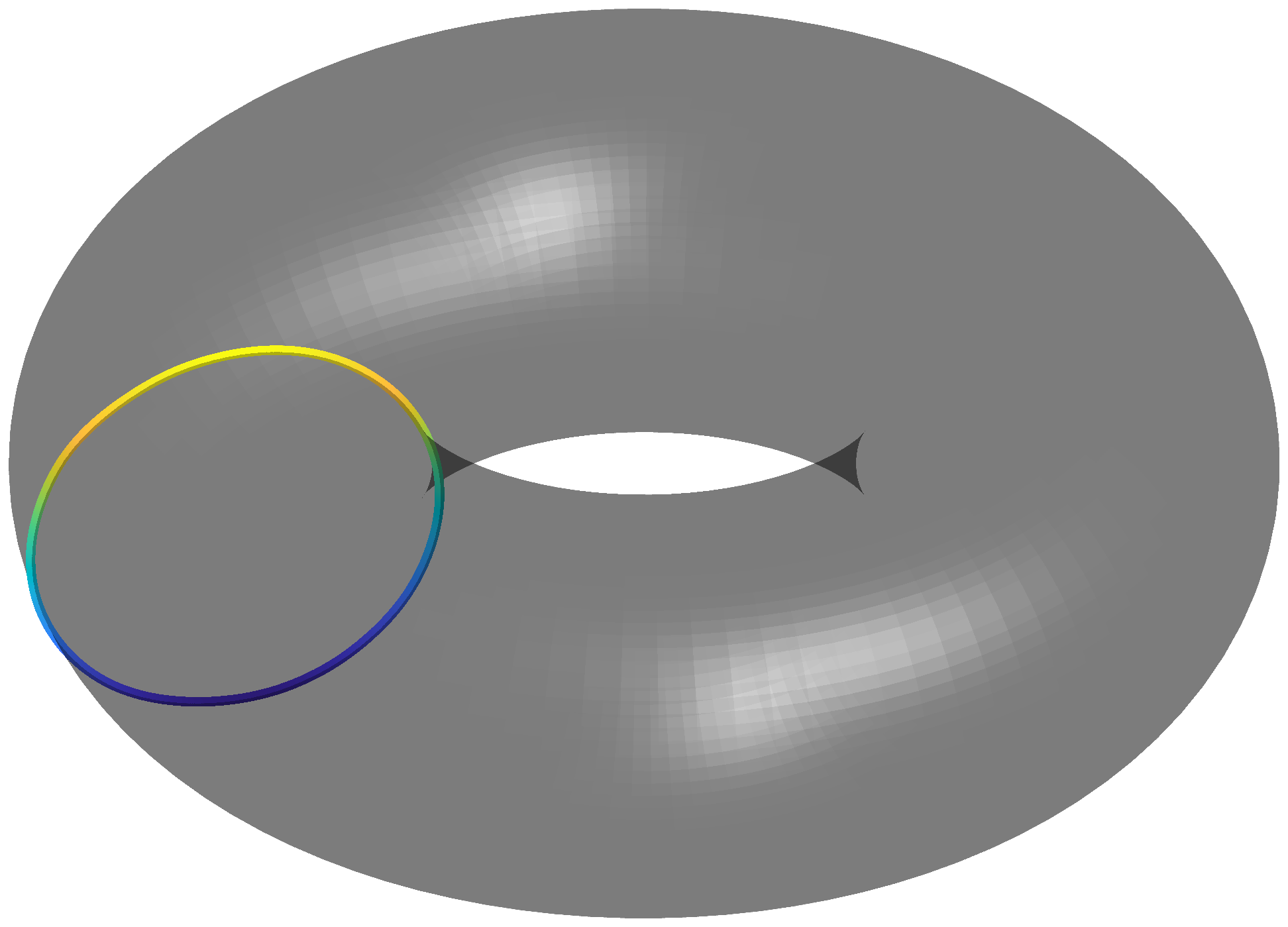}}
    \subfigure{\includegraphics[width=0.25\textwidth]{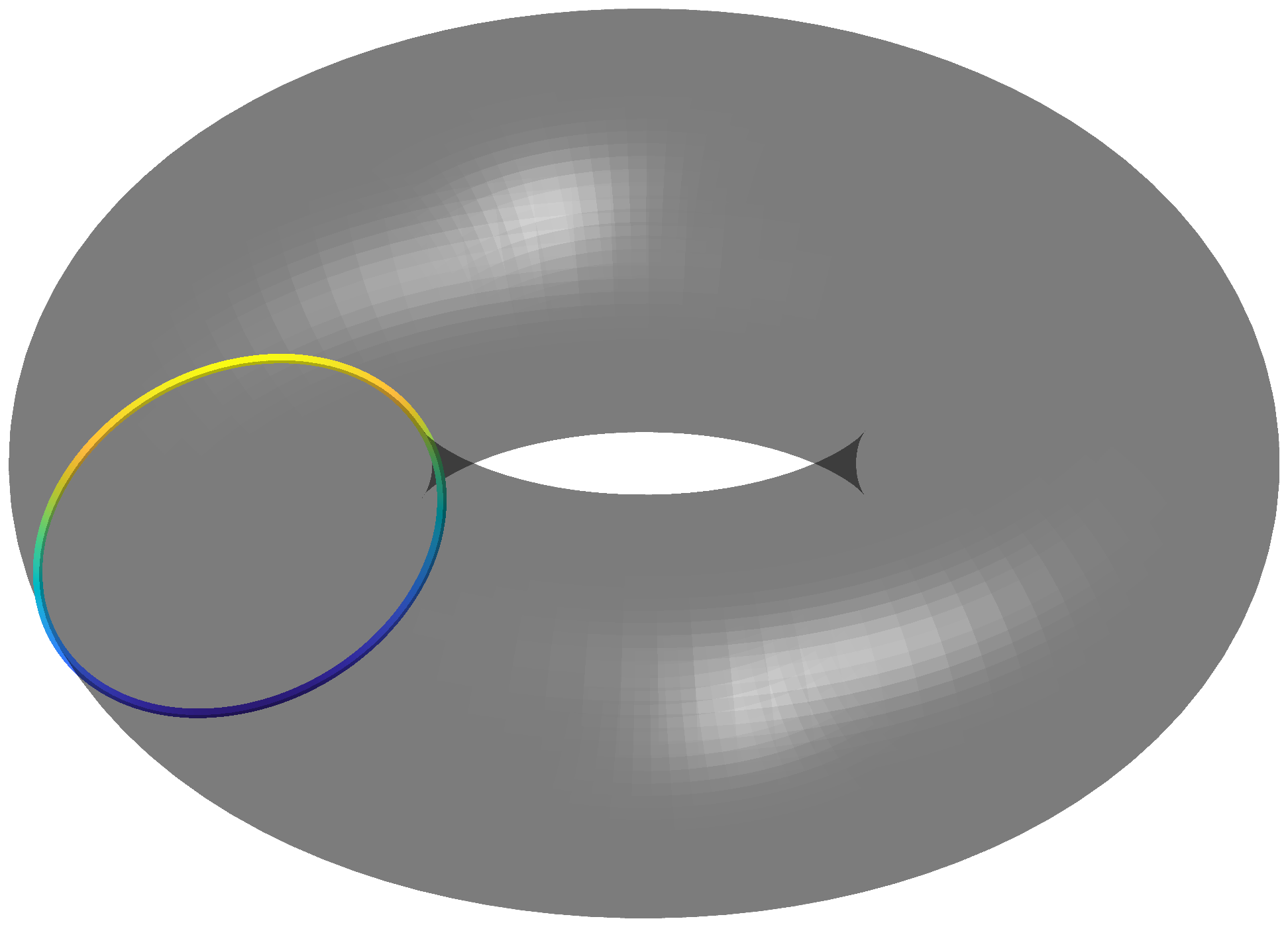}}\\
    \subfigure{\includegraphics[width=0.25\textwidth]{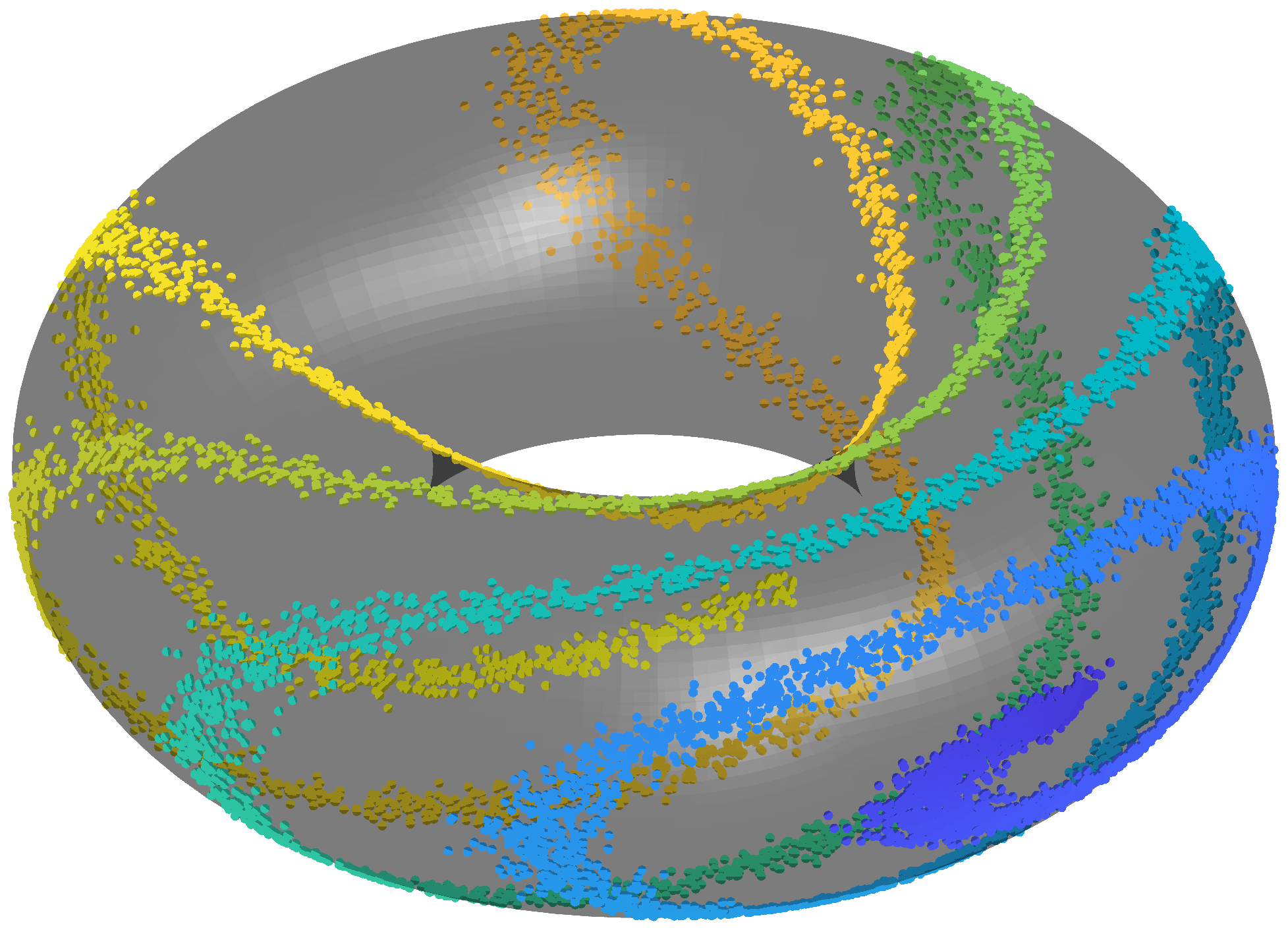}}
    \subfigure{\includegraphics[width=0.25\textwidth]{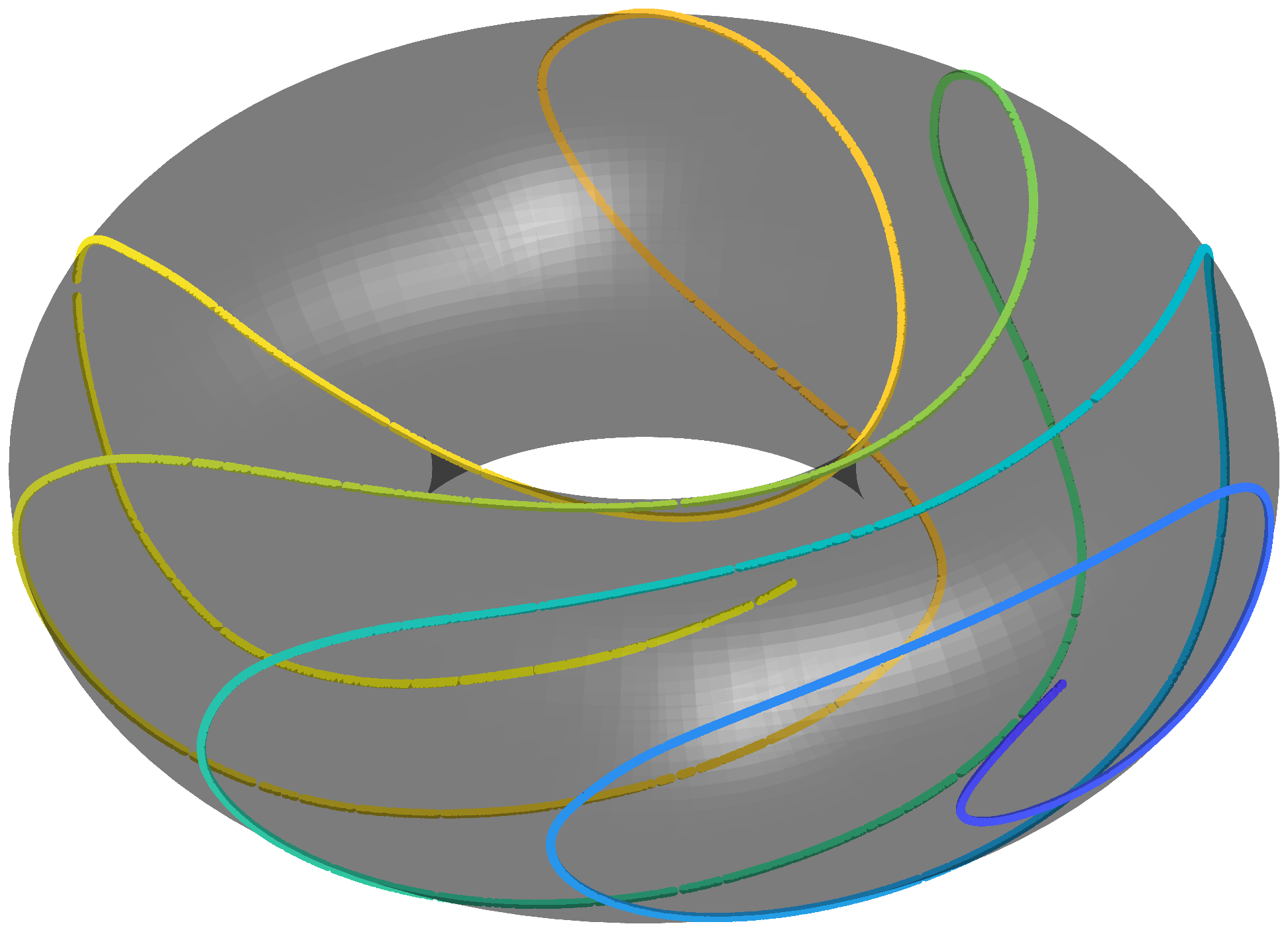}}
    \subfigure{\includegraphics[width=0.25\textwidth]{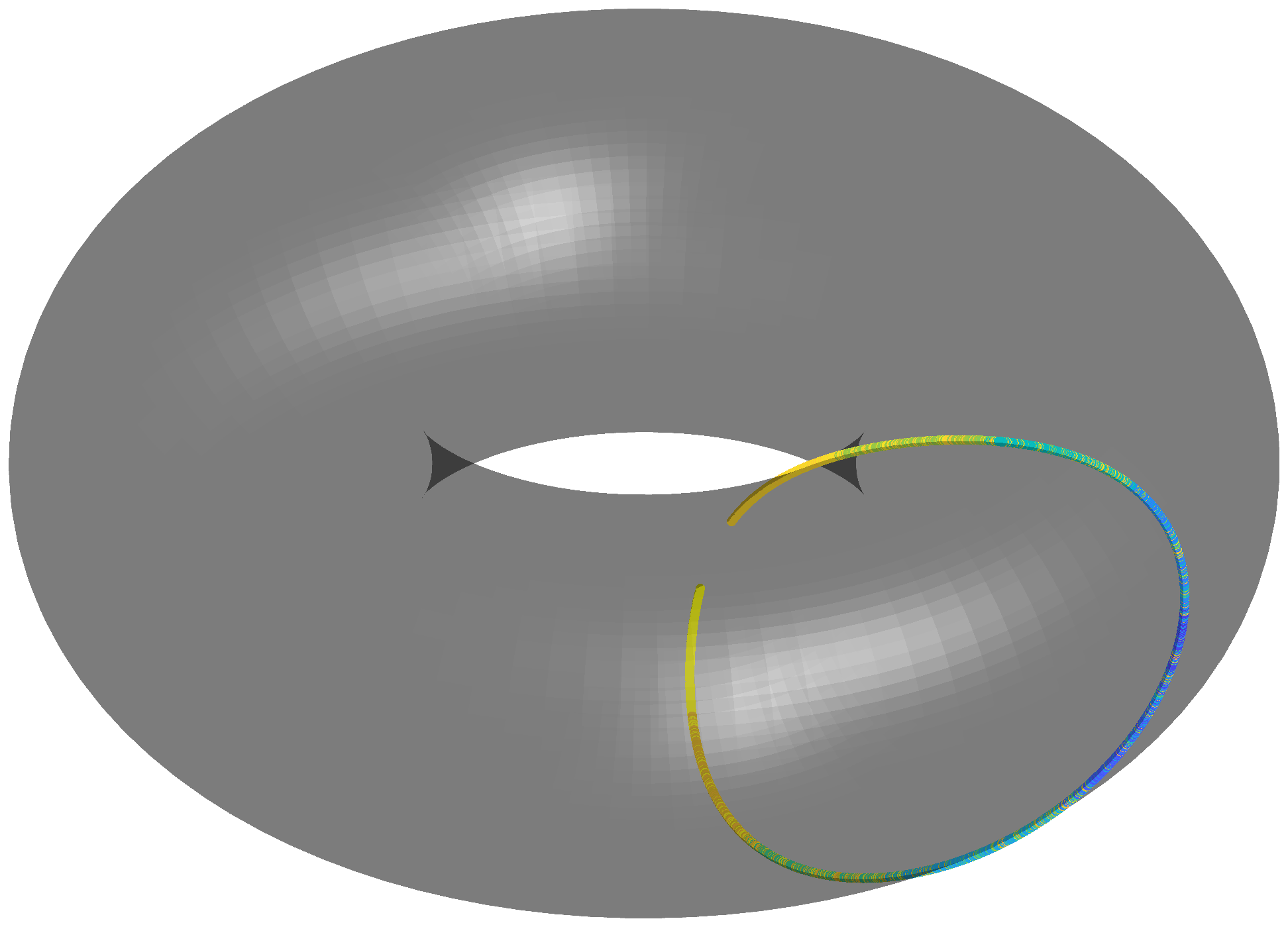}}\\
    \caption{Scatter plots for the torus case with colors added to enhance the visualization of sample adjacency. Each row corresponds to one generative curve. The left column panels display the input data; the middle panels show the results with principal nested submanifolds, and the right column presents the results from the torus principal component analysis.}
    \label{fig:Simu-Torus}
\end{figure}

Then, we investigate the decomposition of the data on torus. In this case, samples are generated with three curves
\[
\begin{aligned}
    \gamma_1(t) &= \left(t,0\right)^\top, \quad t\in(0,2\pi);\\
    \gamma_2(t) &= \left(-2\pi/3,t\right)^\top, \quad t\in(0,2\pi);\\
    \gamma_3(t) &= \left(\frac{t}{10}\cos(t),~\frac{t}{10}\sin(t)\right)^\top, \quad t\in\left(\frac{\pi}{2},\frac{19\pi}{2}\right),
\end{aligned}
\]
with the data generating function in \eqref{eq:simu-shapes-generating}, a set of $n=10^4$ noisy angle pairs in the angle plane for each case, and these angle pairs are embedded into a torus in $\bR^4$ by
\[
x_i = \bigl( \cos(\phi_i),~\sin(\phi_i),~\cos(\psi_i),~\sin(\psi_i)\bigl)^\top, \quad i=1,\dots,n.
\]
For better visualization, all the results are transformed back to the angles and further embedded in a torus in $\bR^3$ by 
\[
(\phi,\psi) \mapsto \bigl(\{1 + 0.5\cos(\psi)\}\cos(\phi),~\{1 + 0.5\cos(\psi)\}\sin(\phi),~0.5\sin(\psi)\bigl)^\top.
\]
The resulting data sets are shown in the left-most column of Figure \ref{fig:Simu-Torus}, which resemble two circles and an involute on the torus respectively. Then, the proposed method is applied to these data sets in $\bR^4$ to project them on a one-dimensional submanifold while the torus principal component analysis is directly applied on the angle pairs.

The projection results of our method are shown in the middle column of Figure \ref{fig:Simu-Torus}. In each case, the points projected by our method are approximately on a curve that passes through the center of the sample, and the neighbourhood relationship between the samples is well maintained. This suggests that our method can fit submanifolds on torus to explain the majority of sample variation. In contrast to our approach, the torus principal component analysis only projects the samples on circles cut from the torus, as shown in the last column of Figure \ref{fig:Simu-Torus}, which mixes points from different regions of the sample. This difference is also inherited from principal component analysis and the principal nested spheres, and highlights the descriptive flexibility of the present construction on the torus.

In summary, these examples illustrate how our method decomposes and performs dimension reduction on data around low-dimensional structures in spheres and tori. This is achieved by adaptively exploiting the local sample covariance structure and using manifold fitting methods. In contrast, methods like principal nested spheres and torus principal component analysis, which also project samples onto subspaces of specific dimensions, impose restrictions on the form of subspaces. This can lead to a less informative geometric summary when samples from different parts are mixed together.

\section{Application to single-cell RNA sequencing data}
In this section, we investigate the application of the proposed method to the single-cell RNA-sequencing data set reported by \citet{scRNAData}. This data set originates from the small intestinal epithelium of six mice. The processing of the samples is detailed in \citep{scRNAData}. In summary, the epithelium tissues undergo cell isolation, crypt isolation, cell sorting, and RNA sequencing, resulting in an expression matrix of 7,216 cells across 15,971 genes. The data is preprocessed by removing lowly expressed genes from individual cells and performing dimension reduction on highly expressed genes with principal component analysis following a logarithmic transformation.
 This initial PCA representation follows the preprocessing pipeline of \citet{scRNAData} and is used here as the input representation for subsequent geometric analysis, rather than as a claim that it provides a unique biological description of the data. 
The 13 significant components are then used to construct a 200-nearest-neighbour graph, which serves as the input for Infomap \citep{infomap}. This process clusters the cells into the 15 groups listed in Table~\ref{tab:Cell-count}. These assignments are treated as reference annotations for interpretation, not as supervised ground truth.
At the same time, single-cell RNA-sequencing data are often organized not only by discrete cell classes but also by continuous developmental relationships. We therefore use the 13 components to examine how annotation-associated patterns and continuous biological organization appear across a sequence of dimensions.

\begin{figure}[htbp]
    \centering
    \subfigure[]{\includegraphics[height=0.42\textwidth]{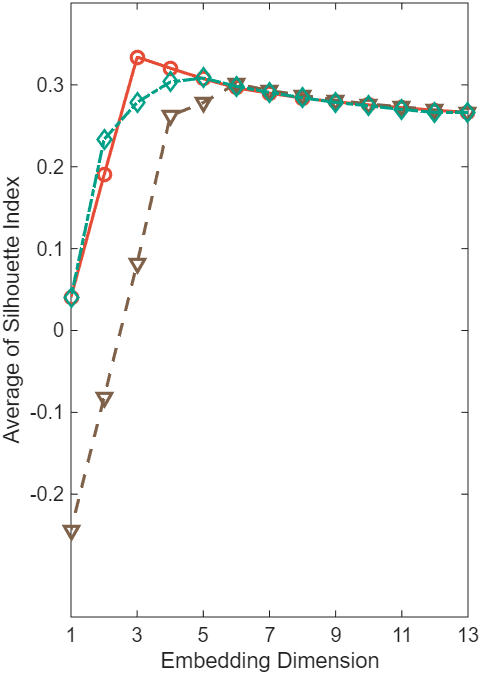}\label{fig:Real-metric-SI}}
    \subfigure[]{\includegraphics[height=0.42\textwidth]{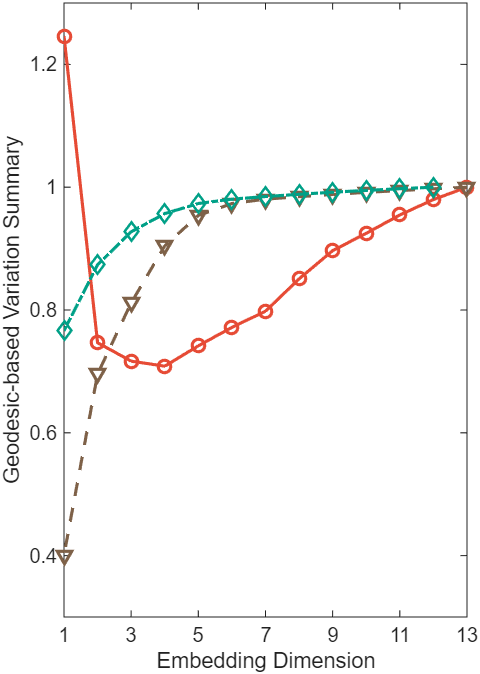}\label{fig:Real-metric-var}}
    \subfigure[]{\includegraphics[height=0.42\textwidth]{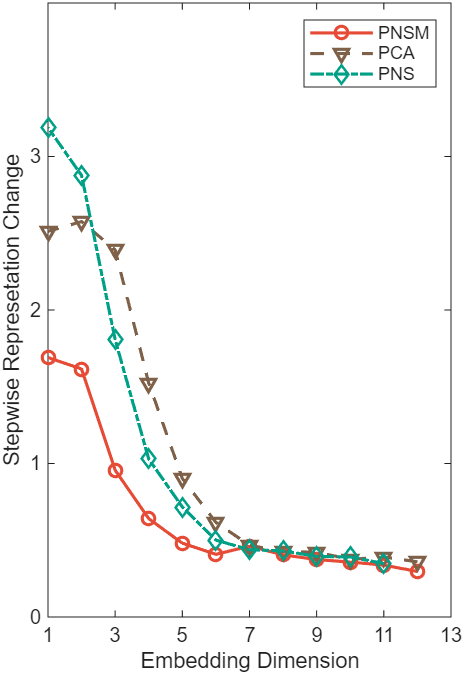}\label{fig:Real-metric-mse}}
    \caption{The average silhouette index (a), geodesic-based variation summary across dimensions (b), and stepwise representation change between consecutive levels of the hierarchy (c) for principal nested submanifolds (PNSM), principal component analysis (PCA), and principal nested spheres (PNS)}
    \label{fig:Real-metric}
\end{figure}

Prior to dimensionality reduction and subsequent analysis, the data undergoes further processing. To support stable estimation of local principal directions, we calculate pairwise Euclidean distances, construct neighbourhoods with $r=6$, and remove points with fewer than 25 neighbours. This refinement leaves 6,528 samples for further analysis. Their distribution is summarized in the Appendix and illustrated through a two-dimensional t-distributed stochastic neighbour embedding in Figure~\ref{fig:Real-tsne}(a). Compared with the original sample set (Figure 1b of \citealt{scRNAData}), this support filter predominantly excludes Enteroendocrine and Tuft cells, which are farther from the majority and more broadly dispersed. The data-specific radius calibration and post-selection diagnostics are reported in Section~\ref{sec:radius-diagnostics}.

\subsection{Radius calibration and stability diagnostics}
\label{sec:radius-diagnostics}

The radius $r$ is a resolution parameter that forms part of the population target, so different radii generally define different scale-specific hierarchies. At a fixed radius, a nonzero relevant eigengap concerns population directional identifiability. Finite-sample reliability additionally depends on neighbourhood support and on local covariance error being small relative to that gap; nearby-radius projector changes provide an empirical stability check. Small radii can consequently be under-supported and unstable, whereas large radii pool across changing local covariance structure and lose locality. Under an auxiliary manifold-plus-ambient-variation interpretation, sampling sparsity and effective ambient variation give heuristic lower-scale guidance, while the geometric scale over which the relevant eigenspaces change gives upper-scale guidance; curvature or reach may help interpret the latter. These diagnostics are formulated directly from the observed distribution and neither require nor estimate a separate noise level or latent reach.

For the single-cell RNA analysis, let $\rho_i(k)$ be the distance from raw observation $i$ to its $k$th nonself neighbour. With $k=25$, target coverage $\tau=0.90$, and candidate grid $\mathcal R=\{3,3.5,\ldots,9\}$, we used the geometry-only calibration
\begin{equation}
\widehat r
=
\min\left\{
r\in\mathcal R:
\frac{1}{7216}\sum_{i=1}^{7216}
\ID\{\rho_i(25)\leq r\}\geq 0.90
\right\}.
\label{eq:radius-calibration}
\end{equation}
This selected $r=6$, the smallest grid value reaching the target: $6528/7216=0.9047$ of the raw observations met the support criterion, and the empirical 90th percentile of the 25th-neighbour distance was 5.767. Here $k=25$ is a finite-sample support guard, approximately twice the ambient dimension, while 90\% coverage prevents a sparse tail from forcing a much larger radius. These choices and the grid are practical design decisions for this analysis, not universal thresholds or an optimality result.

\begin{figure}[htbp]
    \centering
    \includegraphics[width=\textwidth]{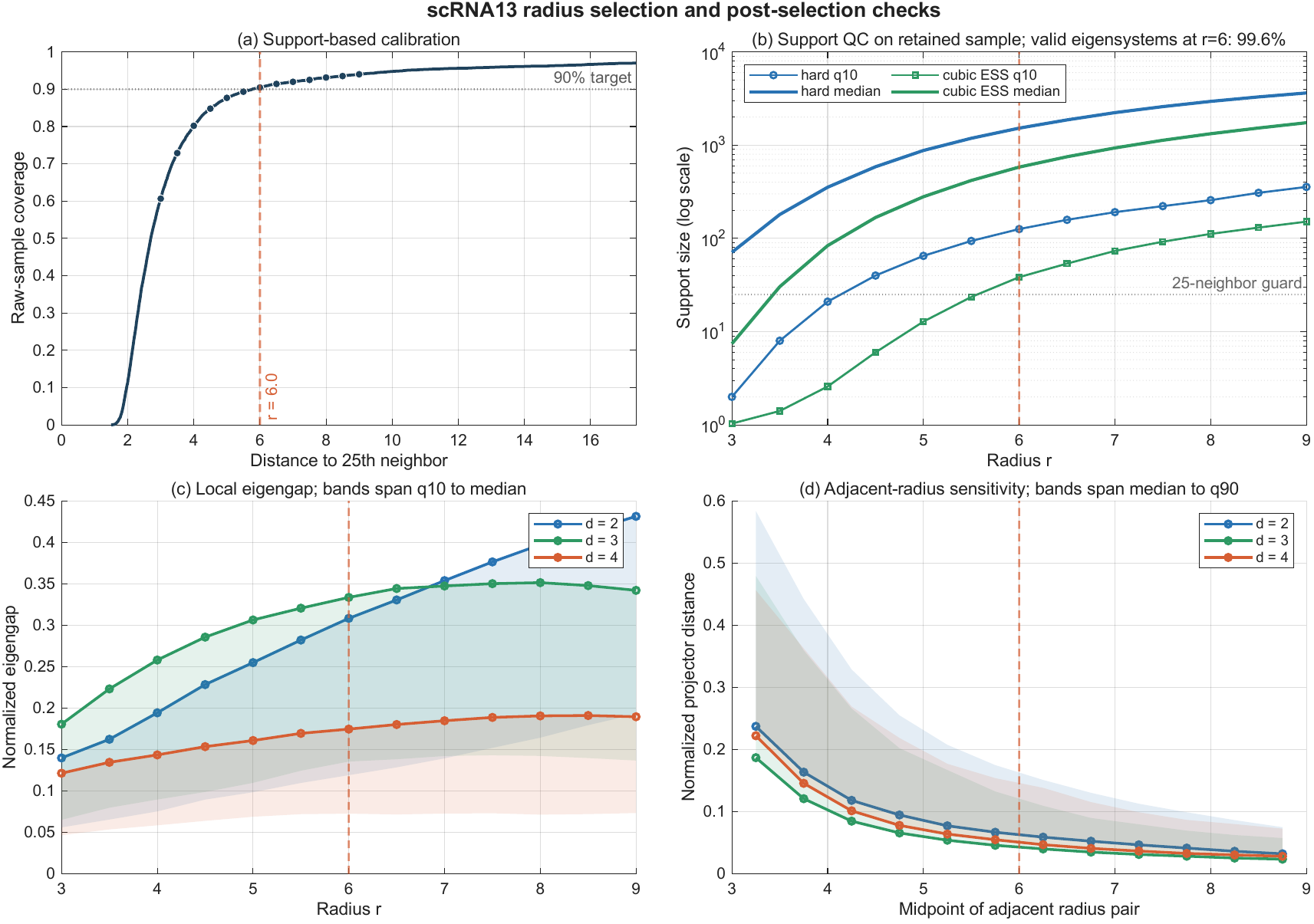}
    \caption{Radius calibration and post-selection checks for the single-cell analysis. (a) Fraction of raw observations whose 25th nonself neighbour lies within each candidate radius; the horizontal line marks the 90\% target and the vertical line marks the selected $r=6$. (b) Hard-neighbourhood support and cubic-weight effective sample size among retained observations. (c) Normalized local eigengaps. (d) Changes in local normal-space projectors between adjacent radii. Panels (b)--(d) are conditional diagnostics after selection, not a universal radius rule.}
    \label{fig:radius-selection}
\end{figure}

We next performed post-selection checks on the retained sample. At $r=6$, the 10th percentile and median neighbourhood sizes were 126 and 1523, respectively, and valid local eigensystems were obtained for 6501 of the 6528 observations. At target dimension $d$, the normalized eigengap is $(\lambda_{D-d+1}-\lambda_{D-d})/(\lambda_{D-d+1}+\lambda_{D-d})$. For adjacent radii, we compared the corresponding local normal-space projectors using their Frobenius distance normalized by its maximal value, $\{2\min(D-d,d)\}^{1/2}$. Table~\ref{tab:radius-qc} reports normalized eigengaps and projector changes around the selected radius.

\begin{table}[htbp]
\centering
\tbl{Post-selection radius diagnostics. Eigengap entries are the 10th percentile and median; projector-change entries are the median and 90th percentile across observations.}
{\begin{tabular}{lccc}
\toprule
Diagnostic & $d=2$ & $d=3$ & $d=4$ \\
\midrule
Normalized eigengap at $r=6$ & 0.1190 / 0.3082 & 0.1354 / 0.3335 & 0.0726 / 0.1746 \\
Projector change, $5.5\to6$ & 0.0669 / 0.1749 & 0.0459 / 0.1325 & 0.0549 / 0.1539 \\
Projector change, $6\to6.5$ & 0.0590 / 0.1510 & 0.0401 / 0.1094 & 0.0469 / 0.1383 \\
\bottomrule
\end{tabular}}
\label{tab:radius-qc}
\end{table}

The support summaries are conditional checks after applying \eqref{eq:radius-calibration}, not independent evidence that $r=6$ is uniquely preferable. Over the examined grid, comparison with global principal-component normal spaces did not indicate collapse of the local normal-space structure to a global linear summary. More generally, neighbourhood support, stability of leading local eigendirections or eigengaps, sensitivity to nearby radii, and preservation of nonlinear local structure should be considered jointly when calibrating the scale.

\begin{figure}[htbp]
    \centering
   \includegraphics[width=0.9\textwidth]{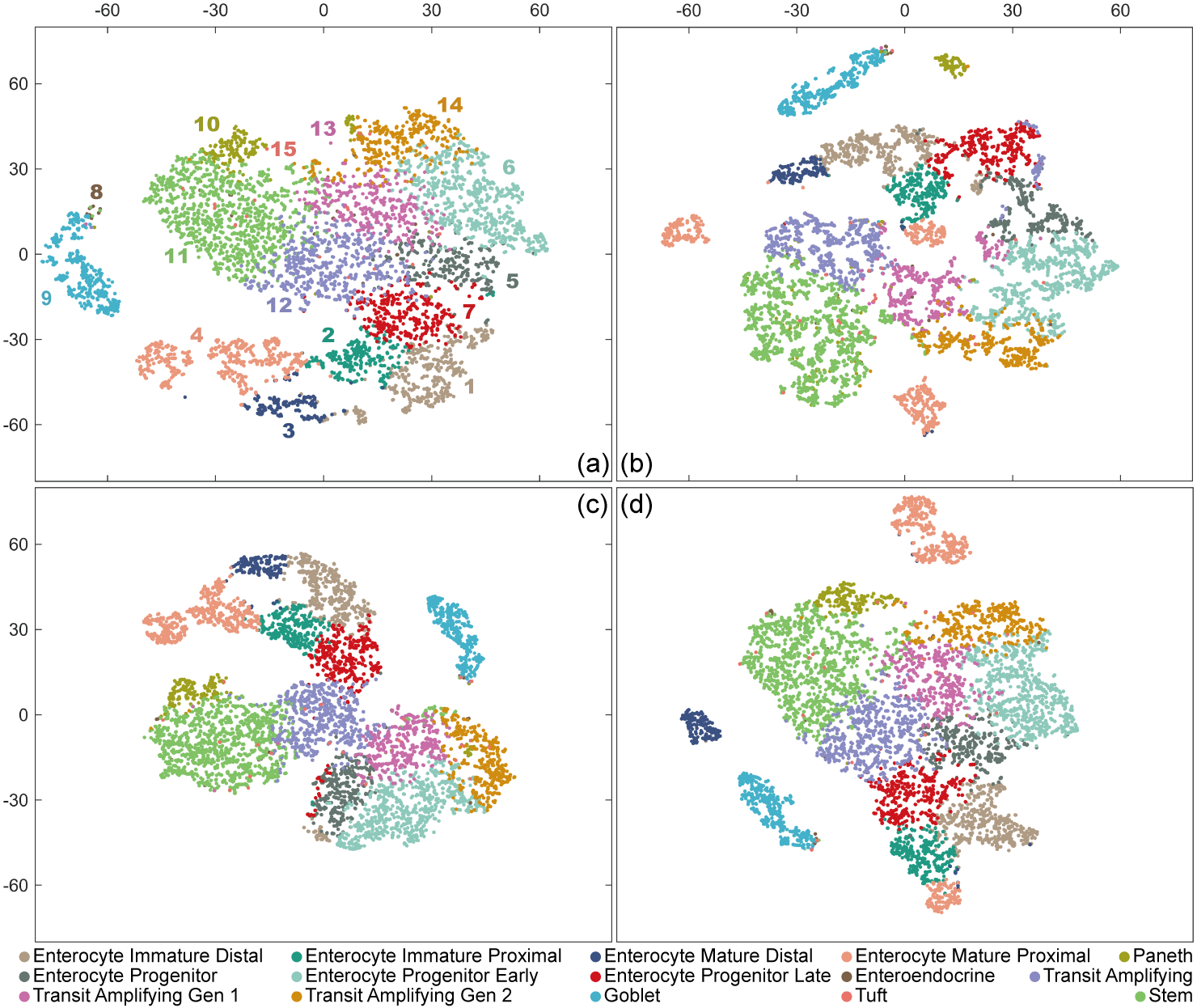}
    \caption{Scatter plots for the two-dimensional t-distributed stochastic neighbour embedding of (a) cleaned data set; (b) the principal nested submanifolds with $d=3$; (c) principal component analysis with $d=6$; (d) the principal nested spheres with $d=5$.}
    \label{fig:Real-tsne}
\end{figure}

To further examine the low-dimensional geometric organization of the cleaned data set, we apply our proposed method, principal component analysis, and principal nested spheres. In our method, we set the radius to $r=6$ and sequentially fit submanifolds with codimensions ranging from 1 to 12. Here $r$ controls the locality, or smoothing scale, of the construction, and the present choice retains local neighbourhoods with sufficient support for stable local covariance estimation while remaining local in the data cloud. Neighbourhood support, local eigendirection and nearby-radius stability diagnostics in Section~\ref{sec:radius-diagnostics} support this empirical choice without treating it as optimal. Different fixed values of $r$ correspond to different smoothing scales and hence to different nested hierarchies. For principal component analysis, we compute the covariance matrix of the samples and project the samples onto linear subspaces of dimensions 1 to 12. To apply principal nested spheres, we normalize each sample, use the unit vectors on $\cS^{12}$ as inputs for the algorithm to fit nested spheres with dimensions 1 to 11, and then restore the norms of the output vectors to their original values. Eventually, we obtain projections of the samples in various dimensions using different methods. Rather than searching for a single optimal dimension, we use these representations to study how broad annotation-associated organization and within-group geometric variation appear at different levels of the hierarchy.

In the dimension reduction process, we summarize the resulting representations using three quantities: the average silhouette index, a geodesic-based variation summary, and the stepwise representation change at each projection step. For principal component analysis and principal nested spheres, the plotted variation summaries are the corresponding Euclidean or embedding-based analogues; for principal nested submanifolds, the quantity is instead a graph/geodesic-based summary estimated from the minimum sum of squared shortest paths on the nearest-neighbour graph of the samples as an approximation to geodesic distances. Consequently, this quantity need not be bounded by one or monotone in the target dimension and is used here as a descriptive summary of the fitted hierarchy rather than as a strict variance decomposition. The stepwise representation change is likewise descriptive: it averages the Euclidean change between consecutive levels of the nested hierarchy and should not be interpreted as a standard reconstruction mean squared error. Figure~\ref{fig:Real-metric-SI} shows that the average silhouette index for our method initially increases, reaches a peak at $d=3$, and then declines. The curves for the other methods similarly rise and fall, but peak at higher dimensions. We use these summaries to identify method-specific levels at which the reference annotations are most visible, not to rank the methods or identify a uniquely best dimension. More generally, different levels of a hierarchy can emphasize different aspects of biological organization while remaining mutually compatible.

Figures~\ref{fig:Real-metric-var} and \ref{fig:Real-metric-mse} display the geodesic-based variation summary in each subspace and the stepwise representation change between consecutive levels, respectively. For our method, the curve is approximately linearly decreasing at higher dimensions, with a marked shift between dimensions 3 and 4 that results in a substantial increase at lower dimensions. As noted above, this graph/geodesic-based quantity is best read as a descriptive summary of how the hierarchy organizes the data across dimensions. In contrast, the other two methods exhibit only minor changes at the initial dimensions but experience more substantial changes when the dimensionality is below 5. All three methods show gradually increasing stepwise representation changes as the dimension is reduced. Section~\ref{sec:scrna-allocation} reports sample-wise residual-allocation diagnostics for principal nested submanifolds and principal component analysis.

Figure~\ref{fig:Real-tsne} shows separately computed two-dimensional t-distributed stochastic neighbour embeddings. These qualitative layouts are not coordinate-comparable across panels or dimensions. Additional layouts in the Appendix carry the same caveat; residual-allocation diagnostics provide more direct adjacent-level comparisons.

\subsection{Residual allocation across adjacent levels}
\label{sec:scrna-allocation}

For each retained cell, let
\[
R_{i,d}=\|\Delta_{i,d}\|,
\qquad
S_{i,d}=\frac{R_{i,d}}{\sum_{\ell}R_{i,\ell}}.
\]
Thus, $S_{i,d}$ is the relative contribution of transition $d+1\to d$ to that cell's total sequential path movement. All cells had positive total movement, and the shares summed to one up to numerical precision. Because adjacent principal nested submanifold increments need not be orthogonal, these shares are neither an energy decomposition nor explained variance. For example, the 10th percentile, median, and 90th percentile of
\[
\frac{\sum_d\|\Delta_{i,d}\|^2}{\left\|\sum_d\Delta_{i,d}\right\|^2}
\]
were 0.732, 0.935, and 1.124 for principal nested submanifolds, compared with 1.000, 1.000, and 1.000 for principal component analysis.

\begin{figure}[htbp]
    \centering
    \includegraphics[width=\textwidth]{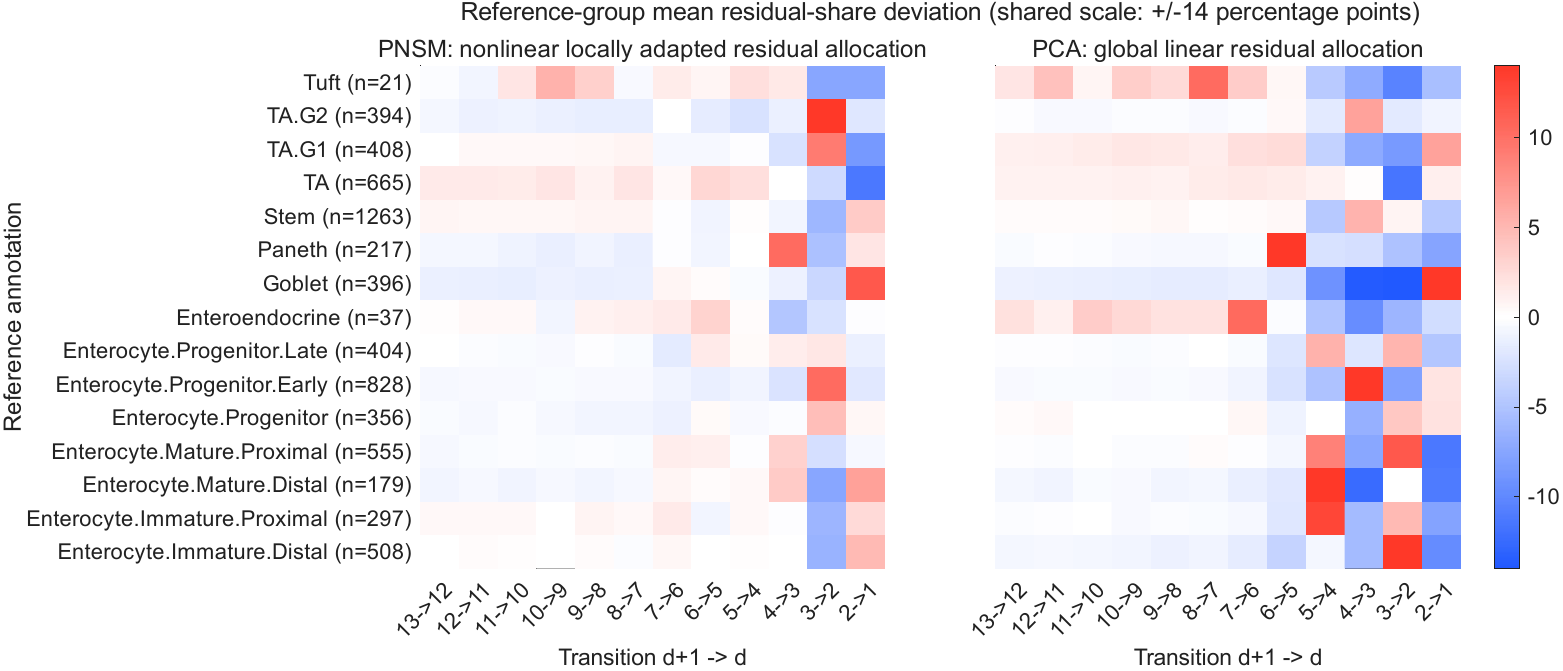}
    \caption{Reference-annotation group means minus the method-specific overall mean of the within-cell residual share, in percentage points, across adjacent-level transitions. Principal nested submanifolds give a nonlinear locally adapted allocation (left), while principal component analysis gives a global linear allocation (right). The shared colour scale is clipped at $\pm14$ percentage points. These shares allocate sequential path movement and are not explained variance or an energy decomposition.}
    \label{fig:residual-allocation-heatmaps}
\end{figure}

The reference annotations are used only to interpret how this allocation varies among biologically described groups. Table~\ref{tab:scrna-allocation} summarizes classical one-way eta-squared associations between the shares and the annotations. The strongest principal nested submanifold association occurs at $3\to2$. Principal component analysis also has strong annotation-associated residual structure, with a larger maximum eta-squared at $2\to1$. The effect sizes describe different allocations and are not a ranking of the methods.

\begin{table}[htbp]
\centering
\tbl{Association between within-cell residual allocation and the reference annotations.}
{\begin{tabular}{lccc}
\toprule
Method & Strongest transition & Maximum eta-squared & Eta-squared at $3\to2$ \\
\midrule
PNSM & $3\to2$ & 0.2577 & 0.2577 \\
PCA  & $2\to1$ & 0.8087 & 0.5834 \\
\bottomrule
\end{tabular}}
\label{tab:scrna-allocation}
\end{table}

At the principal nested submanifold $3\to2$ transition, Table~\ref{tab:scrna-groups} reports selected deviations from the method-specific overall mean share. The intervals use 1000 within-annotation bootstrap resamples.

\begin{table}[htbp]
\centering
\tbl{Selected reference-annotation deviations in the principal nested submanifold $3\to2$ allocation of sequential path movement, in percentage points.}
{\begin{tabular}{lrrc}
\toprule
Reference annotation & $n$ & Difference & Bootstrap 95\% interval \\
\midrule
Transit Amplifying Gen 2       & 394 &  14.385 & [12.625, 15.905] \\
Enterocyte Progenitor Early    & 828 &  10.357 & [9.464, 11.302] \\
Transit Amplifying Gen 1       & 408 &   9.286 & [7.950, 10.528] \\
Enterocyte Mature Distal       & 179 &  -7.464 & [$-9.018$, $-5.812$] \\
Tuft                           &  21 &  -7.441 & [$-10.987$, $-3.695$] \\
Enterocyte Immature Distal     & 508 &  -6.453 & [$-7.259$, $-5.601$] \\
\bottomrule
\end{tabular}}
\label{tab:scrna-groups}
\end{table}

These results identify a transition-specific pattern associated with the reference annotations; they do not treat the annotations as supervised ground truth or identify a true intrinsic dimension. The small retained Tuft ($n=21$) and Enteroendocrine ($n=37$) groups require particular caution. Overall, principal component analysis supplies a strong global linear nested comparator, whereas principal nested submanifolds supply a complementary nonlinear and locally adapted residual hierarchy.

\subsection{Same-machine wall-clock runtime}
\label{sec:runtime}

The MATLAB prototype is the original implementation used to generate the reported experimental results. The optimized CPU implementation carries out the same procedure using CPU-oriented optimization and parallelization, whereas the optimized GPU implementation uses GPU acceleration for the same local geometric computations. Thus, the three versions differ in computational environment and parallelization strategy, not in the underlying principal nested submanifold procedure.

Table~\ref{tab:runtime-scrna} reports wall-clock runtimes for the single-cell RNA analysis with $n=6528$, ambient dimension 13, and $r=6$. The runs used the same Windows workstation with an Intel Core Ultra 7 265KF processor, 20 logical processors, and an NVIDIA GeForce RTX 5070 GPU under MATLAB R2026a Update 2. Data loading, plotting, parallel-pool initialization, and post-fit variation summaries were excluded. The principal component analysis and principal nested sphere entries are rounded averages of three timed runs; each principal nested submanifold entry is a completed full run.

\begin{table}[htbp]
\centering
\tbl{Same-machine wall-clock runtime for the single-cell RNA analysis.}
{\begin{tabular}{lll}
\toprule
Method & Implementation & Wall-clock runtime \\
\midrule
PCA  & MATLAB & 0.44 seconds \\
PNS  & MATLAB & 2.04 seconds \\
PNSM & MATLAB prototype & 12.36 hours \\
PNSM & optimized CPU implementation & 47 minutes \\
PNSM & optimized GPU implementation & 13 minutes \\
\bottomrule
\end{tabular}}
\label{tab:runtime-scrna}
\end{table}

Table~\ref{tab:runtime-simulations} gives the corresponding timings for the nine simulation cases, each with $n=10{,}000$ and $r=0.5$. The Euclidean and spherical cases have ambient dimension 3, and the torus cases use the four-dimensional Euclidean embedding. The applicable comparator is principal component analysis for Euclidean cases, principal nested spheres for spherical cases, and torus principal component analysis for toroidal cases. In the table, Prototype, CPU, and GPU refer to the MATLAB prototype, optimized CPU implementation, and optimized GPU implementation, respectively. The MATLAB prototype was timed for one complete run per case; comparator, CPU, and GPU entries are means of three timed runs after warm-up.

\begin{table}[htbp]
\centering
\tbl{Same-machine wall-clock runtime for the nine simulation cases, in seconds.}
{\setlength{\tabcolsep}{3.5pt}
\begin{tabular}{@{}llrrrr@{}}
\toprule
\multirow{2}{*}{Simulation} & \multirow{2}{*}{Comparator} & \multicolumn{4}{c}{Wall-clock time (s)} \\
\cmidrule(lr){3-6}
 & & Comparator & Prototype & CPU & GPU \\
\midrule
Euclidean line          & PCA  & 0.375 & 17.986 & 4.809 & 18.579 \\
Euclidean circle        & PCA  & 0.365 & 14.449 & 4.356 & 18.667 \\
Euclidean involute      & PCA  & 0.362 & 25.444 & 4.827 & 19.070 \\
Sphere great circle     & PNS  & 0.340 & 11.378 & 2.725 & 11.612 \\
Sphere tennis           & PNS  & 0.376 &  9.480 & 2.428 & 11.574 \\
Sphere involute         & PNS  & 0.329 &  8.835 & 2.532 & 11.546 \\
Torus T1                & tPCA & 0.843 & 57.801 & 4.314 & 14.373 \\
Torus T2                & tPCA & 0.861 & 57.380 & 4.444 & 14.131 \\
Torus involute          & tPCA & 1.149 & 39.331 & 2.351 & 13.953 \\
\bottomrule
\end{tabular}}
\label{tab:runtime-simulations}
\end{table}

For the simulation suite, data generation, plotting, one-time initialization, and post-fit variation summaries were excluded. The MATLAB prototype used a 20-worker process pool, and the optimized CPU implementation used 20 CPU threads.

These timings are implementation- and hardware-dependent wall-clock comparisons rather than language-independent measures of method complexity. Because principal nested submanifolds involve iterative local geometric computations, their runtime is more affected by parallelization and hardware configuration than that of linear methods such as principal component analysis. The optimized implementations preserve the same underlying procedure; software optimization is not a separate statistical contribution of this work.

\section{Discussion}
In this paper, we study nonlinear dimension reduction and data decomposition through the principal nested submanifolds. Rather than producing only a nonlinear representation at one chosen target dimension, the method constructs a nested geometric hierarchy across dimensions, extending the layered interpretation of principal component analysis to a nonlinear setting. We show that the population principal submanifolds are smooth submanifolds of the ambient space and provide algorithms for fitting them. Through several simulation studies and an analysis of a real-world single-cell RNA sequencing data set, we illustrate the interpretability and potential applicability of the resulting hierarchy across various fields.

The principal nested submanifolds approach is inspired by the manifold fitting framework and uses local covariance information to define each principal nested submanifold as a level set of a smooth function. These level sets are naturally nested and, under mild conditions, form low-dimensional manifolds in the ambient space. The population object is defined directly through the observed distribution together with the fixed localization scale $r$. Thus, the present paper formulates this hierarchy as a distribution-dependent statistical summary rather than as exact latent-manifold recovery or a universally canonical object. This yields a complementary geometric alternative to methods such as \citet{PNS}, \citet{tpca-aoas}, \citet{donnell1994analysis}, \citet{panaretos2014principal}, and \citet{yao2024principalsubmanifolds}.

In numerical studies, the principal nested submanifolds method provides an informative geometric summary. The simulation section uses specially designed cases in Euclidean spaces, spheres, and tori to illustrate how the method complements traditional principal component analysis, principal nested spheres, and torus principal component analysis. It captures complex structures and highlights the role of nested nonlinear geometry. In addition, when analysing single-cell RNA sequencing data, the principal nested submanifolds method suggests distinct patterns in the data while organizing them into a layered hierarchy across dimensions. This advances our understanding of nonlinear space decomposition and provides a fundamental framework for future advances in nonlinear embedding and complex data analysis. With potential applications in fields as diverse as neuroscience and machine learning, our framework offers an additional perspective on data structures and their dynamics.

Several challenges remain. First, the proposed method relies on iterative local geometric computations, so its runtime is more implementation- and hardware-dependent than that of linear methods such as principal component analysis; wall-clock runtimes are reported in Section~\ref{sec:runtime}. Second, unlike principal component analysis and principal nested spheres, nonlinear residuals are harder to interpret geometrically. Finally, each fixed choice of radius $r$ defines one nested hierarchy under a shared geometric scale, and it would be interesting to study how different values of $r$ correspond to different resolutions of the data, or how $r$ might vary across dimensions without losing a coherent nested interpretation.

\section*{Acknowledgement}
Zhigang Yao has been supported by the Singapore Ministry of Education Tier 2 grant (A-0008520-00-00 and A-8001562-00-00) and the Tier 1 grant (A8000987-00-00 and A-8002931-00-00) at the National University of Singapore. Much of this work was carried out while Jiaji Su was a postdoctoral researcher supported by grant A-8001562-00-00 at the National University of Singapore. Jiaji Su is also supported by the Fundamental Research Funds for the Central Universities (No. 0208005203020) at Chongqing University.

\section*{Appendices}
\setcounter{section}{1}
\setcounter{theorem}{0}
\setcounter{lemma}{0}
\setcounter{equation}{0}
\setcounter{definition}{0}
\setcounter{figure}{0}

\renewcommand{\thesection}{\Alph{section}}
\renewcommand{\thetheorem}{\thesection.\arabic{theorem}}
\renewcommand{\theequation}{\thesection.\arabic{equation}}
\renewcommand{\thelemma}{\thesection.\arabic{lemma}}
\renewcommand{\thedefinition}{\thesection.\arabic{definition}}
\renewcommand{\thefigure}{\thesection.\arabic{figure}}
\subsection{Proof omitted from the main text}
\begin{proof}
    Let $p$ be a fixed point on $\cM_{r,d}$, and $z\in\cB(p,\eta)$ be another point of interest. For the sake of simplicity, we denote the collection of principal directions at $z$ as $V_z = (v_{r,z,1},\dots,v_{r,z,D-d})$, the summation of projection matrices as 
    \[
        \Psi_z = V_z V_z^\top = \sum_{k=1}^{D-d}\Pi_{r,z,k},
    \]
    and the summation of bias vectors as the bias function 
    \[
        b(z) = \Psi_z (z - \mu_z) = \sum_{k=1}^{D-d}b_{r,k}(z),
    \]
    for all $z$ of interest. To prove Theorem \ref{thm:population}, we need the following lemmas.
    \begin{lemma}\label{lemma:1st-order-b}
        With Conditions 1--3 and the set up of Theorem \ref{thm:population}, the Jacobian matrix of $b(z)$ satisfies
        \[
            \|\nJ_b(z) - \Psi_z\| \leq (D-d) (r \ell_3 + \ell_5),
        \]
        for all $z\in\cB(p,\eta)$.
    \end{lemma}
    \begin{proof}
        Let $u$ be a unit vector in $\bR^D$, the directional derivative of $b(z)$ along $u$ is
    \[
        \partial_u b(z) = (\partial_u \Psi_z) (z - \mu_z) + \Psi_z u - \Psi_z (\partial_u\mu_z),
    \]
    which leads to
    \[
         \|\partial_u b(z) - \Psi_z u\| \leq  \|(\partial_u \Psi_z) (z - \mu_z)\| + \|\Psi_z (\partial_u\mu_z)\|.
    \]
    The first term can be upper bounded as
     \begin{align*}
        \|(\partial_u \Psi_z) (z - \mu_z)\|
        &\leq \| z - \mu_z\| \left\|\sum_{k=1}^{D-d} \partial_u \Pi_{r,z,k}\right\|\\
        &\leq \sum_{k=1}^{D-d} \| z - \mu_z\| \left\|\partial_u \Pi_{r,z,k} \right\|\\
        &\leq (D-d) r \ell_3,
    \end{align*}
    where the last inequality is due to the fact that $\mu_z \in \cB(z,r)$.
    For the second term, there is
    \begin{align*}
    \|\Psi_z (\partial_u\mu_z)\|
        &=  \left\|\sum_{k=1}^{D-d} v_{r,z,k}v_{r,z,k}^\top \nJ_\mu(z) u\right\|\\
        &\leq \sum_{k=1}^{D-d}  \|v_{r,z,k}\|\left\|v_{r,z,k}^\top \nJ_\mu(z) \right\| \|u\|\\
        &< (D-d) \ell_5,
    \end{align*}
    where $\nJ_\mu(z)$ denotes the Jacobian matrix of $\mu_z$ at point $z$. Then, by definition
    \begin{align*}
        \|\nJ_b(z) - \Psi_z\| 
        &= \sup_{u\in\bR^D} \frac{\|\nJ_b(z)u - \Psi_z u\|}{\|u\|}\\
        &< (D-d) (r \ell_3 + \ell_5),
    \end{align*}
    for all $z\in\cB(p,\eta)$, which completes the proof.
    \end{proof}
    
    \begin{lemma}\label{lemma:2nd-order-b}
        With Conditions 1--3 and the set up of Theorem \ref{thm:population}, for two unit vectors $u,u^\prime \in \bR^D$, 
        \[
            \|\partial_u \partial_{u^\prime} b(z)\| \leq (D-d) \bigl\{r\ell_4 + \ell_2 +2 \ell_3(1+\ell_1)\bigl\}
        \]
        for all $z\in\cB(p,\eta)$.
    \end{lemma}
    \begin{proof}
        By definition and the triangle inequality, an upper bound of $ \|\partial_u\partial_{u^\prime} b(z)\|$ is given by
        \begin{align*}
        \|\partial_u\partial_{u^\prime} b(z)\| 
        \leq& \|(\partial_u\partial_{u^\prime} \Psi_z) (z-\mu_z)\| + \|\Psi_z \partial_u\partial_{u^\prime}  (z-\mu_z)\| \\
        &+ \|(\partial_u \Psi_z) \partial_{u^\prime}(z-\mu_z)\| + \|(\partial_{u^\prime} \Psi_z) \partial_u(z-\mu_z)\|.
    \end{align*}
    
    For the first term,
    \begin{align*}
        \|(\partial_u\partial_{u^\prime} \Psi_z) (z-\mu_z)\|
        & = \left\| \sum_{k=1}^{D-d}  \partial_u\partial_{u^\prime} \Pi_{r,z,k}\right\| \| z - \mu_z\| \\
        &\leq \sum_{k=1}^{D-d} r\left\| \partial_u\partial_{u^\prime} \Pi_{r,z,k}\right\| \\
        &\leq (D-d) r \ell_4.
    \end{align*}
    For the second term,
    \begin{align*}
        \|\Psi_z \partial_u\partial_{u^\prime}  (z-\mu_z)\|
        &= \left\|\sum_{k=1}^{D-d} v_{r,z,k}v_{r,z,k}^\top (\partial_u \partial_{u^\prime}\mu_z)\right\|\\
        &\leq \sum_{k=1}^{D-d} \left\|v_{r,z,k}^\top (\partial_u \partial_{u^\prime}\mu_z)\right\|\\
        &\leq (D-d) \ell_2.
    \end{align*}
    For the last two terms,
    \begin{align*}
        \|(\partial_u \Psi_z) \partial_{u^\prime}(z-\mu_z)\|
        &=\|(\partial_u \Psi_z) (u^\prime-\partial_{u^\prime}\mu_z)\|\\
        &\leq \|\partial_u \Psi_z\| \bigl(1 + \|\partial_{u^\prime}\mu_z\|\bigl)\\
        &\leq (D-d)\ell_3 (1+\ell_1).
    \end{align*}
    Hence, 
    \[
        \|\partial_u\partial_{u^\prime} b(z)\| \leq (D-d) \bigl\{r\ell_4 + \ell_2 +2 \ell_3(1+\ell_1)\bigl\},
    \]
    for all $z\in\cB(p,\eta)$.
    \end{proof}
    \begin{lemma}[Constant-Rank Level Set Theorem (Theorem 5.12 of \citep{lee2012smooth})]
        Let $\mathcal{S}_1$ and $\mathcal{S}_2$ be smooth manifolds, and let $\Phi:\,\mathcal{S}_1 \to\mathcal{S}_2$ be a smooth map with constant rank $k$. Each level set of $\Phi$ is a properly embedded submanifold of codimension $k$ in $\mathcal{S}_1$.
    \end{lemma}
    \begin{lemma}[Implicit Function Theorem (Theorem C.40 of \citep{lee2012smooth})]
        Let $\mathcal U\subseteq \bR^{d_1} \times \bR^{d_2}$ be an open subset, and let $(s,t) = (s_1,\dots,s_{d_1}, t_1,\dots,t_{d_2})$ denote the standard coordinates on $\mathcal U$. Suppose $\Phi:\,\mathcal U\to \bR^{d_2}$ is a smooth function, $(s^*,t^*)\in \mathcal U$, and $\Phi(s^*,t^*) = 0$. If the $d_2\times d_2$ matrix 
        \[
            \left(\frac{\partial \Phi_i}{\partial t_j} (s^*,t^*)\right)
        \]
        is nonsingular, then there exist neighbourhoods $\mathcal{V}_1 \subseteq \bR^{d_1}$ of $s^*$ and $\mathcal{V}_2 \subseteq \bR^{d_2}$ of $t^*$ and a smooth function $\varphi:\, \mathcal{V}_1 \to \mathcal{V}_2$ such that $\Phi^{-1}(0) \cap \mathcal{V}_1 \times \mathcal{V}_2$ is the graph of $\varphi$, that is, $\Phi(s,t) = 0$ for $(x,y) \in \mathcal{V}_1 \times \mathcal{V}_2$ if and only if $t = \varphi(s)$.
    \end{lemma}
    After introducing these preparative lemmas, let us consider an ancillary function $h(z): \cB(p,\eta)\to\bR^{D-d}$ given by 
        \[
            h(z) = V_p^\top b(z) = V_p^\top \Psi_z(z - \mu_z).
        \]
    
    \vspace{5pt}
    \underline{We first prove for all $z\in\cB(p,\eta)$, $b(z) = 0$ if and only if $h(z) = 0$.}
    
    Since it is obvious that
    \[
    \{z:b(z) = 0\} \subset \{z:h(z) = 0\},
    \]
    it is sufficient to show the other direction.\\
    Assume there is a point $q\in \cB(p,\eta)$ such that $h(q) = 0$ but $b(q)\neq 0$. Then, there is
    \[
    \begin{aligned}
        \|\Psi_p - \Psi_q\| 
        &= \sup_{u\in\bR^D}\frac{\|(\Psi_p - \Psi_q)u\|}{\|u\|} \\
        &\geq \frac{\|(\Psi_p - \Psi_q)b(q)\|}{\|b(q)\|} \\
        &= \frac{\|V_pV_p^\top b(q)-\Psi_q b(q)\|}{\|b(q)\|}\\
        &= \frac{\|V_p h(q)-b(q)\|}{\|b(q)\|} \\
        &= \frac{\|0-b(q)\|}{\|b(q)\|}\\
        &= 1,
    \end{aligned}
    \]
    which, together with \eqref{eq:def-eta} and the fact that $\|p-q\|<\eta$, leads to the result that
    \begin{align*}
        \sup_{z \in \cB(p,\eta)} \|\partial_u \Psi_{z}\|
        &\geq \|\Psi_p - \Psi_q\| / \|p-q\|\\
        &\geq \eta^{-1}\\
        &\geq \frac{(D-d) \bigl\{r\ell_4 + \ell_2 +2 \ell_3(1+\ell_1)\bigl\}}{1 - (D-d) (r \ell_3 + \ell_5)}\\
        &>(D-d) \bigl\{r\ell_4 + \ell_2 +2 \ell_3(1+\ell_1)\bigl\}\\
        &>2(D-d)\ell_3
    \end{align*}  
    is contradictory to the assumption that 
    \begin{align*}
        \sup_{z \in \cB(p,\eta)} \|\partial_u \Psi_{z}\| 
        &\leq \sum_{k=1}^{D-d}  \sup_{z \in \cB(p,\eta)} \left\|\partial_u \Pi_{r,z,k} \right\|\\
        &\leq (D-d)\ell_3.
    \end{align*}
    Hence, $b(z)=0$ if and only if $h(z) = 0$ for $z\in\cB(p,\eta)$.

    \vspace{5pt}
    \underline{Next, we investigate the Jacobian matrix of $h(z)$.}
    
    For any $z \in \cB(p,\eta)$,
    \begin{align*}
        \nJ_h(z)
        & = V_p^\top \nJ_b(z)\\
        & = V_p^\top \{\nJ_b(z) - \nJ_b(p)\} + V_p^\top \{\nJ_b(p) - \Psi_p\} + V_p^\top \Psi_p.
    \end{align*}
    Since $V_p^\top \Psi_p = V_p^\top$, there is
    \begin{align*}
        \|\nJ_h(z) - V_p^\top\|
        &\leq  \|V_p^\top \{\nJ_b(z) - \nJ_b(p)\}\| + \|V_p^\top \{\nJ_b(p) - \Psi_p\}\|\\
        &\leq  \|\nJ_b(z) - \nJ_b(p)\| + \|\nJ_b(p) - \Psi_p\|\\
        &\leq \|z - p\| \sup_{\substack{z^\prime \in\cB(p,\eta) \\\|u\|= \|u^\prime\| =1}} \|\partial_u \partial_{u^\prime} b(z^\prime)\| + \sup_{z^\prime \in\cB(p,\eta)}\|\nJ_b(z^\prime) - \Psi_{z^\prime}\|\\
        &< \eta (D-d) \bigl\{r\ell_4 + \ell_2 +2 \ell_3(1+\ell_1)\bigl\}+ (D-d) (r \ell_3 + \ell_5)\\
        &<1,
    \end{align*}
    which suggests that the maximal difference between the singular values of $\nJ_h(z)$ and $V_p^\top$ is less than $1$. Because the column of $V_p$ are orthogonal unit vectors and thus its singular values are all $1$, there is $\rank\bigl( \nJ_h(z) \bigl) = D-d$. According to the \textit{Constant-Rank Level Set Theorem}, the level set 
    \[
    \{z\in\cB(p,\eta):b(z) = 0\} = \{z\in\cB(p,\eta):h(z) = 0\} = \cM_{r,d} \cap \cB(p,\eta)
    \]
    is a $d$-dimensional embedded submanifold of $\cB(p,\eta) \subseteq \bR^D$. Since this result holds for every $p\in\cM_{r,d}$, we can state that $\cM_{r,d}$ is topologically a $d$-dimensional submanifold of $\bR^D$.

    \vspace{5pt}
    \underline{Moreover, if $\cM_{r,d}$ is a connected set, $\reach(\cM_{r,d})$ can be bounded below.}
    
    Recall that, the reach of $\reach(\cM_{r,d})$ can be given by
    \[
        \reach(\cM_{r,d}) = \inf_{z,z^\prime \in \cM_{r,d}} \frac{\|z - z^\prime\|^2}{2\nd (z^\prime,T_z\cM_{r,d})}. 
    \]
    We first consider the case where the two points are close, that is, for $p\in\cM_{r,d}$, we consider another point $q \in\cM_{r,d}$ with $p\neq q$ and $\|p - q\| \leq \eta$. Since we have already shown that $\cB(p,\eta) \cap \cM_{r,d}$ homeomorphic to a $d$-dimensional Euclidean space, let the collection of the unit basis vectors of $T_p \cM_{r,d}$ be $W_p \in \bR^{D\times d}$, and $W_p^\perp$ be its complement in $\bR^D$.
    Without loss of generality, let the column span space of $W_p$ be corresponding to the first $d$ coordinates, the column span space of $W_p^\perp$ be corresponding to the last $D-d$ coordinates, and for each $z\in\cB(p,\eta)$ let
    \begin{equation}\label{eq:def-coord}
        s_z = W_p^\top (z - p)\in \bR^d,\quad t_z = (W_p^{\perp})^\top (z - p) \in \bR^{D-d}, 
    \end{equation}
    denote the subset of coordinates. Since $\cM_{r,d}$ is locally given by the root set of $h(z)$, there is
    \[
        \nJ_h(p) W_p = 0.
    \]

    According to the \textit{Implicit Function Theorem} and the results above, there exists a function $g(\cdot)$ such that, for  $z\in\cB(p,\eta)$, $h(z) = h(s_z, t_z) = 0$ if and only if $t_z = g(s_z)$. Then, by calculating the derivate of $h(s,g(s)) = 0$, there is
    \begin{align*}
        0 & = \nD_s h (s,g(s)) + \nD_t h (s,g(s)) \nD_s g (s)\\
        & = \nJ_h\bigl(s,g(s)\bigl) W_p + \nJ_h\bigl(s,g(s)\bigl) W_p^\perp  \nJ_g(s),
    \end{align*}
    which leads to the result that the Jacobian of $g(s)$ can be expressed as
    \[
        \nJ_g(s) = -\left\{\nJ_h\bigl(s,g(s)\bigl) W_p^\perp\right\}^{-1} \nJ_h\bigl(s,g(s)\bigl) W_p,
    \]
    and hence,
    \begin{align*}
        \|\nJ_g(s)\| 
        &= \|\left\{\nJ_h\bigl(s,g(s)\bigl) W_p^\perp\right\}^{-1} \nJ_h\bigl(s,g(s)\bigl) W_p\|\\
        &\leq \|\left\{\nJ_h\bigl(s,g(s)\bigl) W_p^\perp\right\}^{-1}\| \| \bigl\{\nJ_h\bigl(s,g(s)\bigl) - \nJ_h(p)\bigl\} W_p\|\\
        &\leq \left\{ \lambda_{D-d} \Bigl( \nJ_h\bigl(s,g(s)\bigl) W_p^\perp\Bigl) \right\}^{-1} \|V_p^\top\|\| \bigl\{\nJ_b\bigl(s,g(s)\bigl) - \nJ_b(p)\bigl\}\| \|W_p\|,
    \end{align*}
    where $\lambda_{D-d}(A)$ denotes the $(D-d)$th largest singular of matrix $A$. Since
    \begin{align*}
        \lambda_{D-d} \Bigl( \nJ_h\bigl(s,g(s)\bigl) W_p^\perp\Bigl)
        &\geq \lambda_{D-d} \Bigl( V^\top_p \nJ_b\bigl(s,g(s)\bigl)\Bigl) \lambda_{\min}(W_p^\perp)\\
        &\geq \lambda_{\min}(V^\top_p) \lambda_{D-d} \Bigl( \nJ_b\bigl(s,g(s)\bigl)\Bigl) \lambda_{\min}(W_p^\perp)\\
        &\geq 1 - (D-d) (r \ell_3 + \ell_5),
    \end{align*}
    the norm of $\nJ_g(s)$ is further bounded as
    \begin{align*}
        \|\nJ_g(s)\|
        &\leq \left\{\sigma_{D-d} \bigl( J_b\bigl(s,g(s) \bigl)\right\}^{-1} \|\nJ_b\bigl(s,g(s)\bigl) - \nJ_b(p)\|\\
        &\leq \|\bigl(s,g(s)\bigl) - p\|  \frac{(D-d) \bigl\{r \ell_4 + \ell_2 +2 \ell_3(1+\ell_1)\bigl\}}{1 - (D-d) (r \ell_3 + \ell_5)}.
    \end{align*}
    
    Therefore, for any $p, q \in \cM_{r,d}$, let $s_p$ and $s_q$ be defined as in \eqref{eq:def-coord}, there is
    \begin{equation}\label{eq:reach-bound-1}
    \begin{aligned}
        \frac{\|p - q\|^2}{2\nd (q,T_p\cM_{r,d})}
        &= \frac{\|p - q\|^2}{2\|g(s_p) - g(s_q)\|}  \\
        &\geq \frac{\|p - q\|^2}{2\|s_p - s_q\| \|\bigl(s_q,g(s_q)\bigl) - p\|} \frac{1 - (D-d) (r \ell_3 + \ell_5)}{(D-d) \bigl\{r \ell_4 + \ell_2 +2 \ell_3(1+\ell_1)\bigl\}} \\
        &>\frac{\|p - q\|^2 \eta}{2\|p - q\| \|\bigl(s_q,g(s_q)\bigl) - p\|} \\
        &\geq \frac{\eta}{2}.
    \end{aligned}   
    \end{equation}
    Moreover, since $\nd (q,T_p\cM_{r,d}) \leq \|p - q\|$, for $p, q \in \cM_{r,d}$ with $\|p - q\|> \eta$ there is
    \begin{equation}\label{eq:reach-bound-2}
         \frac{\|p - q\|^2}{2\nd (q,T_p\cM_{r,d})} \geq \frac{\|p - q\|}{2} \geq \frac{\eta}{2}.   
    \end{equation}
    Combining the result in \eqref{eq:reach-bound-1} and \eqref{eq:reach-bound-2}, for any $p\neq q \in \cM_{r,d}$, 
    \[
        \frac{\|p - q\|^2}{2\nd (q,T_p\cM_{r,d})} \geq \frac{\eta}{2},
    \]
    which suggests that $\cM_{r,d}$ is a $\mathcal{C}^1$ manifold with $\reach(\cM_{r,d}) \geq \eta/2$.
\end{proof}

\subsection{Asymptotic analysis of empirical principal submanifolds}
\label{sec:appendix-consistency}

In this section, we take $x_1,\dots,x_n$ to be independent and identically distributed from $F_X$, fix the target dimension $d$, write $m=D-d$, let $0<c<1$ be the constant in \eqref{def:M_hat}, and suppress the dependence on $d$ in the notation. We study the discrepancy between the hard population principal submanifold $\cM_{r,d}$ and the componentwise empirical estimator in \eqref{def:M_hat}. The argument separates a stochastic term from deterministic approximation terms due to centring, localization, and projector smoothing. Throughout this section, the norm of a matrix is its operator norm.

\subsubsection{Hard and smooth localization}

Let
\[
W(u)=\frac{1}{\vol\{\cB(0,1)\}}\ID(\|u\|\le 1),
\qquad
W_r(u)=r^{-D}W(u/r).
\]
Since the normalization cancels in the numerator and denominator, the original hard-neighbourhood definitions of $\mu_{r,z}$ and $\Sigma_{r,z}$ can be equivalently written as
\begin{align}
\mu_{r,z}
&=
\frac{\int x W_r(z-x)\,dF_X(x)}
{\int W_r(z-x)\,dF_X(x)},
\label{eq:def-hard-mu-kernel}
\\
\Sigma_{r,z}
&=
\frac{\int (x-\mu_{r,z})(x-\mu_{r,z})^\top W_r(z-x)\,dF_X(x)}
{\int W_r(z-x)\,dF_X(x)}.
\label{eq:def-hard-sigma-kernel}
\end{align}

The sample covariance in the main text is centred at its anchor. Its population counterpart is
\begin{align*}
\Sigma^{\rm anc}_{r,x}
&=
\frac{\int (u-x)(u-x)^\top W_r(x-u)\,dF_X(u)}
{\int W_r(x-u)\,dF_X(u)}\\
&=
\Sigma_{r,x}
+
\{\mu_{r,x}-x\}\{\mu_{r,x}-x\}^\top.
\end{align*}
Let $\Pi^{\rm anc}_{r,x,k}$ denote its $k$th ordered rank-one eigenprojector, with the eigenvalues arranged in non-decreasing order.

Let $K:\bR^D\to[0,\infty)$ be the normalized version of the compactly supported cubic taper in (3), so that $K(u)\propto(1-\|u\|^2)^3$ for $\|u\|\le 1$ and $K(u)=0$ otherwise. Thus, $K$ is radially symmetric and satisfies
\[
\supp(K)\subseteq \cB(0,1),
\qquad
\int_{\bR^D}K(u)\,du=1,
\]
and we define
\[
K_r(u)=r^{-D}K(u/r).
\]
The associated \emph{smooth population} local moments are
\begin{align}
\mu^{\sm}_{r,z}
&=
\frac{\int x K_r(z-x)\,dF_X(x)}
{\int K_r(z-x)\,dF_X(x)},
\label{eq:def-smooth-pop-mu}
\\
\Sigma^{\sm}_{r,z}
&=
\frac{\int (x-\mu^{\sm}_{r,z})(x-\mu^{\sm}_{r,z})^\top K_r(z-x)\,dF_X(x)}
{\int K_r(z-x)\,dF_X(x)}.
\label{eq:def-smooth-pop-sigma}
\end{align}
Let $v^{\sm}_{r,z,1},\dots,v^{\sm}_{r,z,D}$ be the unit eigenvectors of $\Sigma^{\sm}_{r,z}$ arranged in non-decreasing order of the corresponding eigenvalues, and let
\[
\Pi^{\sm}_{r,z,k}
=
v^{\sm}_{r,z,k}(v^{\sm}_{r,z,k})^\top.
\]
The hard and direct smooth-covariance projector fields and their bias fields are
\[
\Psi_{r,z}=\sum_{k=1}^m\Pi_{r,z,k},
\qquad
g_r(z)=\Psi_{r,z}(z-\mu_{r,z}),
\]
and
\[
\Psi^{\rm dir}_{r,z}=\sum_{k=1}^m\Pi^{\sm}_{r,z,k},
\qquad
g^{\rm dir}_r(z)=\Psi^{\rm dir}_{r,z}(z-\mu^{\sm}_{r,z}),
\]
respectively.

The anchor-centred componentwise population smoother corresponding to the sample-point projectors is
\[
A^{\rm anc}_{r,z,k}
=
\frac{\int K_r(z-x)\Pi^{\rm anc}_{r,x,k}\,dF_X(x)}
{\int K_r(z-x)\,dF_X(x)},
\qquad
\Psi^{\rm cw,anc}_{r,z}
=
\sum_{k=1}^m\bP_1(A^{\rm anc}_{r,z,k}).
\]
We retain the mean-centred population analogues of the componentwise and block projector smoothers. For $k=1,\dots,m$, let
\[
A^{\sm}_{r,z,k}
=
\frac{\int K_r(z-x)\Pi_{r,x,k}\,dF_X(x)}
{\int K_r(z-x)\,dF_X(x)},
\qquad
A^{\sm}_{r,z}=\sum_{k=1}^m A^{\sm}_{r,z,k},
\]
and set
\[
\Psi^{\rm cw}_{r,z}
=
\sum_{k=1}^m\bP_1(A^{\sm}_{r,z,k}),
\qquad
\Psi^{\rm blk}_{r,z}
=
\bP_m(A^{\sm}_{r,z}).
\]
Here and below, for a symmetric smoother matrix $A$, its eigenvalues are denoted in non-increasing order by
\[
\rho_1(A)\ge\cdots\ge\rho_D(A),
\]
whereas the covariance eigenvalues $\lambda_j$ remain in non-decreasing order. Define
\[
g^{\rm cw,anc}_r(z)=\Psi^{\rm cw,anc}_{r,z}(z-\mu^{\sm}_{r,z}),
\qquad
g^{\rm cw}_r(z)=\Psi^{\rm cw}_{r,z}(z-\mu^{\sm}_{r,z}),
\qquad
g^{\rm blk}_r(z)=\Psi^{\rm blk}_{r,z}(z-\mu^{\sm}_{r,z}).
\]
The deterministic terms used below are
\begin{align*}
R_*
&=
\sup_{z\in\cZ}\|z-\mu^{\sm}_{r,z}\|,\\
B_r^{\rm ctr}
&=
R_*\sup_{z\in\cZ}
\|\Psi^{\rm cw,anc}_{r,z}-\Psi^{\rm cw}_{r,z}\|,\\
B_r^{\rm agg}
&=
R_*\sup_{z\in\cZ}
\|\Psi^{\rm cw}_{r,z}-\Psi^{\rm blk}_{r,z}\|,\\
B_r^{\rm blk}
&=
R_*\sup_{z\in\cZ}
\|\Psi^{\rm blk}_{r,z}-\Psi^{\rm dir}_{r,z}\|,\\
B_r^{\rm loc}
&=
\sup_{z\in\cZ}\|g^{\rm dir}_r(z)-g_r(z)\|.
\end{align*}
and the total deterministic approximation term is
\[
B_r=B_r^{\rm ctr}+B_r^{\rm agg}+B_r^{\rm blk}+B_r^{\rm loc}.
\]

We now write the actual empirical construction in a form used in the proof. The normalized weights in (3) can be written as
\[
\alpha_i(z)
=
\frac{K_r(z-x_i)}{\sum_{j=1}^n K_r(z-x_j)}.
\]
For the sample-point projectors $\widehat\Pi_{r,i,k}$ defined in the main text, let
\begin{align*}
\widehat A_{r,z,k}
&=
\sum_{i=1}^n\alpha_i(z)\widehat\Pi_{r,i,k},
\qquad k=1,\dots,m,\\
\widehat A_{r,z}
&=
\sum_{k=1}^m\widehat A_{r,z,k}
=
\sum_{i=1}^n\alpha_i(z)
\left\{\sum_{k=1}^m\widehat\Pi_{r,i,k}\right\}.
\end{align*}
For each $i$, $\sum_{k=1}^m\widehat\Pi_{r,i,k}$ is a rank-$m$ orthogonal projector, so $\widehat A_{r,z}$ is a convex weighted average of such projectors.
The actual componentwise empirical projector field and an auxiliary block smoother are
\[
\widehat\Psi^{\rm cw}_{r,z}
=
\sum_{k=1}^m\bP_1(\widehat A_{r,z,k}),
\qquad
\widehat\Psi^{\rm blk}_{r,z}
=
\bP_m(\widehat A_{r,z}).
\]
Thus, retaining the notation of \eqref{def:M_hat},
\[
\widehat\Pi_{r,z,k}=\bP_1(\widehat A_{r,z,k}).
\]
For the empirical local mean, write
\begin{align}
\widehat q_r(z)
&=
\frac{1}{n}\sum_{i=1}^n K_r(z-x_i),
\label{eq:def-qhat}
\\
\widehat \mu_{r,z}
&=
\frac{\frac{1}{n}\sum_{i=1}^n x_i K_r(z-x_i)}
{\widehat q_r(z)},
\label{eq:def-muhat-appendix}
\end{align}
Then the empirical field in \eqref{def:M_hat} is
\[
\widehat g_r(z)
=
\widehat\Psi^{\rm cw}_{r,z}(z-\widehat\mu_{r,z}).
\]
For later moment bounds only, define the auxiliary direct smooth empirical covariance
\begin{align}
\widehat \Sigma_{r,z}
&=
\frac{\frac{1}{n}\sum_{i=1}^n (x_i-\widehat\mu_{r,z})(x_i-\widehat\mu_{r,z})^\top K_r(z-x_i)}
{\widehat q_r(z)}.
\label{eq:def-sigmahat-appendix}
\end{align}
For the asymptotic result, take a compact set $\cZ\subseteq\{z\in\bR^D:\nd(z,\supp(F_X))<cr\}$ containing $\cM_{r,d}$. Compactness and the definition of the support imply that $\cZ\subseteq\bigcup_{i=1}^n\cB(x_i,cr)$ with probability tending to one. On this event, the restriction of the estimator in \eqref{def:M_hat} to $\cZ$ is
\begin{equation}
\label{eq:def-empifold-kernel}
\widehat{\cM}_{r,d}
=
\left\{
z\in\cZ:\widehat g_r(z)=0
\right\},
\end{equation}

\subsubsection{Assumptions}

\begin{assumption}\label{asm:consistency-mass}
There exists a compact set $\cZ\subseteq\{z\in\bR^D:\nd(z,\supp(F_X))<cr\}$ containing $\cM_{r,d}$ and a constant $c_0>0$ such that for all $z\in\cZ$,
\[
\int W_r(z-x)\,dF_X(x)\ge c_0,
\qquad
\int K_r(z-x)\,dF_X(x)\ge c_0.
\]
\end{assumption}

\begin{assumption}\label{asm:consistency-smooth}
The density $f$ of $F_X$ is bounded on
\[
\cZ^{+}:=\{x\in\bR^D:\nd(x,\cZ)\le r\},
\]
and is Lipschitz continuous on $\cZ^{+}$. The kernel $K$ is nonnegative, compactly supported, radially symmetric, and belongs to $\cC^2(\bR^D)$.
\end{assumption}

\begin{assumption}\label{asm:consistency-gap}
For a covariance matrix $\Sigma$, write its eigenvalues as $\lambda_1(\Sigma)\le\cdots\le\lambda_D(\Sigma)$. There exist constants $\gamma_0,\gamma_1,\gamma_2>0$ such that the following conditions hold.
\begin{enumerate}
    \item[(i)] Let $\cZ_X^+=\cZ^+\cap\supp(F_X)$, the population domain containing the sample centres that can receive positive weight from points in $\cZ$. The first $m$ adjacent covariance eigengaps are uniformly separated:
    \[
    \inf_{x\in\cZ\cup\cZ_X^+}\min_{1\le k\le m}
    \{\lambda_{k+1}(\Sigma_{r,x})-\lambda_k(\Sigma_{r,x})\}
    \ge \gamma_0,
    \]
    and
    \[
    \inf_{x\in\cZ_X^+}\min_{1\le k\le m}
    \{\lambda_{k+1}(\Sigma^{\rm anc}_{r,x})-
    \lambda_k(\Sigma^{\rm anc}_{r,x})\}
    \ge \gamma_0,
    \]
    and
    \[
    \inf_{z\in\cZ}\min_{1\le k\le m}
    \{\lambda_{k+1}(\Sigma^{\sm}_{r,z})-
    \lambda_k(\Sigma^{\sm}_{r,z})\}
    \ge \gamma_0.
    \]
    \item[(ii)] The population projector smoothers satisfy the coherence conditions
    \[
    \inf_{z\in\cZ}\min_{1\le k\le m}
    \{\rho_1(A^{\rm anc}_{r,z,k})-\rho_2(A^{\rm anc}_{r,z,k})\}
    \ge \gamma_1,
    \]
    and
    \[
    \inf_{z\in\cZ}\min_{1\le k\le m}
    \{\rho_1(A^{\sm}_{r,z,k})-\rho_2(A^{\sm}_{r,z,k})\}
    \ge \gamma_1,
    \]
    and
    \[
    \inf_{z\in\cZ}
    \{\rho_m(A^{\sm}_{r,z})-\rho_{m+1}(A^{\sm}_{r,z})\}
    \ge \gamma_2.
    \]
    \item[(iii)] Let
    \[
    \mathcal I_r(\cZ)
    =
    \{i:K_r(z-x_i)>0\text{ for some }z\in\cZ\}.
    \]
    The sample-point projectors used in \eqref{def:M_hat} satisfy the uniform compatibility condition
    \[
    \max_{i\in\mathcal I_r(\cZ)}
    \max_{1\le k\le m}
    \|\widehat\Pi_{r,i,k}-\Pi^{\rm anc}_{r,x_i,k}\|
    =
    \cO_p\!\left(\sqrt{\frac{\log n}{nr^D}}\right).
    \]
\end{enumerate}
\end{assumption}

\begin{assumption}\label{asm:consistency-transverse}
The set $\cM_{r,d}$ is a compact embedded $d$-dimensional $\cC^1$ submanifold. There exists $\eta_0>0$ such that each point of $U_{\eta_0}(\cM_{r,d})$ has a unique nearest point on $\cM_{r,d}$, $\overline{U_{\eta_0}(\cM_{r,d})}\subset\cZ$, and $g_r$ is of class $\cC^1$ on
\[
U_{\eta_0}(\cM_{r,d})
:=
\{z\in\bR^D:\nd(z,\cM_{r,d})<\eta_0\},
\]
and there exist constants $\kappa_0>0$ and $L_0>0$ with the following properties.

\begin{enumerate}
    \item[(i)] For every $p\in\cM_{r,d}$, let $T_p\cM_{r,d}$ be the tangent space and $N_p\cM_{r,d}$ be the orthogonal complement. Then the restriction
    \[
    A_p:=\nJ_{g_r}(p)\big|_{N_p\cM_{r,d}}:N_p\cM_{r,d}\to \mathrm{Im}\{\nJ_{g_r}(p)\}
    \]
    is invertible and satisfies
    \[
    \sigma_{\min}(A_p)\ge \kappa_0.
    \]
    \item[(ii)] For all $z,z^\prime\in U_{\eta_0}(\cM_{r,d})$,
    \[
    \|\nJ_{g_r}(z)-\nJ_{g_r}(z^\prime)\|
    \le
    L_0\|z-z^\prime\|.
    \]
    Define
    \[
    \rho_0
    =
    \min\left\{\frac{\eta_0}{2},\frac{\kappa_0}{2L_0}\right\}.
    \]
    \item[(iii)] The population root is separated away from this tube:
    \[
    c_{\rm sep}
    :=
    \inf_{z\in\cZ\setminus U_{\rho_0}(\cM_{r,d})}
    \|g_r(z)\|>0.
    \]
    \item[(iv)] For $p\in\cM_{r,d}$, let $Q_p$ denote the orthogonal projection onto
    $E_p:=\mathrm{Im}\{\nJ_{g_r}(p)\}$. There exist constants $\gamma_3,\kappa_1>0$ such that, for every $p\in\cM_{r,d}$ and $\nu\in N_p\cM_{r,d}$ with $\|\nu\|\le\rho_0$,
    \[
    \rho_m(\Psi^{\rm cw,anc}_{r,p+\nu})\ge\gamma_3
    \]
    and
    \[
    \|Q_pu\|\ge\kappa_1\|u\|,
    \qquad
    u\in\mathrm{Im}\{\Psi^{\rm cw,anc}_{r,p+\nu}\}.
    \]
    \item[(v)] With
    \[
    \varepsilon_*
    =
    \min\left\{\frac{c_{\rm sep}}{2},
    \frac{\kappa_0\rho_0}{4}\right\},
    \]
    the deterministic approximation satisfies $B_r<\varepsilon_*/2$.
\end{enumerate}
\end{assumption}

\begin{remark}
The covariance eigengaps in Assumption~\ref{asm:consistency-gap}(i) identify the individual ordered projectors, whereas part (ii) is a separate local-projector coherence condition for the smoother maps. Part (iii) explicitly links the sample-point construction used in \eqref{def:M_hat} to the anchor-centred population projector field; it is not asserted to follow from the covariance eigengaps alone. Assumption~\ref{asm:consistency-transverse} is a quantitative transversality and root-separation condition. Its range condition ensures that a zero of the locally projected empirical field is also a zero of the full empirical field.
\end{remark}

\subsubsection{Auxiliary local moments and deterministic bias}

Define the hard and smooth localized raw moments
\begin{align*}
q_r(z)
&=
\int W_r(z-x)\,dF_X(x),
&
q_r^{\sm}(z)
&=
\int K_r(z-x)\,dF_X(x),
\\
m_r(z)
&=
\int x W_r(z-x)\,dF_X(x),
&
m_r^{\sm}(z)
&=
\int x K_r(z-x)\,dF_X(x),
\\
S_r(z)
&=
\int xx^\top W_r(z-x)\,dF_X(x),
&
S_r^{\sm}(z)
&=
\int xx^\top K_r(z-x)\,dF_X(x).
\end{align*}
Then
\[
\mu_{r,z}=\frac{m_r(z)}{q_r(z)},
\qquad
\mu^{\sm}_{r,z}=\frac{m_r^{\sm}(z)}{q_r^{\sm}(z)},
\]
and
\[
\Sigma_{r,z}=\frac{S_r(z)}{q_r(z)}-\mu_{r,z}\mu_{r,z}^\top,
\qquad
\Sigma^{\sm}_{r,z}=\frac{S_r^{\sm}(z)}{q_r^{\sm}(z)}-\mu^{\sm}_{r,z}(\mu^{\sm}_{r,z})^\top.
\]

\begin{lemma}\label{lem:deterministic-bias-moments}
Suppose Assumptions~\ref{asm:consistency-mass}--\ref{asm:consistency-gap} hold. Then there exists a constant $C>0$, depending only on $c_0$, $\sup_{\cZ^+}f$, $\sup_{\cZ^+}\|x\|$, $\sup_{\cZ^+}\|xx^\top\|$, and $\gamma_0$, such that
\[
\sup_{z\in\cZ}\|\mu^{\sm}_{r,z}-\mu_{r,z}\|
+
\sup_{z\in\cZ}\|\Sigma^{\sm}_{r,z}-\Sigma_{r,z}\|
\le
C\|K-W\|_{L^1(\bR^D)}.
\]
Consequently, the localization term satisfies
\[
B_r^{\rm loc}
\le
C\|K-W\|_{L^1(\bR^D)}.
\]
\end{lemma}

\begin{proof}
We first control the raw localized moments. Since
\[
W_r(u)=r^{-D}W(u/r),
\qquad
K_r(u)=r^{-D}K(u/r),
\]
a change of variables shows that
\[
\|K_r-W_r\|_{L^1(\bR^D)}=\|K-W\|_{L^1(\bR^D)}.
\]
Hence, for every $z\in\cZ$,
\[
|q_r^{\sm}(z)-q_r(z)|
\le
\int |K_r(z-x)-W_r(z-x)|\,f(x)\,dx
\le
\|f\|_{\infty,\cZ^+}\|K-W\|_{L^1},
\]
and similarly,
\[
\|m_r^{\sm}(z)-m_r(z)\|
\le
\sup_{x\in\cZ^+}\|x\|\,\|f\|_{\infty,\cZ^+}\|K-W\|_{L^1},
\]
\[
\|S_r^{\sm}(z)-S_r(z)\|
\le
\sup_{x\in\cZ^+}\|xx^\top\|\,\|f\|_{\infty,\cZ^+}\|K-W\|_{L^1}.
\]

By Assumption~\ref{asm:consistency-mass}, $q_r(z)\ge c_0$ and $q_r^{\sm}(z)\ge c_0$ on $\cZ$. Therefore,
\[
\mu^{\sm}_{r,z}-\mu_{r,z}
=
\frac{m_r^{\sm}(z)-m_r(z)}{q_r^{\sm}(z)}
+
m_r(z)\left\{\frac{1}{q_r^{\sm}(z)}-\frac{1}{q_r(z)}\right\},
\]
and thus
\[
\sup_{z\in\cZ}\|\mu^{\sm}_{r,z}-\mu_{r,z}\|
\le
C\|K-W\|_{L^1}.
\]
Likewise,
\[
\Sigma^{\sm}_{r,z}-\Sigma_{r,z}
=
\left\{\frac{S_r^{\sm}(z)}{q_r^{\sm}(z)}-\frac{S_r(z)}{q_r(z)}\right\}
-
\left\{
\mu^{\sm}_{r,z}(\mu^{\sm}_{r,z})^\top-\mu_{r,z}\mu_{r,z}^\top
\right\},
\]
so the same argument yields
\[
\sup_{z\in\cZ}\|\Sigma^{\sm}_{r,z}-\Sigma_{r,z}\|
\le
C\|K-W\|_{L^1}.
\]

Under Assumption~\ref{asm:consistency-gap}, the spectral projector map is locally Lipschitz on the set of symmetric matrices with eigengap at least $\gamma_0$. Therefore,
\[
\sup_{z\in\cZ}\|\Pi^{\sm}_{r,z,k}-\Pi_{r,z,k}\|
\le
C\sup_{z\in\cZ}\|\Sigma^{\sm}_{r,z}-\Sigma_{r,z}\|
\le
C\|K-W\|_{L^1},
\]
uniformly in $k=1,\dots,m$. Since
\[
g^{\rm dir}_r(z)-g_r(z)
=
\sum_{k=1}^m
\left[
(\Pi^{\sm}_{r,z,k}-\Pi_{r,z,k})(z-\mu^{\sm}_{r,z})
+
\Pi_{r,z,k}(\mu_{r,z}-\mu^{\sm}_{r,z})
\right],
\]
the preceding bounds and compactness of $\cZ$ yield
\[
B_r^{\rm loc}
=
\sup_{z\in\cZ}\|g^{\rm dir}_r(z)-g_r(z)\|
\le
C\|K-W\|_{L^1}.
\]
\end{proof}

\subsubsection{Uniform stochastic bounds for the empirical local moments}

Define the empirical raw localized moments
\begin{align*}
\widehat m_r(z)
&=
\frac{1}{n}\sum_{i=1}^n x_i K_r(z-x_i),
\\
\widehat S_r(z)
&=
\frac{1}{n}\sum_{i=1}^n x_i x_i^\top K_r(z-x_i).
\end{align*}
Then
\[
\widehat\mu_{r,z}=\frac{\widehat m_r(z)}{\widehat q_r(z)},
\qquad
\widehat\Sigma_{r,z}
=
\frac{\widehat S_r(z)}{\widehat q_r(z)}-\widehat\mu_{r,z}\widehat\mu_{r,z}^\top.
\]

\begin{lemma}\label{lem:uniform-raw-moment}
Suppose Assumptions~\ref{asm:consistency-mass}--\ref{asm:consistency-smooth} hold. If $nr^D/\log n\to\infty$, then
\[
\sup_{z\in\cZ}|\widehat q_r(z)-q_r^{\sm}(z)|
=
\cO_p\!\left(\sqrt{\frac{\log n}{nr^D}}\right),
\]
\[
\sup_{z\in\cZ}\|\widehat m_r(z)-m_r^{\sm}(z)\|
=
\cO_p\!\left(\sqrt{\frac{\log n}{nr^D}}\right),
\]
and
\[
\sup_{z\in\cZ}\|\widehat S_r(z)-S_r^{\sm}(z)\|
=
\cO_p\!\left(\sqrt{\frac{\log n}{nr^D}}\right).
\]
\end{lemma}

\begin{proof}
We only prove the bound for $\widehat q_r$, since the arguments for $\widehat m_r$ and $\widehat S_r$ are identical coordinatewise.

Let $\delta_n\in(0,1)$ be a mesh parameter that may decrease with $n$. Since $\cZ$ is compact, there exists a finite $\delta_n r$-net $\{z_1,\dots,z_N\}\subseteq\cZ$ with cardinality
\[
N\le C(\cZ)(\delta_n r)^{-D}.
\]
For each $j$, the random variable $K_r(z_j-X)$ is bounded by $\|K\|_\infty r^{-D}$ and has variance of order $r^{-D}$ because
\[
E\{K_r(z_j-X)^2\}
=
\int K_r(z_j-x)^2 f(x)\,dx
\le
Cr^{-D}.
\]
Therefore, Bernstein's inequality gives
\[
\Pr\left(
\left|
\widehat q_r(z_j)-q_r^{\sm}(z_j)
\right|
> t
\right)
\le
2\exp\left\{
-Cn\min(t^2r^D,tr^D)
\right\}.
\]
Taking
\[
t=M\sqrt{\frac{\log n}{nr^D}}
\]
and using the union bound over $j=1,\dots,N$, we get
\[
\max_{1\le j\le N}
\left|
\widehat q_r(z_j)-q_r^{\sm}(z_j)
\right|
=
\cO_p\!\left(\sqrt{\frac{\log n}{nr^D}}\right).
\]

To pass from the net to the whole of $\cZ$, note that
\[
|\widehat q_r(z)-\widehat q_r(z^\prime)|
\le
\frac{1}{n}\sum_{i=1}^n
|K_r(z-x_i)-K_r(z^\prime-x_i)|
\le
Cr^{-D-1}\|z-z^\prime\|,
\]
because $K\in\cC^1_c$ and hence $K_r$ is Lipschitz with constant of order $r^{-D-1}$. The same bound holds for $|q_r^{\sm}(z)-q_r^{\sm}(z^\prime)|$. Choosing $z_j$ from the $\delta_n r$-net so that $\|z-z_j\|\le \delta_n r$, we obtain
\[
|\widehat q_r(z)-q_r^{\sm}(z)|
\le
|\widehat q_r(z_j)-q_r^{\sm}(z_j)|+C\delta_n r\cdot r^{-D-1}
=
|\widehat q_r(z_j)-q_r^{\sm}(z_j)|+C\delta_n r^{-D}.
\]
Taking the mesh size $\delta_n r$ to decrease with $n$, so that the Lipschitz discretization error is of the same order as the Bernstein bound, yields the asserted uniform rate; the resulting logarithmic covering factor is absorbed in the displayed $\log n$ term. Thus,
\[
\sup_{z\in\cZ}|\widehat q_r(z)-q_r^{\sm}(z)|
=
\cO_p\!\left(\sqrt{\frac{\log n}{nr^D}}\right).
\]
The proofs for $\widehat m_r$ and $\widehat S_r$ are identical after applying the same argument to the bounded envelopes $xK_r(z-x)$ and $xx^\top K_r(z-x)$ on $\cZ^+$.
\end{proof}

\begin{lemma}\label{lem:uniform-moments}
Suppose Assumptions~\ref{asm:consistency-mass}--\ref{asm:consistency-smooth} hold. If $nr^D/\log n\to\infty$, then
\[
\sup_{z\in\cZ}\|\widehat\mu_{r,z}-\mu^{\sm}_{r,z}\|
=
\cO_p\!\left(\sqrt{\frac{\log n}{nr^D}}\right),
\]
and
\[
\sup_{z\in\cZ}\|\widehat\Sigma_{r,z}-\Sigma^{\sm}_{r,z}\|
=
\cO_p\!\left(\sqrt{\frac{\log n}{nr^D}}\right).
\]
\end{lemma}

\begin{proof}
By Lemma~\ref{lem:uniform-raw-moment},
\[
\sup_{z\in\cZ}|\widehat q_r(z)-q_r^{\sm}(z)|
=
\cO_p\!\left(\sqrt{\frac{\log n}{nr^D}}\right).
\]
Since $q_r^{\sm}(z)\ge c_0$ on $\cZ$ by Assumption~\ref{asm:consistency-mass}, it follows that
\[
\inf_{z\in\cZ}\widehat q_r(z)\ge \frac{c_0}{2}
\]
with probability tending to one. Therefore,
\[
\widehat\mu_{r,z}-\mu^{\sm}_{r,z}
=
\frac{\widehat m_r(z)-m_r^{\sm}(z)}{\widehat q_r(z)}
+
m_r^{\sm}(z)\left\{
\frac{1}{\widehat q_r(z)}-\frac{1}{q_r^{\sm}(z)}
\right\},
\]
and hence
\[
\sup_{z\in\cZ}\|\widehat\mu_{r,z}-\mu^{\sm}_{r,z}\|
\le
C\left\{
\sup_{z\in\cZ}\|\widehat m_r(z)-m_r^{\sm}(z)\|
+
\sup_{z\in\cZ}|\widehat q_r(z)-q_r^{\sm}(z)|
\right\}.
\]
The asserted rate follows from Lemma~\ref{lem:uniform-raw-moment}.

Similarly,
\[
\widehat\Sigma_{r,z}-\Sigma^{\sm}_{r,z}
=
\left\{
\frac{\widehat S_r(z)}{\widehat q_r(z)}
-
\frac{S_r^{\sm}(z)}{q_r^{\sm}(z)}
\right\}
-
\left\{
\widehat\mu_{r,z}\widehat\mu_{r,z}^\top
-
\mu^{\sm}_{r,z}(\mu^{\sm}_{r,z})^\top
\right\},
\]
so the same argument yields
\[
\sup_{z\in\cZ}\|\widehat\Sigma_{r,z}-\Sigma^{\sm}_{r,z}\|
=
\cO_p\!\left(\sqrt{\frac{\log n}{nr^D}}\right).
\]
\end{proof}

\subsubsection{Spectral perturbation and the bias--variance decomposition}

\begin{lemma}\label{lem:spectral-perturbation}
Suppose Assumptions~\ref{asm:consistency-mass}--\ref{asm:consistency-gap} hold and $nr^D/\log n\to\infty$. Then
\[
\sup_{z\in\cZ}\max_{1\le k\le m}
\|\widehat A_{r,z,k}-A^{\rm anc}_{r,z,k}\|
=
\cO_p\!\left(\sqrt{\frac{\log n}{nr^D}}\right).
\]
Consequently,
\[
\sup_{z\in\cZ}
\|\widehat\Psi^{\rm cw}_{r,z}-\Psi^{\rm cw,anc}_{r,z}\|
=
\cO_p\!\left(\sqrt{\frac{\log n}{nr^D}}\right).
\]
For the actual empirical field in \eqref{def:M_hat},
\[
\sup_{z\in\cZ}\|\widehat g_r(z)-g_r^{\rm cw,anc}(z)\|
=
\cO_p\!\left(\sqrt{\frac{\log n}{nr^D}}\right).
\]
\end{lemma}

\begin{proof}
For $k=1,\dots,m$, define the oracle empirical numerator and its population counterpart by
\[
\widehat N^{\rm anc}_{r,z,k}
=
\frac{1}{n}\sum_{i=1}^nK_r(z-x_i)\Pi^{\rm anc}_{r,x_i,k},
\qquad
N^{\rm anc}_{r,z,k}
=
\int K_r(z-x)\Pi^{\rm anc}_{r,x,k}\,dF_X(x).
\]
Since the projector entries are bounded, the same Bernstein and net argument as in Lemma~\ref{lem:uniform-raw-moment}, applied coordinatewise, gives
\[
\sup_{z\in\cZ}\max_{1\le k\le m}
\|\widehat N^{\rm anc}_{r,z,k}-N^{\rm anc}_{r,z,k}\|
=
\cO_p\!\left(\sqrt{\frac{\log n}{nr^D}}\right).
\]
Moreover,
\[
\overline A^{\rm anc}_{r,z,k}
:=
\sum_{i=1}^n\alpha_i(z)\Pi^{\rm anc}_{r,x_i,k}
=
\frac{\widehat N^{\rm anc}_{r,z,k}}{\widehat q_r(z)},
\qquad
A^{\rm anc}_{r,z,k}
=
\frac{N^{\rm anc}_{r,z,k}}{q_r^{\sm}(z)}.
\]
Assumption~\ref{asm:consistency-mass} and Lemma~\ref{lem:uniform-raw-moment} therefore imply
\[
\sup_{z\in\cZ}\max_{1\le k\le m}
\|\overline A^{\rm anc}_{r,z,k}-A^{\rm anc}_{r,z,k}\|
=
\cO_p\!\left(\sqrt{\frac{\log n}{nr^D}}\right).
\]
On the other hand, convexity of the weights gives
\begin{align*}
\|\widehat A_{r,z,k}-\overline A^{\rm anc}_{r,z,k}\|
&\le
\sum_{i=1}^n\alpha_i(z)
\|\widehat\Pi_{r,i,k}-\Pi^{\rm anc}_{r,x_i,k}\|\\
&\le
\max_{i\in\mathcal I_r(\cZ)}
\|\widehat\Pi_{r,i,k}-\Pi^{\rm anc}_{r,x_i,k}\|.
\end{align*}
The first assertion now follows from Assumption~\ref{asm:consistency-gap}(iii).

The smoother gaps in Assumption~\ref{asm:consistency-gap}(ii) and the Davis--Kahan theorem yield, uniformly over $z\in\cZ$,
\[
\|\bP_1(\widehat A_{r,z,k})-\bP_1(A^{\rm anc}_{r,z,k})\|
\le
C\|\widehat A_{r,z,k}-A^{\rm anc}_{r,z,k}\|.
\]
Since $m$ is fixed, summing over $k$ proves the projector-field bound. Finally,
\begin{align*}
\|\widehat g_r(z)-g_r^{\rm cw,anc}(z)\|
\le{}&
\|\widehat\Psi^{\rm cw}_{r,z}-\Psi^{\rm cw,anc}_{r,z}\|
\|z-\widehat\mu_{r,z}\|\\
&+
\|\Psi^{\rm cw,anc}_{r,z}\|
\|\widehat\mu_{r,z}-\mu^{\sm}_{r,z}\|.
\end{align*}
The first factor involving the mean is uniformly bounded with probability tending to one, and $\|\Psi^{\rm cw,anc}_{r,z}\|\le m$. Lemma~\ref{lem:uniform-moments} completes the proof.
\end{proof}

\begin{lemma}[Componentwise--block bridge]\label{lem:componentwise-block-bridge}
For every $z$ for which the empirical weights are defined,
\begin{align*}
\|\widehat\Psi^{\rm cw}_{r,z}-\widehat\Psi^{\rm blk}_{r,z}\|
\leq{}&
\sum_{k=1}^m
\|\bP_1(\widehat A_{r,z,k})-\widehat A_{r,z,k}\|\\
&+
\|\bP_m(\widehat A_{r,z})-\widehat A_{r,z}\|.
\end{align*}
The analogous inequality holds for the mean-centred population smoothers after replacing the hatted quantities by $A^{\sm}_{r,z,k}$, $A^{\sm}_{r,z}$, $\Psi^{\rm cw}_{r,z}$, and $\Psi^{\rm blk}_{r,z}$. In particular,
\[
\|\bP_1(A_{r,z,k}^{\sm})-A_{r,z,k}^{\sm}\|
=
1-\rho_1(A_{r,z,k}^{\sm}),
\]
and
\[
\|\bP_m(A_{r,z}^{\sm})-A_{r,z}^{\sm}\|
=
\max\{1-\rho_m(A_{r,z}^{\sm}),\rho_{m+1}(A_{r,z}^{\sm})\}.
\]
\end{lemma}

\begin{proof}
Because $\widehat A_{r,z}=\sum_{k=1}^m\widehat A_{r,z,k}$, the triangle inequality gives
\begin{align*}
\|\widehat\Psi^{\rm cw}_{r,z}-\widehat\Psi^{\rm blk}_{r,z}\|
={}&
\left\|
\sum_{k=1}^m\bP_1(\widehat A_{r,z,k})
-
\bP_m(\widehat A_{r,z})
\right\|\\
\leq{}&
\sum_{k=1}^m
\|\bP_1(\widehat A_{r,z,k})-\widehat A_{r,z,k}\|
+
\|\widehat A_{r,z}-\bP_m(\widehat A_{r,z})\|.
\end{align*}
The mean-centred population inequality is identical. Each $A_{r,z,k}^{\sm}$ is positive semidefinite with trace one, so $1-\rho_1(A_{r,z,k}^{\sm})$ dominates every remaining eigenvalue. Also, $A_{r,z}^{\sm}$ is a convex average of rank-$m$ orthogonal projectors, and hence has trace $m$ and eigenvalues in $[0,1]$. The two displayed identities follow from the spectral decomposition. No equality between the componentwise and block smoothers is asserted.
\end{proof}

\begin{proposition}[Bias--variance decomposition]
\label{prop:bias-variance}
Suppose Assumptions~\ref{asm:consistency-mass}--\ref{asm:consistency-gap} hold and $nr^D/\log n\to\infty$. Then
\[
\sup_{z\in\cZ}\|\widehat g_r(z)-g_r(z)\|
\le
R_{n,r}+B_r,
\]
where
\[
R_{n,r}
=
\sup_{z\in\cZ}\|\widehat g_r(z)-g_r^{\rm cw,anc}(z)\|
=
\cO_p\!\left(\sqrt{\frac{\log n}{nr^D}}\right).
\]
The deterministic term is
\[
B_r=B_r^{\rm ctr}+B_r^{\rm agg}+B_r^{\rm blk}+B_r^{\rm loc}.
\]
Only the localization component has the explicit kernel bound
\[
B_r^{\rm loc}
\le
C\|K-W\|_{L^1(\bR^D)}
\]
from Lemma~\ref{lem:deterministic-bias-moments}; the centring, aggregation, and block terms are not included in this kernel-mismatch bound.
\end{proposition}

\begin{proof}
For every $z\in\cZ$, the relevant fields satisfy
\begin{align*}
\|\widehat g_r(z)-g_r(z)\|
\leq{}&
\|\widehat g_r(z)-g_r^{\rm cw,anc}(z)\|
+
\|g_r^{\rm cw,anc}(z)-g_r^{\rm cw}(z)\|\\
&+
\|g_r^{\rm cw}(z)-g_r^{\rm blk}(z)\|
+
\|g_r^{\rm blk}(z)-g_r^{\rm dir}(z)\|
+
\|g_r^{\rm dir}(z)-g_r(z)\|.
\end{align*}
The first term is controlled by Lemma~\ref{lem:spectral-perturbation}. Since the anchor-centred componentwise, mean-centred componentwise, block, and direct smooth fields use the same mean $\mu^{\sm}_{r,z}$, the suprema of the remaining four terms are bounded by $B_r^{\rm ctr}$, $B_r^{\rm agg}$, $B_r^{\rm blk}$, and $B_r^{\rm loc}$, respectively. Taking the supremum proves the result.
\end{proof}

\subsubsection{Hausdorff stability of the root set}

We now convert the uniform field bound in Proposition~\ref{prop:bias-variance} into a Hausdorff bound between $\widehat{\cM}_{r,d}$ and $\cM_{r,d}$.

\begin{theorem}\label{thm:hausdorff-stability}
Suppose Assumptions~\ref{asm:consistency-mass}--\ref{asm:consistency-transverse} hold and $nr^D/\log n\to\infty$. Then, with probability tending to one, whenever
\[
\Delta_{n,r}
=
\sup_{z\in\cZ}\|\widehat g_r(z)-g_r(z)\|
\le \varepsilon_*,
\]
the empirical root set $\widehat{\cM}_{r,d}$ is contained in $U_{\rho_0}(\cM_{r,d})$ and
\[
\textnormal{d}_{\textnormal{H}}(\widehat{\cM}_{r,d},\cM_{r,d})
\le
\frac{4}{\kappa_0}\Delta_{n,r}.
\]
Consequently,
\[
\textnormal{d}_{\textnormal{H}}(\widehat{\cM}_{r,d},\cM_{r,d})
\le
C\left\{
\cO_p\!\left(\sqrt{\frac{\log n}{nr^D}}\right)+B_r
\right\}.
\]
\end{theorem}

\begin{proof}
We first establish two local growth bounds. Fix $p\in\cM_{r,d}$, write $N_p=N_p\cM_{r,d}$ and $E_p=\mathrm{Im}\{\nJ_{g_r}(p)\}$, and let $U_p:E_p\to N_p$ be an isometry. For $\nu\in N_p$ with $\|\nu\|\le\rho_0$, the fundamental theorem of calculus gives
\[
g_r(p+\nu)
=
A_p\nu
+
\int_0^1
\{\nJ_{g_r}(p+t\nu)-\nJ_{g_r}(p)\}\nu\,dt.
\]
Assumption~\ref{asm:consistency-transverse} therefore yields
\begin{equation}
\label{eq:normal-growth}
\|g_r(p+\nu)\|
\ge
\frac{\kappa_0}{2}\|\nu\|.
\end{equation}
For the map
\[
H_p(\nu)
=
U_pQ_pg_r(p+\nu):
N_p\longrightarrow N_p,
\]
we likewise have
\[
H_p(\nu)
=
U_pA_p\nu
+
U_pQ_p\int_0^1
\{\nJ_{g_r}(p+t\nu)-\nJ_{g_r}(p)\}\nu\,dt,
\]
and hence
\begin{equation}
\label{eq:projected-normal-growth}
\|H_p(\nu)\|
\ge
\frac{\kappa_0}{2}\|\nu\|,
\qquad
\|\nu\|\le\rho_0.
\end{equation}

We next record why the projected empirical equation is sufficient. Lemma~\ref{lem:spectral-perturbation} gives
\[
\sup_{z\in\cZ}
\|\widehat\Psi^{\rm cw}_{r,z}-\Psi^{\rm cw,anc}_{r,z}\|
=o_p(1).
\]
The rank and range conditions in Assumption~\ref{asm:consistency-transverse}(iv), together with standard subspace perturbation, imply that, with probability tending to one,
\begin{equation}
\label{eq:empirical-range-compatibility}
\|Q_pu\|
\ge
\frac{\kappa_1}{2}\|u\|,
\qquad
u\in\mathrm{Im}\{\widehat\Psi^{\rm cw}_{r,p+\nu}\},
\end{equation}
uniformly over $p\in\cM_{r,d}$ and $\nu\in N_p$ with $\|\nu\|\le\rho_0$. Indeed, $\rho_m(\Psi^{\rm cw,anc}_{r,p+\nu})\ge\gamma_3$ and $\widehat\Psi^{\rm cw}_{r,p+\nu}$ has rank at most $m$, so its range converges uniformly to the population range. Since $\widehat g_r(z)$ belongs to $\mathrm{Im}\{\widehat\Psi^{\rm cw}_{r,z}\}$, \eqref{eq:empirical-range-compatibility} shows that
\[
Q_p\widehat g_r(p+\nu)=0
\quad\Longrightarrow\quad
\widehat g_r(p+\nu)=0.
\]

We now prove the two-sided bound on this high-probability event. Suppose $\Delta_{n,r}\le\varepsilon_*$. If $\widehat z\in\cZ\setminus U_{\rho_0}(\cM_{r,d})$, Assumption~\ref{asm:consistency-transverse}(iii) gives
\[
\|\widehat g_r(\widehat z)\|
\ge
c_{\rm sep}-\Delta_{n,r}
\ge
\frac{c_{\rm sep}}{2}>0.
\]
Thus every $\widehat z\in\widehat{\cM}_{r,d}$ lies in $U_{\rho_0}(\cM_{r,d})$. Let $p$ be its nearest point on $\cM_{r,d}$ and write $\widehat z=p+\nu$ with $\nu\in N_p$. Since $\widehat g_r(\widehat z)=0$, \eqref{eq:normal-growth} gives
\[
\nd(\widehat z,\cM_{r,d})
=
\|\nu\|
\le
\frac{2}{\kappa_0}\Delta_{n,r}.
\]
Consequently,
\begin{equation}
\label{eq:one-sided-hd}
\sup_{\widehat z\in\widehat{\cM}_{r,d}}
\nd(\widehat z,\cM_{r,d})
\le
\frac{2}{\kappa_0}\Delta_{n,r}.
\end{equation}

For the reverse direction, fix $p\in\cM_{r,d}$. If $\Delta_{n,r}=0$, then $\widehat g_r=g_r$ on $\cZ$ and the assertion is immediate. Otherwise, set
\[
\rho
=
\frac{4}{\kappa_0}\Delta_{n,r}
\le
\rho_0
\]
and we define
\[
\widehat H_p(\nu)
=
U_pQ_p\widehat g_r(p+\nu).
\]
On the boundary $\|\nu\|=\rho$, \eqref{eq:projected-normal-growth} gives
\[
\|H_p(\nu)\|
\ge
2\Delta_{n,r},
\]
whereas
\[
\|\widehat H_p(\nu)-H_p(\nu)\|
\le
\Delta_{n,r}.
\]
On the same high-probability event, the empirical smoother gaps remain positive, so $\widehat H_p$ is continuous on the normal ball. Thus $\widehat H_p$ and $H_p$ have the same topological degree there. The projected normal-growth bound shows that $0$ is the only zero of $H_p$ in this ball. Since
\[
\nJ_{H_p}(0)=U_pA_p
\]
is invertible, the degree is $\pm1$. Therefore, $\widehat H_p$ has a zero $\nu_p$ in the ball. By \eqref{eq:empirical-range-compatibility}, this is a zero of the full field $\widehat g_r$. Hence
\begin{equation}
\label{eq:other-side-hd}
\sup_{p\in\cM_{r,d}}
\nd(p,\widehat{\cM}_{r,d})
\le
\frac{4}{\kappa_0}\Delta_{n,r}.
\end{equation}
Combining \eqref{eq:one-sided-hd} and \eqref{eq:other-side-hd} proves the conditional assertion.

Finally, Proposition~\ref{prop:bias-variance} gives
\[
\Delta_{n,r}
\le
\cO_p\!\left(\sqrt{\frac{\log n}{nr^D}}\right)+B_r.
\]
Assumption~\ref{asm:consistency-transverse}(v) ensures that this is at most $\varepsilon_*$ with probability tending to one. Substitution completes the proof.
\end{proof}

\begin{remark}
Theorem~\ref{thm:hausdorff-stability} is intentionally stated for a fixed radius $r$. This matches the focus of the paper on an empirical nested sequence defined under a common geometric scale. In this regime, the stochastic error vanishes as $n\to\infty$. The term $B_r^{\rm ctr}$ accounts for the difference between the anchor-centred and mean-centred population projector smoothers, $B_r^{\rm loc}$ reflects the hard-versus-smooth localization, and $B_r^{\rm agg}$ and $B_r^{\rm blk}$ account for the projector-smoothing approximations used in \eqref{def:M_hat}. Defining the population target through the same smooth kernel would remove $B_r^{\rm loc}$, but would not by itself force the other three terms to vanish. If the empirical hierarchy were instead constructed with a different radius, say $r^\ast\neq r$, then the target would itself change from $\cM_{r,d}$ to $\cM_{r^\ast,d}$ rather than merely introducing additional estimation error around the same object. This is one reason the paper adopts a common fixed-$r$ viewpoint. Allowing the radius to vary with the dimension may be of future interest, but it would require rethinking how to preserve a coherent nested interpretation across dimensions. When $B_r=o(1)$ along the considered sequence of models or localization schemes, the displayed bounds imply uniform field convergence and Hausdorff consistency in probability.
\end{remark}

\subsection{Cell counts for the single-cell RNA data set}
\begin{table}[htbp]
\centering
    \tbl{Number of cells in each group before and after cleaning}
    {\begin{tabular}{clcc}
        \toprule\\[-13pt]
        Group     & Cell Type                          & Count (Original) & Count (Cleaned) \\\hline
        1         & Enterocyte Immature Distal         & 512              & 508             \\
        2         & Enterocyte Immature Proximal       & 297              & 297             \\
        3         & Enterocyte Mature Distal           & 241              & 179             \\
        4         & Enterocyte Mature Proximal         & 581              & 555             \\
        5         & Enterocyte Progenitor              & 356              & 356             \\
        6         & Enterocyte Progenitor Early        & 829              & 828             \\
        7         & Enterocyte Progenitor Late         & 404              & 404             \\
        8         & Enteroendocrine                    & 310              & 37              \\
        9         & Goblet                             & 510              & 396             \\
        10        & Paneth                             & 260              & 217             \\
        11        & Stem                               & 1267             & 1263            \\
        12        & Transit Amplifying                 & 665              & 665             \\
        13        & Transit Amplifying Gen 1           & 408              & 408             \\
        14        & Transit Amplifying Gen 2           & 410              & 394             \\
        15        & Tuft                               & 166              & 21              \\\hline\\[-10pt]
        \textbf{} & {\textbf{Total}} & \textbf{7216}    & \textbf{6528}  \\[-3pt]\bottomrule
    \end{tabular}}
    \label{tab:Cell-count}
\end{table}

\subsection{Algorithms for fitting the principal nested submanifolds}
\begin{algorithm}[H]
    \caption{An algorithm for the principal nested submanifolds}\label{alg:pnsb}
    \begin{algorithmic}
    \Require The noisy data set $\cY=\{y_1,\dots,y_n\} \subset \bR^D$,
    the radius parameter $r$, embedding set $\mathcal{E}\subset\bR^{D^\prime}$
    \State $\cX_D=\{x_1,\dots,x_n\}\gets$ embedding of $\cY$ in $\mathcal{E}$
    \For{$x_i$ in $\cX_D$}
        \State Calculate $\{\widehat\Pi_{r,i,1},\dots,\widehat\Pi_{r,i,D^\prime}\}$ with $(x_i,\ \cX_D,\ r)$
    \EndFor
    \For{$d = D-1 : 1$}
        \State $\cX_d \gets$ projections of $\cX_{d+1}$ on $\cM_d$ with $\bigl(\cX_D,\ r,\ \{\widehat\Pi_{r,i,1},\dots,\widehat\Pi_{r,i,D^\prime-d}\}_{i=1}^n\bigl)$
        \State $\cY_d \gets$ map $\cX_d$ back to $\bR^{D}$
    \EndFor
    \Ensure $\{\cY_1,\dots,\cY_{D-1}\}$
    \end{algorithmic}
\end{algorithm}

\begin{algorithm}[H]
    \caption{Project a point $z$ onto $\cM_d$}\label{alg:proj}
    \begin{algorithmic}
    \Require The point to project $z$, the noisy data set $\cX_D=\{x_1,\dots,x_n\}$, 
    the radius parameter $r$, the collection of projection matrices at each point in $\cX_D$, $\{\widehat\Pi_{r,i,1},\dots,\widehat\Pi_{r,i,D^\prime-d}\}_{i=1}^n$, the embedding set $\mathcal{E}$ 
    \For{$\text{iter} = 1$: MaxIter}
        \State $\alpha_i(z) \gets$ calculate weights defined in \eqref{eq:alpha} with $(z,\ \cX_D,\ r)$
        \State $\Delta_z \gets \sum\alpha_i(z)x_i - z$
        \State $F_z \gets 0$
        \For{$t = 1$ to $D^\prime-d$}
            \State $\widehat\Pi_z^{(t)} \gets \dsp_1\left(\sum\alpha_i(z)\Pi_i^{(t)}\right)$
            \State $F_z \gets F_z + \widehat\Pi_z^{(t)} \Delta_z$
        \EndFor
        \If{$\| F_z\| < \epsilon$}
            \State \textbf{break}
        \EndIf
        \State $z \gets z + F_z$
    \EndFor
    \State $z\gets \argmin_{z^\prime \in \mathcal{E}} \|z - z^\prime\|$
    \Ensure $z$
    \end{algorithmic}
\end{algorithm}

\subsection{Additional figures of the real data study}
The layouts below are computed separately at each target dimension and are intended as qualitative visualizations; their coordinates are not directly comparable across panels.
\begin{figure}[htbp]
    \centering
    \subfigure[]{\includegraphics[width=0.3\textwidth]{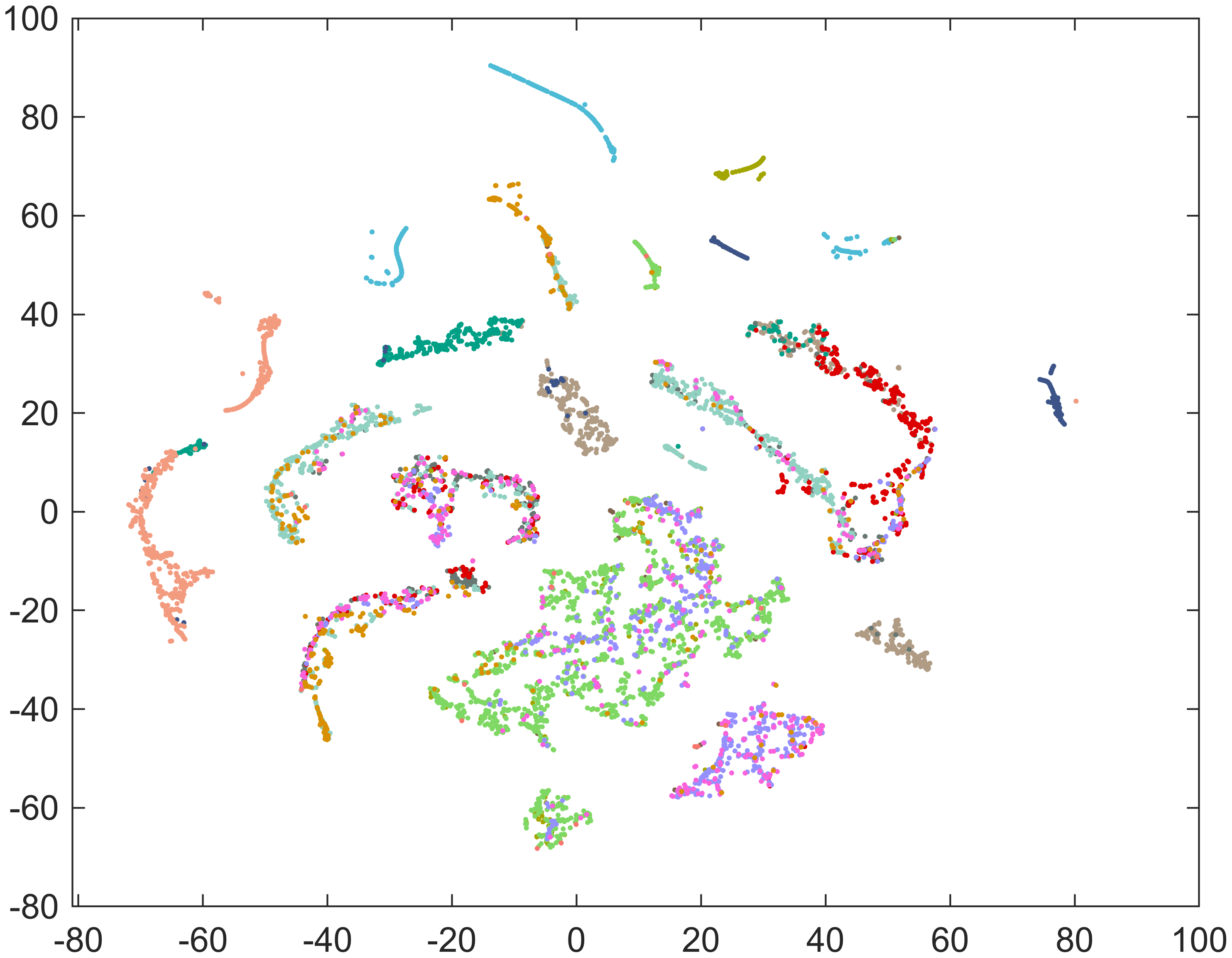}}
    \subfigure[]{\includegraphics[width=0.3\textwidth]{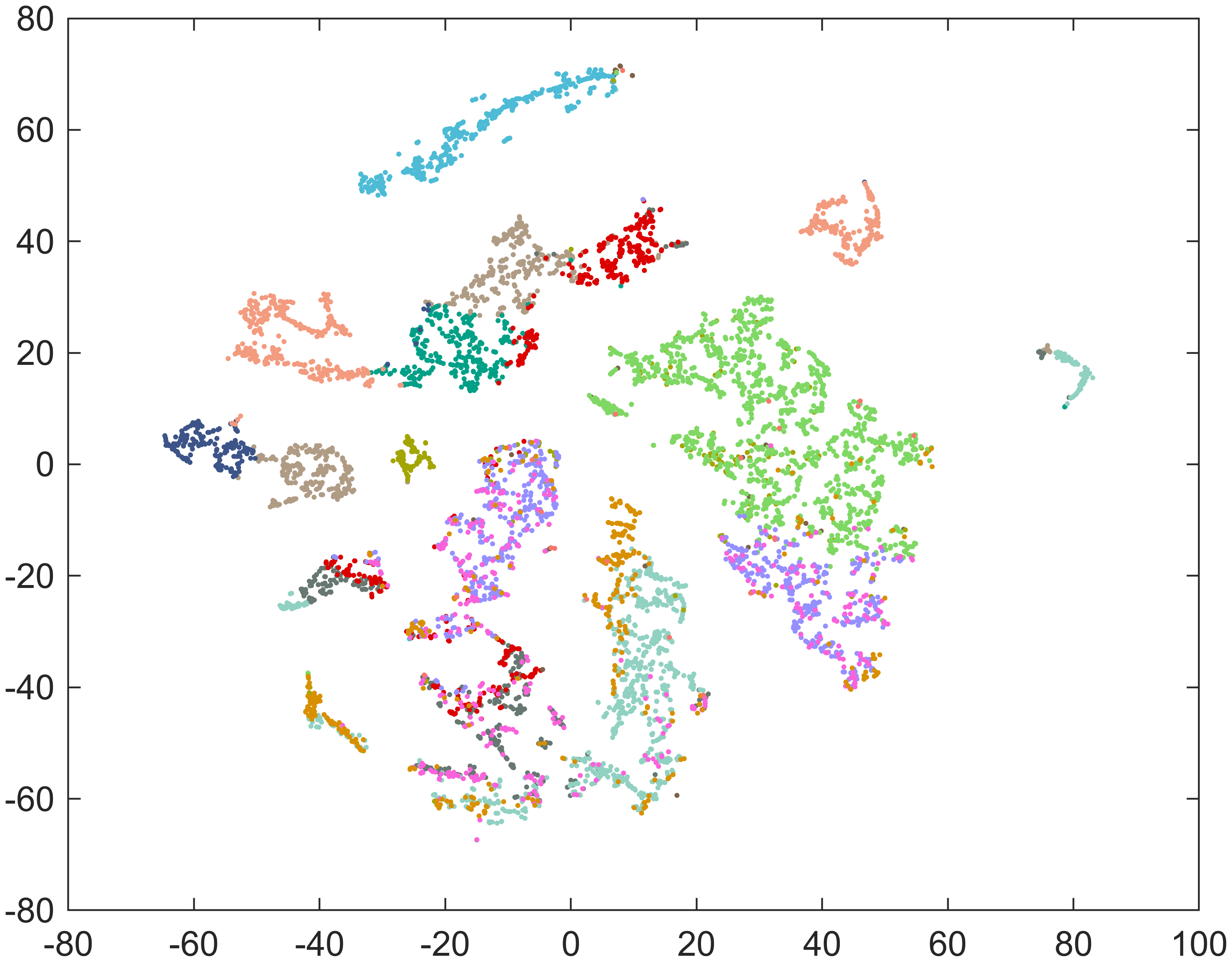}}
    \subfigure[]{\includegraphics[width=0.3\textwidth]{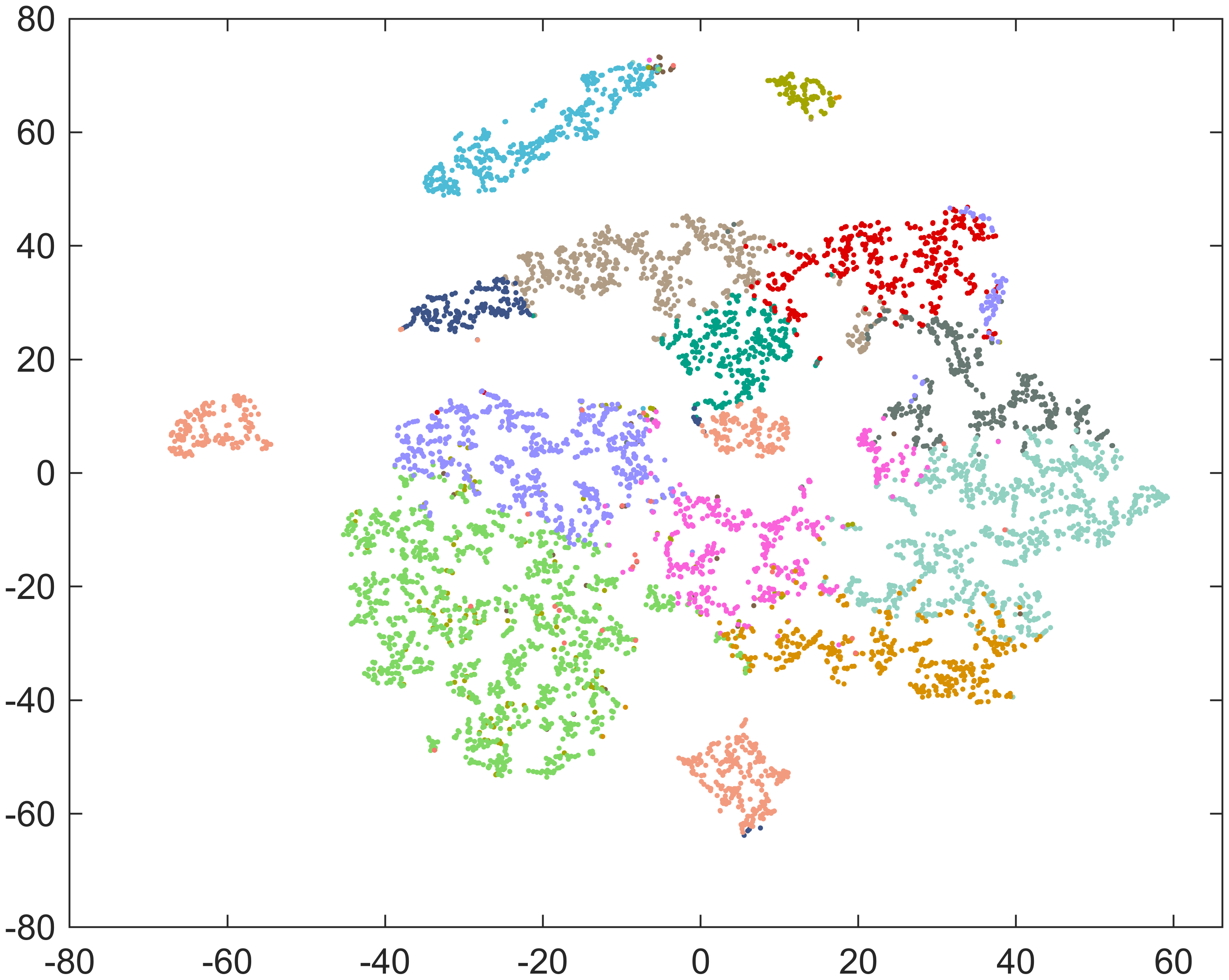}}
    \subfigure[]{\includegraphics[width=0.3\textwidth]{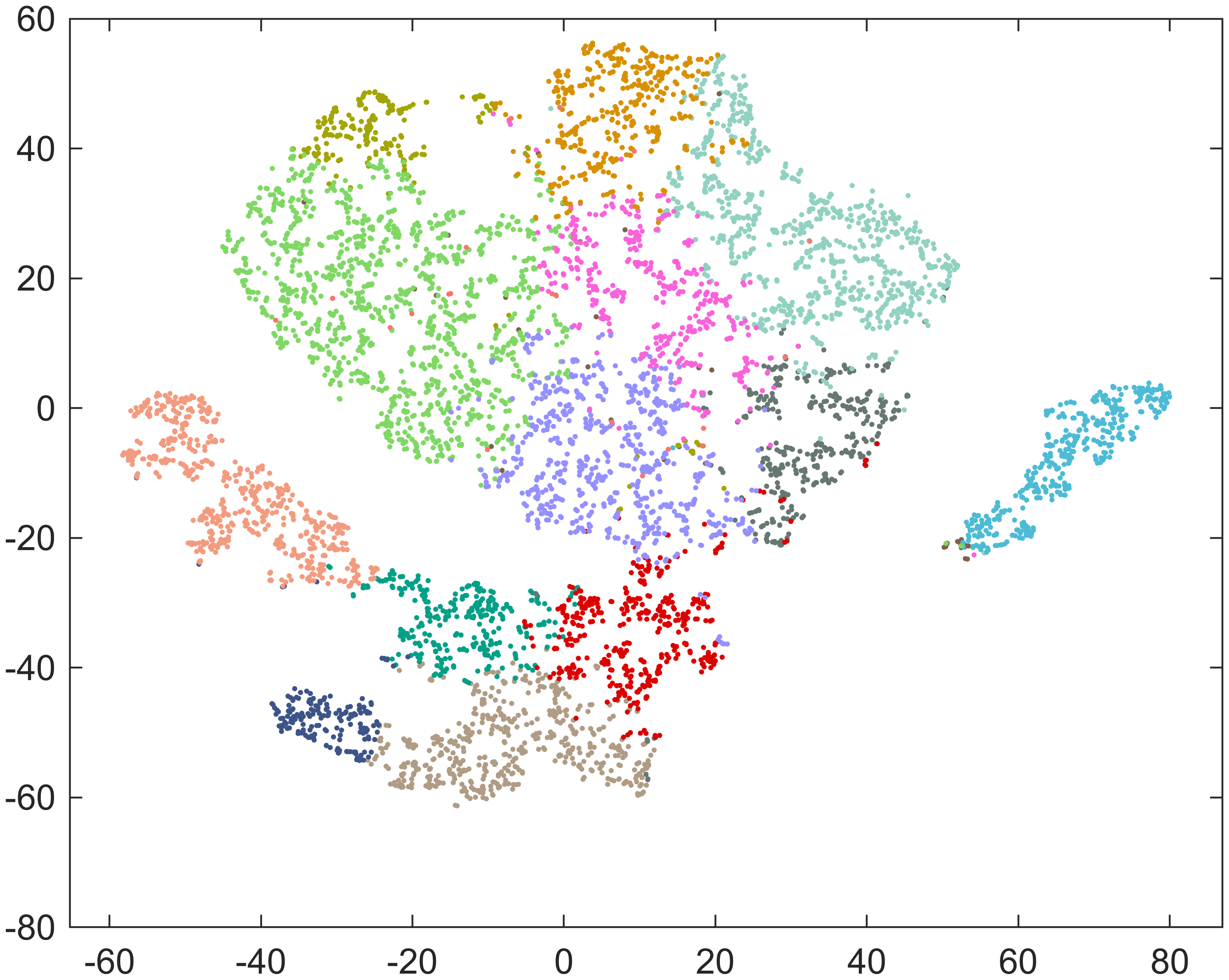}}
    \subfigure[]{\includegraphics[width=0.3\textwidth]{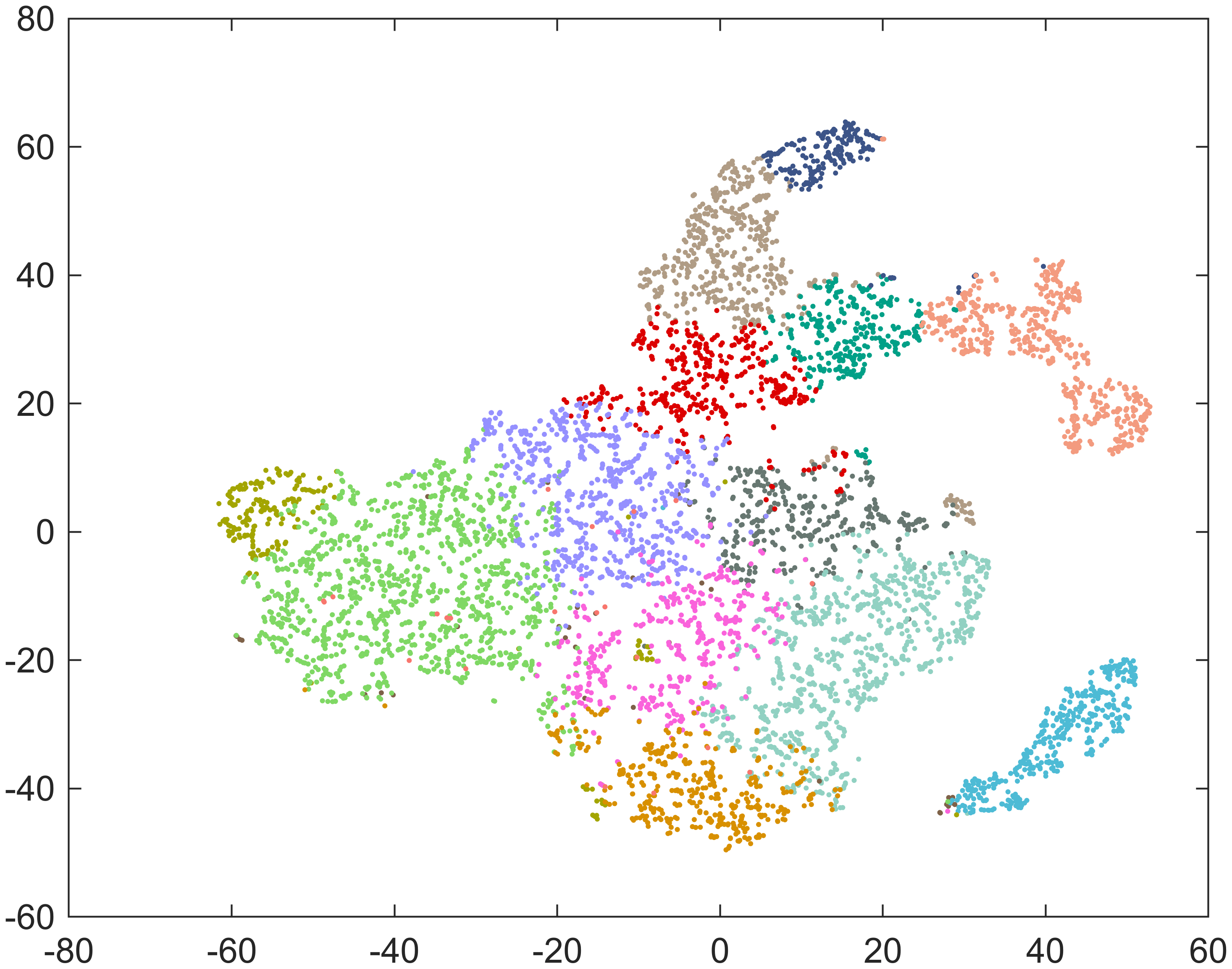}}
    \subfigure[]{\includegraphics[width=0.3\textwidth]{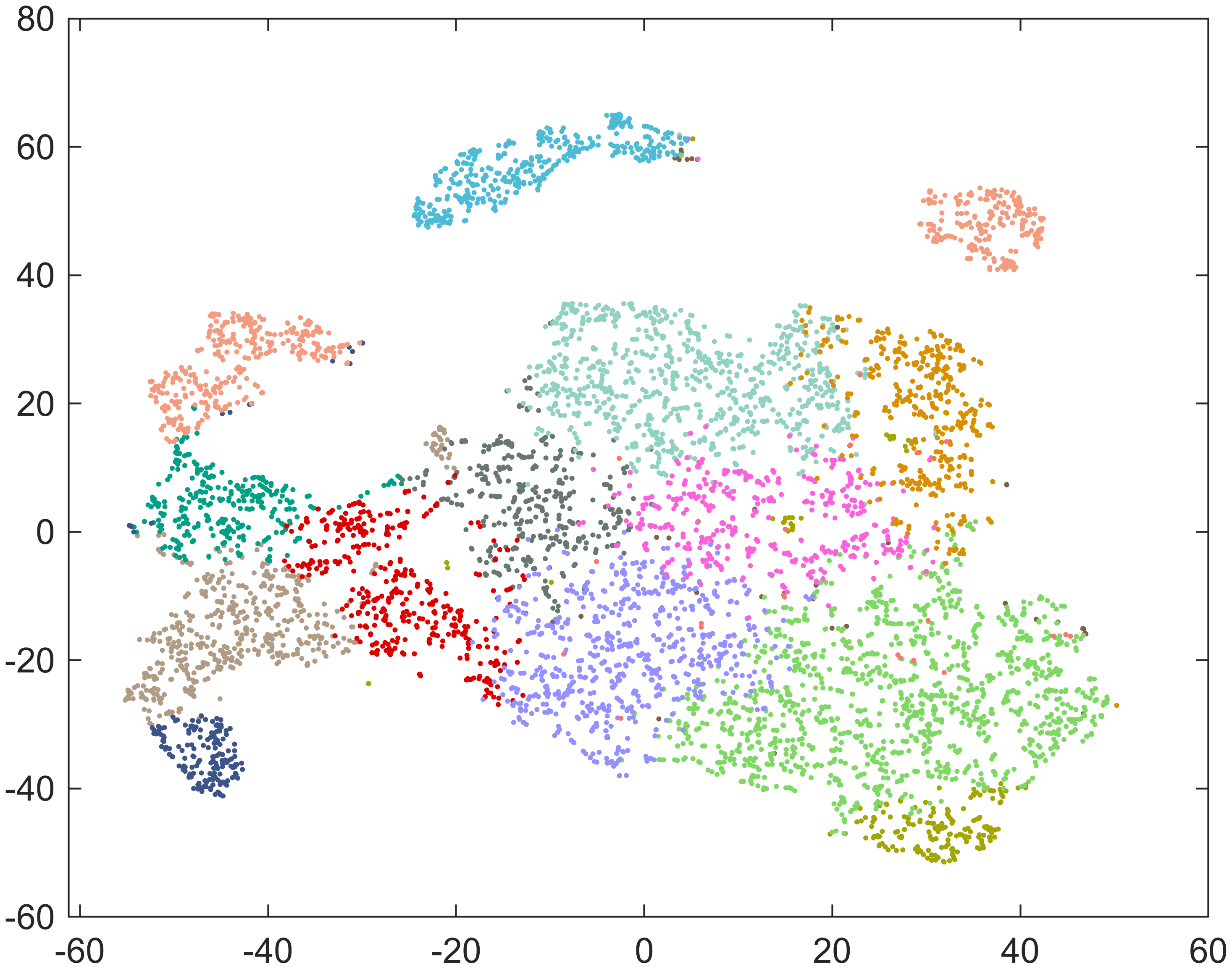}}
    \subfigure[]{\includegraphics[width=0.3\textwidth]{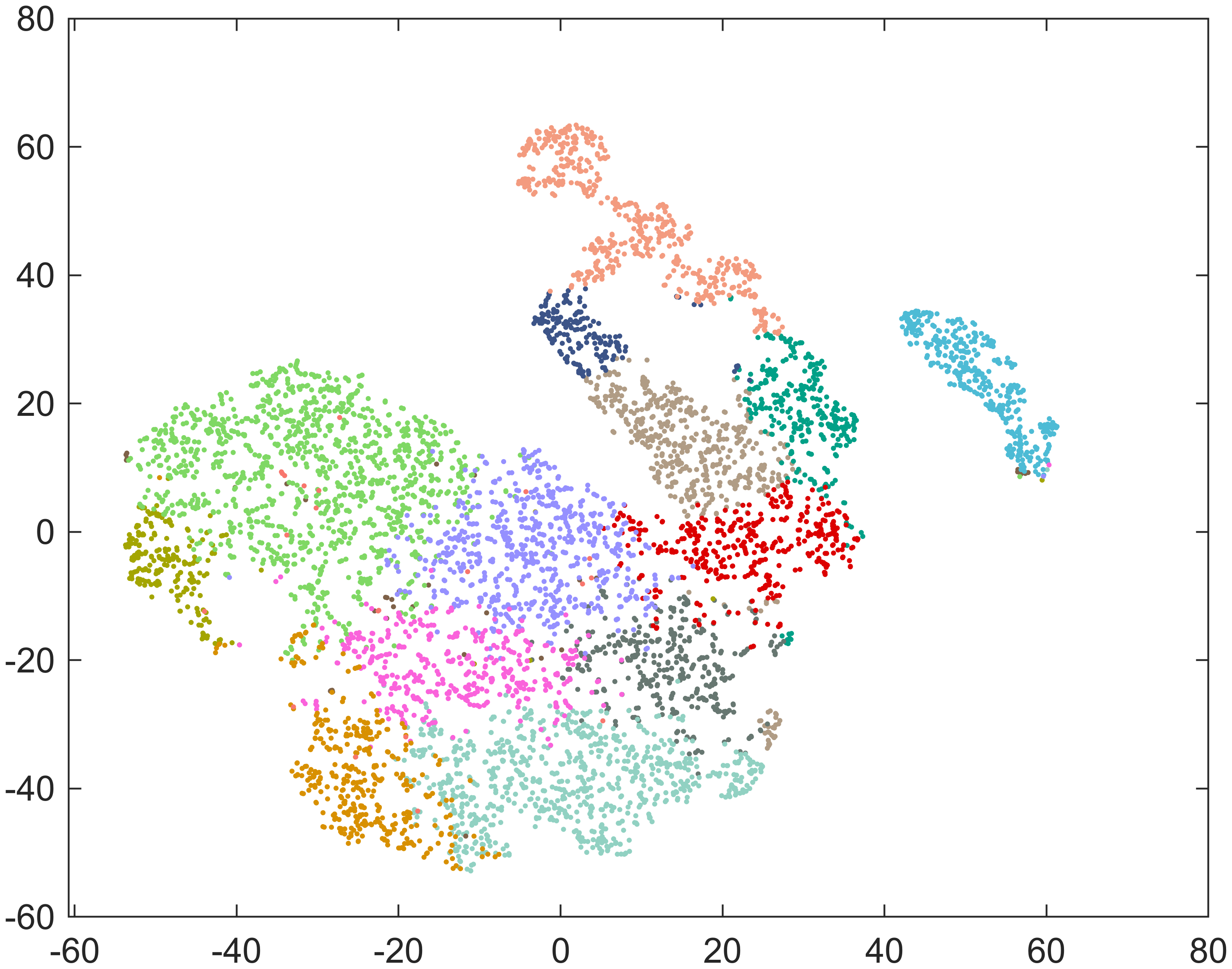}}
    \subfigure[]{\includegraphics[width=0.3\textwidth]{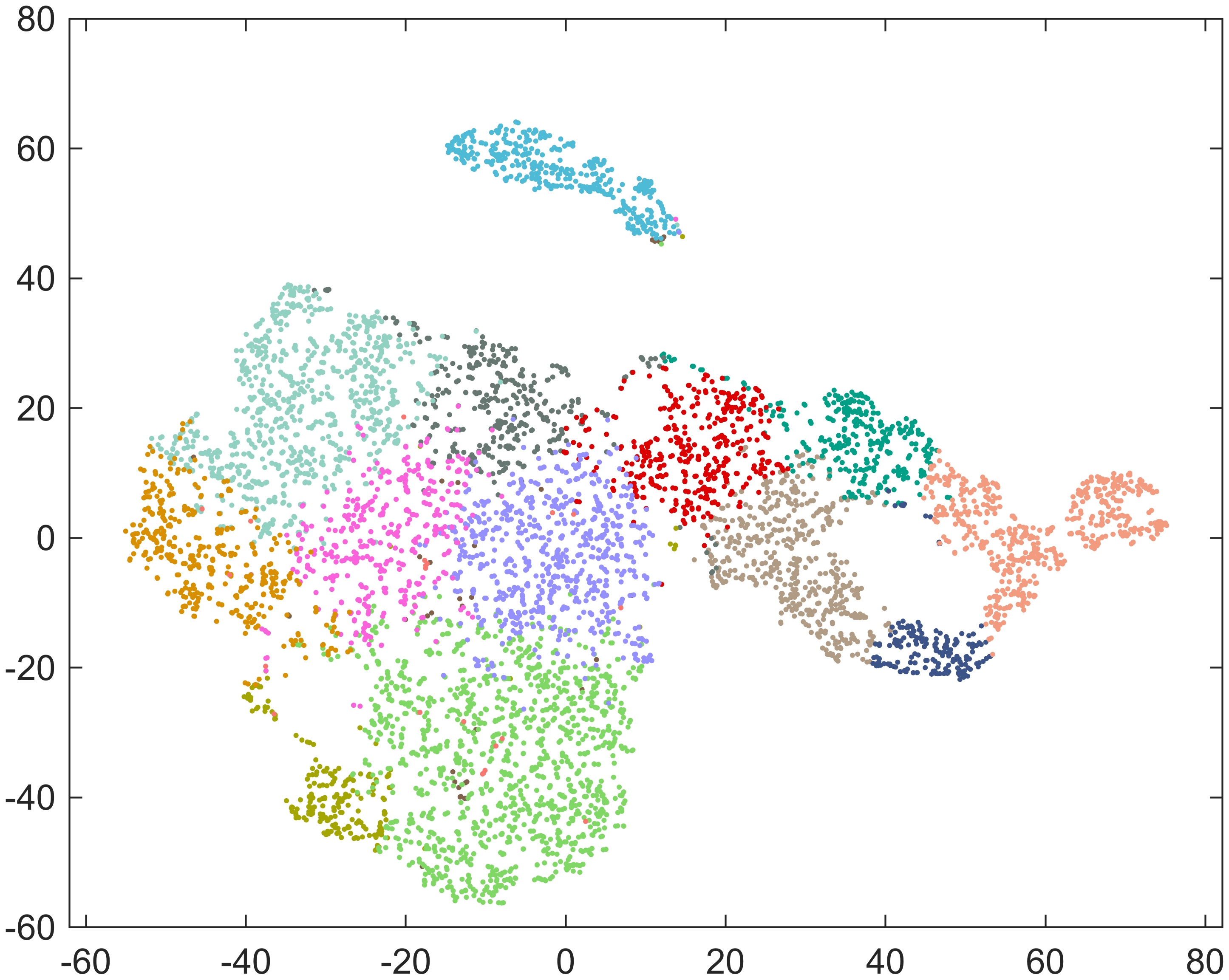}}
    \subfigure[]{\includegraphics[width=0.3\textwidth]{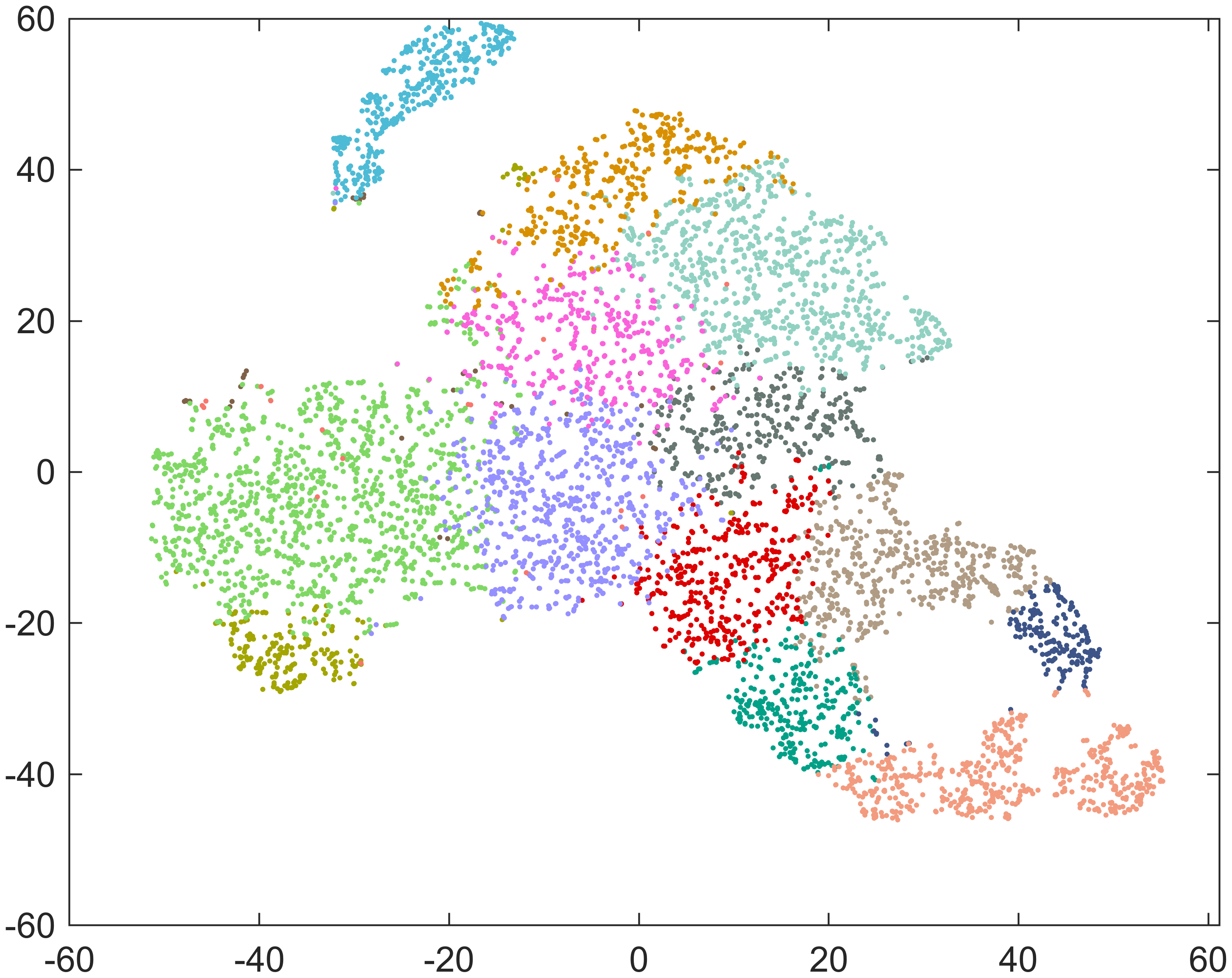}}
    \subfigure[]{\includegraphics[width=0.3\textwidth]{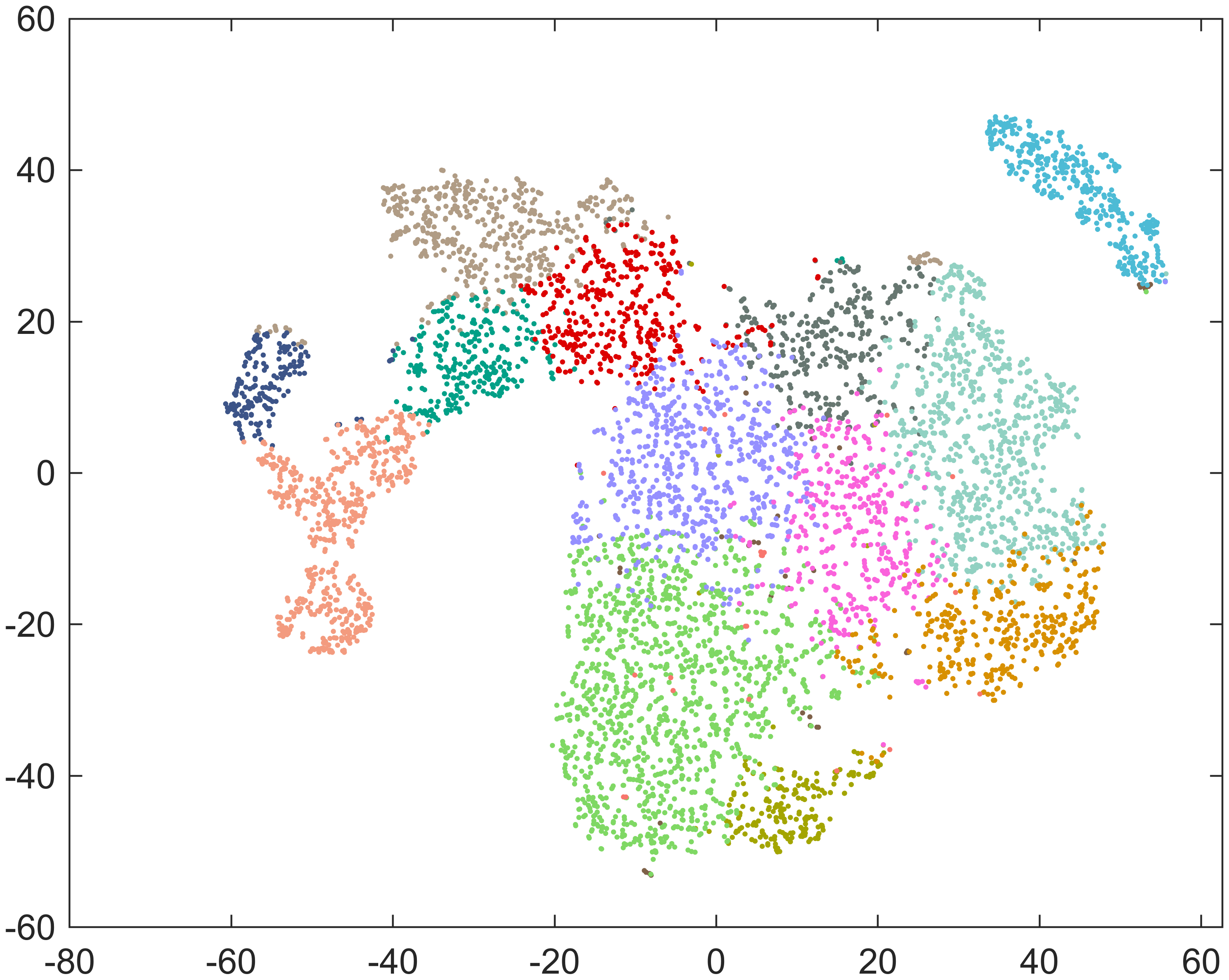}}
    \subfigure[]{\includegraphics[width=0.3\textwidth]{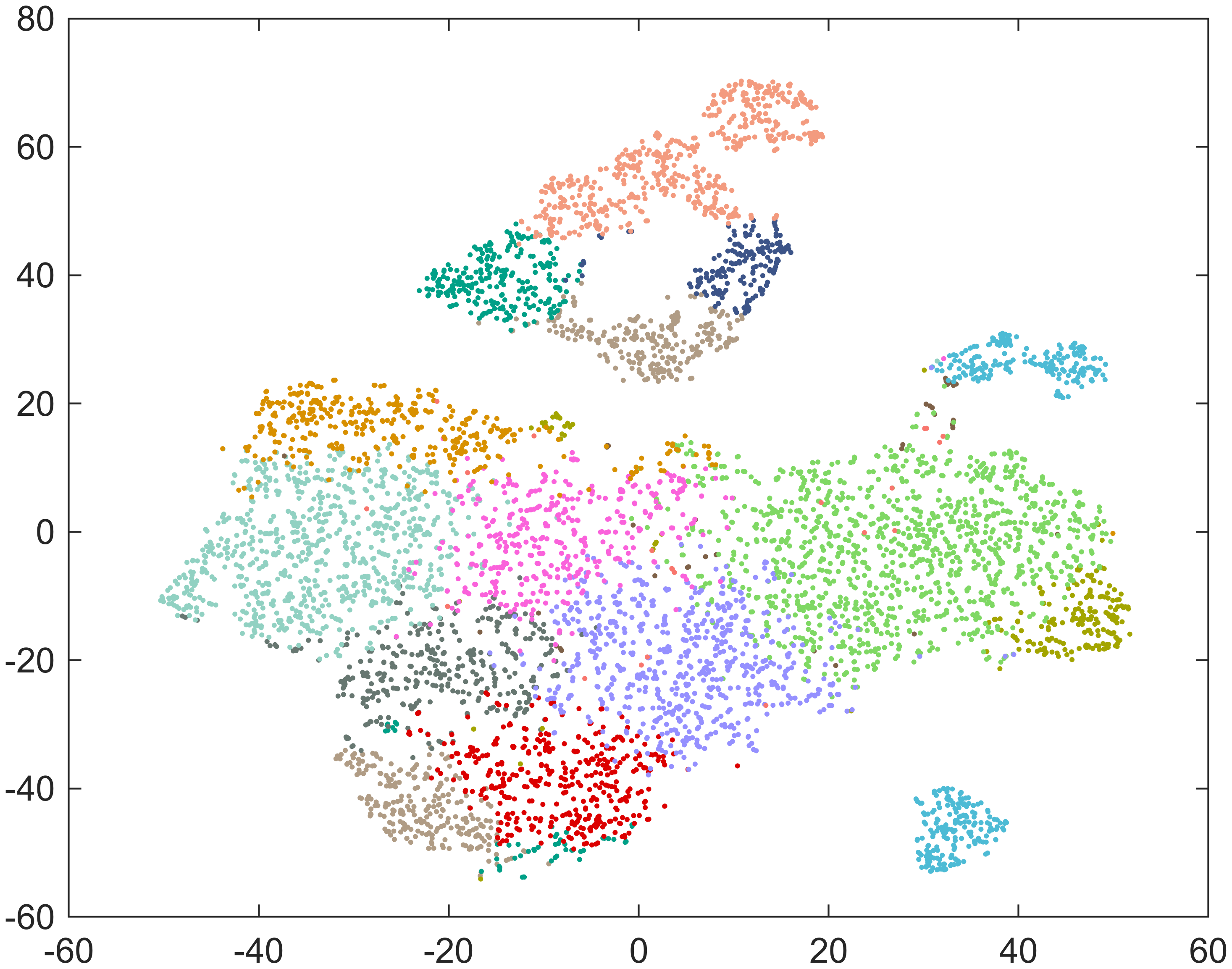}}
    \subfigure[]{\includegraphics[width=0.3\textwidth]{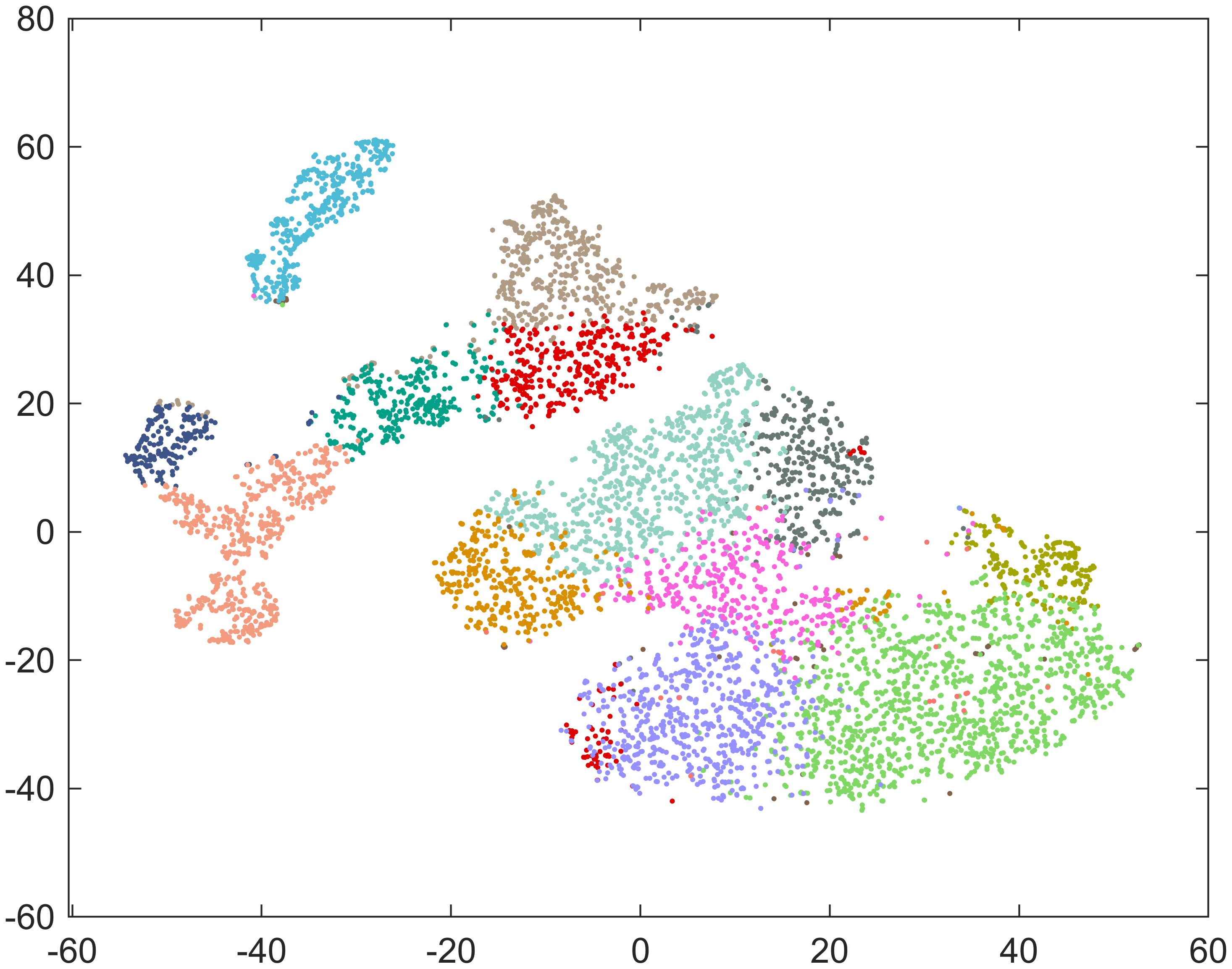}}
    \caption{Scatter plots for the two-dimensional t-distributed stochastic neighbour embedding of the result from the principal nested submanifolds with $d=1,\dots,12$ in (a)--(l).}
\end{figure}
\begin{figure}[htbp]
    \centering
    \subfigure[]{\includegraphics[width=0.3\textwidth]{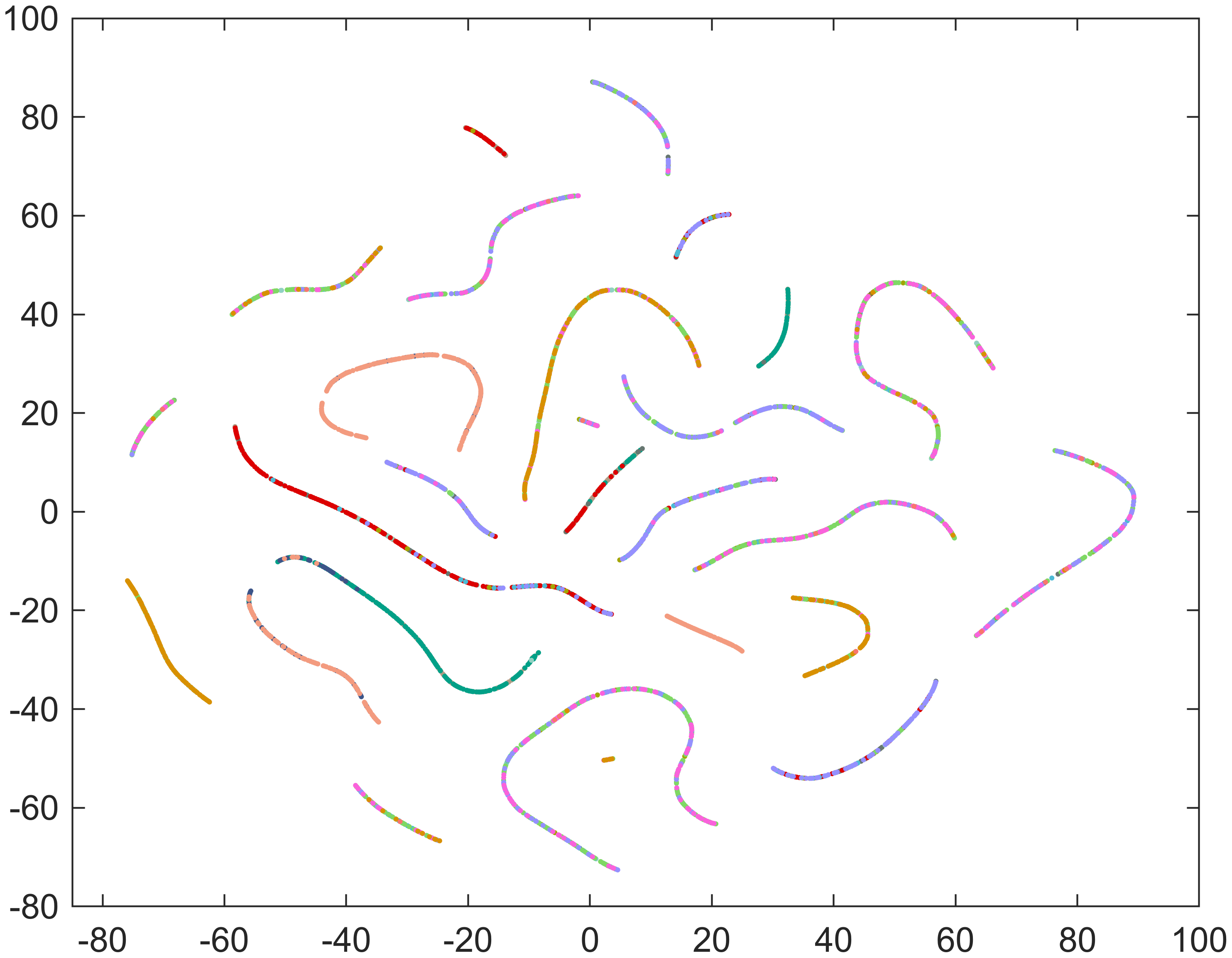}}
    \subfigure[]{\includegraphics[width=0.3\textwidth]{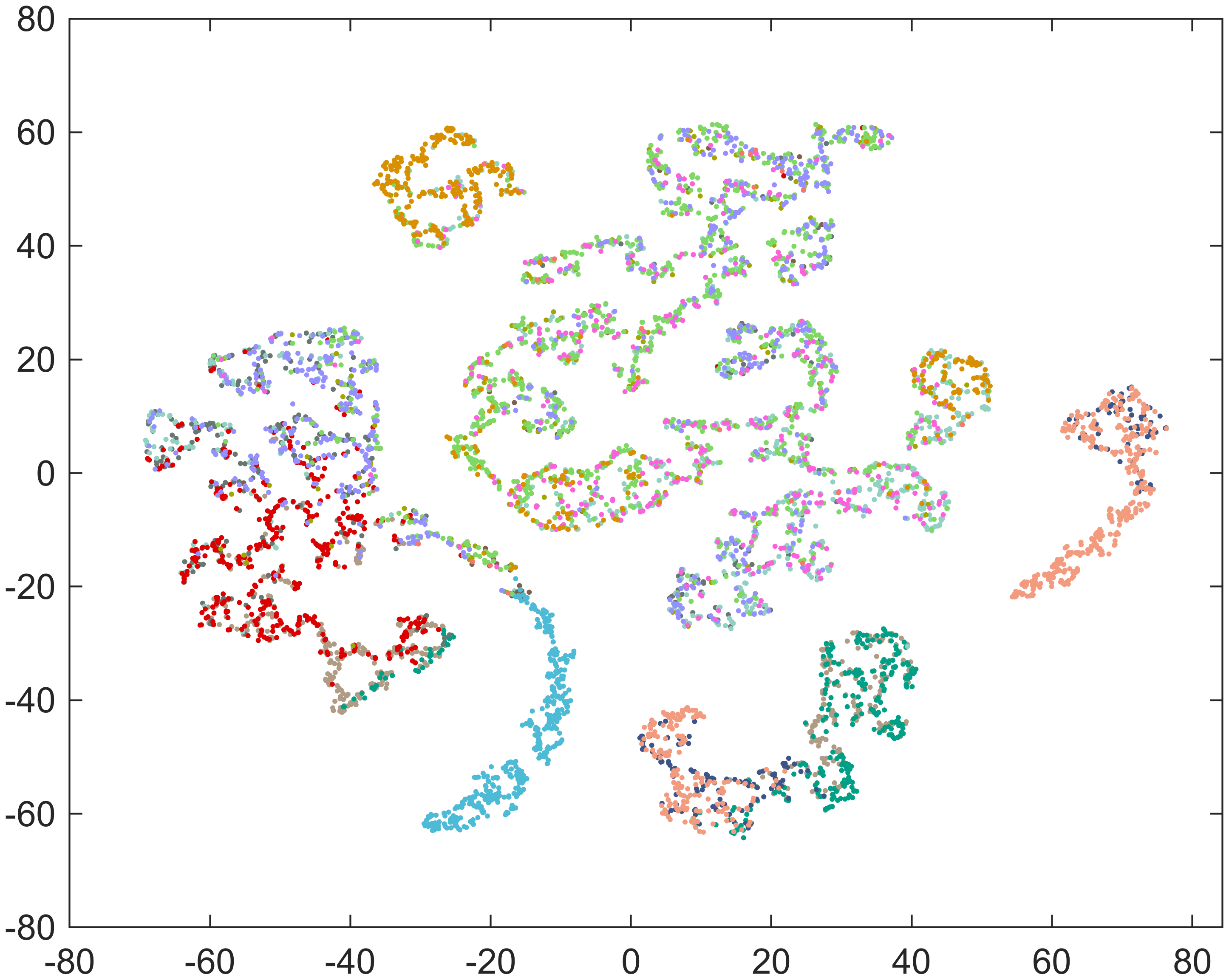}}
    \subfigure[]{\includegraphics[width=0.3\textwidth]{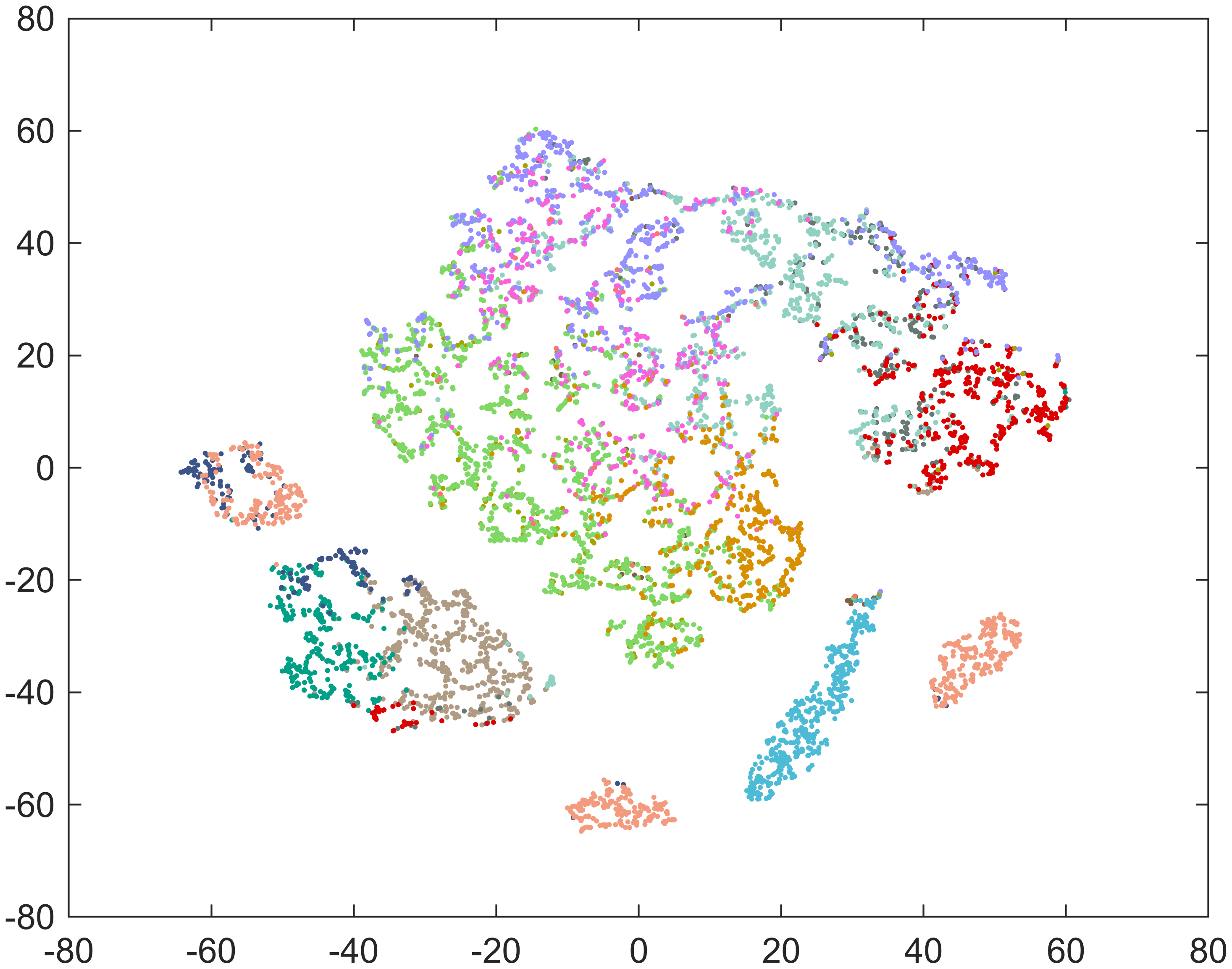}}
    \subfigure[]{\includegraphics[width=0.3\textwidth]{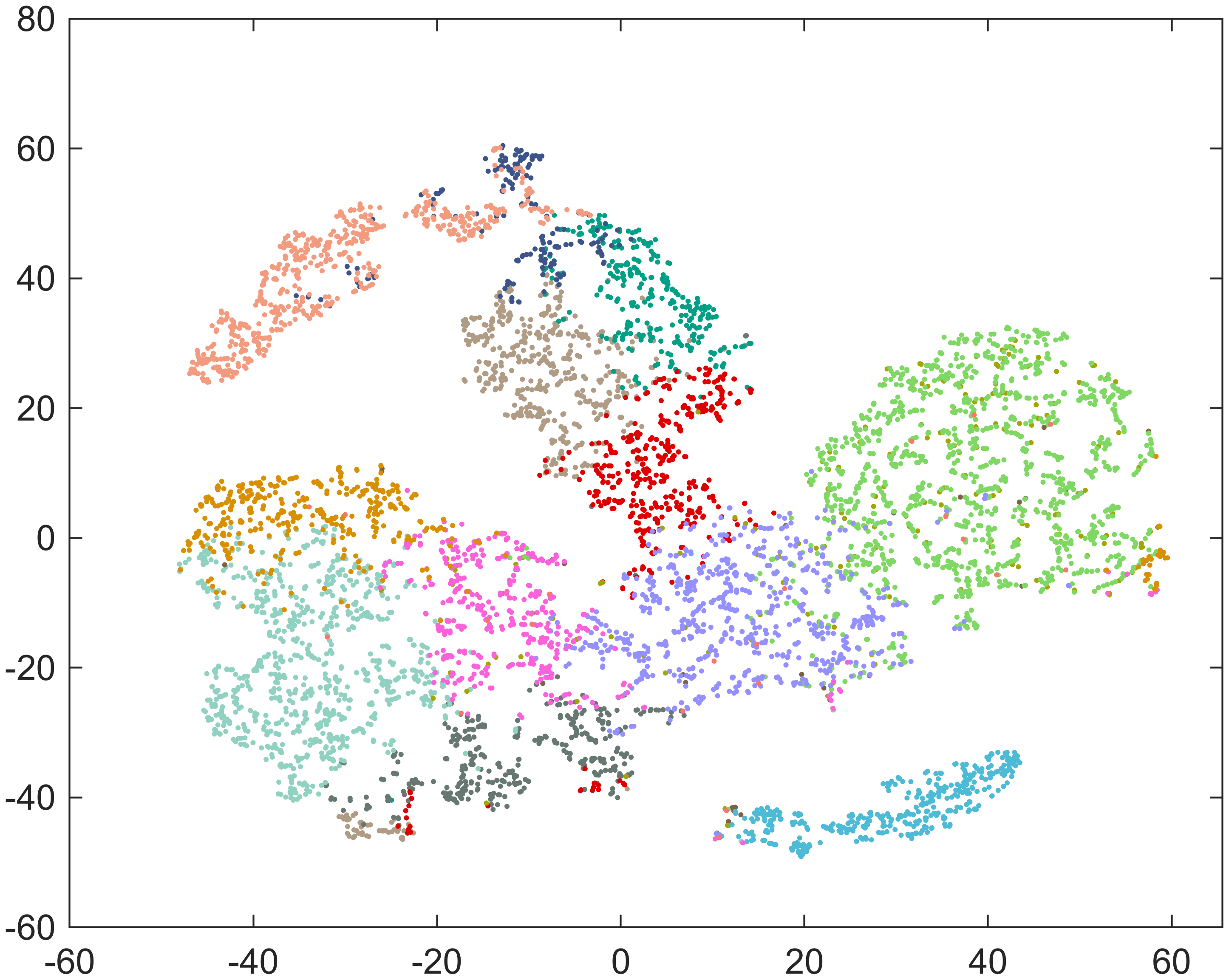}}
    \subfigure[]{\includegraphics[width=0.3\textwidth]{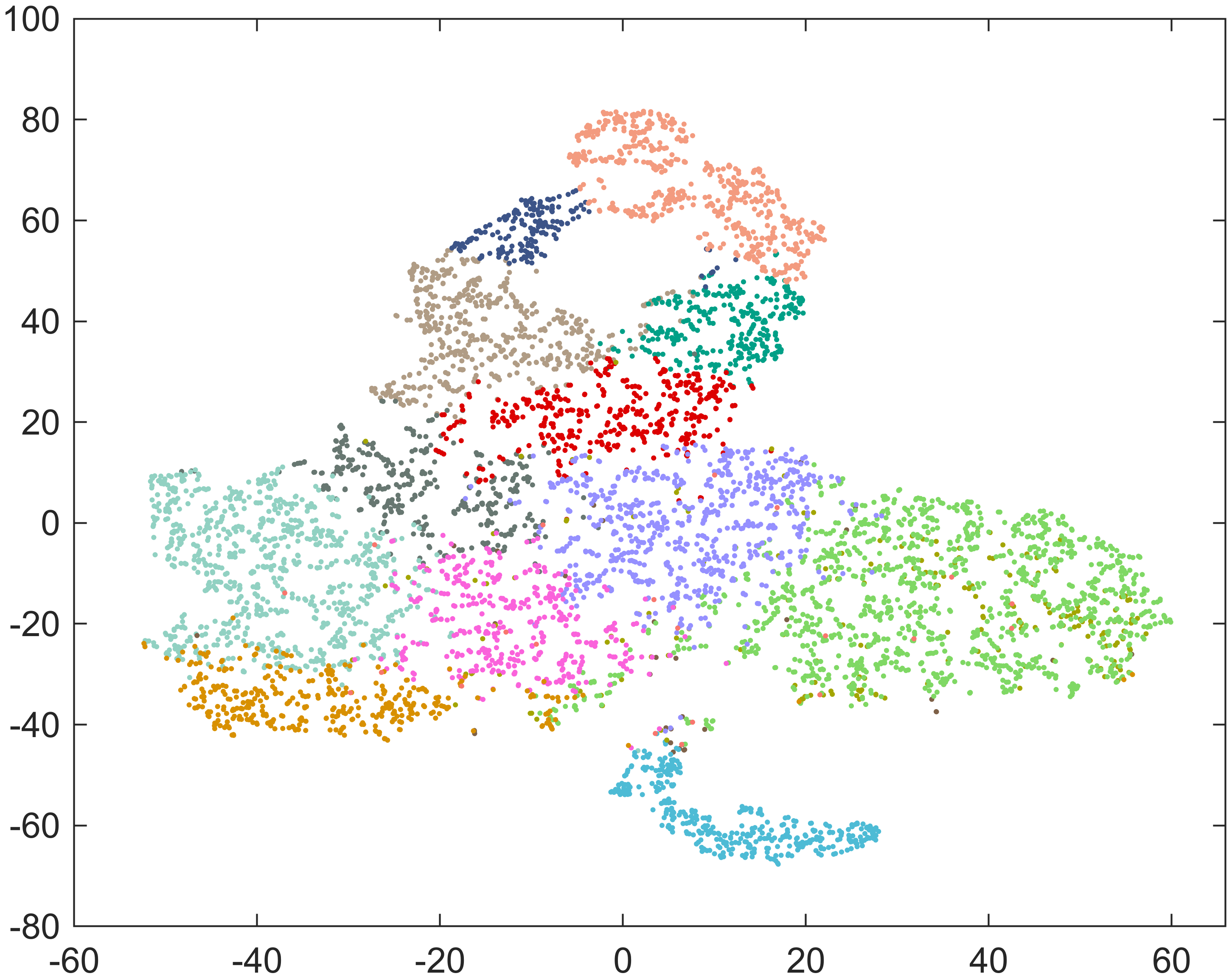}}
    \subfigure[]{\includegraphics[width=0.3\textwidth]{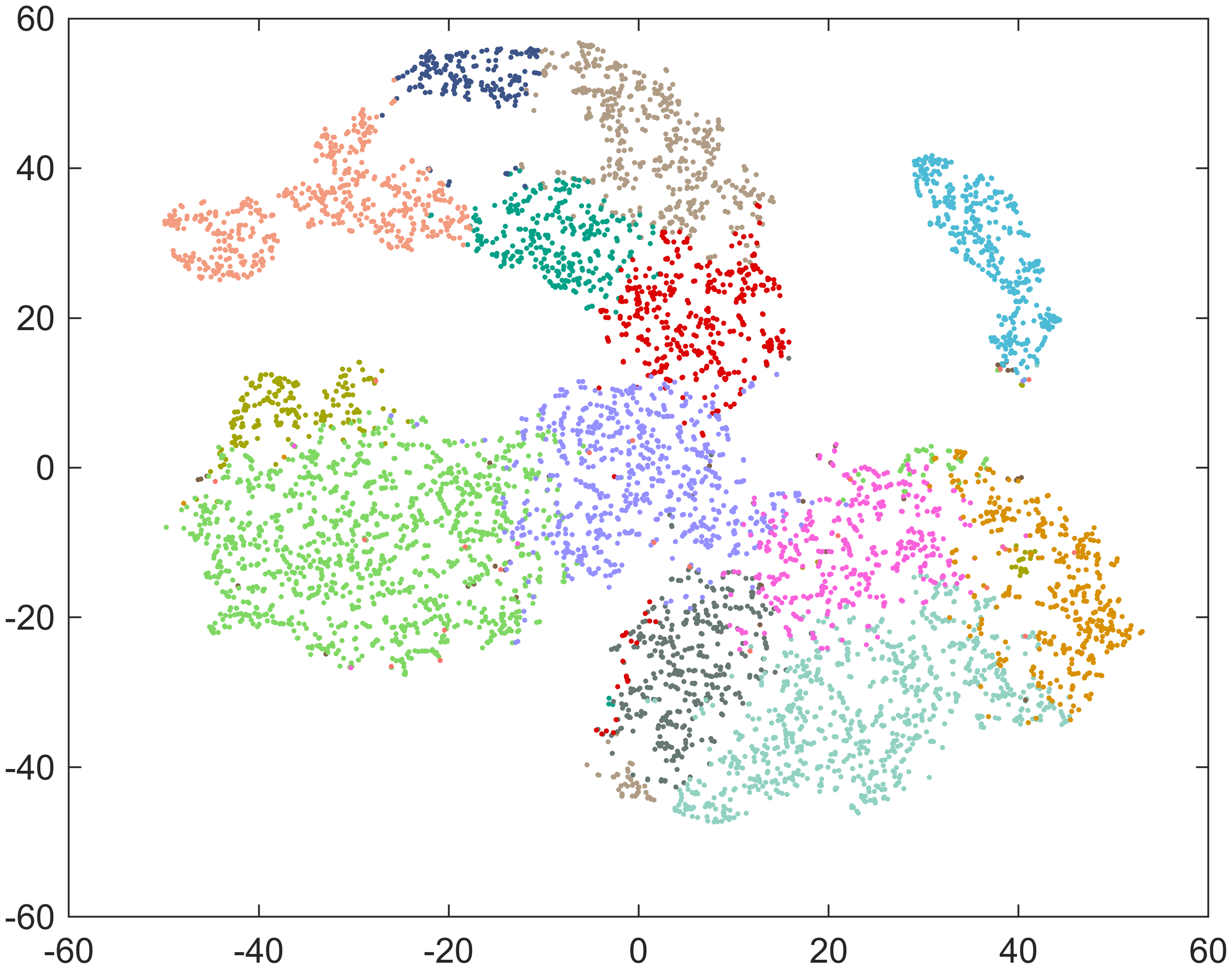}}
    \subfigure[]{\includegraphics[width=0.3\textwidth]{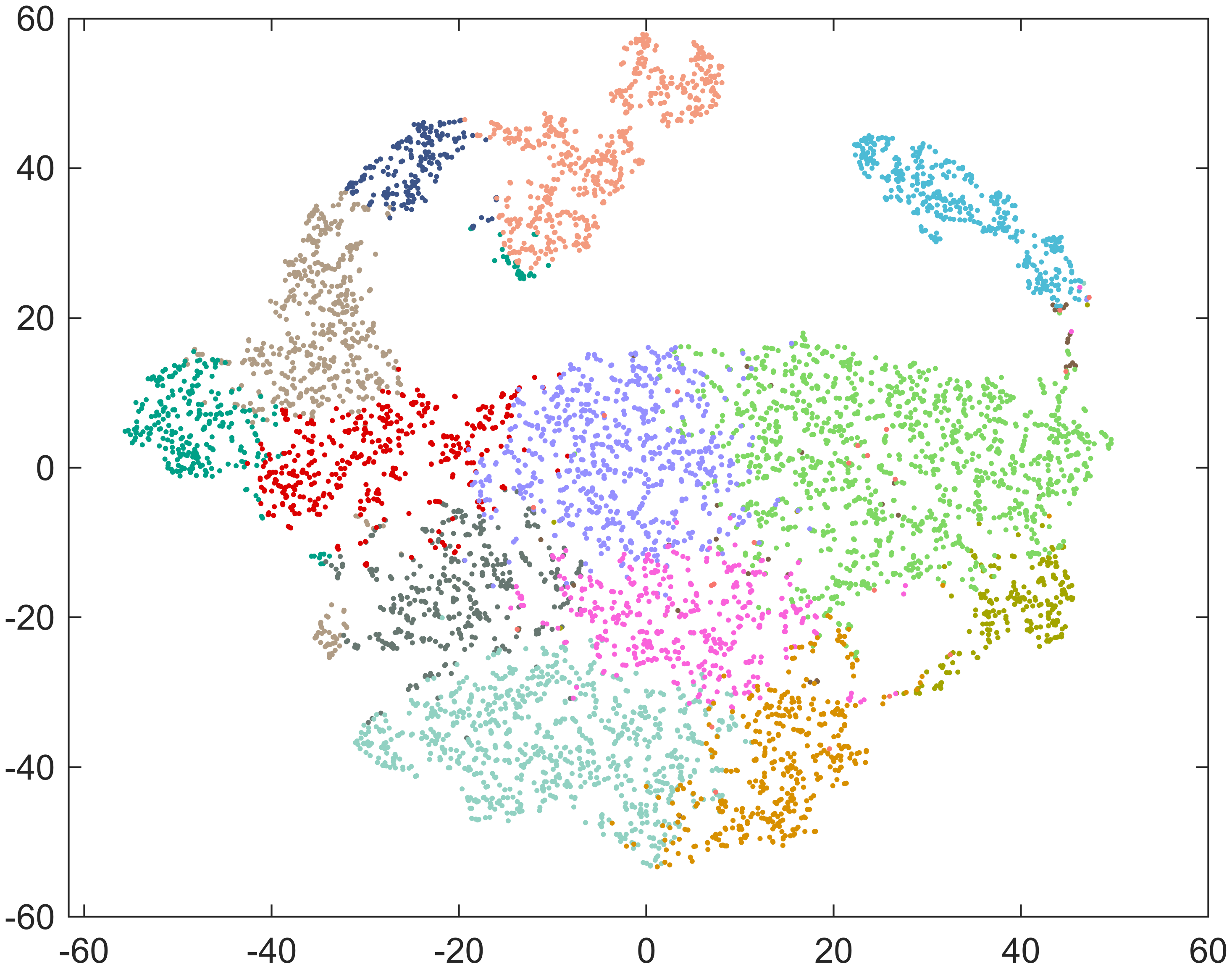}}
    \subfigure[]{\includegraphics[width=0.3\textwidth]{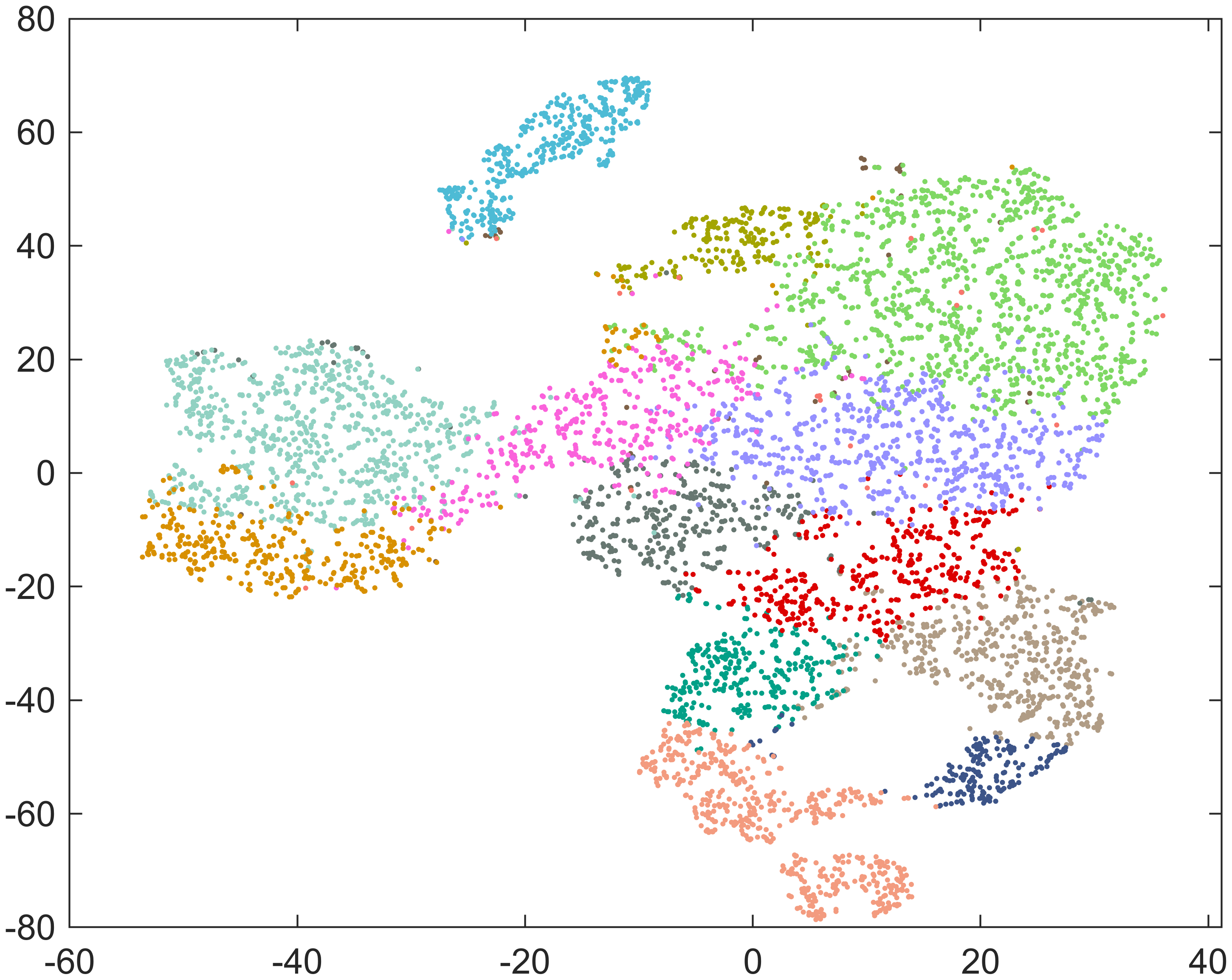}}
    \subfigure[]{\includegraphics[width=0.3\textwidth]{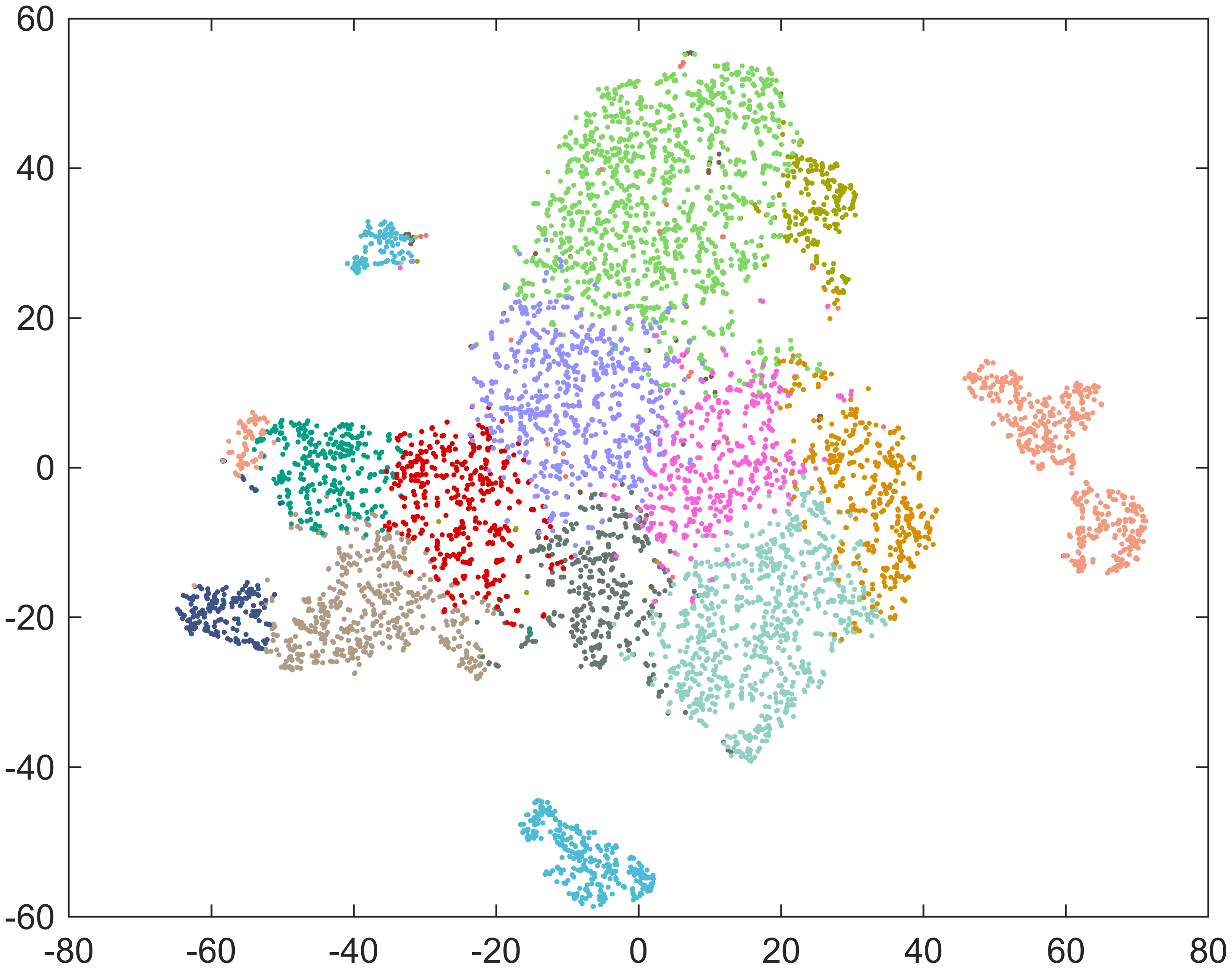}}
    \subfigure[]{\includegraphics[width=0.3\textwidth]{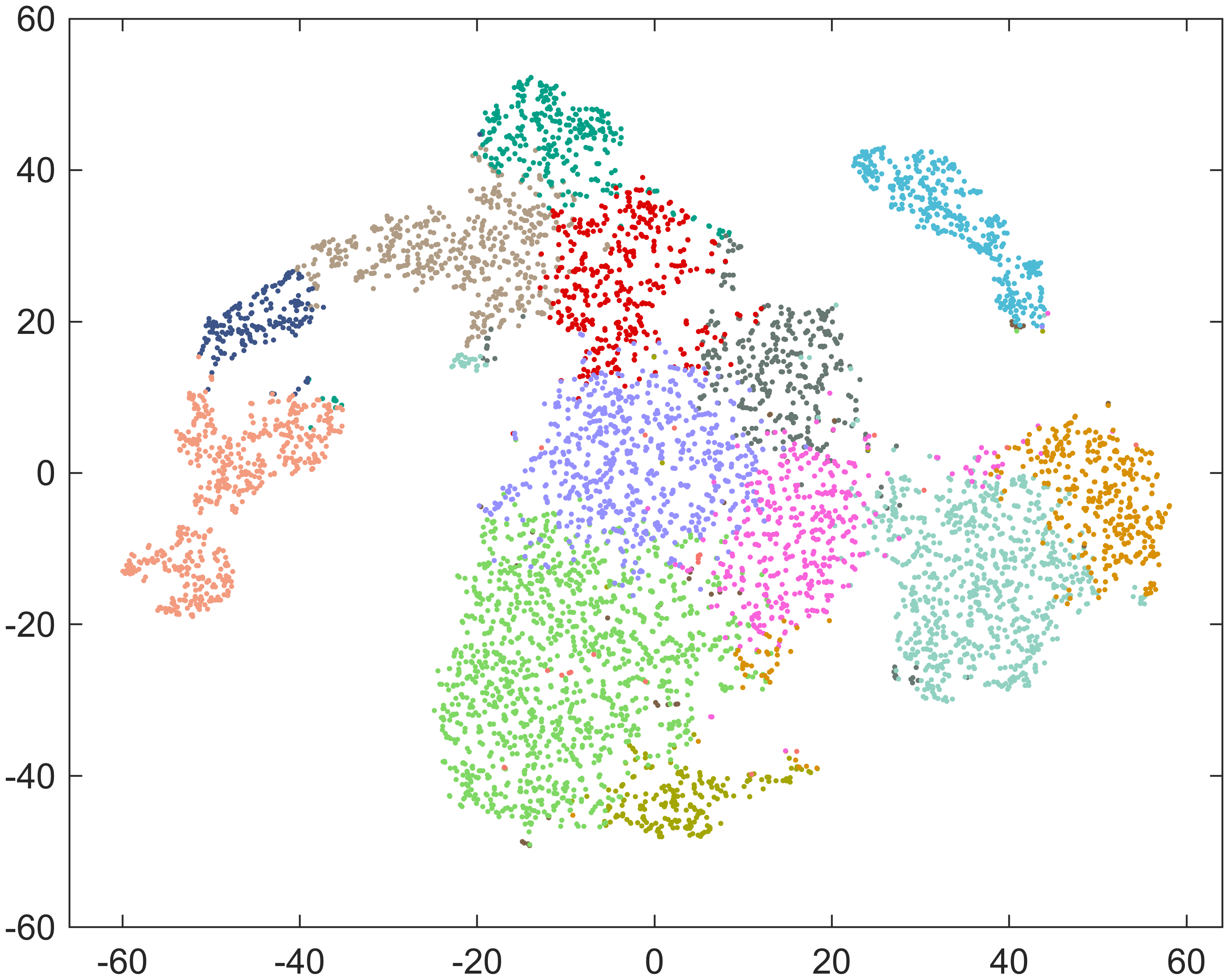}}
    \subfigure[]{\includegraphics[width=0.3\textwidth]{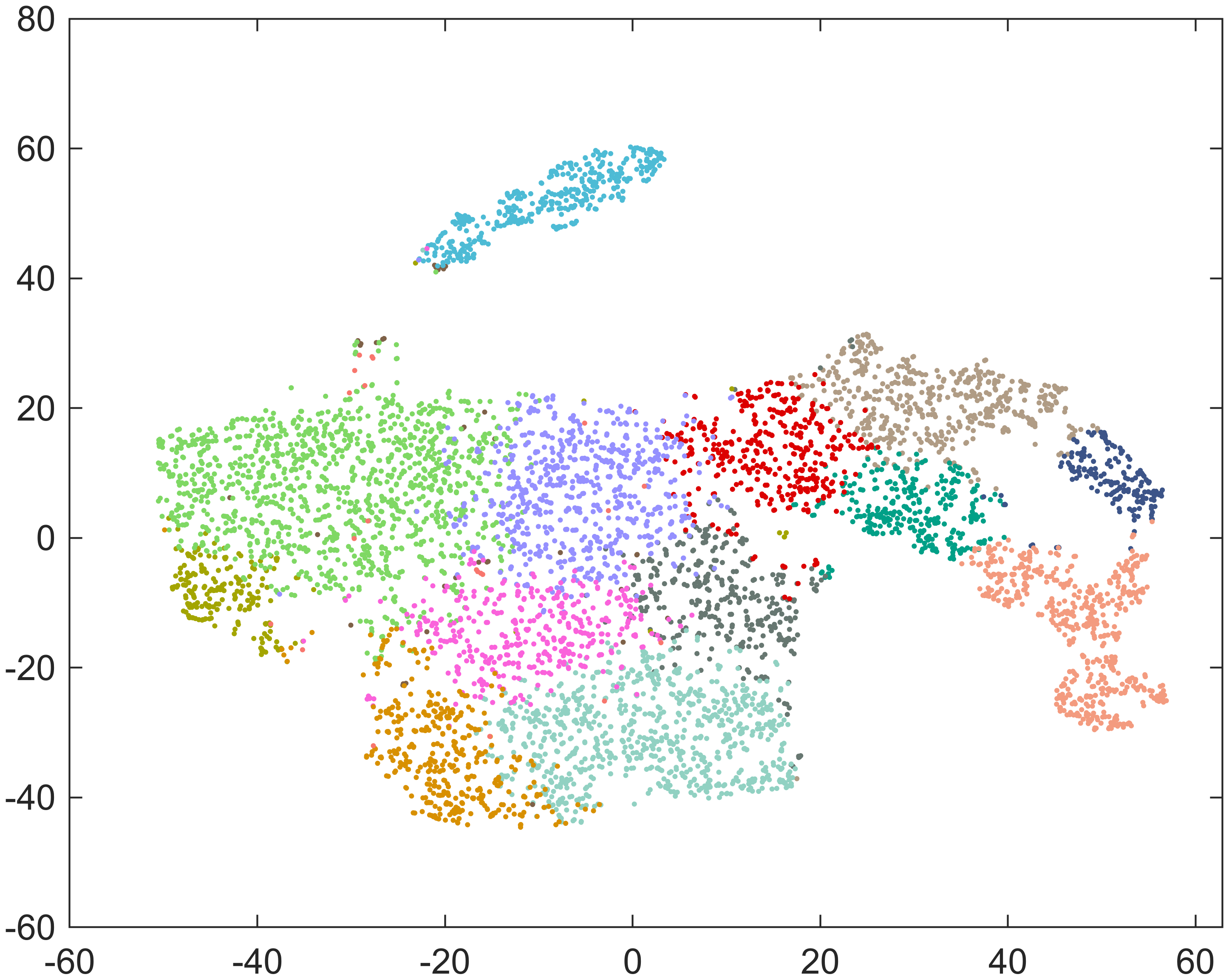}}
    \subfigure[]{\includegraphics[width=0.3\textwidth]{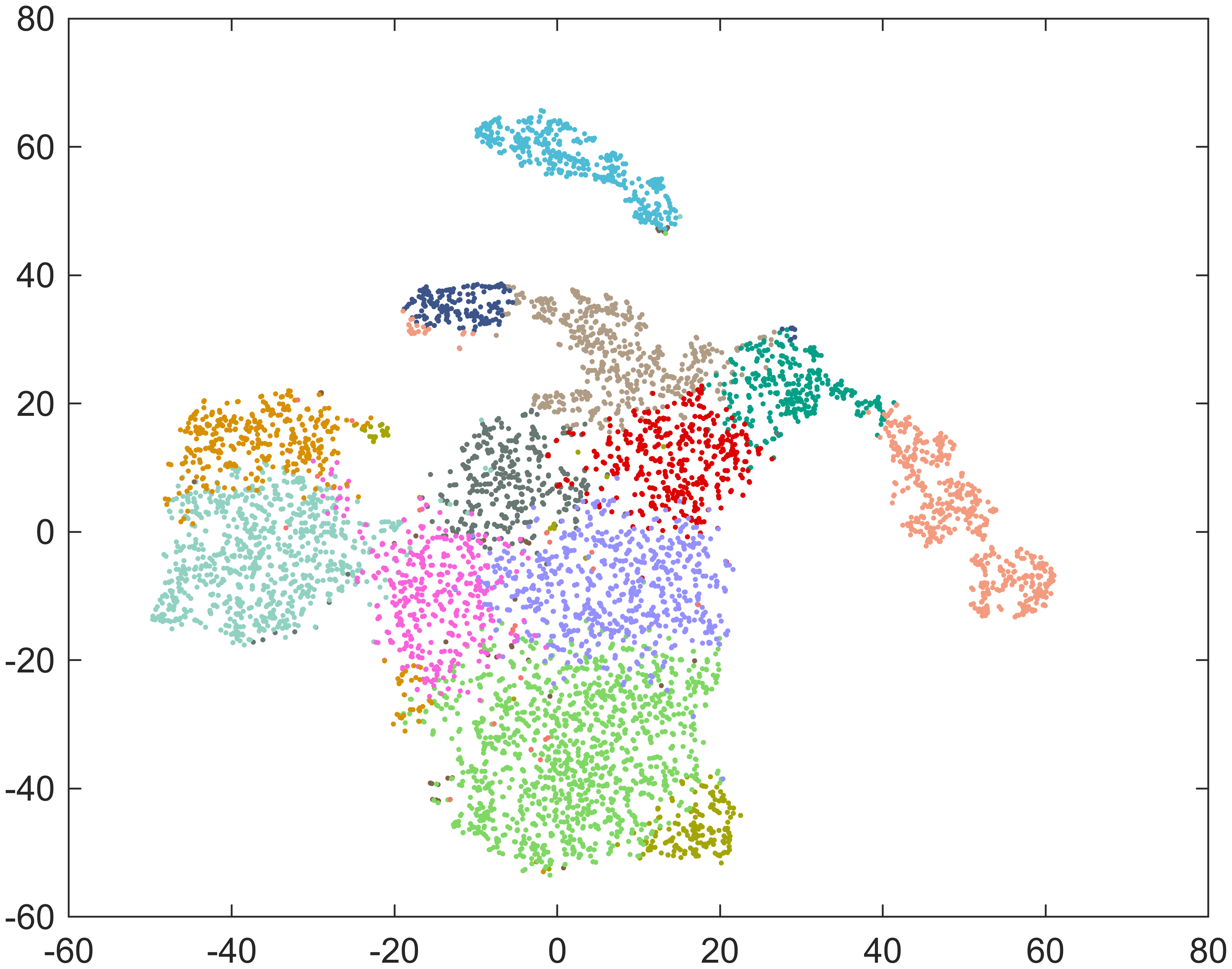}}
    \caption{Scatter plots for the two-dimensional t-distributed stochastic neighbour embedding of the result from the principal component analysis with $d=1,\dots,12$ in (a)--(l).}
\end{figure}
\begin{figure}[htbp]
    \centering
    \subfigure[]{\includegraphics[width=0.3\textwidth]{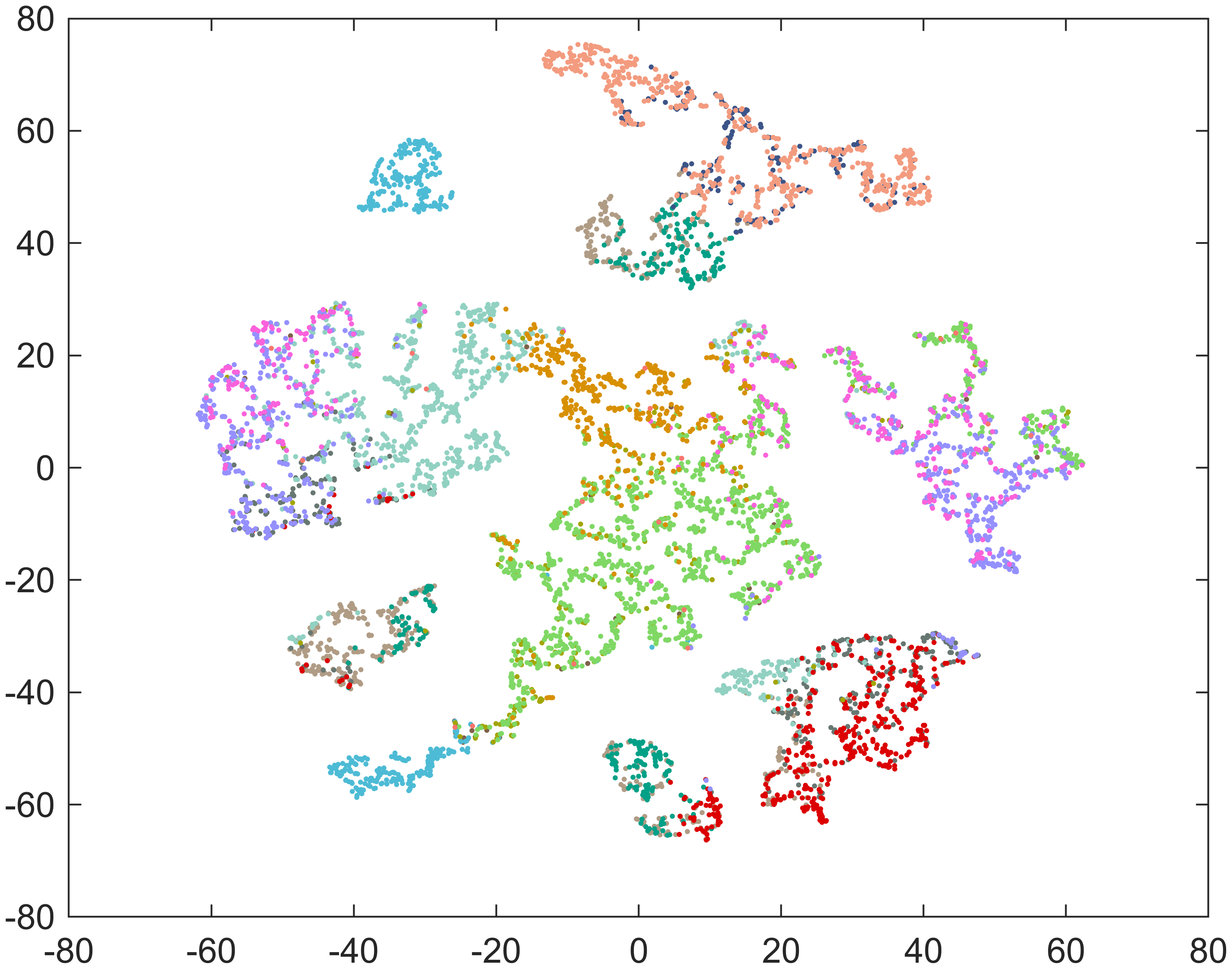}}
    \subfigure[]{\includegraphics[width=0.3\textwidth]{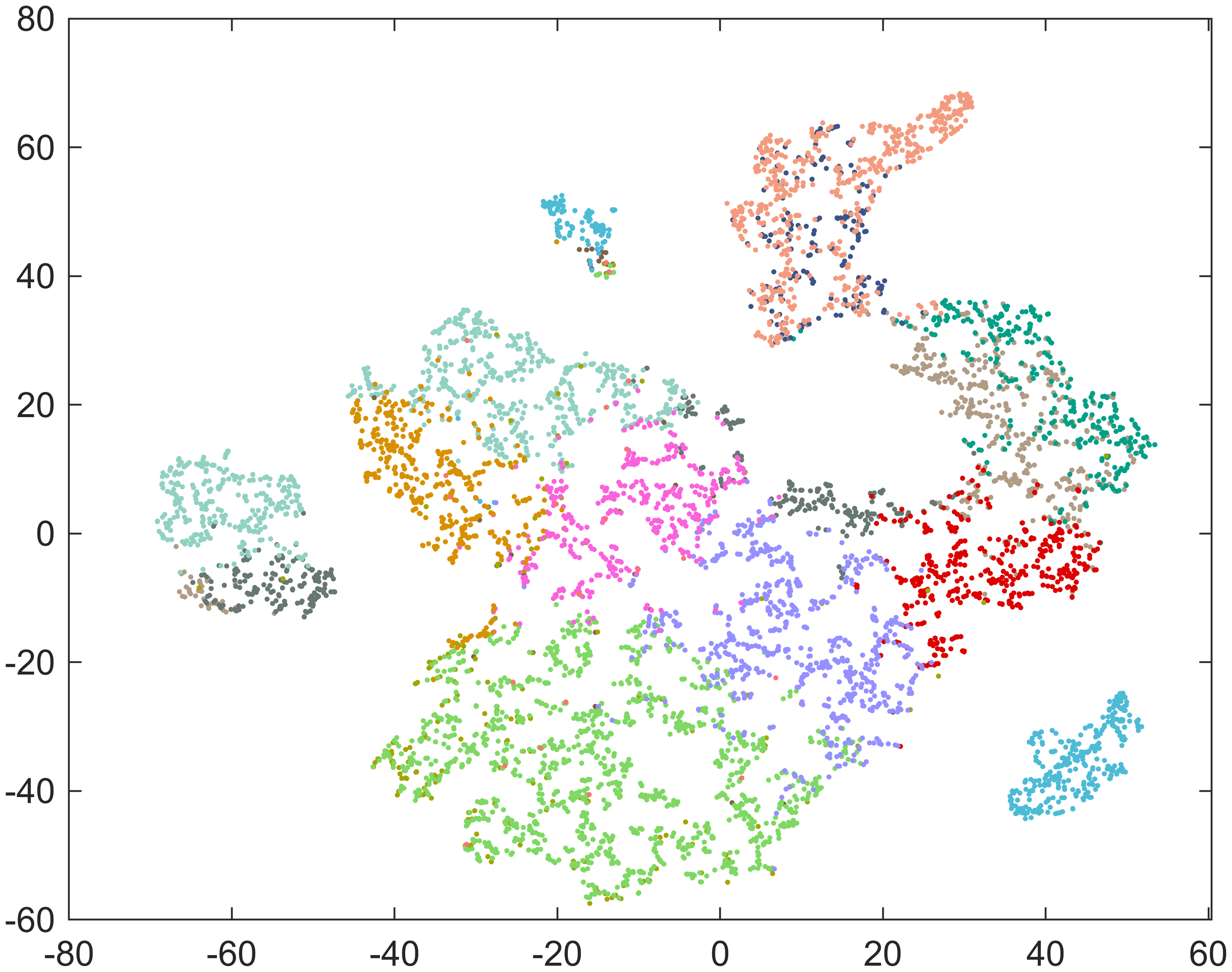}}
    \subfigure[]{\includegraphics[width=0.3\textwidth]{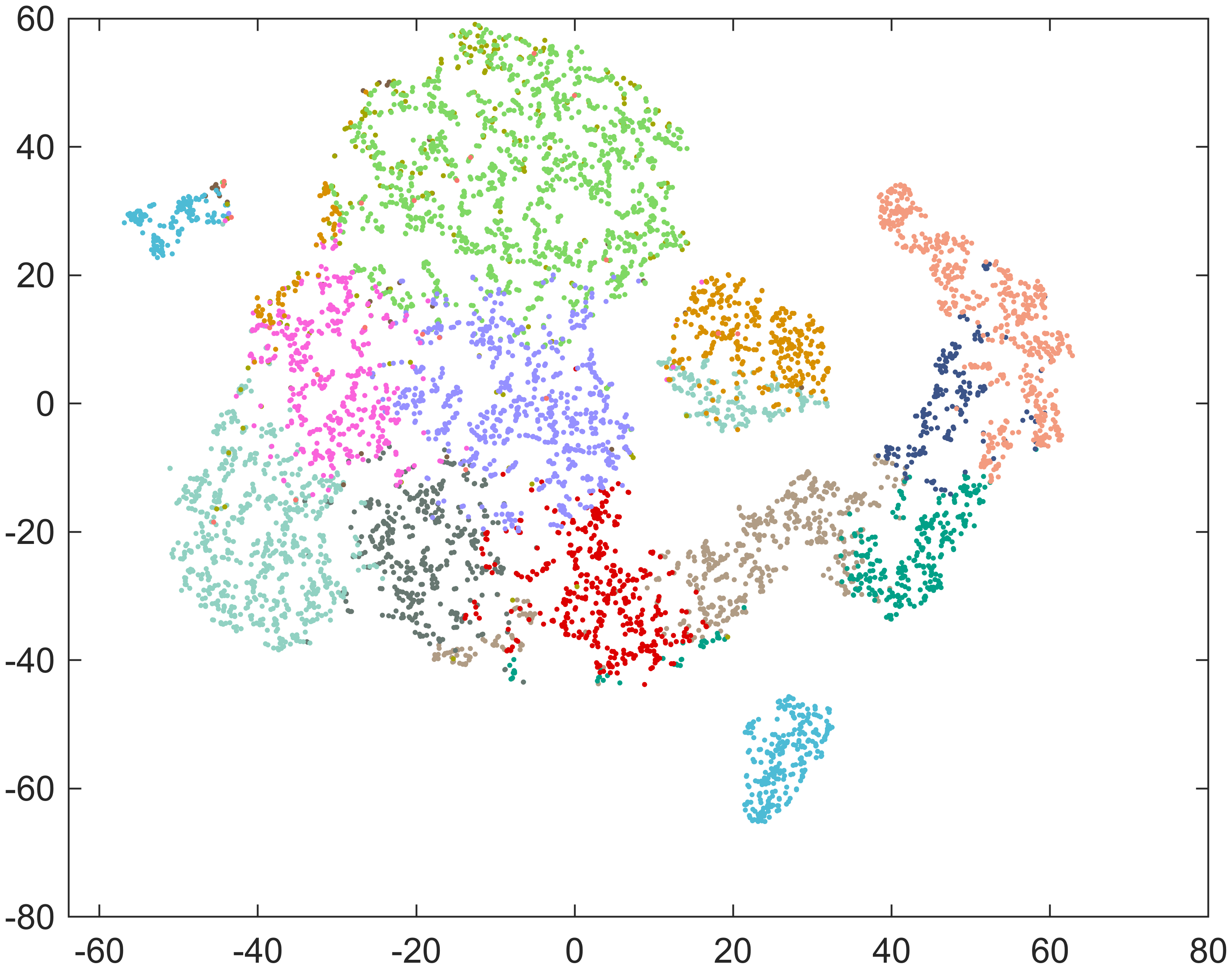}}
    \subfigure[]{\includegraphics[width=0.3\textwidth]{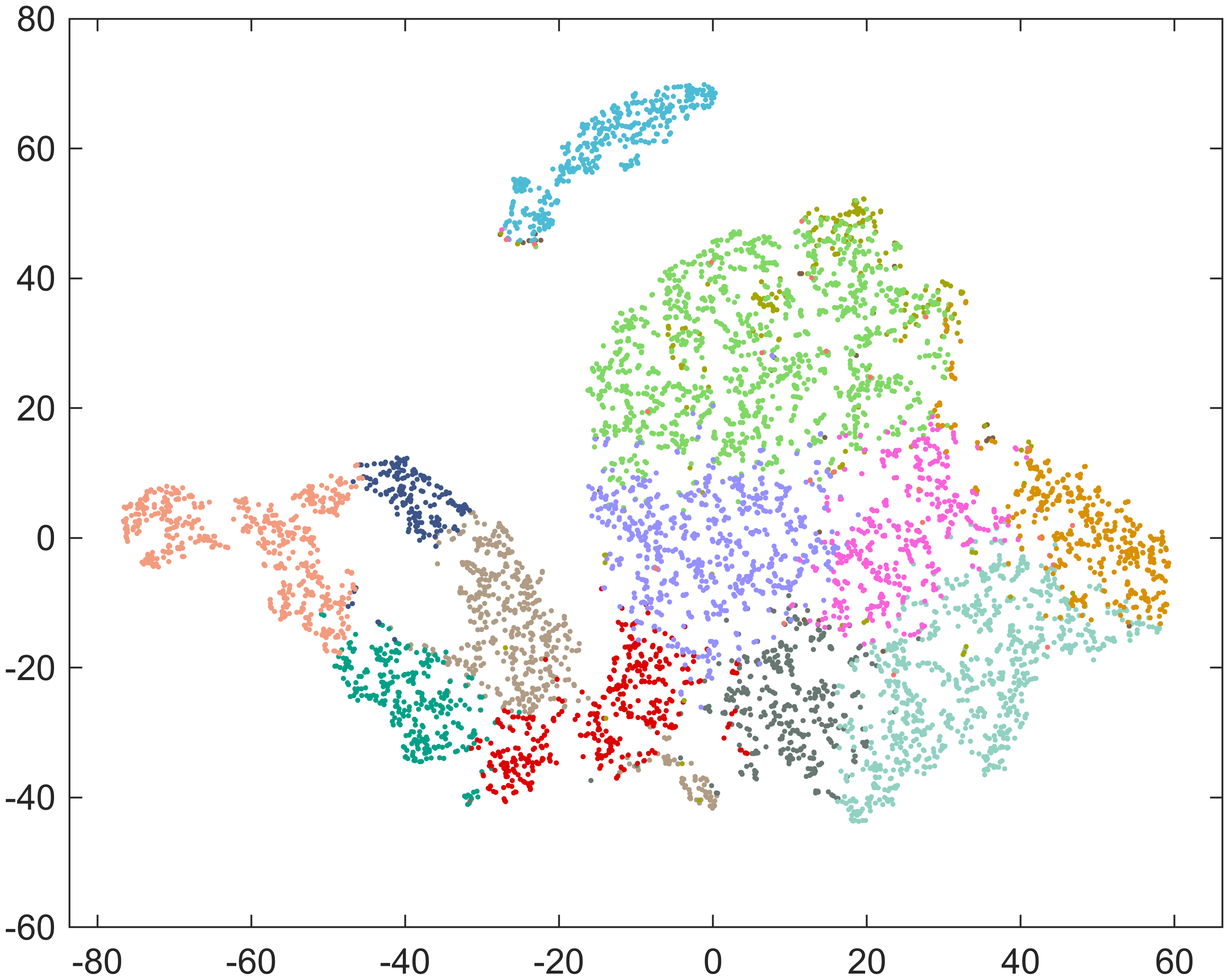}}
    \subfigure[]{\includegraphics[width=0.3\textwidth]{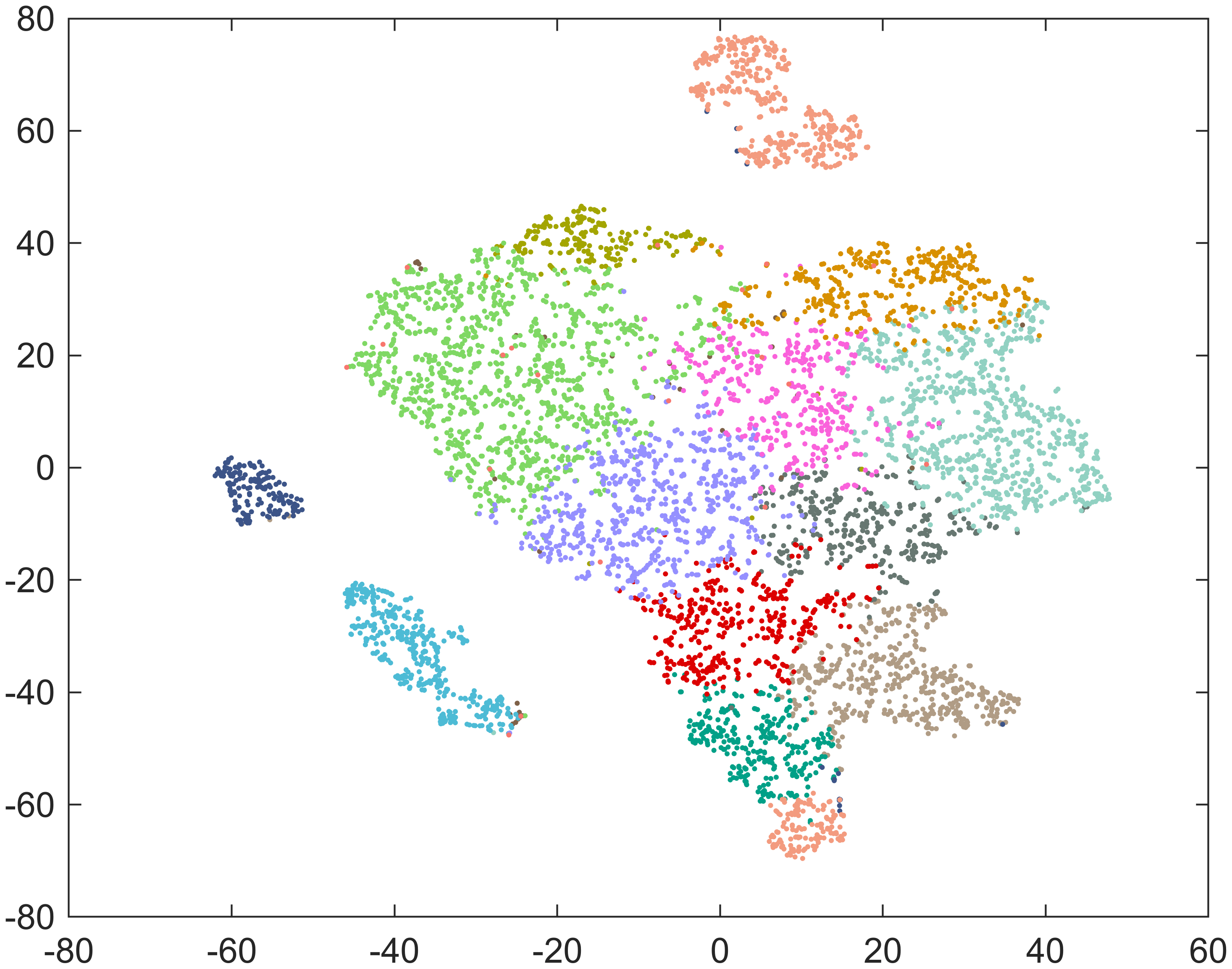}}
    \subfigure[]{\includegraphics[width=0.3\textwidth]{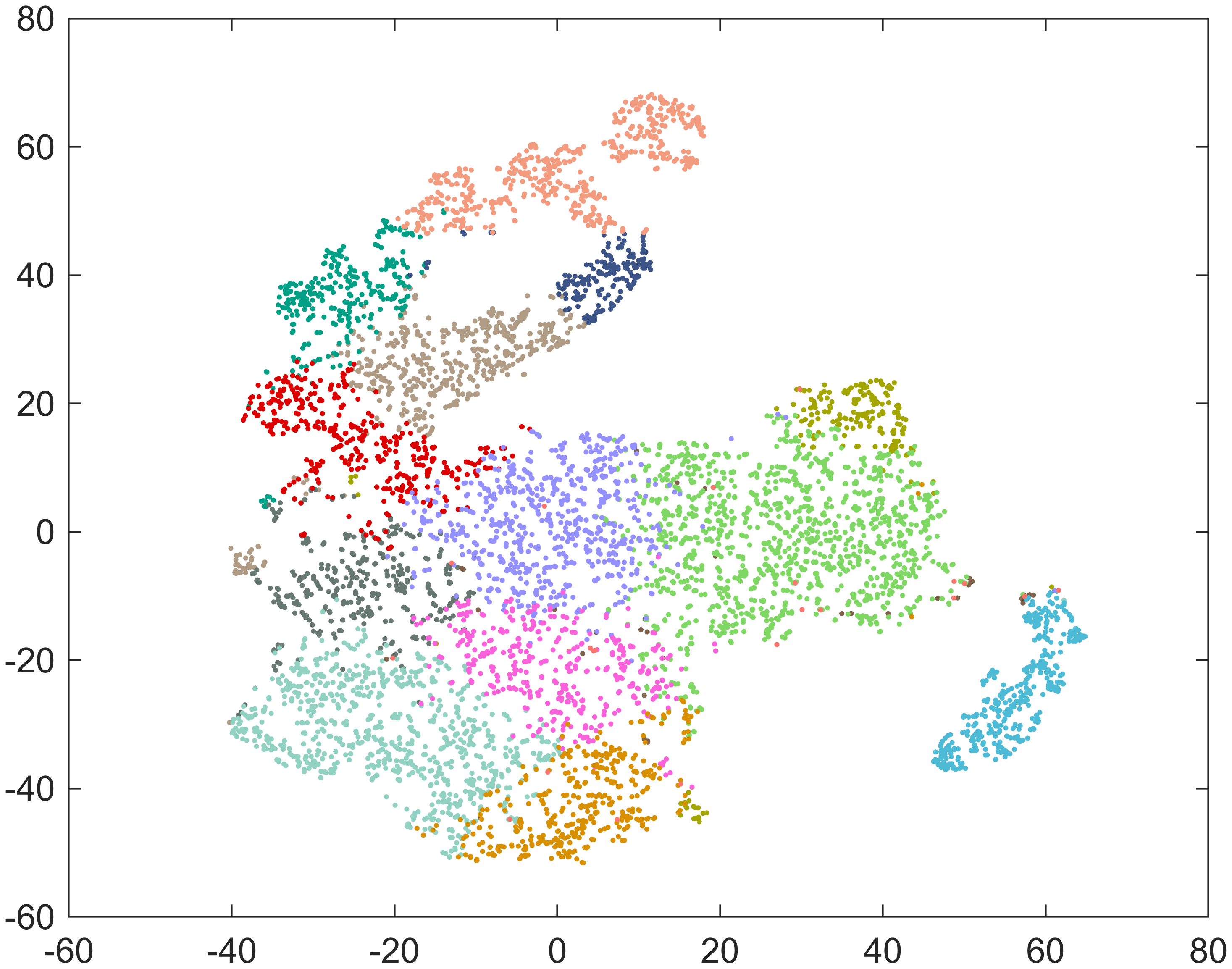}}
    \subfigure[]{\includegraphics[width=0.3\textwidth]{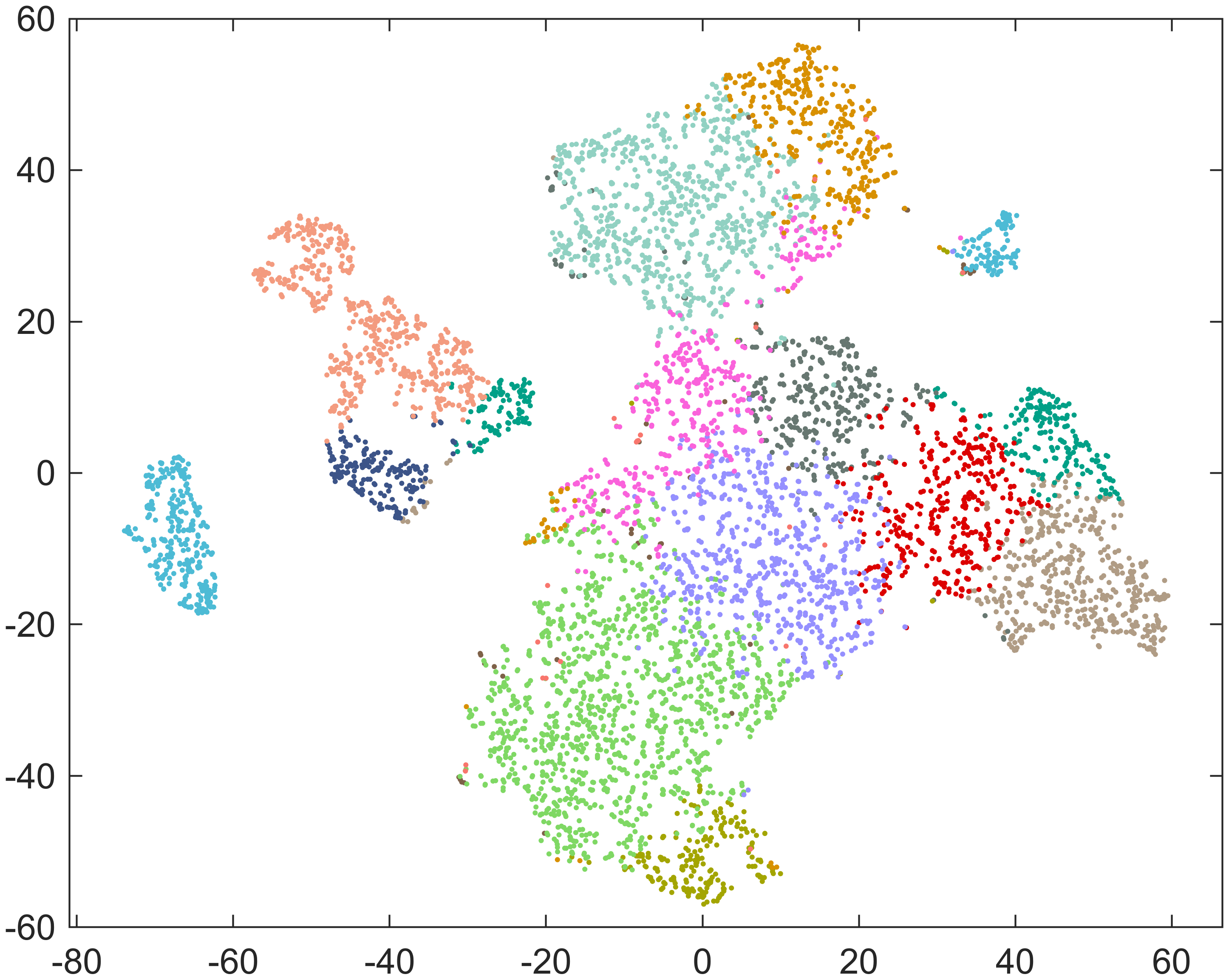}}
    \subfigure[]{\includegraphics[width=0.3\textwidth]{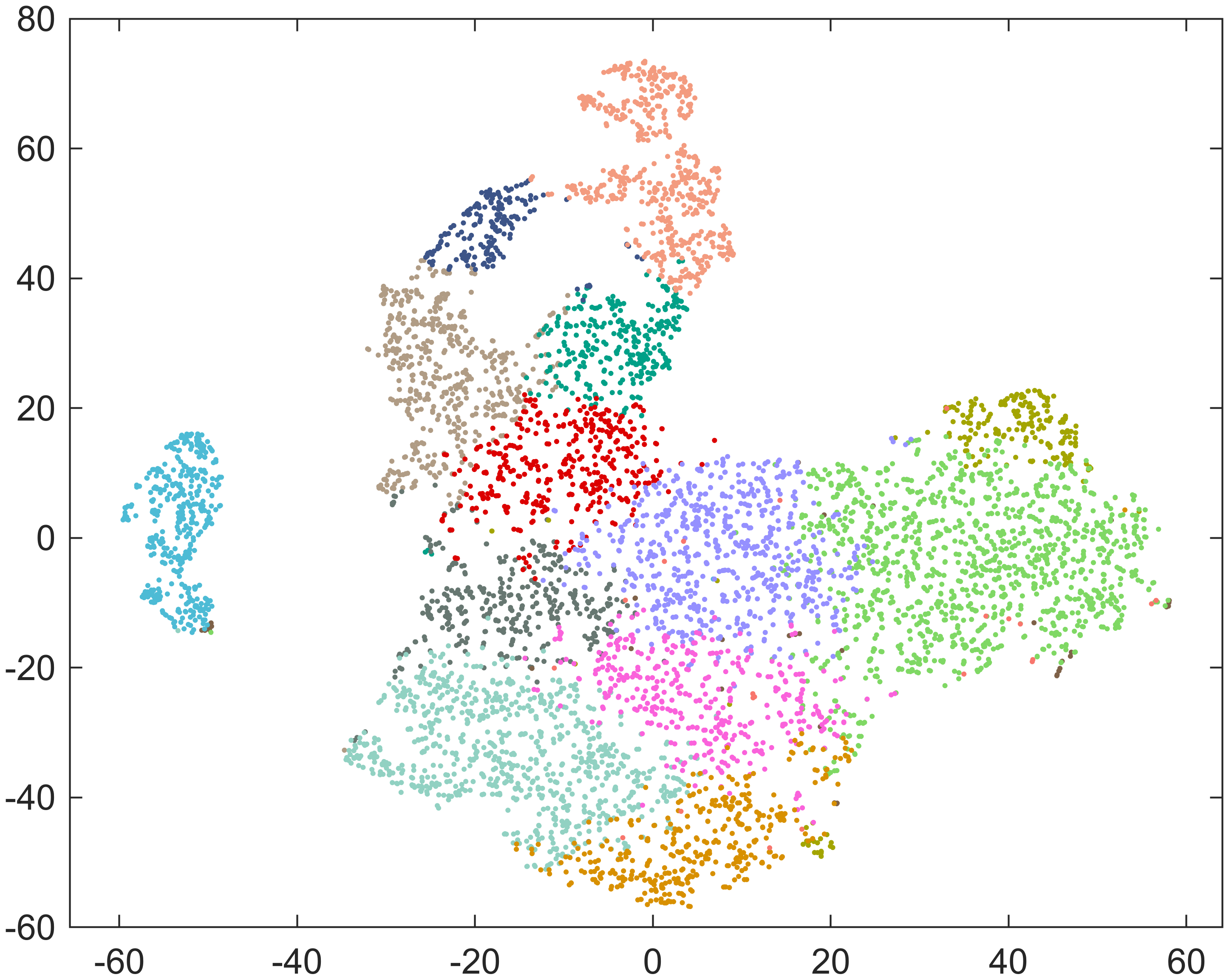}}
    \subfigure[]{\includegraphics[width=0.3\textwidth]{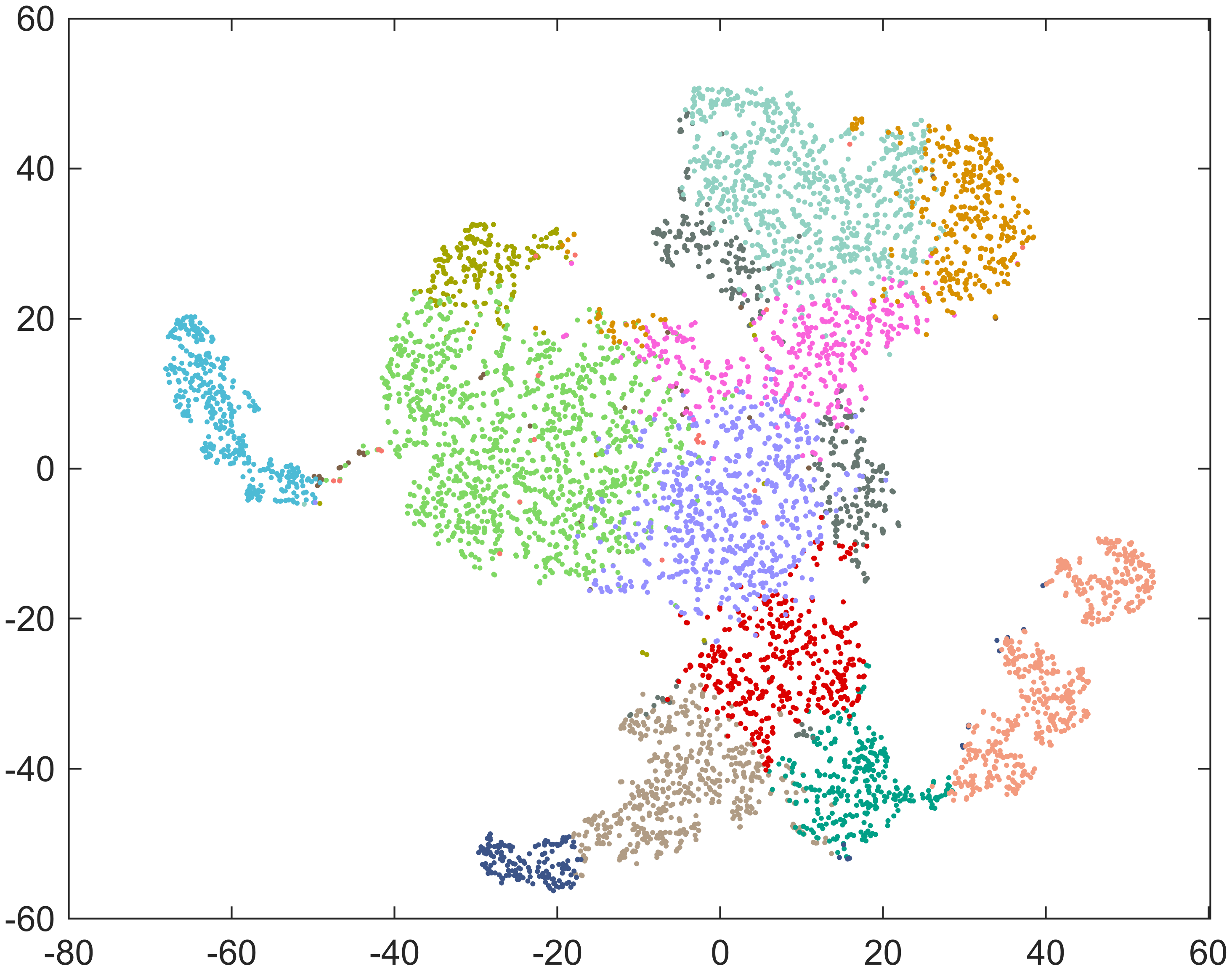}}
    \subfigure[]{\includegraphics[width=0.3\textwidth]{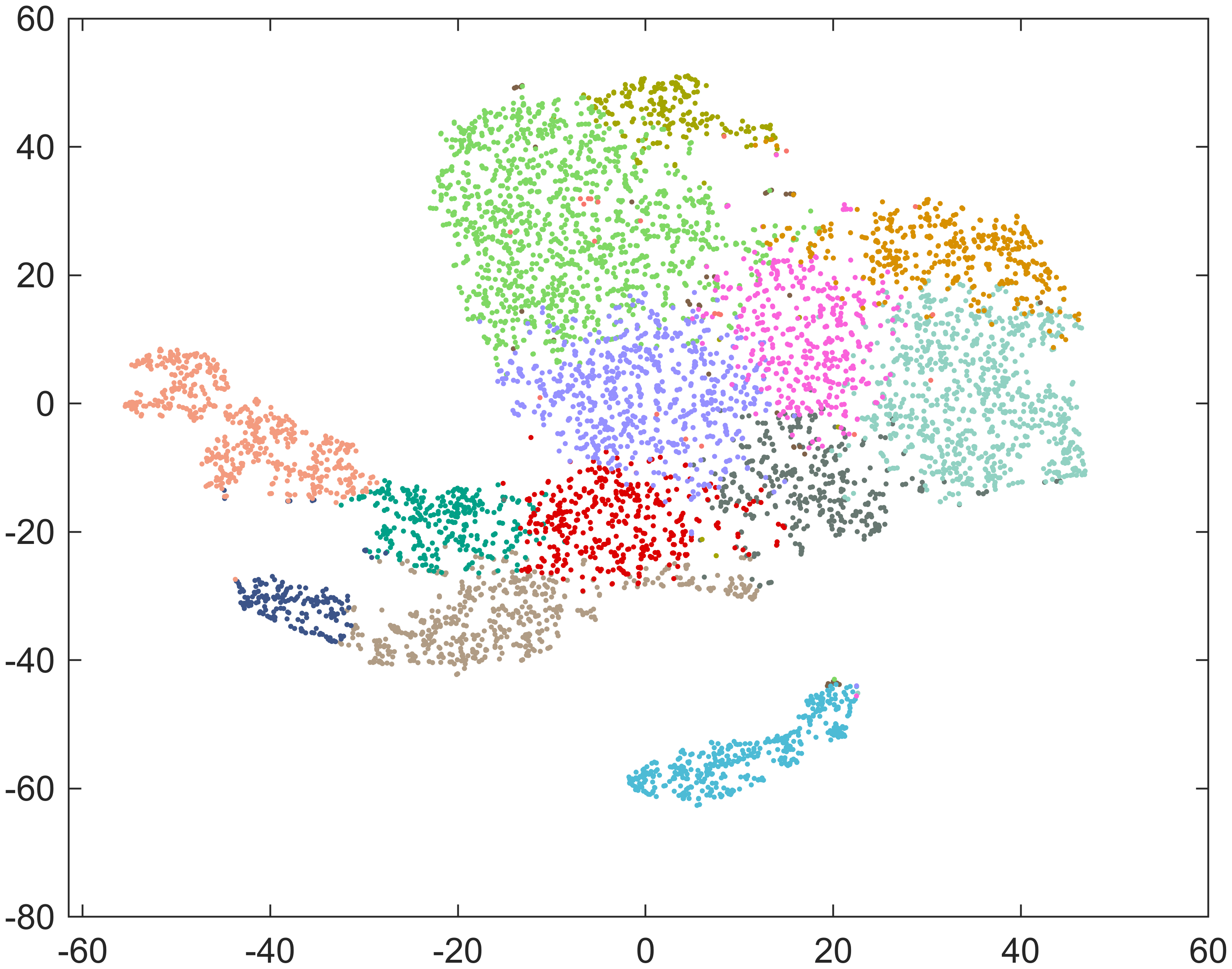}}
    \subfigure[]{\includegraphics[width=0.3\textwidth]{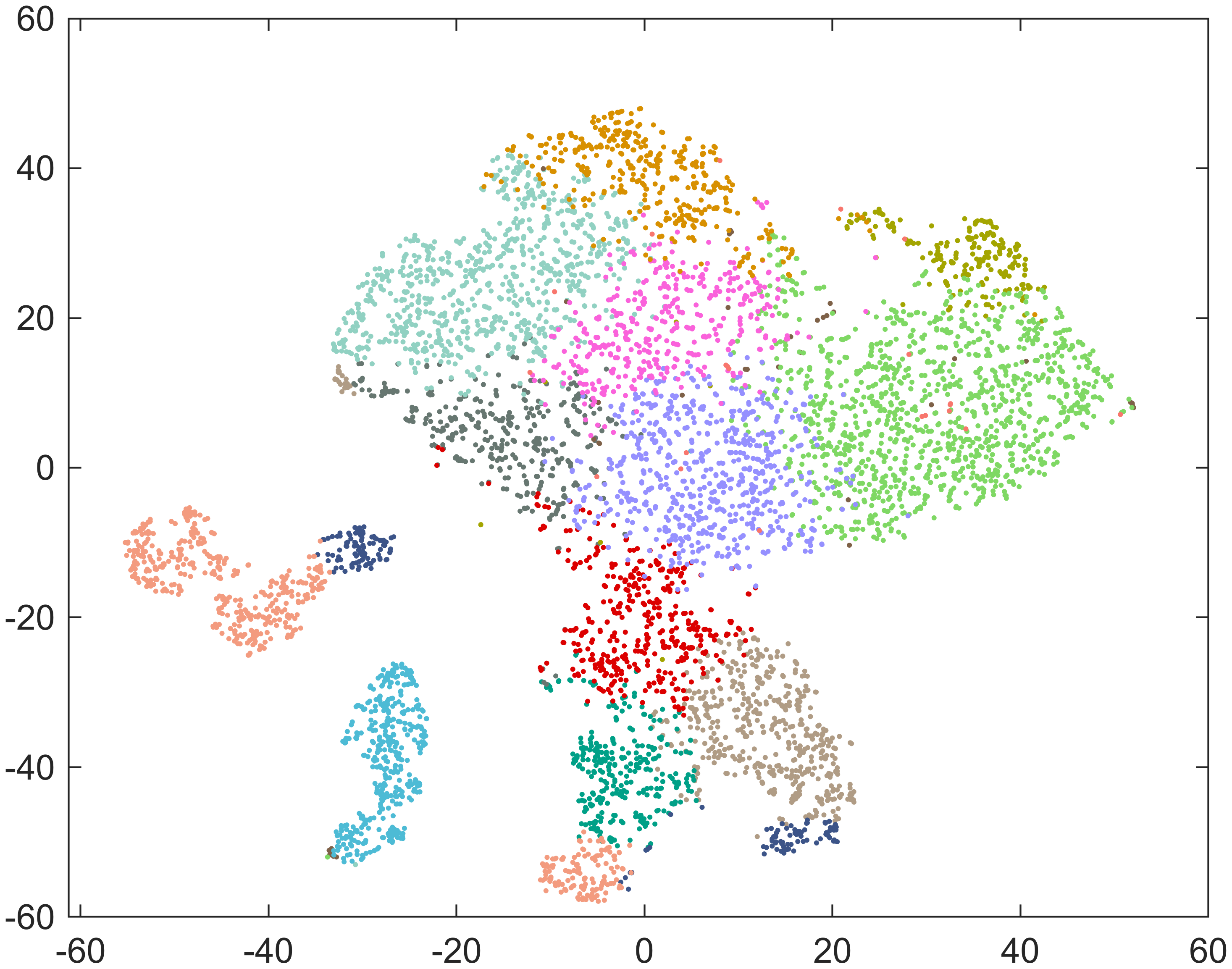}}
    \subfigure[]{\includegraphics[width=0.3\textwidth]{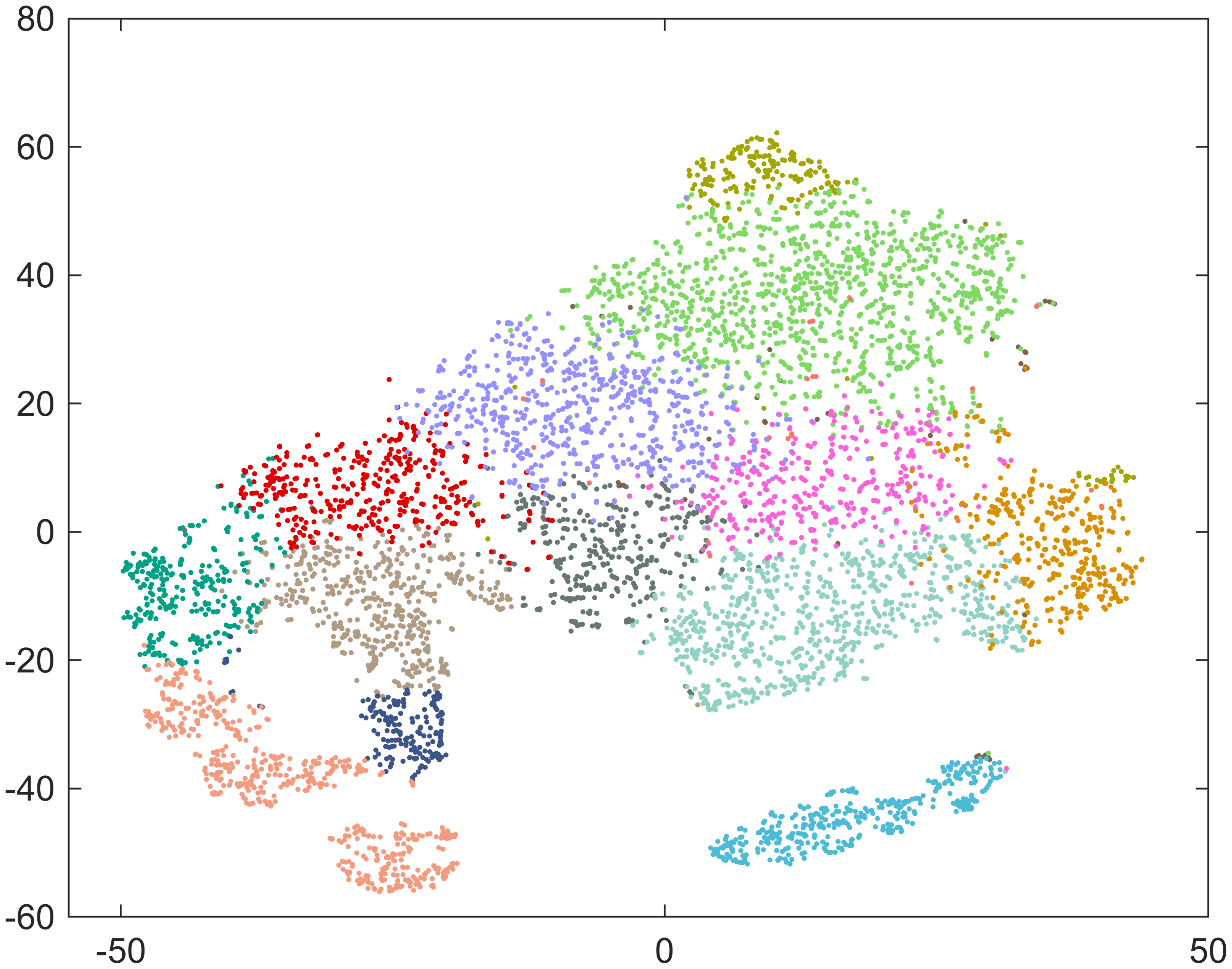}}
    \caption{Scatter plots for the two-dimensional t-distributed stochastic neighbour embedding of the result from the principal nested spheres with $d=1,\dots,12$ in (a)--(l).}
\end{figure}

\bibliographystyle{unsrtnat}
\bibliography{_references}
\end{document}